\documentclass[10pt]{article}
\usepackage[utf8]{inputenc}
\usepackage[american,british]{babel}
\usepackage{amsmath,amsthm,amsfonts,amssymb,amscd,dsfont}
\usepackage{mathrsfs}
\usepackage{mathtools}
\usepackage{stmaryrd}
\usepackage{physics}
\usepackage{subcaption}
\usepackage{graphicx}
\usepackage{hyperref}
\hypersetup{
    linkcolor=blue,
    citecolor=red,
    hypertexnames=false
    }    
\usepackage[capitalize]{cleveref}  
\Crefname{equation}{}{}

\usepackage{enumitem}
\usepackage[table]{xcolor}
    
\usepackage{fullpage}
\usepackage{lastpage}  

\usepackage{xfrac}

\usepackage{todonotes}

\usepackage{tikz}
\usetikzlibrary{decorations.markings}
\usetikzlibrary{positioning}

\numberwithin{equation}{section}

\newtheoremstyle{break}%
  {}{}%
  {\itshape}{}%
  {\bfseries}{}
  {\newline}{}

\theoremstyle{break}

\newtheorem{theorem}{Theorem}[section]

\newtheorem{corollary}[theorem]{Corollary}
\newtheorem{lemma}[theorem]{Lemma}
\newtheorem{proposition}[theorem]{Proposition}

\newtheorem{maintheorem}{Theorem}

\crefname{maintheorem}{theorem}{theorems}

\theoremstyle{plain}
\newtheorem{remark}[theorem]{Remark}

\newcommand{\V}[1]{\mathbf{#1}}
\newcommand{\Cal}[1]{\mathcal{#1}}
\newcommand{\Bold}[1]{\boldsymbol #1}

\newcommand{\N}{\mathbb{N}}

\newcommand{\Z}{\mathbb{Z}}
\newcommand{\R}{\mathbb{R}}
\newcommand{\C}{\mathbb{C}}
\newcommand{\E}{\mathbb{E}}

\renewcommand{\S}{\mathbb{S}}

\renewcommand{\P}{\mathbb{P}}

\newcommand{\Kcal}{\mathcal{K}}
\newcommand{\Acal}{\mathcal{A}}
\newcommand{\Dcal}{\mathcal{D}}

\newcommand{\bigO}{\mathcal{O}}

\newcommand{\Scal}{\mathcal{S}}
\newcommand{\Ccal}{\mathcal{C}}

\newcommand{\Mcal}{\mathcal{M}}

\newcommand{\Zcal}{\mathcal{Z}}
\newcommand{\Rcal}{\mathcal{R}}

\newcommand{\xx}{\V{x}}

\newcommand{\zz}{\V{z}}

\newcommand{\Si}{\boldsymbol\sigma}

\newcommand{\eps}{\epsilon}

\newcommand{\vp}{\varphi}
\newcommand{\tilV}{\tilde{V}}

\newcommand{\df}{\operatorname{d}\!}
\newcommand{\dx}{\df x}
\newcommand{\dy}{\df y}
\newcommand{\dz}{\df z}
\newcommand{\du}{\df u}

\newcommand{\ds}{\df s}
\newcommand{\dr}{\df r}

\newcommand{\partd}[1]{\dfrac{\partial}{\partial #1}}

\newcommand{\der}[1]{\frac{d}{d #1}}
\newcommand{\dder}[2]{\frac{d^#2}{d #1^#2}}

\renewcommand{\d}{\, \mathrm{d}}

\DeclareMathOperator{\diam}{\mathrm{diam}}

\DeclareMathOperator{\supp}{\mathrm{supp}}

\DeclareMathOperator*{\esssup}{ess\,sup}

\newcommand{\cbr}[1]{\left\lbrace #1 \right\rbrace}
\newcommand{\rbr}[1]{\left( #1 \right)}
\newcommand{\abr}[1]{\left[ #1 \right]}

\newcommand\cbo[1]{\mathopen{\left\{\rule{0pt}{#1}\right.}}
\newcommand\cbc[1]{\mathclose{\left\}\rule{0pt}{#1}\right.}}

\renewcommand{\abs}[1]{\left | #1 \right |} 

\renewcommand{\norm}[1]{\left\lVert#1\right\rVert}

\newcommand{\1}{\mathbf{1}}
\newcommand{\half} {\frac{1} {2} }

\newcommand{\sG}{\mathrm{sG}}

\begin{document}
\title{The massless sine-Gordon model with logarithmic correlations in arbitrary dimension}
\author{Sami Vihko$^1$}

\date{$^1$University of Helsinki, Department of Mathematics and Statistics, P.O. Box 68, FIN-00014 University of Helsinki, Finland \\[2ex]
    \today
}

\maketitle

\tableofcontents

\newpage

\begin{abstract}
In this article, we study the sine-Gordon model with logarithmic correlations in arbitrary dimension $d \ge 1$. The model is defined as a Euclidean field theory whose interaction term in the classical action is
\[
S_{\mathrm{int}}(\varphi) := 2z \int_{\Lambda} \cos(\sqrt{\beta}\,\varphi)\,\dx,
\]
where $\varphi$ is a logarithmically correlated Gaussian field on a compact domain $\Lambda \subset \mathbb{R}^d$ that coincides with the Gaussian free field when $d=2$, and where $z \in \mathbb{R}$ and $\beta>0$. We treat both the massive and massless cases, with the main emphasis on the massless regime. We construct the field with no boundary conditions, together with its charge and gradient correlation functions, and we show that the partition function is renormalizable. The results are non-perturbative, holding for all $z \in \mathbb{R}$. If $d=1$, our results are valid for the full subcritical range $\beta\in (0,4\pi)$ and if $d\geq 2$, the results are valid for $\beta\in (0,(d+1)2\pi)$. Our analysis is carried out directly in the continuum: we perform renormalization via a scale decomposition and control the partition function by controlling the renormalized potential that arises in this procedure, using the iterated Mayer expansion of Brydges and Kennedy~\cite{BrKe87a}.
\end{abstract}

\begin{section}{Introduction}
\label{sec:introduction}
In two dimensions, the canonical sine-Gordon model is formally defined by the measure
\begin{equation}
\label{eq:formal definition of sine-Gordon measure}
\Zcal^{-1}\exp(2z\int_{\R^2}\cos(\sqrt{\beta}\vp(x))\dx)\nu(\d\vp),
\end{equation}
where $\Zcal$ is a normalizing constant called the partition function, the reference measure $\nu$ is the law of the two-dimensional Gaussian free field (GFF), either massive or massless, and $\beta\in (0,8\pi)$ and $z\in\R$ are the fundamental parameters of the theory. If we study the model on a bounded domain, we may also consider its restriction to the boundary. In this sense, the one-dimensional model can be thought of as a \emph{boundary model} of the two-dimensional case. For simplicity, we consider only the case in which the boundary has a non-empty intersection with the real line, and we confine our constructions to this intersection.

In this paper, we consider only interactions restricted to compact domains $\Lambda\subset \R^d$, rather than to the whole space. We renormalize the model on $\Lambda$ in the UV-regime and do not attempt to take the limit $\Lambda\to\R^d$. However, our massive reference field is defined on the whole space, so our approach does not specify any boundary conditions for the model.

Since the GFF is logarithmically correlated in two dimensions, the sine-Gordon model is naturally logarithmic in the settings above. However, if $d\geq 3$, the canonical model is no longer logarithmically correlated, and to obtain a logarithmic model we will replace the two-dimensional GFF with another logarithmically correlated reference field. This logarithmic model no longer corresponds to any canonical model of quantum field theory (QFT), but it can be thought of as a representation of a $d$-dimensional log-gas; see \cite[Section 1.2]{LaRhVa23a}. Thus, the model is also interesting from the point of view of statistical mechanics, and we view it as a statistical rather than a quantum field theory. The main physical examples of log-gases are the two-dimensional Coulomb and Yukawa gases. The equivalence is at the level of partition functions if we take the grand canonical partition function on the Coulomb or Yukawa side. Furthermore, the parameter $\beta$ has the physical meaning of an inverse temperature, when viewed through this equivalence. Thus, the correspondence maps the high-temperature regime of the gases to the small-$\beta$ regime of the sine-Gordon model, and vice versa. Other examples also arise mainly in dimensions $d=1$ or $d=2$; see \cite{Fo10a,Se19a} and references therein. We note that the correspondence of the Coulomb/Yukawa gas to the canonical sine-Gordon model (with the $d$-dimensional GFF as the free field) holds in arbitrary dimension.

In dimensions $d=1$ or $2$, the log-gases exhibit a countable sequence of collapse points $\beta=\beta_N:=(1-1/N)4d\pi$, where $N$ is even. Physically, these correspond to particles clumping into larger and larger multipoles, and the final transition point $\beta_\infty=4d\pi$ is analogous to the Kosterlitz–Thouless (KT) \cite{KoTh73a,FrSp81a,DiHu91a,Fa12a} transition point in $d=2$, where the gas can be thought of as collapsing onto itself. This multipole picture was first proposed in \cite{GaNi85a}
and recently established rigorously for the two‑dimensional Coulomb gas in the canonical ensemble \cite{BoSe25a}. On the sine-Gordon side, the dipole point $\beta_2=2d\pi$ is where the model stops being a finite field theory. In the regime $\beta\in (d,2d)2\pi$, the model is super-renormalizable, but as $\beta$ crosses the thresholds $\beta_N$, more and more counterterms are needed to renormalize the partition function. For $\beta\in [\beta_N,\beta_{N+2})$, $N$ counterterms are needed. At the final transition point $\beta=4d\pi$, the model becomes strictly renormalizable, and for $\beta>4d\pi$ the model is non-renormalizable or trivial in the RG-sense. In two dimensions, the point $\beta=8\pi$ has been analyzed in \cite{NiPe89a}. The non-renormalizability for $\beta>8\pi$, or the RG-triviality, has not been rigorously established mathematically, as far as the author is aware. However, in the physics literature, Coleman, for example, has established that for $\beta>8\pi$ the theory has no ground state and is therefore physically nonsensical \cite{Co75a}. With a UV-cutoff, the effective theory has been shown to be asymptotically free in the IR-regime \cite{DiHu91a,DiHu00a}.

 Mathematically, it is an interesting question whether such phenomenology persists for log-gases if $d\geq 3$. This is suggested in the introduction of \cite{LaRhVa23a}. However, our analysis shows that, at least for our canonical choice of reference field on the sine-Gordon side, something already happens at the point $\beta=(d+1)2\pi< \beta_4=3d\pi$ if $d\geq 3$; see also \Cref{rem: renormalizability in higher dimensions}.

The two-dimensional sine-Gordon model is also interesting from other points of view. It is an example of a non-conformal integrable field theory, where integrability means that its correlation functions can be explicitly computed. Physicists have many predictions for these correlation functions; see for example \cite{Za79a,BeLe94a,LuZa97a,FaFrLu99a} and the references therein. Besides the correlation functions, much about the model has been conjectured by physicists; see for example the book \cite[Part 4]{Mu10a} and references therein. However, a lot less has been mathematically proven. The classical sine-Gordon equation is also integrable; see for example \cite{FaTa87a}. The model also exhibits a phenomenon called bosonization. Bosonization means that there is a duality or a canonical mapping between bosonic and fermionic field theories. For massless sine-Gordon, this is known as the Coleman correspondence \cite{Co75a}, and the corresponding fermionic model is known as the massive Thirring model, see \cite{Di98a,BeFaMa09a,BaWe24a} and references therein. This correspondence is at the level of correlation functions. The sine-Gordon model also has connections, or parallel features, to imaginary multiplicative chaos, the Ising model, and imaginary Liouville CFT; see for example \cite{JuSaWe20a,BaMaWe25a,PaViWe25a,GuCoKu25a} and references therein.

Our analysis follows the construction of the model in two dimensions with $\beta\in (0,6\pi)$ given in \cite{BaWe24a}. Thus, it is natural to consider a fixed canonical reference field $\vp$, which recovers the model treated in \cite{BaWe24a} when $d=2$. The canonical log-correlated reference field we introduce is the field whose covariance operator is given by the (rescaled) Green's function of the fractional massive Laplacian
\begin{equation}
K^m(x,y):=[\lambda_d(-\Delta+m^2)^{\frac{d}{2}}]^{-1}(x,y)=\int_0^\infty C_s^m(x,y)\ds,
\end{equation}
where $C_\cdot^m \colon [0,\infty)\times\R^d\times\R^d\to \R$ denotes the massive kernel
\begin{equation}
C_s^m(x,y):=e^{-m^2s}C_s(x,y):=\frac{1}{4\pi s}e^{-m^2s-\frac{|x-y|^2}{4s}}.
\end{equation}
The rescaling constant is $\lambda_d = (2^{\,2-d}\,\pi^{\,1-\frac{d}{2}})/\Gamma\!\left(\frac{d}{2}\right)$, but this is unimportant for our analysis; we adopt this rescaling merely to obtain, in arbitrary dimension, kernels of the same form as those in \cite{BaWe24a} for $d=2$. For $d=2$, the kernels are simply the heat kernels. We also introduce the notation $C_s$ for the massless kernel, as it will be needed later. Note that these kernels are themselves covariance functions. One may also write $K^m(x,y)=\frac{1}{2\pi}K_0(m|x-y|)$, where $K_0$ denotes the modified Bessel function of order zero. The well-known asymptotic formula $K_0(x)=-\gamma-\log(\half x)+\bigO(x)$ shows that the covariance indeed has a logarithmic singularity $K^m(x,y)\sim-\frac{1}{2\pi}\log(m|x-y|)$ both at the diagonal of the spatial variables and as $m\to 0$. 

\begin{remark}
In the canonical model, for $d>2$ the reference field would be the $d$-dimensional massive GFF with covariance given by the Green's function of the ordinary massive Laplacian
\begin{equation}
\begin{split}
(-\Delta+m^2)^{-1}(x,y):=\int_0^\infty H_s^m(x,y)\ds=\frac{1}{(2\pi)^{\frac{d}{2}}}\rbr{\frac{m}{|x-y|}}^{\frac{d}{2}-1}K_{\frac{d}{2}-1}(m|x-y|)
\end{split}
\end{equation}
where $H_t^m(x,y):=(4\pi t)^{-d/2}e^{-m^2t-|x-y|^2/(4t)}$
is the massive heat kernel and $K_{\nu}$ again denotes the modified Bessel function of order $\nu$.

This covariance has a power-law singularity $(-\Delta+m^2)^{-1}(x,y)\sim C_d|x-y|^{-d+2}$ rather than a logarithmic one. From the Coulomb point of view, $d=3$ would be the physically natural dimension, but unfortunately our methods do not work for $d\geq 3$ in this case.
\end{remark}

\begin{remark}
From now on, we will abuse terminology and refer to the reference field as the free field of the theory or model, or simply as the free field.
\end{remark}

We adopt the following regularization. Consider a Gaussian field $\vp_{m,\eps}$ with covariance function $K_\eps^m\colon \R^d\times\R^d\to\R$ given by
\begin{equation}
\label{eq:def of Keps}
K_\eps^m(x,y):=\E[\vp_{m,\eps}(x)\vp_{m,\eps}(y)]\equiv\E_{m,\eps}[\vp(x)\vp(y)]:=\int_{\eps^2}^\infty C_s^m(x,y)\ds.
\end{equation}
We have introduced two equivalent notations: either the dependence on the parameters $m,\eps$ is attached to the field itself, or it is encoded in the expectation symbol. In the latter case, the subscripts specify the probability measure, and $\vp$ is then merely a placeholder for the field with this law. A similar convention will be used for the associated measures and the truncated expectation symbol $\E^T$ to be introduced in the subsection below. When expectations involve more than one type of field, we will explicitly indicate with respect to which field the expectation is taken. A Kolmogorov–Chentsov-type argument (see Proposition B.2 in the appendix of \cite{LaRhVa15a}) shows that this field is smooth; that is, the law $\nu_{m,\eps}$ of $\vp_{m,\eps}$ is a Gaussian measure supported on $C^\infty(\R^d)$. Moreover, this measure is symmetric, translation invariant, and rotation invariant.

The regularized sine-Gordon measure is then defined by
\begin{equation}
\nu_{\sG(\beta,z,\Lambda|m,\eps)}(\d\vp):=\frac{1}{\Zcal_0(\beta,z,\Lambda|m,\eps)}\exp\!\Bigl(2z\int_\Lambda\eps^{-\frac{\beta}{4\pi}}\cos(\sqrt{\beta}\vp(x))\dx\Bigr)\nu_{m,\eps}(\d\vp),
\end{equation}
and the corresponding expectation will be denoted by $\E_{\sG(\beta,z,\Lambda|m,\eps)}$. In the expression for the bare (or non-renormalized) regularized partition function,
\begin{equation}
\Zcal_0(\beta,z,\Lambda|m,\eps):=\E_{m,\eps}\abr{e^{2z\int_\Lambda\eps^{-\frac{\beta}{4\pi}}\cos(\sqrt{\beta}\vp(x))\dx}},
\end{equation}
we have omitted the $\sG$ subscript for brevity. Since a renormalized partition function $\Zcal_R$ will also appear later, we shall simply refer to $\Zcal_0$ as the partition function.

The charge and gradient correlation functions (and their mixtures) that we study are correlations of the objects $\eps^{-\frac{\beta}{4\pi}}e^{i\sigma\sqrt{\beta}\vp(x)}$ and $D\vp(x)$, where $\sigma\in\{-1,1\}$ and $D$ denotes a first-order constant-coefficient linear partial differential operator (or simply the derivative when $d=1$). More precisely, we consider smeared versions of the truncated correlations of these quantities; see the next subsection for the precise definitions.

\begin{remark}
\label{rem: eps to power of beta over four pi}
Note that the factor $\eps^{-\frac{\beta}{4\pi}}$ is just Wick ordering and an additional multiplicative mass renormalization. That is,
\begin{equation}
\eps^{-\frac{\beta}{4\pi}}e^{i\sqrt{\beta}\sigma\vp_{m,\eps}(x)}=:e^{i\sqrt{\beta}\sigma\vp_{m,\eps}(x)}:\Delta(m),
\end{equation}
and the cosine is a sum of such terms. The factor $\Delta(m)$ can be derived explicitly, but we do not need it. Alternatively, one can think of $\eps^{-\frac{\beta}{4\pi}}$ simply as the way $1/\E_{m,\eps}[e^{i\sqrt{\beta}\vp(x)}]$ diverges as $\eps\to 0$.
\end{remark}

\begin{subsection}{Main results}
\label{sec:main results}
To make sense of the $\eps\to 0$ and $m,\eps\to 0$ limits of the objects above beyond the region $\beta\in (0,2d\pi)$, we need to specify the renormalization scheme we will use.

First, we introduce some notation and terminology. For a field $\vp\colon\R^d\to\R$, a map $F\colon \R\to \C$ (such that the integral below exists) and a suitable test function $f$, we denote
\begin{equation}
\label{eq:def integration against a test function notation}
F(\vp)(f):=\int F(\vp(x))f(x)\dx,
\end{equation}
where the integration is always over the whole space $\R^d$. In the case of functions $f\colon \R^d\times\{-1,1\}\to \C$, we allow maps $F\colon \R\times \{-1,1\}\to \C$, denote $F(x,\sigma)\equiv F_\sigma(x)$ and define
\begin{equation}
F(\vp)(f):=\sum_{\sigma\in\{-1,1\}}\int F_\sigma(\vp(x))f(x,\sigma)\dx.
\end{equation}
Note that if $F$ is the identity function, we simply have $\vp(f)=\int\vp(x)f(x)\dx$. In addition, we also use similar notation for first order differential operators $D$ acting on the field. In this case, we interpret it in the distributional sense unless otherwise specified. Thus, for a first-order (constant-coefficient) differential operator $D$ we define
\begin{equation}
(D\vp)(g):=-\int \vp(x)(D g)(x)\dx
\end{equation}
for test functions $g\in C_c^\infty(\R^d)$. 

By smeared correlation functions, or moments of objects $F_i(\vp(\cdot))$, we mean the expectations
\begin{equation}
\E\abr{\prod_{i=1}^nF_i(\vp)(f_i)}:=\E\abr{\prod_{i=1}^n\int F_i(\vp(x_i))f_i(x_i)\dx_i}
\end{equation}
or analogously if $\sigma$ dependence is present. Here the expectation is taken with respect to the law of the field $\vp$.

The term truncated correlation functions refers to cumulants; see Appendix \ref{sec:Def and prop of cumulants} for definitions and crucial properties. We will use the following notation for them. Let $\{X_i\}_{i\in I}$ be a collection of generic random variables and let $\kappa(\cdot)$ denote their joint cumulants. For $J:=\{i_1,i_2,\dots,i_j\}\subset I$, we denote
\begin{equation}
\E^T[X_j\mid j\in J]:=\kappa(X_{i_1};X_{i_2};\dots;X_{i_j}).
\end{equation}
Note that this should not be confused with the notation for conditional expectations, which do not appear in this article. For two different collections $\{X_i\}_{i\in I}$ and $\{Y_j\}_{j\in J}$ that need to be distinguished, we use the formal product notation
\begin{equation}
\E^T\abr{\prod_{i\in \tilde{I}}X_i\prod_{j\in \tilde{J}}Y_j}:=\kappa(X_{i_1};\dots;X_{i_n};Y_{j_1};\dots;Y_{j_k})
\end{equation}
for arbitrary finite sets $\tilde{I}=\{i_1,i_2,\dots,i_n\}\subset I$ and $\tilde{J}=\{j_1,j_2,\dots,j_k\}\subset J$. An analogous notation may also be used if one of the products is replaced by a power of a single random variable.

We are interested in the $m,\eps\to 0$ limits of the correlation functions
\begin{equation}
\Ccal_{\sG(\beta,z,\Lambda|m,\eps)}^{n,k}(\V{f},\V{g}):=\E_{\sG(\beta,z,\Lambda|m,\eps)}^T\abr{\prod_{j=1}^n\eps^{-\frac{\beta}{4\pi}}e^{i\sqrt{\beta}\sigma_j\vp}(f_j)\prod_{l=1}^k[D_l\vp](g_l)},
\end{equation}
where $f_j\in L_c^\infty(\R^d\times \{-1,1\})$ (compactly supported, essentially bounded functions), $j=1,2,\dots,n$, and $g_l\in C_c^\infty(\R^d)$, $l=1,2,\dots,k$. We choose the differential operators $D_l$ to be partial derivatives with respect to arbitrary components of $x\in\R^d$, or simply the derivative when $d=1$.

We call the limits of the above correlation functions the smeared and truncated mixed charge and gradient correlation functions of our sine-Gordon model. Possible further restrictions or conditions on the test functions will be specified in the theorems below. We also note that the linear constant‑coefficient partial differential operators mentioned above can be constructed from the partial derivatives. The natural operators for $d=2$ are the holomorphic or Wirtinger derivatives $\half(\partial_x\pm i\partial_y)$. Analogous correlation functions of the free field will play a crucial role in the analysis, and these are discussed in Appendix \ref{sec:Correlation functions of the reference field.}.

Next, we define the renormalized partition function and the setup in which most of our analysis will take place. First, let $\beta\in (0,4d\pi)$ and let $N\equiv N(\beta)$ be such that $\beta\in[\beta_N,\beta_{N+2})$, where, as mentioned in the introduction, $\beta_N:=(1-1/N)4d\pi$. For $k\in \N$ with $k\leq N/2$ and $\eta\in L_c^\infty(\R^d\times \{-1,1\})$ put
\begin{equation}
\label{eq:definition of the multiplicative counter terms}
Z_k(\eta):=\exp(-\frac{1}{(2k)!}\sum_{\Si\in\{-1,1\}^{2k}}\1_{\{\sum_{i=1}^{2k}\sigma_i=0\}}\int_{\R^{2kd}}\rbr{\prod_{i=1}^{2k}\eta(x_i,\sigma_i)}\E_{m,\eps}^T\abr{\eps^{-\frac{\beta}{4\pi}}e^{i\sqrt{\beta}\sigma_j\vp(x_j)}\,\big|\, j\in[2k]}\d\xx),
\end{equation}
where $[2k]:=\{1,2,\dots,2k\}$ is the discrete interval from $1$ to $2k$, $\Si:=(\sigma_1,\sigma_2,\dots,\sigma_{2k})$, and $\d\xx:=\dx_1\dx_2\dots\dx_{2k}$, with $\dx_j$ denoting the $d$‑dimensional Lebesgue measure.

Then put
\begin{equation}
\label{eq:definition of the renormalized partition function of the sG model}
\Zcal_{R}(\beta,z,\Lambda|m,\eps):=\Zcal_0(\beta,z,\Lambda|m,\eps)\prod_{k=1}^{N/2}Z_k(-z\1_\Lambda).
\end{equation}
The notation above is chosen so that we may carry out the general analysis with arbitrary $\eta\in L_c^\infty(\R^d\times\{-1,1\})$. The generalized partition function introduced at the beginning of \Cref{sec: the renormalized potential} will then depend on $\eta$ instead of $z$ and $\Lambda$.
\begin{remark}
For $d=1$, we will always first fix $\beta\in (0,4\pi)$ and then let $N$ be as above. For $d\geq 2$, we will always have $N=0$ or $N=2$, since we will be working with $\beta\in (0,(d+1)2\pi)$. 
\end{remark}

Now we are ready to state the main results.

\begin{maintheorem}[Renormalizability of the partition function]
\label{thm:renormalizability of the partition function}
Fix $\beta\in (0,4\pi)$ if $d=1$ and $\beta\in (0,(d+1)2\pi)$ if $d\geq 2$, and let $N\equiv N(\beta)$ be such that $\beta\in [\beta_N,\beta_{N+2})$. Also assume that $z\in \R$, $m\in (0,\infty)$, and that $\Lambda\subset\R^d$ is compact. Then the limits
\begin{equation}
\begin{split}
\Zcal_R(\beta,z,\Lambda|m)&:=\lim_{\eps\to 0}\Zcal_R(\beta,z,\Lambda|m,\eps)
\\
\Zcal_R(\beta,z,\Lambda)&:=\lim_{m\to 0}\lim_{\eps\to 0}\Zcal_R(\beta,z,\Lambda|m,\eps)
\end{split}
\end{equation}
exist and are finite.

Furthermore, the maps $z\mapsto \Zcal_R(\beta,z,\Lambda)$ and $z\mapsto \Zcal_R(\beta,z,\Lambda|m)$ can be extended to entire functions on the complex plane $\C$, and they are positive for $z\in\R$. The map $z\mapsto \Zcal_R(\beta,z,\Lambda)$ is also an even function.
\end{maintheorem}

\begin{proof}
This is precisely the content of the first three statements of \Cref{thm:convergence of the partition function}, proved in \Cref{sec:convergence of the partition function}, for $\eta=-z\1_{\Lambda}$.
\end{proof}

\begin{maintheorem}[Existence of the correlation functions]
\label{thm:existence of the sine Gordon correlation functions}
Fix $\beta\in (0,4\pi)$ if $d=1$ and $\beta\in (0,(d+1)2\pi)$ if $d\geq 2$, and let $N\equiv N(\beta)$ be such that $\beta\in[\beta_N,\beta_{N+2})$. Let $z\in \R$, $n,k\in\Z_+$, $m\in (0,\infty)$, and let $\Lambda\subset \R^d$  be compact. Suppose $f_j\in L_c^\infty(\R^d\times \{-1,1\})$, $g_l\in C_c^\infty(\R^d)$, $j=1,2,\dots, n$, $l=1,2,\dots,k$. Then we have:

\begin{enumerate}
\item If either $n>N(\beta)$, $k\geq 1$, or $d\geq 2$, and, in the last case, $n=N(\beta)=2$ and $\supp(f_1)\cap\supp(f_2)=\emptyset$, then the limit
\begin{equation}
\begin{split}
\Ccal_{\sG(\beta,z,\Lambda)}^{n,k}(\V{f},\V{g})&\,\,\equiv\E_{\mathrm{sG}(\beta,z,\Lambda)}^T\abr{\prod_{j=1}^n:e^{i\sqrt{\beta}\sigma_j\vp}:(f_j)\prod_{l=1}^k[D_l\vp](g_l)}
\\
&:=\lim_{m\to 0}\lim_{\eps\to 0}\Ccal_{\sG(\beta,z,\Lambda|m,\eps)}^{n,k}(\V{f},\V{g})
\end{split}
\end{equation}
exists and is finite.

\item Assuming the conditions in part 1. above, $\Ccal_{\mathrm{sG}(\beta,z,\Lambda)}^{n,k}$ has an analytic extension (in $z$) to some neighborhood $ U\equiv U(\Lambda)$ of the real axis $\R$.
\item Finally, assuming the conditions of part 1. above, we have, for any $q\geq 0$,
\begin{equation}
\begin{split}
\dder{z}{q}\bigg|_{z=0}&\Ccal_{\mathrm{sG}(\beta,z,\Lambda)}^{n,k}(\V{f},\V{g})
\\
&=
\E^T\abr{\rbr{:e^{i\sqrt{\beta}\vp}:(\1_\Lambda)+:e^{-i\sqrt{\beta}\vp}:(\1_\Lambda)}^q\prod_{j=1}^n:e^{i\sqrt{\beta}\sigma_j\vp}:(f_j)\prod_{l=1}^k(D_l\vp)(g_l)}
\\
\end{split}
\end{equation}
where the expectation is now with respect to the law of the free field, and the RHS is defined by an analogous limit of the regularized free field, as in part 1.
\end{enumerate}

Furthermore, analogous statements hold for fixed $m$.
\end{maintheorem}

\begin{remark}
\label{rem:dot notations}
The $:(\cdot):$ notations above are formal analogues of Wick ordering adapted to this non-Gaussian setting, as used in \cite{BaWe24a}.
\end{remark}
\begin{remark}
If the conditions in part~1 above fail, we expect the associated correlation functions to diverge. We conjecture, however, that they converge after an additional multiplicative renormalization by a factor $g(\beta,\eps,n)$. Such a result would be analogous to the convergence of the one point charge correlation function in \cite[Theorem~3.1, part~(iv)]{BaWe24a}, and the cases $d\geq 2$ should follow readily from the proof given there. We leave the analysis of the case $d=1$ to future work.
\end{remark}
\begin{maintheorem}[Existence of the field]
\label{thm:existence of the field}
Fix $\beta \in (0,4\pi)$ if $d=1$ and $\beta\in (0,(d+1)2\pi)$ if $d\geq 2$, and let $N\equiv N(\beta)$ be such that $\beta\in [\beta_N,\beta_{N+2})$. Let $z\in\R$, $m\in (0,\infty)$, and let $\Lambda\subset\R^d$ be compact. Suppose $f\in C_c^\infty(\Lambda)\equiv \Dcal(\Lambda)$. Then the limit
\begin{equation}
\lim_{\eps\to 0}\E_{\sG(\beta,z,\Lambda|m,\eps)}\abr{e^{w\vp(f)}}
\end{equation}
exists and defines an entire function of $w$. If, in addition, $\int f=0$, then the limit
\begin{equation}
\lim_{m\to 0}\lim_{\eps\to 0}\E_{\sG(\beta,z,\Lambda|m,\eps)}\abr{e^{w\vp(f)}}
\end{equation}
also exists and defines an entire function of $w$ that is an even function of $z$.

Furthermore, for $w=i$ these are Fourier transforms or characteristic functions of probability measures on $\Dcal'(\Lambda)$ and $\sfrac{\Dcal'(\Lambda)}{\R}$, respectively.
\end{maintheorem}

\begin{remark}
We will prove the above theorems for $m\in (0,1)$, since other values can be obtained by rescaling space. Also, the main difficulties arise in the limit $m\to 0$, where our main interest also lies. Similarly, we will restrict $\eps \in (0,1)$ in our analysis.
\end{remark}

\subsection{Strategy of the proofs, main novelties and the structure of the rest of the article}
We closely follow the strategy used in \cite{BaWe24a} to prove the analogous results for $d=2$ and $\beta\in (0,6\pi)$ (which is exactly our parameter range $\beta\in (0,(d+1)2\pi)$ when $d=2$). We also make crucial use of the Onsager-type inequality introduced in \cite{LaRhVa23a} in our inductive arguments, especially in the case $d=1$.

In \Cref{sec: the renormalized potential}, we introduce a crucial tool, the renormalized potential. First, we use the scale decomposition $\vp_{m,\eps}\overset{d}{=}\vp_{m,(\eps^2,t)}+\vp_{m,\sqrt{t}}$, where the fields on the right-hand side are independent Gaussian fields with covariances $K_\eps^m-K_{\sqrt{t}}^m$ and $K_{\sqrt{t}}^m$, respectively; recall the notation from \Cref{eq:def of Keps}. We then integrate out the UV-rough part, that is, scales from $\eps^2$ to $t$, and introduce the so-called renormalized potential $V_t$, which has a series representation known as the iterated Mayer expansion discovered in \cite{BrKe87a}. This series is written in terms of recursively defined functions $\tilV_t^n\colon (\R^{d}\times\{-1,1\})^n\to \R$, which turn out to be nothing more than the charge cumulants of the field $\vp_{m,(\eps^2,t)}$. We control the functions $\tilV_t^n$ inductively using the recursion. The Mayer expansion needs to be renormalized for $\beta>2d\pi$. The counterterms used in the renormalized partition function produce terms involving $\tilV_\infty^n$, which are equal to the full free field charge cumulants. Thus, we also need to control the difference $\tilV_\infty^n-\tilV_t^n$ when $n$ is even and less than or equal to $N$. These yield geometric convergence of the iterated Mayer expansion for suitable values of $t$.

In \Cref{sec:Analysis of the partition function} we use the iterated Mayer expansion for the renormalized potential to control the renormalized partition function. Furthermore, we derive an everywhere convergent series representation, in powers of $z$, of the renormalized partition function. Using the estimates for $\tilV_t^n$ and $\tilV_\infty^n-\tilV_t^n$ obtained in \Cref{sec: the renormalized potential}, we prove that this power series also has a limit as we first take $\eps\to 0$ and then take $m\to 0$. This yields \Cref{thm:renormalizability of the partition function}.

In \Cref{sec:The analysis of the correlation functions and the existence of the field and t}, we prove \Cref{thm:existence of the sine Gordon correlation functions,thm:existence of the field}. In both proofs, after the crucial step of using the Cameron-Martin-Girsanov transformation, we write everything in terms of the renormalized partition function and the free field correlation functions. Controlling these allows us to prove the theorems.

The free field correlation functions are introduced and analyzed in Appendix \ref{sec:Correlation functions of the reference field.}. We control the charge cumulants of the free field by writing them as $\tilV_1^n+[\tilV_\infty^n-\tilV_1^n]$ and separately controlling the contribution from each term, using the estimates established in \Cref{sec: the renormalized potential}.

Below, we list the main novelties of this work:
\begin{itemize}
\item In $d=1$, we improve the results in \cite{LaRhVa23a} by studying the analytic properties of the partition function and proving the existence and analyticity of the mixed correlation functions. In particular, we treat the gradient correlation functions, which are not considered in \cite{LaRhVa23a} at all, and we work with smeared rather than pointwise correlation functions. We also explicitly control the $m\to 0$ limit. Since $N$ is arbitrary, the inductive arguments controlling the functions $\tilV_t^n$ must be generalized from the case $N=2$ in \cite{BaWe24a}. Here we use the Onsager-type inequality from \cite{LaRhVa23a} for terms with $n\leq N+1$. Moreover, the control of the terms $\tilV_\infty^n-\tilV_t^n$ becomes more difficult in this case, since they can no longer be computed explicitly and we do not have integrable majorants for them that are independent of the mass $m$.

\item In $d=2$, as mentioned, we reproduce results analogous to those in \cite{BaWe24a}. However, we also make the analysis of the renormalized potential more robust and easier to generalize. First of all, we avoid using the infinite-dimensional Polchinski equation in the proof that the renormalized potential coincides with the iterated Mayer expansion. The induction controlling the functions $\tilV_t^n$ is also made simpler: we do not need to treat the term $n=4$ separately, nor do we need to treat separately the cases where the recursion for higher-order terms contains the term $n=2$. Furthermore, all of our estimates are uniform over the charge configuration. We could control the terms $\tilV_t^2$ and $\tilV_t^3$ as in \cite{BaWe24a}, but we choose to use the Onsager-type inequality here as well, to remain consistent with our analysis in $d=1$. The analysis used here generalizes straightforwardly to $d\geq 3$, and is in fact carried out simultaneously for all $d\geq 2$.

\item To the best of our knowledge, the explicit construction is novel for $d\geq 3$. Most notably, we establish that the first collapse point $\beta_2$, that is, the point where the smeared charge two-point function of the free field diverges, scales linearly with the dimension. Our results also imply that the higher collapse points do not scale with the dimension. Indeed, as discussed in \Cref{rem: renormalizability in higher dimensions}, even the charge two-point function cannot be made locally integrable above $(d+2)2\pi<\beta_4$ by our methods.

The author thanks Sarah Jean-Meyer for pointing out (after publication of the preprint on arXiv) that she and M. Gubinelli had already noted in \cite[Remark 1.6]{GuMe26a} that their methods should readily generalize to arbitrary dimension with $\beta \in (0,(d+1)2\pi)$ for the massive model in the whole space. 
\end{itemize}

The guiding philosophy of this article is to use the case $d=1$ as a testing ground for the procedure introduced above, while allowing an arbitrary number of counterterms. Consequently, if the analysis in \Cref{sec: the renormalized potential} can be refined to cover the full regime $\beta \in (0,8\pi)$ in $d=2$, then analogous results would follow from our computations. Our analysis of the partition function and the correlation functions treats all dimensions equally up to the point where more refined estimates would be needed. Everything rests on the estimates in \Cref{sec: the renormalized potential}, in particular \Cref{prop:induction statement,cor:convergence of the Vinfty-Vt terms}. Thus, if we can prove analogues of these in the full range $\beta\in (0,8\pi)$ for $d=2$, everything else should follow readily. As they stand, they do not utilize possible cancellations arising from summing over $\Si\in\{-1,1\}^n$. These are difficult to propagate in the induction using the recursion \Cref{eq:def of V1,eq:def of Vn}. However, we can also represent the functions defined by such a recursion as an expansion over certain types of trees, where the $\Si$-summation might be more tractable. We leave this analysis for future work.
\end{subsection}

\begin{subsection}{Earlier results}
The mathematically rigorous study of the two-dimensional Euclidean sine-Gordon field theory began with Fr\"ohlich in 1976 \cite{Fr76a}, where he constructed the field theory in the finite regime $\beta<4\pi$. There, the only renormalization required is Wick ordering. He also investigated connections to one- and two-dimensional Coulomb and Yukawa gases. In the early 1980s, Benfatto, Gallavotti and Nicol\`o initiated a program to study the UV-behavior of the two-dimensional massive sine-Gordon model on the full range $(0,8\pi)$ \cite{BeGaNi82a,Ni83a,NiReSt86a}. These results use non-perturbative renormalization techniques, but they do not explicitly construct the model or its correlation functions. Instead, they prove the UV-stability of the problem by establishing upper and lower bounds for the partition function.

In the latter half of the 1980s, cluster expansion methods for the model were developed in \cite{Be85a,BrKe87a}. In particular, Brydges and Kennedy initiated the development of several tools that we also use in this article. Dimock and Hurd were the first to apply rigorous renormalization group methods to construct the two-dimensional sine-Gordon model \cite{DiHu93a,DiHu00a}; see also \cite{Pe25a} for a recent account of similar methods. In \cite{DiHu93a} they construct the model for the full range $(0,8\pi)$ and prove analyticity of the field correlations for sufficiently small $|z|$. Dimock and Hurd also observed that the UV- and IR-behavior of the model exhibit a form of duality when the threshold $\beta = 8\pi$ is crossed: the model is asymptotically free in the UV-regime with an IR-cutoff for $\beta < 8\pi$, and the situation is reversed for $\beta > 8\pi$.

The papers \cite{Di98a,BeFaMa09a} construct the model for $\beta < 16\pi/3$ and prove the Coleman correspondence at the free fermion point $\beta = 4\pi$ and, respectively, for the full range $\beta \in (0,16\pi/3)$. Both articles treat only the case of small $|z|$. All of the constructions of the model above work in finite volume for the massless case. However, Fr\"ohlich already constructed the free energy (or equivalently the pressure in the Coulomb gas picture) of the massless model in the finite regime ($\beta\in (0,4\pi)$) and infinite volume in \cite{Fr76a}.

We now arrive at the line of work to which this article belongs methodologically. In \cite{BaWe24a} Bauerschmidt and Webb construct the model and its charge and gradient correlation functions in finite volume for $\beta \in (0,6\pi)$, and for all $z\in\R$, using renormalization techniques originating in \cite{BrKe87a}. They also establish the Coleman correspondence at the free fermion point $\beta=4\pi$ and use this to extend the massless model (as a probability measure on $\Scal'(\R^2)$) to the whole space $\mathbb{R}^2$ at this point. They also construct the massless model as a probability measure on $\sfrac{\Dcal'(\R^2)}{\R}$ without the Coleman correspondence using correlation inequalities from \cite{FrPa78a} for the range $\beta\in (0,6\pi)$.  Relying on similar techniques and together with collaborators, they have also proved other important instances of bosonization at the free fermion point \cite{BaMaWe25a,PaViWe25a}.

For the two-dimensional model, the non-perturbative construction in the full range $\beta\in (0,8\pi)$ (even in finite volume) for both the massless and massive models, and the full plane construction for the massless model outside the free fermion point (even in the perturbative regime), are still elusive in rigorous mathematical literature.

In one dimension, we mention the construction by Lacoin, Rhodes and Vargas of the model for the full subcritical range $\beta \in (0,4\pi)$ in finite volume \cite{LaRhVa23a}. They consider a more general log-gas, which leads to a wider class of covariance kernels for the reference field in the sine-Gordon representation. As a consequence, they do not distinguish between the massive and massless models. They also construct the model and the pointwise charge correlation functions for arbitrary $z \in \mathbb{R}$, but they do not address the analyticity of these in $z$.

One may also consider various extensions or modifications of the two-dimensional model, as well as approaches based on techniques quite different from those employed in the present article, see for example \cite{BaBo21a,Ba22a,GuHaTa25a,GuMe26a,OhRoSo21a,KaPiWi91a,HaSh16a,BrCa26a}, and references therein. Finally, we stress that this list of earlier work is far from complete, as the mathematical literature on the sine-Gordon model is extensive.
\end{subsection}

\begin{subsection}{Notation}
In addition to the notation introduced thus far, we collect some further notation in this section.

For index sets we will use the following notation: $[n]:=\{1,2,\dots,n\}$ and $[[n,k]]:=\{n,n+1,\dots,k\}$ for $n,k\in\N$ with $k>n$. Arbitrary subsets of $[n]$ will be denoted by capital letters like $I$ or $J$, and we will use sets as indices for the spatial variables $x$ and the charges $\sigma\in\{-1,1\}$. That is, for $I:=\{i_1,i_2,\dots,i_j\}$, by $(\xx_I,\Si_I)$ we mean the vector $(x_{i_1},x_{i_2},\dots,x_{i_j},\sigma_{i_1},\sigma_{i_2},\dots,\sigma_{i_j})\in(\R^d)^j\times\{-1,1\}^j$. If $I=[n]$ we omit the subscripts so that $(\xx,\Si)\in (\R^d)^n\times\{-1,1\}^n$. If there is only one type of field ($\vp$) or test function $\eta$ involved in the computation, we will use the shorthand notations $\vp_j:=\vp(x_j)$ and $\eta_j:=\eta(x_j,\sigma_j)$.

In addition to the regularized covariance $K_\eps^m$ introduced in \Cref{eq:def of Keps}, we will also encounter the integrals
\begin{equation}
\begin{split}
\label{eq:def of Kst}
K_{s,t}^m(x,y)&:=\int_s^t C_r^m(x,y)\dr
\\
K_{s,t}(x,y)&:=\int_s^t C_r(x,y)\dr
\end{split}
\end{equation}
for any $t>s>0$. Furthermore, we will use the following shorthand notations
\begin{equation}
\label{eq:some notations for K and the heat kernel}
\begin{split}
\Kcal(x,x)&=:\Kcal(0)
\\
\Kcal(x_i,x_j)&=:\Kcal(i,j)
\\
\Ccal_s(x,x)&=:\Ccal_s(0)
\\
\Ccal_s(x_i,x_j)&=:\Ccal_s(i,j).
\end{split}
\end{equation}
Above, $\Kcal$ represents $K_{\cdot}^m$, $K_{\cdot,\cdot}^m$ or $K_{\cdot,\cdot}$, and $\Ccal_s$ represents $C_s^m$ or $C_s$.

We will denote constants in the estimates by $C$ or $B$, and these may have subscripts or variables ($C_\beta$ and $B(\beta)$ would both denote a constant depending on $\beta$). We allow the constants to change from (in)equality to (in)equality in deriving estimates or making computations, to avoid naming new constants that are unnecessary to distinguish and which we do not need in the future.
\end{subsection}

\end{section}

\section*{Acknowledgements}
The author wishes to thank Christian Webb and Eero Saksman for useful discussions during this work. 

\section*{Funding}
The author was supported by the Emil Aaltonen Foundation, the Academy of Finland through Grant 348452, the Academy of Finland CoE FiRST, and ERC grant
CONFSTAT funded by the European Union. Views and opinions expressed are however those
of the author only and do not necessarily reflect those of the European Union or ERC. Neither
the European Union nor ERC can be held responsible for them.

\begin{section}{The renormalized potential}
\label{sec: the renormalized potential}
We begin the analysis of the construction of our model by introducing an approach that was first used in \cite{BrKe87a} and has recently been employed by Bauerschmidt and Webb, and their collaborators in \cite{BaBo21a,BaHo22a,BaWe24a,BaMaWe25a,PaViWe25a}. First let $\eta\in L_c^\infty(\R^d\times \{-1,1\})$ and $\vp\in C(\R^d)$, and define
\begin{equation}
\label{eq:def of V0}
V_0(\beta,\eta,\vp|\eps):=\eps^{-\frac{\beta}{4\pi}}\sum_{\sigma\in\{-1,1\}}\int_{\R^d}\eta(x,\sigma)e^{i\sqrt{\beta}\sigma\vp(x)}\dx. 
\end{equation}
Next we define the generalized partition function
\begin{equation}
\label{eq:def of the generalized partition function}
\Cal{Z}_0(\beta,\eta| m,\eps):=\E_{m,\eps}\abr{e^{-V_0(\beta,\eta,\vp | \eps)}},
\end{equation}
which will play a crucial role in our analysis. Note that choosing $\eta=-z\1_\Lambda$ yields exactly the partition function $\Zcal_0(\beta,z,\Lambda|m,\eps)$ of our sine-Gordon model. In the proofs of the existence of the sine-Gordon correlation functions and of the field itself, we will choose $\eta$ to depend suitably on certain parameters. The correlation functions will then be obtained by taking logarithmic derivatives of this generalized partition function with respect to these parameters.

To analyze the generalized partition function, we decompose the regularized free field $\vp_{m,\eps}$ into two independent components living on different scales. More precisely, recalling the notation from \Cref{eq:def of Keps,eq:def of Kst}, we write $\vp_{m,(\eps^2,t)}$ for the Gaussian field with covariance $K_{\eps^2,t}^m$ and $\vp_{m,\sqrt{t}}$ for the Gaussian field with covariance $K_{\sqrt{t}}^m$. Since these are independent Gaussians and $K_\eps^m=K_{\eps^2,t}^m+K_{\sqrt{t}}^m$, we have $\vp_{m,\eps}\overset{d}{=}\vp_{m,(\eps^2,t)}+\vp_{m,\sqrt{t}}$. The field $\vp_{m,(\eps^2,t)}$ contains the UV-rough part (in the limit $\eps\to 0$) at small scales, while $\vp_{m,\sqrt{t}}$ is the UV-regular part.

Next, as is customary in renormalization theory, we integrate out the rough part and write
\begin{equation}
\label{eq:partition function in terms of th renormalized potential}
\Zcal_0(\beta,\eta| m,\eps)=\E_{m,\sqrt{t}}\abr{e^{-V_t(\beta,\eta,\vp |m, \eps)}},
\end{equation}  
where
\begin{equation}
\label{eq:def of the renormalized potential}
e^{-V_t(\beta,\eta,\vp_{m,\sqrt{t}}|m,\eps)}:=\E_{m,(\eps^2,t)}\abr{e^{-V_0(\beta,\eta,\vp_{m,\eps}|\eps)}}=\E_{m,(\eps^2,t)}\abr{e^{-V_0(\beta,\eta,\vp+\vp_{m,\sqrt{t}}|\eps)}},
\end{equation}
where we have written $\vp\equiv \vp_{m,(\eps^2,t)}$, and $\E_{m,(\eps^2,t)}$ denotes expectation with respect to the law of $\vp$. The quantity $V_t$ above can be defined for an arbitrary sufficiently regular field $\psi$ independent of $\vp_{m,(\eps^2,t)}$ (or deterministic) in place of $\vp_{m,\sqrt{t}}$ by the rightmost formula. Note that, since the test functions $\eta$ are complex-valued, a priori only the exponential $e^{-V_t}$ is well-defined. However, by looking at the map $z\mapsto e^{-V_0(\beta,z\eta,\vp+\psi|\eps)}$ (and its expectation with respect to the law of $\vp\equiv \vp_{m,(\eps^2,t)}$) we can prove that for small enough $\eta$ also
\begin{equation}
V_t(\beta,\eta,\psi|m,\eps):=\log(\E_{m,(\eps^2,t)}\abr{e^{-V_0(\beta,\eta,\vp+\psi|\eps)}})
\end{equation}
is well-defined; see the beginning of the proof of \Cref{thm:the equivalence of St and Vt} part 2 from \Cref{eq:first step in the proof of theorem 2.1} onwards for details. Alternatively, considering the equivalence with the series $S_t$ introduced in \Cref{eq:def of S} below, we can view $V_t$ as well-defined for any $\eta\in L_c^\infty(\R^d\times \{-1,1\})$ if $t$ is chosen small enough depending on $\eta$.

The quantity $V_t$ is called the renormalized potential in \cite{BaWe24a}. We will adhere to this terminology even though it is somewhat at odds with the rest of our usage of the word “renormalized,” since every other renormalized quantity involves explicit counterterms. Instead of studying the partition function $\mathcal{Z}_0$ directly, it is more convenient to study the renormalized potential $V_t$. Moreover, in order to analyze $V_t$, it turns out to be advantageous to consider the following series. For $\vp\in C(\R^d)$ put 
\begin{equation}
\label{eq:def of S}
S_t(\beta,\eta,\vp|m,\eps):=\sum_{n=1}^\infty\frac{1}{n!}\sum_{\Si\in\{-1,1\}^n}\int_{\R^{nd}}\rbr{\prod_{i=1}^n \eta_i}e^{i\sqrt{\beta}\sum_{j=1}^n\sigma_j\vp_j}\tilde{V}_t^n(\xx,\Si|m,\eps)\d\xx,
\end{equation}
where we have used the shorthand notation $\eta_i:=\eta(x_i,\sigma_i)$ and $\vp_i:=\vp(x_i)$, and we have omitted the dependence on $\beta$ of the functions $\tilV_t^n$. The functions $\tilde{V}_t^n$ satisfy the following recursion. For $n=1$,
\begin{equation}
\begin{split}
\label{eq:def of V1}
\tilde{V}_t^1(x,\sigma|m,\eps)\equiv\tilde{V}_t^1(m,\eps)&:=\exp\!\Big(-\frac{\beta}{2}\abr{K_{\eps^2,t}^m(0)+\frac{1}{4\pi}\log(\eps^2)}\Big)
\\
&=\exp\!\Big(-\frac{\beta}{2}\abr{\int_{\eps^2}^t C_s^m(0)\ds+\frac{1}{4\pi}\log(\eps^2)}\Big),
\end{split}
\end{equation}
where we have used the notations from \Cref{eq:some notations for K and the heat kernel}. For any $n\geq 2$,
\begin{equation}
\begin{split}
\label{eq:def of Vn}
\tilde{V}_t^n(\xx,\Si|m,\eps):=\frac{\beta}{2}\sum_{I\varsubsetneq [n]}\int_{\eps^2}^t&\bigg\{\bigg(\sum_{\substack{i\in I\\ j\in I^c}}\sigma_{ij}C_s^m(i,j)\bigg)\tilde{V}_s^{|I|}(\xx_I,\Si_I|m,\eps)\tilde{V}_s^{|I^c|}(\xx_{I^c},\Si_{I^c}|m,\eps)
\\
&\qquad\times \exp\!\Big(-\frac{\beta}{2}\sum_{k,l=1}^n\sigma_{kl}K_{s,t}^m(k,l)\Big)\bigg\}\ds,
\end{split}
\end{equation}
where we have denoted by $[n]:=\{1,2,\dots,n\}$ the discrete interval from $1$ to $n$, by $I^c$ the complement of $I$ with respect to $[n]$, and by $|I|$ the cardinality of $I$. Furthermore, we have indexed variables by sets: for $k\in\N$ and $J=\{j_1,j_2,\dots,j_k\}$ we denote $\zz_J=(z_{j_1},z_{j_2},\dots,z_{j_k})$ for $z=x,\sigma$. We have also written $\sigma_{pq}:=\sigma_p\sigma_q$ for $\sigma_p,\sigma_q\in\{-1,1\}$.

We note that the sum $\sum_{I\varsubsetneq [n]}$ over non-empty proper subsets $I$ of $[n]$ accounts for all ordered bipartitions $\{A,A'\}$ of the set $[n]$ such that $A\uplus A'=[n]$, where $\uplus$ denotes disjoint union. This means that for fixed $A,A'\varsubsetneq [n]$ satisfying $A\uplus A'=[n]$, the sum contains both $\{A,A'\}$ and $\{A',A\}$.

We readily see that
\begin{equation}
\tilV_t^1(m,\eps)=\E_{m,(\eps^2,t)}\abr{\eps^{-\frac{\beta}{4\pi}}e^{i\sigma\sqrt{\beta}\vp(x)}},
\end{equation}
where $\sigma=\pm 1$. This is the first charge correlation function with respect to the regularized free field at small scales below $t>0$. The equivalence above will generalize to $n>1$ with the truncated correlation functions as stated in \Cref{thm:the equivalence of St and Vt} below. To obtain convergence of $S_t$ in the limit $m,\eps\to 0$ for all relevant values of $\beta$, it needs to be renormalized. First fix $\beta\in (0,4\pi)$ for $d=1$ or $\beta\in (0,(d+1)2\pi)$ for $d\geq 2$, and pick $N\equiv N(\beta)$ such that $\beta\in[\beta_N,\beta_{N+2})$. Then define
\begin{equation}
\label{eq:def of StR}
\begin{split}
S_t^R&(\beta,\eta,\vp|m,\eps)
\\
&:=\sum_{n=1}^\infty\frac{1}{n!}\sum_{\Si\in\{-1,1\}^n}\int_{\R^{nd}}
\rbr{\prod_{i=1}^n \eta_i}
\rbr{
e^{i\sqrt{\beta}\sum_{j=1}^n\sigma_j\vp_j}\tilde{V}_t^n(\xx,\Si|m,\eps)
-
\1_{n\leq N}\1_{\{\sum_{i=1}^n\sigma_i=0\}}
\tilV_\infty^n(\xx,\Si|m,\eps)
}
\d\xx.
\end{split}
\end{equation}
Note that the subtracted $\tilV_\infty^n$ terms turn out to produce exactly the multiplicative counterterms $Z_k$ in the definition \Cref{eq:definition of the renormalized partition function of the sG model} of the renormalized partition function $\Zcal_R$ when we put $\eta=-z\1_{\Lambda}$. As already mentioned, for finite $t$ the functions $\tilV_t^n$ are just the pointwise charge cumulants of the regularized free field at scales below $t$. Furthermore, the full ($t=\infty$) regularized free field charge cumulants exist, so we define $\tilV_\infty^n$ through this equivalence.

Now we are ready to state the main result of this section.
\begin{theorem}
\label{thm:the equivalence of St and Vt}
Let $\beta\in (0,4\pi)$ for $d=1$ and $\beta\in (0,(d+1)2\pi)$ for $d\geq 2$, $\eta\in L_c^\infty(\R^d\times \{-1,1\})$. Then there exists $t^*(\eta)\equiv t^*>0$ such that for $t<t^*$ (this restriction is needed only for the first two claims) the following statements are true.
\begin{enumerate}
\item For $\vp\in C^1(\R^d)$ the renormalized series $S_t^R(\beta,\eta,\vp|m,\eps)$ defined in \Cref{eq:def of StR} converges absolutely and uniformly in $m,\eps\in (0,1)$.
\item For $\vp\in C(\R^d)$ and fixed $m,\eps\in (0,1)$ the series $S_t(\beta,\eta,\vp|m,\eps)$ converges absolutely and equals the renormalized potential $V_t(\beta,\eta,\vp|m,\eps)$.
\item Finally, for fixed $m,\eps\in (0,1)$ and any $t>0$ we have
\begin{equation}
\tilV_t^n(\xx,\Si|m,\eps)
=
(-1)^{\,n-1}\,
\E_{m,(\eps^2,t)}^T\abr{\eps^{-\frac{\beta}{4\pi}}e^{i\sqrt{\beta}\sigma_j\vp(x_j)}\mid j\in[n]}.
\end{equation}
\end{enumerate}
\end{theorem}

\begin{remark}
The regularity assumptions on the field $\vp$ in the above theorem are minimal requirements for this exact argument, but are not necessarily optimal. The Gaussian field $\vp_{m,\sqrt{t}}$ to which the result will be applied is in fact almost surely smooth for $t,m>0$. 
\end{remark}

For the proof of the above theorem we will need additional tools. First we have the basic lemma that serves as the prototype for the behavior of the functions $\tilV_t^n$.

\begin{lemma}
\label{lem:uniform bound for V1}
For any $\beta>0$ we have
\begin{equation}
\label{eq:estimate for V1}
\sup_{m,\eps\in (0,1)}|\tilV_t^1(m,\eps)|\leq C
(t\wedge m^{-2})^{-\frac{\beta}{8\pi}}
\end{equation}
with some constant $C>0$ independent of $m,\eps\in (0,1)$ and $t>0$. Furthermore, for any $t>0$ we can also obtain a constant upper bound, that is,
\begin{equation}
\label{eq:uniform bound for V1}
\sup_{m,\eps\in (0,1)}|\tilV_t^1(m,\eps)|\leq C'
\end{equation}
for some different constant $C'>0$ still independent of $m,\eps \in (0,1)$.  

The constants $C,C'$ may depend on $\beta$, but are finite for all $\beta>0$.
\end{lemma}  

\begin{proof}
Firstly, we have
\begin{equation}
\label{eq:useful form of V1}
\tilV_t^1(m,\eps)
=
\underbrace{
\exp(\frac{\beta}{2}\int_{\eps^2}^1\frac{1-e^{-m^2s}}{4\pi s}\ds)
}_{\leq C}
\exp(-\frac{\beta}{2}\int_1^t\frac{e^{-m^2s}}{4\pi s}\ds),
\end{equation}
where we have simply written the logarithm in the definition as an integral. The constant $C>0$ is finite and can be chosen independent of $m,\eps\in (0,1)$. Consider then the second factor. For $t<m^{-2}$ we have
\begin{equation}
\exp(-\frac{\beta}{2}\int_1^t\frac{e^{-m^2s}}{4\pi s}\ds)
=
\exp(\frac{\beta}{2}\int_1^t\frac{1-e^{-m^2s}}{4\pi s}\ds)
\exp(-\frac{\beta}{2}\int_1^t\frac{\ds}{4\pi s})
\leq
e^{\frac{\beta}{2}m^2t}t^{-\frac{\beta}{8\pi}}
\leq Ct^{-\frac{\beta}{8\pi}},
\end{equation}
where the constant $C>0$ is independent of $m$ and $t$ if $t<m^{-2}$. Consider then the opposite case $t\geq m^{-2}$:
\begin{equation}
\begin{split}
\exp(-\frac{\beta}{2}\int_1^t\frac{e^{-m^2s}}{4\pi s}\ds)
&=
\exp(-\frac{\beta}{2}\int_{m^2}^1\frac{e^{-s}}{4\pi s}\ds)
\underbrace{
\exp(-\frac{\beta}{2}\int_1^{m^2 t}\frac{e^{-s}}{4\pi s}\ds)
}_{\leq 1}
\\
&\leq
\exp(-\frac{\beta}{2}\int_{m^2}^1\frac{\du}{4\pi u})
\underbrace{
\exp(\frac{\beta}{2}\int_{m^2}^1\frac{1-e^{-s}}{4\pi s}\ds)
}_{\leq C \text{ for all } m\in (0,1)}
\\
&\leq C m^{\frac{\beta}{4\pi}}
\\
&=C(m^{-2})^{-\frac{\beta}{8\pi}},
\end{split}
\end{equation}
where the constant $C>0$ is independent of $m$. Thus, by choosing the largest constant in the above analysis, we have
\begin{equation}
\sup_{m,\eps\in (0,1)}|\tilV_t^1(m,\eps)|\leq C(t\wedge m^{-2})^{-\frac{\beta}{8\pi}}. 
\end{equation}
Finally, to obtain a constant bound we simply bound the second exponential in \Cref{eq:useful form of V1} by $1$.
\end{proof}

In the analysis of $\tilV_t^n$ for general $n$, we will need the following norms. For $f\colon \Lambda^n\to \C$ we define the norm $\norm{\cdot}_{n,\Lambda}$ for functions on $\Lambda^n$ with $\Lambda\subset \R^d$ by
\begin{equation}
\label{eq:n norms without the sigma sums}
\norm{f}_{n,\Lambda}:=
\begin{cases} 
\norm{f}_{L^\infty(\Lambda)}:=\esssup_{x\in \Lambda}|f(x)|, &\text{ if $n=1$}
\\
\esssup_{x_1\in\Lambda}\int_{\Lambda^{(n-1)}}|f(\xx)|\d\xx_{n-1}, &\text{ if $n\geq 2$}
\end{cases}
\end{equation}
for suitable functions $f\colon \Lambda^n\to \C$.

If $\Lambda=\R^d$ we omit it from the notation, that is, $\norm{\cdot}_{n,\R^d}\equiv \norm{\cdot}_n$. Above, we have denoted $\d\xx_{n-1}:=\dx_2\dx_3\dots\dx_n$. If the variables are indexed by a set $I$ of size $n$ other than $[n]$, then the supremum is taken with respect to the smallest index, and for permutation-invariant functions we may take it with respect to any of the variables $x_j$ and integrate with respect to the rest. In these cases, we omit the corresponding variables from the integration and use the same notation $\d\xx_{n-1}$ for the integration measure.

Let us then introduce an auxiliary function that we will need in the analysis. Recall that $N\equiv N(\beta)$ is such that $\beta\in [\beta_{N},\beta_{N+2})$. Then set
\begin{equation}
\label{eq:def of Htn}
H_t^n(\xx,\Si):=
\begin{cases}
\frac{d(\xx)}{\sqrt{t}}\wedge 1,\quad &\text{if $n\leq N$ and $\sum_{i=1}^n\sigma_i=0$}
\\
1,\quad &\text{otherwise.} 
\end{cases}
\end{equation}
Above, $d(\xx)$ is the $1$-Wasserstein distance between the distributions of positive and negative charges. First let us denote by $\nu$ the charge distribution 
\begin{equation}
\nu\equiv\nu(n,\xx,\Si):=\sum_{k=1}^n \sigma_k\delta_{x_k},
\end{equation}
where $\xx:=(x_1,x_2,\dots,x_n)\in (\R^d)^n$, $\Si:=(\sigma_1,\sigma_2,\dots,\sigma_n)\in \{-1,1\}^n$ and $\delta_{x_k}$ is the Dirac measure concentrated at the point $x_k\in\R^d$. We denote the positive (negative) charge distribution by $\nu_+$ ($\nu_-$). Next, assume that the configuration is neutral, that is, $\sum_{i=1}^n\sigma_i=0$. Furthermore, w.l.o.g. we may assume that the first $p=n/2$ charges have signs different from the rest. Then we define the $1$-Wasserstein distance between the positive and negative charge distributions by
\begin{equation}
\label{eq:def Wasserstein distance}
d_{1-\mathrm{Wass}}(\nu_+,\nu_-)\equiv d(\xx):=\min_{\rho\in \mathbb{S}_p}\sum_{k=1}^p|x_k-x_{p+\rho(k)}|,
\end{equation} 
where $\mathbb{S}_p$ is the symmetric group, that is, the group of permutations on $[p]=\{1,2,\dots,p\}$.

The function $H_t^n$ arises naturally when we apply the Onsager-type inequality established in \cite{LaRhVa23a} to the exponential in the definition of $\tilV_t^n$. We will present this result, adapted to our parametrization of the problem, in \Cref{sec:Onsager type inequality}. Our task is then to recover the same function in the remaining estimates, so that we can propagate the induction through the recursion for $\tilV_t^n$ and prove the results stated below, which are central to our analysis.

We can now state the auxiliary results needed in the proof of \Cref{thm:the equivalence of St and Vt} and in our proofs of the main results. First, we have

\begin{lemma}
\label{lem:differentiability of Vtn with respect to mass ETC}
Let $d=1$, fix $\beta \in (0,4\pi)$, and let $N\equiv N(\beta)$ be s.t.\ $\beta\in [\beta_{N},\beta_{N+2})$. Suppose also that $\Lambda\subset \R$ is compact. Then for all $t\in (0,\infty)$, $n\leq N$, $\Si\in\{-1,1\}^n$ and $\xx\in\Lambda^n$ outside a Lebesgue null set containing the intersection of $\Lambda^n$ with the diagonal $\{\xx=(x_1,x_2,\dots,x_n)\in \R^{n}\mid x_i=x_j \text{ for some } i\neq j\}$, the functions $\tilV_t^n(\xx,\Si|m,\eps)$ are differentiable with respect to $m$ on the interval $(0,1)$. The above holds for any $\eps\in [0,1)$. Furthermore, if $\eps=0$, then $\tilV_t^n$ is also continuous in $m$ at $m=0$. In particular, in this situation $\tilV_t^n$ is continuous in $m$ on any closed subinterval of $[0,1)$.

Furthermore, there exist functions $\tilde{h}_\cdot^n(\cdot,\cdot)\equiv \tilde{h}_\cdot^n\colon \R_+\times(\R^{n}\times \{-1,1\}^n)\to [0,\infty]$ and $\tilde{g}_\cdot^n(\cdot,\cdot)\equiv \tilde{g}_\cdot^n \colon \R_+\times(\R^n\times \{-1,1\}^n)\to [0,\infty]$ such that $|\tilV_t^n(\cdot,\cdot|m,\eps)|<\tilde{h}_t^n$ and $|\partial_m[\tilV_t^n(\cdot,\cdot|m,\eps)|_{\Lambda}]|<\tilde{g}_t^n|_{\Lambda}$ for all $t>0$ and $n\in\N$. In addition, $\tilde{h}_t^n$ and $\tilde{g}_t^n$ are permutation-invariant, independent of $m,\eps\in(0,1)$, and satisfy the following estimates
\begin{equation}
    \sup_{\Si\in\{-1,1\}^n}\norm{H_t^n(\cdot,\Si)\tilde{h}_t^n(\cdot,\Si)}_{n,\Lambda}\leq [B_1(\beta,N,\Lambda)]^{n-1}C^nn^{n-2}
    \begin{cases}
        t^{-\half}\rbr{t^{\half-\frac{\beta}{8\pi}}}^n, &\text{ if } t<1
        \\
        \max(1,[\log(t)]^{n-1}), &\text{ if } t\geq 1.
    \\
    \end{cases}
\end{equation}
and
\begin{equation}
    \sup_{\Si\in\{-1,1\}^n}\norm{\tilde{g}_t^n(\cdot,\Si)}_{n,\Lambda}\leq [\tilde{B}_1(\beta,N,\Lambda)]^{n-1}C^nn^{n-2}
    \begin{cases}
        t^{\half}\rbr{t^{\half-\frac{\beta}{8\pi}}}^n, &\text{ if } t<1
        \\
        t\max(1,[\log(t)]^{n-1}), &\text{ if } t\geq 1.
    \\
    \end{cases}
\end{equation}
where $\Lambda\subset\R$ is compact.
\end{lemma}

\begin{remark}
For $d=1$, we need the mass derivatives of the functions $\tilV_t^n$ in order to estimate terms of the form
\begin{equation}
e^{-m^2s}\tilV_t^n(\xx,\Si|m,\eps=0)-\tilV_t^n(\xx,\Si|m=0,\eps=0)=\int_0^m\frac{\partial}{\partial \mu}\abr{e^{-\mu^2s}\tilV_t^n(\xx,\Si|\mu,\eps=0)}\d\mu
\end{equation} 
in the proof of \Cref{cor:convergence of the Vinfty-Vt terms}, when estimating the error term. This is needed because, in this case, we do not have mass-independent majorants for the differences $\tilV_\infty^n-\tilV_t^n$ when $n>2$. Furthermore, the uniform majorants $\tilde{g}_t^n(\xx,\Si)$ are needed to justify differentiating in the mass under the $s$-integral in the definition \Cref{eq:def of Vn} of the functions $\tilV_t^n$.
\end{remark}

\begin{remark}
\label{rem:limits of the Vtn functions}
By the third claim of \Cref{thm:the equivalence of St and Vt}, the functions $\tilV_t^n$ are precisely the truncated charge correlation functions of the free field at small scales below $t>0$. The proof of this claim does not rely on any of the estimates in this section. Furthermore, the existence of the limit $\eps \to 0$ (and similarly the limit obtained by taking $m\to 0$ after $\eps\to 0$) holds for the truncated charge correlation function of the full free field (that is, $\tilV_\infty^n$) by \cref{lem:pointwise charge correlations of the free field,rem:pointwise truncated charge correlation functions of the free field}. The same holds for $\tilV_t^n$ by \Cref{lem:pointwise charge correlation functions of the scale decomposed free field}. We define $\lim_{\eps\to 0}\tilV_t^n(\cdot,\cdot|m,\eps)=:\tilV_t^n(\cdot,\cdot|m)$ and $\lim_{m\to 0}\lim_{\eps\to 0}\tilV_t^n(\cdot,\cdot|m,\eps)=:\tilV_t^n(\cdot,\cdot)$ by these limits of the truncated charge correlation functions. Moreover, since a priori we have $m,\eps\in(0,1)$, we also define the values of the functions $\tilV_t^n$ at $\eps=0$ and at $\eps=0=m$ by these limits, that is 
\begin{equation}
\tilV_t^n(\xx,\Si|m=0,\eps=0):=\lim_{m\to 0}\lim_{\eps\to 0}\tilV_t^n(\xx,\Si|m,\eps) \quad\text{ and }\quad \tilV_t^n(\xx,\Si|m,\eps=0)=\lim_{\eps\to 0}\tilV_t^n(\xx,\Si|m,\eps).
\end{equation} 
All the above statements hold for any charge configuration $\Si\in\{-1,1\}^n$. 
\end{remark}

The proofs of the norm estimates for $\tilde{h}_t^n$ in \Cref{lem:differentiability of Vtn with respect to mass ETC} are straightforward modifications of the proofs of the corresponding estimates for $h_t^n$ in \Cref{prop:induction statement} below. The differentiability claim and the existence of the functions $\tilde{g}_t^n$ must be established iteratively and simultaneously, but the estimates are again analogous to those for the functions $g_t^n$ in \Cref{prop:induction statement}. Since \Cref{prop:induction statement} is the main driving force in our analysis, and we also need the techniques from its proof in the proof of \Cref{cor:convergence of the Vinfty-Vt terms}, we prove it first, assuming \Cref{lem:differentiability of Vtn with respect to mass ETC}, and postpone the proof of the latter to the final subsection of this section. Furthermore, \Cref{lem:differentiability of Vtn with respect to mass ETC} is needed only in the case $d=1$. The proof of \Cref{lem:differentiability of Vtn with respect to mass ETC} does not use \Cref{prop:induction statement} itself, but only techniques developed in its proof. Thus, this is not a circular argument. 

\begin{proposition}
\label{prop:induction statement}
Fix $\beta \in (0,4\pi)$ if $d=1$ and $\beta\in (0,(d+1)2\pi)$ if $d\geq 2$, and let $N\equiv N(\beta)$ be s.t.\ $\beta\in [\beta_{N},\beta_{N+2})$. Suppose also that $\Lambda\subset \R^d$ is compact. 

Then for all $t>0$ and $n\in \N$ there exist functions $h_\cdot^n(\cdot,\cdot|m)\equiv h_\cdot^n\colon \R_+\times(\R^{dn}\times \{-1,1\}^n)\to [0,\infty]$ such that $|\tilV_t^n(\cdot,\cdot|m,\eps)|<h_t^n$. In addition, $h_t^n$ is permutation-invariant and independent of $\eps\in(0,1)$. Furthermore, $h_t^n$ is independent of $m$ for $t<m^{-2}$. The functions $h_t^n$ satisfy the following norm estimates for all $t>0$,
\begin{equation}
\label{eq:statement for the norms of the htn functions}
\sup_{\Si\in\{-1,1\}^n}\norm{H_t^n(\cdot,\Si) h_t^n(\cdot,\Si)}_n
\leq [B_d(\beta,N)]^{n-1}C^n n^{n-2}
\begin{cases}
t^{-\frac{d}{2}}\rbr{t^{\frac{d}{2}-\frac{\beta}{8\pi}}}^n, &\text{ if $t<m^{-2}$}
\\
(m^{-2})^{(n-1)\frac{d}{2}-\frac{n\beta}{8\pi}}, &\text{ if $t\geq m^{-2}$ and $d=1$}
\end{cases}
\end{equation}
with some constants $B_d(\beta,N),C>0$, independent of $m,\eps,t$ and $n$, to be chosen suitably during the proof.

Furthermore, if $d=1$, $n\leq N$ and $\beta$ are as above, then there also exist functions $g_\cdot^n(\cdot,\cdot|m)\equiv g_\cdot^n \colon \R_+\times(\R^n\times \{-1,1\}^n)\to [0,\infty]$ that are also permutation-invariant, independent of $\eps\in (0,1)$, and s.t.
 \\
$|\partial_m[\tilV_t^n(\cdot,\cdot|m,\eps)]|<g_t^n$.

The functions $g_t^n$ satisfy the following estimates for all $t>0$
\begin{equation}
\sup_{\Si\in\{-1,1\}^n}\norm{g_t^n(\cdot,\Si|m)}_{n}
\leq [\tilde{B}_1(\beta,N)]^{n-1}C^n n^{n-2}
\begin{cases}
m\, t^{\frac{1}{2}}\rbr{t^{\frac{1}{2}-\frac{\beta}{8\pi}}}^n, &\text{ if $t<m^{-2}$}
\\
(m^{-2})^{n\rbr{\frac{1}{2}-\frac{\beta}{8\pi}}}, &\text{ if $t\geq m^{-2}$}
\end{cases}
\end{equation}
with some constants $\tilde{B}_1(\beta,N),C>0$ independent of $m,\eps,t$ and $n$.

The constant $C$ may initially be different in the estimates for the functions $h_t^n$ and $g_t^n$, but we may choose the maximum of these. The bounds for $t\geq m^{-2}$ are simply the same as for $t<m^{-2}$ with the replacement $t\to m^{-2}$.
\end{proposition}

\begin{remark}
\label{rem:the t geq m2 estimates are not needed for d geq 2}
The statement \Cref{eq:statement for the norms of the htn functions} for $t\geq m^{-2}$ is also true for $d\geq 2$, but we do not need it. Indeed, the estimates for $t\geq m^{-2}$ are used in the proof of \Cref{cor:convergence of the Vinfty-Vt terms}, but as discussed right after stating this corollary, the case $d\geq 2$ can be treated using results from \cite{BaWe24a}.
\end{remark}

\begin{remark}
These bounds imply that the mass derivative of $\tilV_t^n(\cdot,\cdot|m,\eps)$ behaves better than $\tilV_t^n(\cdot,\cdot|m,\eps)$ itself as $\eps\to 0$, since the function $H_t^n$ is needed to cancel the singularity for $h_t^n$, but not for $g_t^n$. However, the opposite is true for the limit $m\to 0$ when $t$ is large. Indeed, the estimate for $h_t^n$ in the regime $t\geq m^{-2}$ can be made independent of $m$ provided the total exponent $(n-1)d/2-n\beta/8\pi$ is non-positive, that is, $\beta\geq (1-1/n)4\pi$ for all $n\leq N$. This is guaranteed by the assumption $\beta\in [\beta_N,\beta_{N+2})$. In contrast, the upper bound for $g_t^n$ in the same regime diverges as $m\to 0$ for all $\beta\in (0,4\pi)$. 
\end{remark}

Finally, we have the following corollary.

\begin{corollary}
\label{cor:convergence of the Vinfty-Vt terms}
Let $d=1$ so that $\beta\in (0,4\pi)$ and let $N\equiv N(\beta)$ be such that $\beta\in[\beta_N,\beta_{N+2})$. Suppose also that $\eta\in L_c^\infty(\R\times\{-1,1\})$, $n\leq N$ and $\sum_{i=1}^n\sigma_i=0$. 
Then there exist functions $F_\cdot^{1,n}(\cdot,\cdot|m)\colon\R_+\times(\R^n\times\{-1,1\}^n)\to [0,\infty]$ independent of $\eps$ such that
\begin{equation}
\label{eq:majorant condition for Vinfty-Vt for d=1}
|\tilV_\infty^n(\xx,\Si|m,\eps)-\tilV_t^n(\xx,\Si|m,\eps)|\leq F_t^{1,n}(\xx,\Si|m)
\end{equation}
and for any compact $\Lambda\subset \R$
\begin{equation}
\label{eq:uniform bound for the majorant in d=1}
\sup_{m\in (0,1)}\sum_{\Si\in\{-1,1\}^n}\int_{\Lambda^{n}}F_t^{1,n}(\xx,\Si|m)\d\xx<\infty.
\end{equation}

Next let $d\geq 2$ and $\beta\in [2d\pi,(d+1)2\pi)$, so that $n=N=2$. Suppose also that $\eta\in L_c^\infty(\R^d\times\{-1,1\})$ and $\sigma_1=-\sigma_2$. Then there also exist functions $F_\cdot^{d,2}(\cdot,\cdot)\colon \R_+\times(\R^{2d}\times\{-1,1\}^2)\to [0,\infty]$ independent of $\eps$ and $m$ such that for any compact $\Lambda\subset\R^d$ we have
\begin{equation}
\label{eq:uniform bound for Vinfty minus Vt for d greater than 2 part 1}
|\tilV_\infty^2(\xx,\Si|m,\eps)-\tilV_t^2(\xx,\Si|m,\eps)|\leq F_t^{d,2}(\xx,\Si),
\end{equation}
and
\begin{equation}
\label{eq:uniform bound for Vinfty minus Vt for d greater than 2 part 2}
\sum_{\Si\in\{-1,1\}^2}\int_{\Lambda^2}F_t^{d,2}(\xx,\Si)\d\xx<\infty.
\end{equation}

Furthermore, for any $d\geq 1$ we have
\begin{equation}
\begin{split}
\label{eq:limits of the difference Vinfty-Vt}
\lim_{m\to 0}\lim_{\eps\to 0}&\sum_{\Si\in\{-1,1\}^n}\int_{\R^{dn}}\rbr{\prod_{i=1}^n\eta_i}\1_{n\leq N}\1_{\sum_{i=1}^n\sigma_i=0}(\tilV_\infty^n(\xx,\Si|m,\eps)-\tilV_t^n(\xx,\Si|m,\eps))\d\xx
\\
&=\sum_{\Si\in\{-1,1\}^n}\int_{\R^{dn}}\rbr{\prod_{i=1}^n\eta_i}\1_{n\leq N}\1_{\sum_{i=1}^n\sigma_i=0}(\tilV_\infty^n(\xx,\Si)-\tilV_t^n(\xx,\Si))\d\xx
\end{split}
\end{equation}
where $\tilV_s^n(\xx,\Si):=\tilV_s^n(\xx,\Si|m=0,\eps=0):=\lim_{m\to 0}\lim_{\eps\to 0}\tilV_s^n(\xx,\Si|m,\eps)$. The claim in \Cref{eq:limits of the difference Vinfty-Vt} also holds for fixed $m$ with obvious changes. Lastly, we also have $\tilV_\infty^n-\tilV_t^n\in L_{loc}^1(\R^{nd}\times\{-1,1\}^n)$.
\end{corollary}

\begin{remark}
\label{rem:finiteness of the limit of the Vinfty-Vt terms}
The limit in \Cref{eq:limits of the difference Vinfty-Vt} is always finite by the bounds \Cref{eq:uniform bound for Vinfty minus Vt for d greater than 2 part 1,eq:uniform bound for Vinfty minus Vt for d greater than 2 part 2,eq:uniform bound for the majorant in d=1}.
\end{remark}

\begin{proof}[Proof of the case $d\geq 2$]
This case does not actually utilize \Cref{prop:induction statement}. The first claim for the $d\geq 2$ case, that is, \Cref{eq:uniform bound for Vinfty minus Vt for d greater than 2 part 1,eq:uniform bound for Vinfty minus Vt for d greater than 2 part 2}, can be established analogously to the proof of \cite[Lemma 5.4]{BaWe24a}, and the bound $\beta< (d+1)2\pi$ follows directly from that argument. Then dominated convergence allows us to pass both limits through the spatial integrals in this setting to obtain \Cref{eq:limits of the difference Vinfty-Vt}. Furthermore, the function $\tilV_t^2$ and the corresponding pointwise limits can be computed explicitly.
\end{proof}

For $d=1$, the proof of this corollary relies on estimates similar to those used in the proof of \Cref{prop:induction statement} and is more involved than in the case $d\geq 2$, since the majorant $F_t^{1,n}$ cannot be chosen independently of the mass (or at least such a choice appears hard to achieve through inductive methods). For this reason, we postpone the detailed proof for the one-dimensional case and first establish \Cref{prop:induction statement}. For now, we only give a skeleton of the argument for the proof of \Cref{cor:convergence of the Vinfty-Vt terms} in the case $d=1$. The estimates \Cref{eq:majorant condition for Vinfty-Vt for d=1,eq:uniform bound for the majorant in d=1}, together with dominated convergence, allow us to take the $\eps\to 0$ limit inside the spatial integral, and the pointwise limits exist by the discussion in \Cref{rem:limits of the Vtn functions}.

For the $m\to 0$ limit, we write
\begin{equation}
\label{eq:decomposition of the Vinfty-Vt terms}
\tilV_\infty^n(\xx,\Si|m)-\tilV_t^n(\xx,\Si|m)
=
\tilV_\infty^n(\xx,\Si)-\tilV_t^n(\xx,\Si)
+
R_t^n(m,\xx,\Si)
\end{equation}
and show that
\begin{equation}
\label{eq:vanishing of the error in the Vinfty-Vt terms}
\lim_{m\to 0}
\sum_{\Si\in\{-1,1\}^n}
\int_{\R^n}
\rbr{\prod_{i=1}^n|\eta_i|}|R_t^n(m,\xx,\Si)|
\d\xx
=0.
\end{equation}

We will now prove \Cref{thm:the equivalence of St and Vt} assuming \Cref{prop:induction statement,cor:convergence of the Vinfty-Vt terms}. After that, we will establish \Cref{prop:induction statement,cor:convergence of the Vinfty-Vt terms} assuming \Cref{lem:differentiability of Vtn with respect to mass ETC}. Finally, in the last subsection, we will prove \Cref{lem:differentiability of Vtn with respect to mass ETC}.

\begin{subsection}{Proof of {{\Cref{thm:the equivalence of St and Vt}}}}
\label{sec:proof of equivalnce of St and Vt}
We will prove each of the claims separately, but first we need to rewrite the series $S_t^R$ in a form that is more tractable for our analysis.

First, let us denote
\begin{equation}
\label{eq: renormalized series for the potential Vt}
S_t^R(\beta,\eta,\vp|m,\eps)=\sum_{n=1}^\infty \Acal_n(\beta,t|m,\eps)
\end{equation}
Then we may write
\begin{equation}
\begin{split}
\label{eq:splitting of the counter terms}
\Acal_n(\beta,t|m,\eps)&=\frac{1}{n!}\sum_{\Si\in\{-1,1\}^n}\int_{\R^{nd}}\rbr{\prod_{i=1}^n\eta_i}\rbr{e^{i\sqrt{\beta}\sum_{j=1}^n\sigma_j\vp_j}-\1_{n\leq N}\1_{\sum_{i=1}^n\sigma_i=0}}\tilV_t^n(\xx,\Si|m,\eps)\d\xx
\\
&\quad +\frac{1}{n!}\sum_{\Si\in\{-1,1\}^n}\int_{\R^{nd}}\rbr{\prod_{i=1}^n\eta_i}\1_{n\leq N}\1_{\sum_{i=1}^n\sigma_i=0}\rbr{\tilV_t^n(\xx,\Si|m,\eps)-\tilV_\infty^n(\xx,\Si|m,\eps)}\d\xx
\\
&\,\,=:\Acal_n^\vp(\beta,t|m,\eps)+\Acal_n^\infty(\beta,t|m,\eps)
\end{split}
\end{equation}
where we have simply added and subtracted $\1_{n\leq N}\1_{\sum_{i=1}^n\sigma_i=0}\tilV_t^n(\xx,\Si|m,\eps)$.

\begin{remark}
\label{rem: renormalizability in higher dimensions}
Analysing the term $\Acal_2^\vp(\beta,t|m,\eps)$ explicitly for $d\geq 2$ gives some indication of how far our methods might extend beyond the restrictions on $\beta$ imposed in this article. In particular, it suggests limits on the renormalizability in the parameter range $\beta>(d+1)2\pi$. First of all, $\tilV_t^2$ can be computed explicitly (see \cite[Equation 4.60]{BaWe24a}). We only consider the neutral part ($\sigma:=\sigma_1=-\sigma_2$) of $\Acal_2^\vp$, since the non-neutral part is bounded for all $\beta>0$. Using this exact form, the neutral part of $\Acal_2^\vp$ becomes
\begin{equation}
\label{eq:sigma summed estimate for the Vt2 term}
\half\sum_{\sigma\in\{-1,1\}}\int_{\R^{2d}}\eta_1\eta_2\rbr{e^{i\sqrt{\beta}\sigma[\vp(x_1)-\vp(x_2)]}-1}[\tilV_t^1(m,\eps)]^2\rbr{1-e^{\beta K_{\eps^2,t}^m(x_1,x_2)}}\dx_1\dx_2
\end{equation}
By applying the cancellation from the $\sigma$ sums (which we do not use in the proofs in this article), we can estimate the expression in \Cref{eq:sigma summed estimate for the Vt2 term}. For simplicity, suppose $\eta(\cdot,1)=\eta(\cdot,-1)$, which yields
\begin{equation}
-\int_{\R^{2d}}\eta_1\eta_2(1-\cos(\sqrt{\beta}[\vp(x_1)-\vp(x_2)]))[\tilV_t^1(m,\eps)]^2\rbr{1-e^{\beta K_{\eps^2,t}^m(x_1,x_2)}}\dx_1\dx_2
\end{equation}

Then, using
\begin{equation}
1-\cos(\sqrt{\beta}[\vp(x_1)-\vp(x_2)])\leq C_\beta (\norm{\grad\vp}_{L^\infty(B_\Lambda)}|x_1-x_2|)^{2-\varepsilon(\beta)}
\end{equation} 
for some $C_\beta>0$ and any $\varepsilon(\beta)>0$ (note that the choice $\varepsilon(\beta)=0$ would cause problems when we apply \Cref{lem:Gaussian tails}). Above, we used arguments and notation similar to those in \Cref{lem:estimate for the field dependent part of the counter term} below. Thus, using \cite[Lemma 4.4]{BaWe24a} and \Cref{lem:uniform bound for V1}, we obtain for $t<m^{-2}$
\begin{equation}
|\Acal_{2,neutral}^\vp(\beta,t|m,\eps)|\leq C_{\beta,t,\Lambda}\norm{\eta}_{L^{\infty}(\R^d)}^2\norm{\grad\vp}_{L^\infty(B_\Lambda)}^{2-\varepsilon(\beta)}\int_{B_\Lambda}|x|^{2-\varepsilon(\beta)}\rbr{\frac{|x|}{\sqrt{t}}\wedge 1}^{-\frac{\beta}{2\pi}}\dx
\end{equation}
for some constant $C_{\beta,t,\Lambda}$ independent of $m$ and $\eps$. By choosing $\varepsilon(\beta)$ suitably, the above integral can be made to converge arbitrarily close to $\beta=(d+2)2\pi$, but not at that value. For $d=2$, $(d+2)2\pi=8\pi$ is the KT-point. Therefore, it is at least in principle possible to extend our methods to the whole subcritical range $\beta\in (0,8\pi)$. However, for large enough $d$ (in fact for $d\geq 5$), $(d+2)2\pi< \beta_4:=3d\pi$ if we assume that the behavior remains scalable beyond the first collapse point, as explained in the introduction. Therefore, there is no hope of extending our methods even up to the second projected collapse point for all $d>2$.
\end{remark}

We need one more lemma before we begin the actual proof.

\begin{lemma}
\label{lem:estimate for the field dependent part of the counter term}
    Let $\Lambda\subset \R^d$ be compact and $\vp\in C^1(\R^d)$. Then we have
    \begin{equation}
    \abs{e^{i\sqrt{\beta}\sum_{j=1}^n\sigma_j\vp_j}-\1_{n\leq N}\1_{\sum_{i=1}^n\sigma_i=0}}\leq 
    \begin{cases}
    2([\sqrt{\beta t}\norm{\grad\vp}_{L^\infty(B_\Lambda)}]\vee 1)\,H_t^n(\xx,\Si), &\text{if } n\leq N,
    \\
    1, &\text{otherwise},
    \end{cases}
    \end{equation}
for all $(\xx,\Si)\in \Lambda^n\times\{-1,1\}^n$. Here $B_\Lambda:=\overline{B_{2R(\Lambda)}(0)}$ and $R(\Lambda)>0$ is chosen so that $\Lambda\subset B_{R(\Lambda)}(0)$, and $H_t^n$ is defined in \Cref{eq:def of Htn}.
\end{lemma}

\begin{proof}
First assume that $n\leq N$ and that the configuration is neutral, that is, $\sum_{i=1}^n\sigma_i=0$. Without loss of generality, we may again assume that the first $n/2=:p$ charges $\sigma_i$ differ from the remaining ones. Then the left-hand side of the claim becomes
\begin{equation}
\label{eq:permutation version of exp-1}
\begin{split}
F(\vp,\xx,\Si)&:=e^{i\sqrt{\beta}\sum_{j=1}^n\sigma_j\vp_j}-1
=e^{i\sqrt{\beta}\sum_{j=1}^p\sigma_j(\vp_j-\vp_{p+\pi(j)})}-1,
\end{split}
\end{equation}
where $\pi\in\mathbb{S}_{p}$ is a permutation, which we may choose to be the one realizing the minimum in the Wasserstein distance $d(\xx)$ defined in \Cref{eq:def Wasserstein distance}.

Using the basic estimate $|e^{ix}-1|=\abs{\int_0^x e^{iy}\,\dy}\le |x|$ together with the triangle inequality and the mean value theorem, we obtain
\begin{equation}
\begin{split}
|F(\vp,\xx,\Si)|&\leq \sqrt{\beta}\sum_{j=1}^p|\vp_j-\vp_{p+\pi(j)}|
\leq \sqrt{\beta}\norm{\nabla\vp}_{L^\infty(B_\Lambda)}\sum_{j=1}^p|x_j-x_{p+\pi(j)}|
=\sqrt{\beta}\norm{\nabla\vp}_{L^\infty(B_\Lambda)}\, d(\xx).
\end{split}
\end{equation}

We may also choose to apply this estimate only for those $\xx\in\Lambda^{n}$ satisfying $d(\xx)<\sqrt{t}$, and use the crude bound $|F(\vp,\xx,\Si)|\leq 2$ otherwise. Hence the final estimate becomes
\begin{equation}
\label{eq:field dependent counterterm estimate}
\begin{split}
|F(\vp,\xx,\Si)|&\leq 2\bigl(\sqrt{\beta}\norm{\grad\vp}_{L^\infty(B_\Lambda)}\sqrt{t}\bigr)^{\1_{\{d(\xx)\leq \sqrt{t}\}}}
\left(\frac{d(\xx)}{\sqrt{t}}\wedge 1\right)
\\
&\leq 2\bigl([\sqrt{\beta t}\norm{\grad\vp}_{L^\infty(B_\Lambda)}]\vee 1\bigr)\,H_t^n(\xx,\Si)
\end{split}
\end{equation}
Because of the $\vee 1$ factor and the definition of $H_t^n$, the same estimate applies to all cases with $n\leq N$ and $\Si\in\{-1,1\}^n$. The case $n>N$ is trivial.
\end{proof}

Now we may begin the actual proof with the first statement. 
\begin{proof}[Proof of {{\Cref{thm:the equivalence of St and Vt}}} 1.]

Recall the notation from \Cref{eq:splitting of the counter terms}
\begin{equation}
S_t^R(\beta,\eta,\vp|m,\eps)=\sum_{n=1}^\infty[\Acal_n^\vp(\beta,t|m,\eps)+\Acal_n^\infty(\beta,t|m,\eps)]
\end{equation}
and note that there are only finitely many ($N$) non-zero $\Acal_n^\infty(\beta,t|m,\eps)$ terms. Thus, the terms $\Acal_n^\vp(\beta,t|m,\eps)$ determine the convergence of the full series, and for the $\Acal_n^\infty$ terms we may use any uniform bound. \Cref{cor:convergence of the Vinfty-Vt terms} provides such a bound. Indeed, we have
\begin{equation}
\label{eq:uniform bound for Acaln}
|\Acal_n^\infty(\beta,t|m,\eps)|\leq\norm{\eta}_\infty^n\sum_{\Si\in\{-1,1\}^n}\int_{\Lambda^n}F_t^{d,n}(\xx,\Si|m)\d\xx<C_d(\beta,t,N,\Lambda)<\infty
\end{equation}
where the constant $C_d(\beta,t,N,\Lambda)>0$ is independent of $m$ and $\eps$, and where $\Lambda\subset \R^d$ is a compact set such that $\supp(\eta)\subset\Lambda\times\{-1,1\}$.

Thus, it remains to show that the series $\sum_{n=1}^\infty\Acal_n^\vp(\beta,t|m,\eps)$ converges absolutely and uniformly in $m,\eps$. 
Without loss of generality, we may assume that $t<m^{-2}$ (at the end we take the minimum of the condition obtained in this way and $m^{-2}$). By \Cref{lem:estimate for the field dependent part of the counter term,prop:induction statement} we have
\begin{equation}
\begin{split}
\label{eq:uniform bound for Acalphi}
|\Acal_n^\vp(\beta,t|m,\eps)|&\leq \frac{1}{n!}(2([\sqrt{\beta t}\norm{\grad\vp}_{L^\infty(B_\Lambda)}]\vee 1)\1_{n\leq N}+\1_{n>N})\norm{\eta}_\infty^n|\Lambda|2^n\sup_{\Si\in\{-1,1\}^n}\norm{H_t^n(\cdot,\Si)h_t^n(\cdot,\Si)}_n
\\
&\leq \frac{1}{n!}(2([\sqrt{\beta t}\norm{\grad\vp}_{L^\infty(B_\Lambda)}]\vee 1)\1_{n\leq N}+\1_{n>N})|\Lambda|\frac{t^{-\frac{d}{2}}}{B_d(\beta,N)}[2\norm{\eta}_{\infty}B_d(\beta,N)Ct^{\frac{d}{2}-\frac{\beta}{8\pi}}]^nn^{n-2}
\\
&\leq (2([\sqrt{\beta t}\norm{\grad\vp}_{L^\infty(B_\Lambda)}]\vee 1)\1_{n\leq N}+\1_{n>N})|\Lambda|\frac{t^{-\frac{d}{2}}}{B_d(\beta,N)}\rbr{\hat{B}(\beta,N)\norm{\eta}_{\infty}t^{\frac{d}{2}-\frac{\beta}{8\pi}}}^n,
\end{split}
\end{equation}
where $\hat{B}(\beta,N):=2eB_d(\beta,N)C$ with $B_d,C>0$ from \Cref{prop:induction statement}. Here we used the facts that $n^{-2}\leq 1$, $n^n/n!\leq e^n$, and $\sum_{\Si\in\{-1,1\}^n}1=2^n$. Then for 
\begin{equation}
\label{eq:def of tilde t}
t<(\hat{B}(\beta,N)\norm{\eta}_\infty)^{-\frac{1}{\frac{d}{2}-\frac{\beta}{8\pi}}}=:\tilde{t}
\end{equation}
we obtain a uniform geometric majorant.

Thus, for $t<t^*:=\min(\tilde{t},m^{-2})$ the series 
\begin{equation}
\sum_{n=1}^\infty\Acal_n^\vp(\beta,t|m,\eps)
\end{equation}
converges absolutely and uniformly in $m,\eps\in (0,1)$. 
\end{proof}

Knowing the equivalence of the series $S_t$ and the renormalized potential $V_t$ would imply the equivalence of the functions $\tilV_t^n$ with the cumulants in part 3 of this theorem, and vice versa. The proof of the former statement, without assuming the latter, would proceed through an infinite-dimensional PDE known as the Polchinski equation, while the proof of the latter, without assuming the former, can be carried out using a finite-dimensional system of ODEs. Both approaches can be traced back to \cite{BrKe87a}. We choose the latter approach and therefore prove the final claim of our theorem next.

\begin{proof}[Proof of {{\Cref{thm:the equivalence of St and Vt}}} 3.]
Consider first the moments corresponding to $\E_{m,(\eps^2,t)}^T[\eps^{-\frac{\beta}{4\pi}}e^{i\sqrt{\beta}\sigma_j\vp_j}\mid j\in[n]]$, where $\vp_j=\vp(x_j)\equiv\vp_{m,(\eps^2,t)}(x_j)$. Namely,
\begin{equation}
\begin{split}
\E_{m,(\eps^2,t)}\abr{\prod_{j=1}^n\eps^{-\frac{\beta}{4\pi}}e^{i\sqrt{\beta}\sigma_j\vp_j}}&=\eps^{-\frac{n\beta}{4\pi}}e^{-\frac{\beta}{2}\sum_{k,l=1}^n\sigma_{kl}K_{\eps^2,t}(k,l)}
\\
&=\rbr{\prod_{j=1}^n\E_{m,(\eps^2,t)}[\eps^{-\frac{\beta}{4\pi}}e^{i\sqrt{\beta}\sigma_j\vp_j}]}e^{-\sum_{1\leq k<l\leq n}\beta\sigma_{kl}K_{\eps^2,t}(k,l)}
\\
&=[\tilV_t^1(\eps)]^n\prod_{1\leq k<l\leq n}e^{-\beta\sigma_{kl}K_{\eps^2,t}(k,l)}
\\
&=[\tilV_t^1(\eps)]^n\prod_{1\leq k<l\leq n}\rbr{1+\abr{e^{-\beta\sigma_{kl}K_{\eps^2,t}(k,l)}-1}}
\\
&=[\tilV_t^1(\eps)]^n\sum_{G\in\Gamma(n)}\prod_{e=(k,l)\in G}\rbr{e^{-\beta \sigma_{kl}K_{\eps^2,t}(k,l)}-1},
\end{split}
\end{equation}
where $\Gamma(n)$ is the set of all Mayer graphs on $[n]$, that is, all unordered pairings $(k,l)$ with $k,l\in [n]$. Above we used the basic Gaussian identity $\E_{m,(\eps^2,t)}[e^{i\sqrt{\beta}\sigma_j\vp_j}]=e^{-\beta/2 K_{\eps^2,t}(0)}$ and the definition \Cref{eq:def of V1} of $\tilV_t^1$. It is not strictly necessary to factor out $(\tilV_t^1)^n$, but we do so in order to apply the results of \cite{BrKe87a} directly. The connected part of the sum above is called the Ursell coefficient or Ursell function $(e^{-U([n],\eps,t)})_c$, where we denote
\begin{align}
U([n],\eps,t)&:=\sum_{1\leq k<l\leq n}\beta\sigma_{kl}K_{\eps^2,t}(k,l)
\\
\rbr{e^{-U([n],\eps,t)}}_c&:=
\begin{cases}
\sum_{G\in\Gamma_c(n)}\prod_{(k,l)\in G}\rbr{e^{-\beta\sigma_{kl}K_{\eps^2,t}(k,l)}-1}, &\text{if $n\geq 2$}
\\
1, &\text{if $n=1$}.
\end{cases}
\end{align}
Here $\Gamma_c(n)$ denotes the set of all connected labeled graphs on $[n]$. We have suppressed the $m$-dependence in $U([n],\eps,t)$. We will also use an analogous definition for any $I\varsubsetneq [n]$.

The Ursell functions are known to represent the truncated correlation functions or cumulants. Thus, since we may factor out $[\tilV_t^1(m,\eps)]^n$ from the cumulants in the representation \Cref{eq:cumulants to moments}, we obtain
\begin{equation}
\E_{m,(\eps^2,t)}^T\abr{\eps^{-\frac{\beta}{4\pi}}e^{i\sqrt{\beta}\sigma_j\vp_j}\mid j\in [n]}=[\tilV_t^1(m,\eps)]^n\rbr{e^{-U([n],\eps,t)}}_c.
\end{equation}

We may likewise factor out $[\tilV_t^1(m,\eps)]^n$ from $\tilV_t^n(\xx,\Si|m,\eps)$ by writing
\begin{equation}
\begin{split}
\tilV_t^n(\xx,\Si|m,\eps)&\,\,=[\tilV_t^1(m,\eps)]^n\tilde{\Cal{V}}_t^n(\xx,\Si|m,\eps)
\\
\tilde{\Cal{V}}_t^n(\xx,\Si|m,\eps)&:=
\begin{cases}
\half\sum_{I\varsubsetneq [n]}\int_{\eps^2}^t\bigg\{\rbr{\sum_{\substack{i\in I\\j\in I^c}}\beta\sigma_{ij}C_s^m(i,j)}\tilde{\Cal{V}}_s^{|I|}(\xx_I,\Si_I|m,\eps)
\\
\qquad\qquad\qquad \times \tilde{\Cal{V}}_s^{|I^c|}(\xx_{I^c},\Si_{I^c}|m,\eps) e^{-\beta\sum_{1\leq k<l\leq n}\sigma_{kl}K_{s,t}(k,l)}\bigg\}\ds, &\text{if $n\geq 2$}
\\
1, &\text{if $n=1$}.
\end{cases}
\end{split}
\end{equation}
To prove this, we may use induction. The base case $n=1$ is immediate from the definition. For the induction step, it suffices to observe that
\begin{equation}
[\tilV_s^1(m,\eps)]^ne^{-\frac{\beta}{2}\sum_{k=1}^n\sigma_{kk}K_{s,t}(k,k)}
=\abr{\tilV_s^1(m,\eps)e^{-\frac{\beta}{2}K_{s,t}(0)}}^n
=[\tilV_t^1(m,\eps)]^n.
\end{equation}

Therefore, it remains to show that $(e^{-U([n],\eps,t)})_c=(-1)^{n-1}\tilde{\Cal{V}}_t^n(\xx,\Si|m,\eps)$. This follows from \cite[Lemma 3.3]{BrKe87a}, but we highlight the main idea. The point is to derive a system of ODEs that both of these satisfy. We have by direct differentiation
\begin{equation}
\begin{split}
\der{t}\tilde{\Cal{V}}_t^n(\xx,\Si|m,\eps)=
\begin{cases}
\half\sum_{I\varsubsetneq [n]}\sum_{\substack{i\in I\\ j\in I^c}}\beta\sigma_{ij}C_t^m(i,j)\tilde{\Cal{V}}_t^{|I|}(\xx_{I},\Si_I|m,\eps)\tilde{\Cal{V}}_t^{|I^c|}(\xx_{I^c},\Si_{I^c}|m,\eps)
\\
\qquad\qquad-\sum_{1\leq k<l\leq n}\beta\sigma_{kl}C_t^m(k,l)\tilde{\Cal{V}}_t^n(\xx,\Si|m,\eps), \qquad \text{if } n\geq 2
\\
0, \qquad\qquad\qquad\qquad\qquad\qquad\qquad\qquad\qquad\qquad \text{if } n=1.
\end{cases}
\end{split}
\end{equation}
where we just used the Leibniz rule. Thus, for any $I\subset[n]$, the function $(-1)^{|I|-1}\tilde{\Cal{V}}_t^{|I|}(\xx_I,\Si_I|m,\eps)$ satisfies the system of ODEs
\begin{equation}
\label{eq:ODE for Ursell functions}
\begin{split}
\der{t}f_t(I,\xx_I)&=-\half\sum_{J\varsubsetneq I}\sum_{\substack{i\in J\\j\in I\setminus J}}\beta \sigma_{ij}C_t^m(i,j)f_t(J,\xx_J)f_t(I\setminus J,\xx_{I\setminus J})
\\
&\qquad\qquad-\sum_{\substack{k,l\in I\\ k<l}}\beta\sigma_{kl}C_t^m(k,l)f_t(I,\xx_I), \qquad \text{if $|I|\geq 2$}
\\
\der{t}f_t(I,\xx_I)&=0,\qquad \text{if $|I|=1$}.
\end{split}
\end{equation}

We must show that the Ursell functions also satisfy the system \Cref{eq:ODE for Ursell functions}, and identify an initial condition shared by both $(-1)^{n-1}\tilde{\Cal{V}}_t^n$ and the Ursell functions that ensures uniqueness of the solution. The correct initial condition is
\begin{equation}
f_{\eps^2}(I,\xx_I)=
\begin{cases}
1, &\text{if $|I|=1$}
\\
0, &\text{otherwise}.
\end{cases}
\end{equation}
We refer to the proof of \cite[Lemma 3.3]{BrKe87a} for the facts that the Ursell functions satisfy this system and that the solution is unique.
\end{proof}

Finally, we can prove the second statement.

\begin{proof}[Proof of {{\Cref{thm:the equivalence of St and Vt}}} 2.]
We will again fix $\beta$ and choose $N(\beta)\equiv N$ such that $\beta\in [\beta_N,\beta_{N+2})$. In the following, the expectation will always be taken with respect to the law of $\vp_{m,(\eps^2,t)}$. Formally, for deterministic (or equivalently independent of $\vp_{m,(\eps^2,t)}$) $\vp\in C_b(\R^d)$ we may write
\begin{equation}
\label{eq:first step in the proof of theorem 2.1}
-V_t(\beta,z\eta,\vp|m,\eps)=\log(\E\abr{e^{-V_0(\beta,z\eta,\vp_{m,(\eps^2,t)}+\vp|\eps)}})=:\log(\E\abr{e^{\zeta Q(\eta,\vp_{m,(\eps^2,t)}+\vp|\eps)}}),
\end{equation}
where 
\begin{equation}
Q(\eta,\vp_{m,(\eps^2,t)}+\vp|\eps)=\eps^{-\frac{\beta}{4\pi}}\sum_{\sigma\in\{-1,1\}}\int_{\R^d}\eta(x,\sigma)e^{i\sqrt{\beta}\sigma(\vp_{m,(\eps^2,t)}(x)+\vp(x))}\dx
\end{equation}
and $\zeta=-z$. 

For fixed $m,\eps\in (0,1)$ we have $|Q|\leq 2|\Lambda|\norm{\eta}_\infty\eps^{-\beta/(4\pi)}$, where $\Lambda\subset \R^d$ is compact and satisfies $\supp(\eta)\subset \Lambda\times \{-1,1\}$, so that $Q$ is a deterministically bounded random variable. Thus, $\zeta\mapsto e^{\zeta Q}$ is an entire function (pointwise in the probability space $\Omega$), and therefore $\zeta\mapsto\E[e^{\zeta Q}]$ is also entire by Fubini’s and Morera’s theorems. Now choose $\zeta$ such that $|\zeta Q|\ll 1$, say $|\zeta Q|<\varepsilon$ for some very small number $\varepsilon$. Then for $|\zeta|=|z|\leq (2|\Lambda|\norm{\eta}_\infty\eps^{-\beta/(4\pi)})^{-1}\varepsilon=:\delta(\varepsilon)$ we have $|e^{\zeta Q}-1|\ll 1$. Thus,
\begin{equation}
\E\abr{e^{\zeta Q(\eta,\vp_{m,(\eps^2,t)}+\vp|\eps)}}=1+\E[e^{\zeta Q(\eta,\vp_{m,(\eps^2,t)}+\vp|\eps)}-1]\neq 0
\end{equation} 
for $\zeta \in B_{\delta(\varepsilon)}(0)$. Therefore, there exists a branch of the logarithm such that $V_t(\beta,z\eta,\vp|m,\eps)$ is analytic in $B_{\delta(\varepsilon)}(0)$ for any $\beta\in\R$. Furthermore, for $\zeta\in\R$ the function $-V_t(\beta,-\zeta\eta,\vp|m,\eps)$ is the cumulant generating function
\begin{equation}
\begin{split}
V_t(\beta,z\eta,\vp|m,\eps)&=V_t(\beta,-\zeta\eta,\vp|m,\eps)
\\
&=-\sum_{n=1}^\infty\frac{\zeta^n}{n!}\kappa^n\abr{Q(\eta,\vp_{m,(\eps^2,t)}+\vp|\eps)}
\\
&=\sum_{n=1}^\infty\frac{z^n}{n!}(-1)^{n-1}\kappa^n\abr{Q(\eta,\vp_{m,(\eps^2,t)}+\vp|\eps)},
\end{split}
\end{equation} 
where $\kappa^n$ denotes the $n^{\text{th}}$-order cumulant of the single random variable $Q$ with respect to the law of $\vp_{m,(\eps^2,t)}$. Then by boundedness and multilinearity we have
\begin{equation}
\kappa^n[Q(\eta,\vp_{m,(\eps^2,t)}+\vp|\eps)]=\sum_{\Si\in\{-1,1\}^n}\int_{\R^{dn}}\rbr{\prod_{i=1}^n\eta_i}e^{i\sqrt{\beta}\sum_{j=1}^n\sigma_j \vp_j}\E^T[\eps^{-\frac{\beta}{4\pi}}e^{i\sqrt{\beta}\sigma_k\vp_{m,(\eps^2,t)}(x_k)}\mid k\in[n]]\d\xx.
\end{equation}

By the proof of part 1, the terms 
\begin{equation}
\frac{1}{n!}\sum_{\Si\in\{-1,1\}^n}\int_{\R^{nd}}\rbr{\prod_{i=1}^nz\eta_i}e^{i\sqrt{\beta}\sum_{j=1}^n\sigma_j\vp_j}\tilV_t^n(\xx,\Si|m,\eps)\d\xx
\end{equation}
of the series $S_t(\beta,z\eta,\vp|m,\eps)$ have geometric majorants on the set 
\begin{equation}
\label{eq:t set where geometric majorant}
\{t\in[\eps^2,\infty)\mid t<t^*\},
\end{equation}
where
\begin{equation}
\label{eq:def of tstar}
t^*:=\min\rbr{m^{-2},\abr{\hat{B}(\beta,N)|z|\norm{\eta}_\infty}^{-\frac{1}{\frac{d}{2}-\frac{\beta}{8\pi}}}},
\end{equation}
with $\hat{B}(\beta,N)=2eB_d(\beta,N)C$ and the constants $B_d$ and $C$ from \Cref{prop:induction statement}. For the terms $n\leq N+1$ we can derive $\eps$- and $m$-dependent upper bounds by a similar iterative method as in the proof of \Cref{prop:induction statement}. For fixed $t$ such that $0<\eps^2\leq t<m^{-2}$ we have that the series $S_t(\beta,z\eta,\vp|m,\eps)$ converges in $B_R(0)$ with 
\begin{equation}
\label{eq:radius of convergence for zS}
R=\bigl(\hat{B}(\beta,N)\norm{\eta}_\infty t^{\frac{d}{2}-\frac{\beta}{8\pi}}\bigr)^{-1},
\end{equation}
and represents an analytic function. Thus, by part 3, $V_t(\beta,z\eta,\vp|m,\eps)=S_t(\beta,z\eta,\vp|m,\eps)$ on $\R\cap B_R(0)\cap B_{\delta(\varepsilon)}(0)$, and by analytic continuation we may extend the equivalence to all of $B_R(0)$. Therefore, choosing $t<t^*$ with $t^*$ defined in \Cref{eq:def of tstar} ensures $1\in B_R(0)$, which yields the result.
\end{proof}

\end{subsection}

\begin{subsection}{Proofs of {{\Cref{lem:differentiability of Vtn with respect to mass ETC,prop:induction statement,cor:convergence of the Vinfty-Vt terms}}}}
\label{sec:proofs of prop 2.1 etc}
The first two subsections will introduce some preliminary estimates. Namely, we will first present the previously mentioned Onsager-type inequality established in \cite{LaRhVa23a}, and afterwards we will analyze the relevant sums of the covariance kernels $C_s$ (the mass exponential can be factored out of all $C_s^m$ and does not affect the estimates). The remainder of the section will then be devoted to the proofs of \Cref{prop:induction statement,cor:convergence of the Vinfty-Vt terms,lem:differentiability of Vtn with respect to mass ETC} in this order. 

\begin{subsubsection}{Onsager-type inequality}
\label{sec:Onsager type inequality}
The topic of this subsection is a tool from \cite{LaRhVa23a}, where the authors analyzed the boundary sine-Gordon model, which is equivalent to our setup with $d=1$, but with a more general covariance kernel in place of $C_s^m$. However, they did not explicitly consider the mass dependence or the limit $m\to 0$. This result is valid for arbitrary dimension $d$, although their other results are not.

Their parametrization of the problem uses a slightly different representation and regularization of the covariance of the reference field. More precisely, they used the regularized covariance kernels
\begin{equation}
\tilde{K}_t(x,y):=\int_0^tQ_r(x,y)\dr,
\end{equation}
where $Q_r$ is a symmetric positive definite kernel such that $\tilde{K}_t$ becomes a logarithmic covariance kernel in the limit $t\to\infty$. The prototype for such a kernel is the translation invariant case $Q_r(x,y)=q(e^r|x-y|)$ with $q\colon[0,\infty)\to\R$. We are going to use the result for our massless kernel $C_u(x,y)$, which fits into this translation invariant framework. The result yields an estimate for the charge weighted sums of the quantities
\begin{equation}
\tilde{K}_{s,t}(x,y):=\int_s^tq(e^r|x-y|)\dr
\end{equation}
The basic substitution $u=e^{-2r}$ then yields, for our parametrization,
\begin{equation}
K_{s,t}(x,y):=\int_s^tC_u(x,y)\du=\int_s^t e^{-\frac{|x-y|^2}{4u}}\frac{\du}{4\pi u}=\frac{1}{2\pi}\int_{\log(1/\sqrt{t})}^{\log(1/\sqrt{s})}q(e^r|x-y|)\dr
\end{equation}
with $q(x)=e^{-x^2/4}$. This $q$ certainly satisfies their assumptions.

Using the above relation we can translate the Onsager-type inequality (\cite[Lemma 4.3]{LaRhVa23a}) to our setup. The parts 1. and 2. are valid for all $n\geq 1$, but the last part is valid only for $n\leq N+1$.
\begin{lemma}[Onsager type inequality]
\label{lemma:onsager type inequality with our parametrization}
Given $n\geq 2$ and $p\leq n/2$ denoting the number of positive charges, there exists a constant $B_n$ depending only on $n$ such that
\begin{enumerate}
\item If $n$ is even and $p=n/2$, then 
\begin{equation}
\begin{split}
-\sum_{1\leq k<l\leq n} \sigma_{kl} K_{s,t}(k,l)&\leq\frac{1}{4\pi}(n-1)\log(\frac{\sqrt{t}}{\sqrt{s}})-\frac{1}{4\pi}\log(\frac{d(\xx)}{\sqrt{t}}\wedge 1)+B_n,
\end{split}
\end{equation}
where $d(\xx)$ is the Wasserstein distance introduced in \Cref{eq:def Wasserstein distance}.
\item 
In all other cases we have
\begin{equation}
-\sum_{1\leq k<l\leq n}\sigma_{kl}K_{s,t}(k,l)\leq \frac{1}{4\pi} (n-1)\log(\frac{\sqrt{t}}{\sqrt{s}})+B_n.
\end{equation}
\item Furthermore, for $n\leq N+1$ (the cases for which we need this result) we can write this using the functions $H_t^n$ for all charge configurations simultaneously
\begin{equation}
\begin{split}
-\sum_{1\leq k<l\leq n} \sigma_{kl} K_{s,t}(k,l)&\leq\frac{1}{4\pi}(n-1)\log(\frac{\sqrt{t}}{\sqrt{s}})-\frac{1}{4\pi}\log(H_t^n(\xx,\Si))+B_n.
\end{split}
\end{equation}
\end{enumerate}
\end{lemma}

\end{subsubsection}

\begin{subsubsection}{Estimates for sums of covariance kernels}
\label{sec:estimates for sums of covariance kernels}
The first estimate in this section stems from \cite[pages 17–18]{LaRhVa23a}, but we rewrite it in our parametrization and notation so that it can be used in our recursion \Cref{eq:def of Vn} to prove \Cref{prop:induction statement}. The second estimate is novel and is tailored to our setup. 

First, let us pick a smooth monotonically decreasing decay function $f_d\colon[0,\infty)\to [0,\infty)$ such that for $x\in \R^d$, $f_d(|x|)=\bigO((1+|x|)^{-(d+1)})$ as $|x|\to\infty$ and $f_d(|x|)\leq 1$ for all $x$. The Gaussian obviously satisfies this in all dimensions, but we work with $f_d$ for the first estimate since it also applies to more general kernels satisfying the derivative estimates below. Furthermore, we will continue to denote by $f_d$ this decay function in the estimates later in this section even though it would be simply the Gaussian. In addition, we need the fact that for all $d\geq 1$ there exist a function $f_d$ as described above and satisfying 
\begin{equation}
\label{eq:estimate for derivatives of the massive heat kernel}
|(\partial_y^{\nu}\partial_x^{\mu} C_s^m)(x,y)|=e^{-m^2s}|(\partial_y^{\nu}\partial_x^{\mu} C_s)(x,y)|\leq C\frac{e^{-{m^2s}}}{s^{1+\frac{|\nu|+|\mu|}{2}}}f_d\rbr{\frac{|x-y|}{\sqrt{s}}}
\end{equation}
where $\mu$ and $\nu$ are multi-indices, $|\mu|$ denotes the order of $\mu$, and $C>0$ is a constant independent of $m,s,x$ and $y$. The above can be proved by induction. In this section, we need it only for $|\nu|,|\mu|\in\{0,1\}$, which is similar to the condition for the kernels in \cite{LaRhVa23a}.

Let us define 
\begin{equation}
\label{eq:def of GsI function}
G_s^I(\xx,\Si):=\bigg|\sum_{\substack{i\in I\\j\in I^c}}\sigma_{ij}C_s(i,j)\bigg|\Rightarrow e^{-m^2s}G_s^I(\xx,\Si)=\bigg|\sum_{\substack{i\in I\\j\in I^c}}\sigma_{ij}C_s^m(i,j)\bigg|
\end{equation}
where $I\subset [n]$, with $I^c$ denoting the complement with respect to $[n]$, and we have used the shorthand notation from \Cref{eq:some notations for K and the heat kernel}, that is, $C_s(i,j):=C_s(x_i,x_j)$.

Now we can state the first estimate of this section
\begin{lemma}
\label{lem:estimate for GI function}
For $s>0$ we have
\begin{equation}
G_s^I(\xx,\Si)\leq \frac{C}{s} H_s^{|I|}(\xx_I,\Si_I)H_s^{|I^c|}(\xx_{I^c},\Si_{I^c})\sum_{\substack{i\in I\\ j\in I^c}}f_d\rbr{\frac{|x_i-x_j|}{\sqrt{s}}},
\end{equation}
where $
C>0$ is a universal constant and $H_t^n(\xx,\Si)$ is the function defined in \Cref{eq:def of Htn}.
\end{lemma}

\begin{proof}
Let $p$ be the number of positive charges in $I$ and $p^c$ the number of positive charges in $I^c$. The choice of positive rather than negative is arbitrary. If we denote $I=\{i_1,i_2,\dots,i_{|I|}\}$ and $I^c=\{j_1,j_2,\dots,j_{|I^c|}\}$, we may assume without loss of generality that the first $p$ of the $\sigma_{i_k}$ are positive, and analogously for $I^c$. Then, if $p=|I|/2$ and $p^c=|I^c|/2$, we can write
\begin{equation}
\sum_{\substack{i\in I\\ j\in I^c}}\sigma_{ij}C_s(i,j)=\sum_{k=1}^p\sum_{l=1}^{p^c}\sigma_{i_kj_l}(C_s(i_k,j_l)-C_s(i_k,j_{p^c+l})-C_s(i_{p+k},j_l)+C_s(i_{p+k},j_{p^c+l})).
\end{equation}
If instead $p=|I|/2$ but $p^c\neq |I^c|/2$ (the reverse case is analogous), then we write
\begin{equation}
\sum_{\substack{i\in I\\ j\in I^c}}\sigma_{ij}C_s(i,j)=\sum_{k=1}^{p}\sum_{l=1}^{|I^c|}\sigma_{i_kj_l}(C_s(i_k,j_l)-C_s(i_{p+k},j_l)).
\end{equation}
If $p\neq|I|/2$ and $p^c\neq |I^c|/2$, we do not do any partial pairing, but use the estimate \Cref{eq:naive estimate for the sum of covariance kernels}.

We may always reorder the pairings of opposite charges so that we write $i_{p+\pi(k)}$ instead of $i_{p+k}$, where $\pi$ is a permutation on $I$. Moreover, we may choose $\pi$ to be the permutation that realizes the minimum in the Wasserstein distance $d(\xx_I)$ defined in \Cref{eq:def Wasserstein distance}. The same can be done for $I^c$. For simplicity of notation, and without loss of generality, we assume that the minimizing permutation is the identity.

By a simple mean value theorem estimate (this can be reduced to a situation similar to the ``second-order'' mean value theorem on a square $Q\subset \R^2$, treated for example in \cite[Theorem 9.40]{Ru76a}), we obtain
\begin{equation}
\begin{split}
|C_s(i_k,j_l)-C_s(i_k,j_{p^c+l})&-C_s(i_{p+k},j_l)+C_s(i_{p+k},j_{p^c+l})|
\\
&\leq \frac{C}{s}\frac{|x_{i_k}-x_{i_{p+k}}|}{\sqrt{s}}\frac{|x_{j_l}-x_{j_{p^c+l}}|}{\sqrt{s}}\sup_{\substack{u_k\in [x_{i_k},x_{i_{p+k}}]\\v_l\in[x_{j_l},x_{j_{p^c+l}}]}}f_d\rbr{\frac{|u_k-v_l|}{\sqrt{s}}},
\end{split}
\end{equation}
where $[a,b]$ denotes the line segment between $a,b\in\R^d$, and $C>0$ is a universal constant. We also have $|x_{i_k}-x_{i_{p+k}}||x_{j_l}-x_{j_{p^c+l}}|\leq d(\xx_I)d(\xx_{I^c})$. If $d(\xx_I)\leq \sqrt{s}$, then also $|x_{i_k}-x_{i_{p+k}}|\leq \sqrt{s}$, so up to a universal multiplicative constant we may replace $u_k$ by either endpoint of the segment $[x_{i_k},x_{i_{p+k}}]$. The same applies to $I^c$. Indeed, consider replacing $u_k$ and $v_l$ with $x_{i_k}$ and $x_{j_l}$
\begin{equation}
f_{d}\rbr{\frac{|u_k-v_l|}{\sqrt{s}}}= f_{d}\rbr{\frac{|u_k-x_{i_k}+x_{i_k}-(v_l-x_{j_l}+x_{j_l})|}{\sqrt{s}}}\leq f_{d}\rbr{\abs{\frac{|x_{i_k}-x_{j_l}|}{\sqrt{s}}-2}},
\end{equation}
where we have used the fact that in this situation $|u_k-x_{i_k}|\leq |x_{i_k}-x_{i_{p+k}}|\leq \sqrt{s}$ (similarly with $v_l,x_{j_l}$ and $x_{j_{p^c+l}}$), triangle inequality and monotonicity of $f_{d}$. Furthermore, there exists a constant $B$ (possibly dependent on $d$) such that $f_{d}(|x-2|)\leq B f_d(|x|)$.

We may symmetrize and use the above pairing estimate only on scales where $d(\xx_I),d(\xx_{I^c})\leq \sqrt{s}$. If only one of these conditions holds, we apply the pairing estimate only to that set. If both fail, we use the naive bound
\begin{equation}
\label{eq:naive estimate for the sum of covariance kernels}
G_s^I(\xx,\Si)\leq \frac{C}{s}\sum_{\substack{i\in I\\ j\in I^c}}f_d\rbr{\frac{|x_i-x_j|}{\sqrt{s}}}
\end{equation}
for some universal constant $C>0$.

Thus, choosing the largest constant $C$ among the possible cases, we may write, in the case $p=|I|/2$ and $p^c=|I^c|/2$,
\begin{equation}
\label{eq: scale dependent estimate for the sum of Cs the fully neutral case}
G_s^I(\xx,\Si)
\leq 
\begin{cases}
\frac{C}{s}\frac{d(\xx_I)}{\sqrt{s}}\frac{d(\xx_{I^c})}{\sqrt{s}}\sum_{\substack{i\in I\\ j\in I^c}}f_d\rbr{\frac{|x_i-x_j|}{\sqrt{s}}}, &\text{if } d(\xx_I),d(\xx_{I^c})<\sqrt{s} 
\\
\frac{C}{s}\frac{d(\xx_I)}{\sqrt{s}}\sum_{\substack{i\in I\\ j\in I^c}}f_d\rbr{\frac{|x_i-x_j|}{\sqrt{s}}},  &\text{if } d(\xx_{I^c})\geq \sqrt{s}, \text{but } d(\xx_I)<\sqrt{s}
\\
\frac{C}{s}\frac{d(\xx_{I^c})}{\sqrt{s}}\sum_{\substack{i\in I\\ j\in I^c}}f_d\rbr{\frac{|x_i-x_j|}{\sqrt{s}}}, &\text{if } d(\xx_{I})\geq \sqrt{s}, \text{but }d(\xx_{I^c})<\sqrt{s}
\\
\frac{C}{s}\sum_{\substack{i\in I\\ j\in I^c}}f_d\rbr{\frac{|x_i-x_j|}{\sqrt{s}}}, &\text{otherwise}.
\end{cases}
\end{equation} 

Similarly, if $p=|I|/2$ but $p^c\neq |I^c|/2$ (the opposite case is similar), we have
\begin{equation}
\label{eq: scale dependent estimate for the sum of Cs the partially neutral case}
G_s^I(\xx,\Si)
\leq
\begin{cases}
\frac{C}{s}\frac{d(\xx_I)}{\sqrt{s}}\sum_{\substack{i\in I\\ j\in I^c}}f_d\rbr{\frac{|x_i-x_j|}{\sqrt{s}}}, &\text{if } d(\xx_{I})<\sqrt{s}
\\
\frac{C}{s}\sum_{\substack{i\in I\\ j\in I^c}}f_d\rbr{\frac{|x_i-x_j|}{\sqrt{s}}}, &\text{otherwise}.
\end{cases}
\end{equation}

If $\sum_{i\in I}\sigma_i\neq 0\neq\sum_{j\in I^c}\sigma_j$, we use the estimate \Cref{eq:naive estimate for the sum of covariance kernels}. Furthermore, if $|I|\leq N$ and $|I^c|> N$ (the opposite case is similar) we also choose to do the pairing argument above only with respect to $I$ even though the configuration of $I^c$ would also be neutral. Similarly, if $|I|,|I^{c}|> N$ we may also choose to use the estimate \Cref{eq:naive estimate for the sum of covariance kernels}. Thus, we can rewrite our estimates in terms of the function $H_t^{n}(\xx,\Si)$.
\end{proof}

The other estimate we need is given in the lemma below. 
\begin{lemma}
\label{lem:estimate for the full sum of the massive heat kernels}
If $s>0$ and $n\leq N$, we have
\begin{equation}
\abs{\sum_{k,l=1}^n\sigma_{kl}C_s(k,l)}=\sum_{k,l=1}^n\sigma_{kl}C_s(k,l)\leq n^2\frac
{B}{s}H_s^n(\xx,\Si)=n^2\frac{B}{s}
\begin{cases}
\frac{d(\xx)}{\sqrt{s}},\quad \text{if } d(\xx)\leq \sqrt{s} \text{ and } \sum_{i=1}^n\sigma_i=0
\\
1,\qquad \text{otherwise}
\end{cases}
\end{equation}
for some universal constant $B>0$. The first equality holds since $C_s$ is a covariance function.

Furthermore, if in addition $\sum_{i=1}^n\sigma_i=0$ and $s$ is bounded from below, that is, $s\geq t>0$ for some fixed $t$ and $\xx\in \Lambda^n$ with compact $\Lambda\subset\R^d$, then we have
\begin{equation}
\sum_{k,l=1}^n\sigma_{kl}C_s(k,l)\leq n^2\frac
{C(t,\Lambda)}{s^{\frac{3}{2}}}H_s^n(\xx,\Si)
\end{equation}
with some finite constant $C(t,\Lambda)>0$. As before we use the shorthand notations in \Cref{eq:some notations for K and the heat kernel}.
\end{lemma} 

\begin{proof}
Firstly, if $\sum_{i=1}^n\sigma_i\neq 0$, the claim follows readily from the fact that the sum has $n^2$ terms and $C_s(\cdot,\cdot)\leq s^{-1}$ so that in this case we may pick $B=1$.

Then for the rest of the proof we assume that $\sum_{i=1}^n\sigma_i=0$. We claim that the terms can be grouped so that there are $n$ contributions of the form $C_s(0)-C_s(a,b_a)$  for $a=1,2,\dots,n$ with $\sigma_{a b_a}=-1$, and $n^2/2-n$ contributions of the form $\pm (C_s(a,b)-C_s(a,c))$ (or either of the pairs (a,b) and (a,c) in reversed order) with $\sigma_{ab}\neq\sigma_{ac}$. In addition, for each such choice we require $|x_a-x_{b_a}|\leq d(\xx)$ and $|x_b-x_c|\leq d(\xx)$.

Without loss of generality, we may again assume that $\sigma_i=+1$ for $i=1,2,\dots,n/2=p$ and $\sigma_i=-1$ for $i=p+1,\dots,2p$. Let $\tau\in \S_p$ be such that 
\begin{equation}
d(\xx)=\sum_{i=1}^p|x_i-x_{p+\tau(i)}|.
\end{equation}
The pairings can now be written explicitly. We write them for the case where $\tau$ is the identity permutation, since in that case it is easiest to see that all pairs appear exactly once. The general case is obtained by replacing $p+i$ with $p+\tau(i)$ (and similarly with $j$) everywhere. Note that here we do not mod out by symmetries, so for instance both $C_s(a,b)$ and $C_s(b,a)$ will both appear in the sums below. We have 
\begin{equation}
\begin{split}
\sum_{i,j=1}^n\sigma_{ij}C_s(i,j)&=\sum_{i=1}^{p}[C_s(0)-C_s(i+p,i)]+\sum_{i=1}^p[C_s(0)-C_s(i,i+p)]
\\
&\quad+\sum_{\substack{1\leq i,j\leq p\\i\neq j}}[C_s(i,j)-C_s(i,p+j)]
\\
&\quad+\sum_{1\leq i<j\leq p}[C_s(p+i,p+j)-C_s(p+i,j)+C_s(p+j,p+i)-C_s(p+j,i)]
\end{split}
\end{equation}
The pairs $(a,b_a)$ are identified in the first row as the pairs $(i,p+i)$ or $(p+i,i)$. The second and third row yield that the triplets $(a,b,c)$ are $(i,j,p+j)$, $i\neq j$, and $(p+i,p+j,j)$ or $(p+j,p+i,i)$ with $i<j$.
Now the first claim follows by the same mean value theorem argument used in the proof of \Cref{lem:estimate for GI function}, together with the bound $f_d\leq 1$.

For the second claim, we now apply the mean value theorem explicitly since we cannot use the general decay function $f_d$ here. We may write
\begin{equation}
\label{eq:translation invariant scale decomposition of the massive heat kernel}
C_s(x,y)=\frac{1}{4\pi s}F\rbr{\frac{x-y}{\sqrt{s}}}
\end{equation}
with $F(x)=e^{-\frac{1}{4} x^2}$. Since $F$ is smooth, the mean value theorem gives
\begin{equation}
\begin{split}
C_s(0)-C_s(a,b_a)&=\frac{1}{4\pi s}\rbr{F(0)-F\rbr{\frac{x_a-x_{b_a}}{\sqrt{s}}}}
= -\frac{1}{4\pi s}\frac{(x_a-x_{b_a})}{\sqrt{s}}\cdot (\grad F)(u),
\\
C_s(a,b)-C_s(a,c)&=\frac{1}{4\pi s}\rbr{F\rbr{\frac{x_a-x_b}{\sqrt{s}}}-F\rbr{\frac{x_a-x_c}{\sqrt{s}}}}
= -\frac{1}{4\pi s}\frac{(x_b-x_c)}{\sqrt{s}}\cdot (\grad F)(v),
\end{split}
\end{equation}
for some $u\in [0,(x_a-x_{b_a})/\sqrt{s}]$ and $v\in [(x_a-x_b)/\sqrt{s},(x_a-x_c)/\sqrt{s}]$.

Moreover,
\begin{equation}
(\grad F)(u)=-\frac{u}{2}F(u)=-\frac{\tilde{u}}{2\sqrt{s}}F\rbr{\frac{\tilde{u}}{\sqrt{s}}}
\end{equation}
for some $\tilde{u}\in[0,(x_a-x_{b_a})]$. For $d(\xx)<\sqrt{s}$ we estimate
\begin{equation}
|C_s(0)-C_s(x_a,x_{b_a})|\leq \frac{1}{s^{\frac{3}{2}}}\frac{d(\xx)}{\sqrt{s}}
\sup_{x_a,x_{b_a}\in\Lambda}\sup_{s\geq t}\sup_{\tilde{u}\in [0,(x_a-x_{b_a})]}
\frac{|\tilde{u}|}{8\pi}F\rbr{\frac{\tilde{u}}{\sqrt{s}}}.
\end{equation}
For $d(\xx)\geq \sqrt{s}$ we estimate
\begin{equation}
|C_s(0)-C_s(x_a,x_{b_a})|\leq \frac{1}{s^{\frac{3}{2}}}
\sup_{x_a,x_{b_a}\in\Lambda}\sup_{s\geq t}\sup_{\tilde{u}\in[0,(x_a-x_{b_a})]}
\frac{|x_a-x_{b_a}|}{\sqrt{s}}\frac{|\tilde{u}|}{8\pi}F\rbr{\frac{\tilde{u}}{\sqrt{s}}}.
\end{equation}
Both suprema above are finite, and the result follows by choosing $C(t,\Lambda)$ to be their maximum; the remaining estimate is treated in a similar vein. For $d=1$ the necessary modifications are straightforward (we simply use the ordinary derivative instead of the gradient).
\end{proof}

\end{subsubsection}

\begin{subsubsection}{Proof of {{\Cref{prop:induction statement}}} in the case $d\geq 2$}
\label{sec:the case d geq 2 in the proof os induction prop}
We will begin this section by defining the functions $h_t^n$. The definition is the same in all dimensions, including $d=1$. Thus, until further notice we keep $d\geq 1$ arbitrary. The functions $h_t^n$ will be defined recursively, or inductively (in the case $t\geq m^{-2}$ we do not even need the induction principle, since the number of terms is finite). After defining these functions we will prove \Cref{prop:induction statement} in the case $d\geq 2$. This case is considerably simpler than the case $d=1$, which is why we treat it first.

Now we turn to defining the functions $h_t^n$. The base case $n=1$ is obtained by simply choosing $h_t^1$ to be the right-hand side of \Cref{eq:estimate for V1} in \Cref{lem:uniform bound for V1}. Next we make the induction hypothesis that the claim holds for all $k<n$. Thus, for any $t>0$ we have
\begin{equation}
\label{eq:starting point for defining htn}
|\tilV_t^n(\xx,\Si|m,\eps)|\leq \frac{\beta}{2}\sum_{I\varsubsetneq[n]}\int_0^te^{-m^2s}G_s^I(\xx,\Si)h_s^{|I|}(\xx_I,\Si_I|m)h_s^{|I^c|}(\xx_{I^c},\Si_{I^c}|m)e^{-\frac{\beta}{2}\sum_{k,l=1}^n \sigma_{kl}K_{s,t}^m(k,l)}\ds,
\end{equation}
where we used the notation $G_s^I(\xx,\Si):=|\sum_{i\in I,\,j\in I^c}\sigma_{ij}C_s(i,j)|$ from \Cref{sec:estimates for sums of covariance kernels}. Also recall from \Cref{eq:some notations for K and the heat kernel} that $K_{s,t}^m$ is the integral of $C_\cdot^m$ from $s$ to $t$.

To make the functions $h_t^n$ independent of $m$ for $t< m^{-2}$ we will need one more estimate before defining them. Namely, we need a good estimate for the last exponential factor involving the terms $K_{s,t}^m$ above when $n\leq N+1$. This is provided in the lemma below.
\begin{lemma}
\label{lem:estimate for the exponential in tilVtn}
Let $n\leq N+1$ and $t<m^{-2}$. Then we have
\begin{equation}
\label{eq:estimate for the exponential in tilVtn}
\exp(-\frac{\beta}{2}\sum_{k,l=1}^n\sigma_{kl} K_{s,t}^m(k,l))\leq \bar{B}_N\rbr{\frac{s}{t}}^\frac{\beta}{8\pi}[H_t^n(\xx,\Si)]^{-\frac{\beta}{4\pi}},
\end{equation}
where 
\begin{equation}
\bar{B}_N:=\max_{n=1,2,\dots,N,N+1}\exp(\frac{n^2\beta}{8\pi}+\beta B_n).
\end{equation} 
and $B_n$ is the constant from \Cref{lemma:onsager type inequality with our parametrization}.
\end{lemma}

\begin{proof}
We may write
\begin{equation}
\exp(-\frac{\beta}{2}\sum_{k,l=1}^n\sigma_{kl} K_{s,t}^m(k,l))=\exp(-\frac{\beta}{2}\sum_{k,l=1}^n\sigma_{kl}[K_{s,t}^m(k,l)-K_{s,t}(k,l)])\exp(-\frac{\beta}{2}\sum_{k,l=1}^n\sigma_{kl}K_{s,t}(k,l)).
\end{equation}
We will apply the Onsager-type inequality from \Cref{lemma:onsager type inequality with our parametrization} to the second factor, but let us first bound the first factor. We have
\begin{equation}
\begin{split}
\exp(-\frac{\beta}{2}\sum_{k,l=1}^n\sigma_{kl}[K_{s,t}^m(k,l)-K_{s,t}(k,l)])&\leq\exp(\frac{\beta}{2}\sum_{k,l=1}^n\int_s^t|1-e^{-m^2r}|e^{-\frac{|x_k-x_l|^2}{4r}}\frac{\dr}{4\pi r})
\\
&\leq \exp(\frac{\beta n^2m^2}{8\pi}\int_0^t\dr)
\\
&\leq \exp(\frac{\beta n^2}{8\pi}),
\end{split}
\end{equation}
where the last step holds if $m^2t\leq 1$, that is, $t<m^{-2}$. The application of \Cref{lemma:onsager type inequality with our parametrization} and $K_{s,t}(0):=K_{s,t}(x,x)=1/(4\pi)\log(t/s)$ yields
\begin{equation}
\begin{split}
\exp(-\frac{\beta}{2}\sum_{k,l=1}^n\sigma_{kl}K_{s,t}(k,l))&\leq \exp(-\frac{n\beta}{8\pi}\log(\frac{t}{s})+\frac{\beta(n-1)}{4\pi}\log(\frac{\sqrt{t}}{\sqrt{s}})-\frac{\beta}{4\pi}\log(H_t^n(\xx,\Si))+\beta B_n)
\\
&=\rbr{\frac{s}{t}}^\frac{\beta}{8\pi}[H_t^n(\xx,\Si)]^{-\frac{\beta}{4\pi}}e^{\beta B_n}.
\end{split}
\end{equation}
\end{proof}

It might be surprising that we need this estimate for the first $N+1$ terms and not just the first $N$ (unless $N=0$, in which case the $(N+1)$th term is covered by \Cref{lem:uniform bound for V1}), but this is indeed the case.

Now we are ready to define the functions $h_t^n$. For $t<m^{-2}$ and $2\leq n\leq N+1$ we use \Cref{lem:estimate for the exponential in tilVtn} for the exponential involving $K_{s,t}^m$ terms in \Cref{eq:starting point for defining htn} and define
\begin{equation}
\label{eq:def of htn for n leq N+1 and t leq m^{-2}}
h_t^n(\xx,\Si):=\frac{\beta}{2}\frac{\bar{B}_N[H_t^n(\xx,\Si)]^{-\frac{\beta}{4\pi}}}{t^{\frac{\beta}{8\pi}}}\sum_{I\varsubsetneq[n]}\int_0^ts^{\frac{\beta}{8\pi}}G_s^I(\xx,\Si)h_s^{|I|}(\xx_I,\Si_I)h_s^{|I^c|}(\xx_{I^c},\Si_{I^c})\ds.
\end{equation}
For $n>N+1$ we define
\begin{equation}
\label{eq:def of htn for n > N+1 and t leq m^{-2}}
h_t^n(\xx,\Si):=\frac{\beta}{2}\sum_{I\varsubsetneq[n]}\int_0^tG_s^I(\xx,\Si)h_s^{|I|}(\xx_I,\Si_I)h_s^{|I^c|}(\xx_{I^c},\Si_{I^c})\ds.
\end{equation}

Finally, for $t\geq m^{-2}$ we only need to define the functions $h_t^n$ for $n\leq N$ and $d=1$. We will split the $s$-integral at $m^{-2}$. For the part where $s\geq m^{-2}$ we bound the exponential involving the terms $K_{s,t}^m$ by $1$. For the part where $s<m^{-2}$ we use the following estimate
\begin{equation}
\label{eq:estimate for the exponential in htn for large t}
\begin{split}
\exp(-\frac{\beta}{2}\sum_{k,l=1}^n\sigma_{kl}K_{s,t}^m(k,l))&= \exp(-\frac{\beta}{2}\sum_{k,l=1}^n\sigma_{kl}K_{s,m^{-2}}^m(k,l))\underbrace{\exp(-\frac{\beta}{2}\sum_{k,l=1}^n\sigma_{kl}K_{m^{-2},t}^m(k,l))}_{\leq 1} 
\\
&\leq \bar{B}_N[H_{m^{-2}}^n(\xx,\Si)]^{-\frac{\beta}{4\pi}}m^{\frac{\beta}{4\pi}}s^{\frac{\beta}{8\pi}},
\end{split}
\end{equation}
where we estimated the first exponential factor in the same way as in the case $t<m^{-2}$, that is, we used \Cref{lem:estimate for the exponential in tilVtn}. Thus, for $t\geq m^{-2}$ define
\begin{equation}
\begin{split}
\label{eq:def of htn for n leq N and t geq m^{-2}}
h_t^n(\xx,\Si|m)&:= \frac{\beta}{2}\bar{B}_N[H_{m^{-2}}^n(\xx,\Si)]^{-\frac{\beta}{4\pi}}m^{\frac{\beta}{4\pi}}\sum_{I\varsubsetneq[n]}\int_0^{m^{-2}}s^{\frac{\beta}{8\pi}}G_s^I(\xx,\Si)h_s^{|I|}(\xx_I,\Si_I)h_s^{|I^c|}(\xx_{I^c},\Si_{I^c})\ds
\\
&\quad +\frac{\beta}{2}\sum_{I\varsubsetneq[n]}\int_{m^{-2}}^te^{-m^2s}G_s^I(\xx,\Si)h_s^{|I|}(\xx_I,\Si_I|m)h_s^{|I^c|}(\xx_{I^c},\Si_{I^c})\ds
\end{split}
\end{equation}
The permutation invariance in all cases is clear, because we sum over all ordered bipartite partitions of $[n]$. To obtain the norm estimates we need to treat the different values of $t$ and $n$ separately. The case $t\geq m^{-2}$ is needed only in the case $d=1$ as discussed in \Cref{rem:the t geq m2 estimates are not needed for d geq 2}.

For the rest of this subsection, we always have $d\geq 2$, $t<m^{-2}$ and $N=0$ or $N=2$. However, the first part, that is, the case $n>N+1$ is similar also in $d=1$ so we keep the notation $N$ arbitrary for this part. This case is also sufficient in all dimensions when $N=0$, that is, when $\beta\in (0,2\pi d)$, since as mentioned above the case $N+1=1$ is then covered by \Cref{lem:uniform bound for V1}.

In the case $n>N+1$, we obtain 
\begin{equation}
\label{eq:First estimate for the n>N case in 1d}
\begin{split}
\norm{H_t^n(\cdot,\Si)h_t^n(\cdot,\Si)}_n &\leq \esssup_{x_1\in\R^d}\int_{\R^{(n-1)d}}\frac{\beta}{2}\sum_{I\varsubsetneq [n]}\int_0^tG_s^I(\xx,\Si)h_s^{|I|}(\xx_I,\Si_I)h_s^{|I^c|}(\xx_{I^c},\Si_{I^c})\ds\d\xx_{n-1}
\\
&\leq\frac{\beta\tilde{C}}{2}\sum_{I\varsubsetneq [n]}\sum_{\substack{i\in I \\ j\in I^c}}\int_0^t\frac{1}{s}\bigg\{\esssup_{x_1\in\R^d}\int_{\R^{(n-1)d}}\bigg[f_d\rbr{\frac{|x_i-x_j|}{\sqrt{s}}}
\\
&\qquad\qquad\qquad\qquad\times H_s^{|I|}(\xx_I,\Si_I)h_s^{|I|}(\xx_I,\Si_I)H_s^{|I^c|}(\xx_{I^c},\Si_{I^c})h_s^{|I^c|}(\xx_{I^c},\Si_{I^c})\bigg]\d\xx_{n-1}\bigg\}\ds,
\end{split}
\end{equation}  
where we immediately used the fact that $H_t^n\equiv 1$ in this case. We also used \Cref{lem:estimate for GI function}, and denoted the constant from it by $\tilde{C}$. Henceforth, we will always denote by $\tilde{C}$ this constant in all similar estimates. 

Next, consider only the part inside the $\{\}$-brackets and fix $I\varsubsetneq[n]$, $i\in I$ and $j\in I^c$. Without loss of generality, assume that $1\in I$. We have
\begin{equation}
\begin{split}
\label{eq:splitting of the n norm}
\esssup_{x_1\in\R^d}&\int_{\R^{(n-1)d}}f_d\rbr{\frac{|x_i-x_j|}{\sqrt{s}}}H_s^{|I|}(\xx_I,\Si_I)h_s^{|I|}(\xx_I,\Si_I)H_s^{|I^c|}(\xx_{I^c},\Si_{I^c})h_s^{|I^c|}(\xx_{I^c},\Si_{I^c})\d\xx_{n-1}
\\
&\leq\esssup_{x_1\in\R^d}\int_{\R^{2d}}\1_{i\neq1}f_d\rbr{\frac{|x_i-x_j|}{\sqrt{s}}}\dx_j\rbr{\int_{\R^{d(|I|-2)}}H_s^{|I|}(\xx_I,\Si_I)h_s^{|I|}(\xx_I,\Si_I)\prod_{\substack{k\in I \\ i\neq k\neq 1}}\dx_k} \dx_i
\\
&\qquad\qquad\qquad\qquad\qquad\qquad\qquad\qquad \times \esssup_{x_j\in\R^d}\int_{\R^{d(|I^c|-1)}}H_s^{|I^c|}(\xx_{I^c},\Si_{I^c})h_s^{|I^c|}(\xx_{I^c},\Si_{I^c})\prod_{\substack{l\in I^c \\ j\neq l}}\dx_l
\\
&\quad+\esssup_{x_1\in\R^d}\int_{\R^d}\1_{i=1}f_d\rbr{\frac{|x_i-x_j|}{\sqrt{s}}}\dx_j\rbr{\int_{\R^{d(|I|-1)}}H_s^{|I|}(\xx_I,\Si_I)h_s^{|I|}(\xx_I,\Si_I)\prod_{\substack{k\in I \\ i\neq k\neq 1}}\dx_k} 
\\
&\qquad\qquad\qquad\qquad\qquad\qquad\qquad\qquad \times \esssup_{x_j\in\R^d}\int_{\R^{d(|I^c|-1)}}H_s^{|I^c|}(\xx_{I^c},\Si_{I^c})h_s^{|I^c|}(\xx_{I^c},\Si_{I^c})\prod_{\substack{l\in I^c \\ j\neq l}}\dx_l.
\end{split}
\end{equation}
The remaining $x_j$ integral yields, by simple scaling and translation, a factor proportional to $s^{d/2}$. Thus, we obtain the upper bound
\begin{equation}
\label{eq: I and Ic norms}
C_ds^{\frac{d}{2}}\norm{H_s^{|I|}(\cdot,\Si_I)h_s^{|I|}(\cdot,\Si_I)}_{|I|}\norm{H_s^{|I^c|}(\cdot,\Si_{I^c})h_s^{|I^c|}(\cdot,\Si_{I^c})}_{|I^c|},
\end{equation}
with some constant $C_d>0$ depending only on the dimension $d$. Henceforth, in all similar estimates we will denote by $C_ds^{\frac{d}{2}}$ the contribution coming from the remaining integral when splitting the norm $\norm{\cdot}_n$ into $\norm{\cdot}_{|I|}$ and $\norm{\cdot}_{|I^c|}$ as above.

Using this and the induction hypothesis we obtain
\begin{equation}
\begin{split}
\norm{H_t^n(\cdot,\Si)h_t^n(\cdot,\Si)}_n&\leq \half \beta \tilde{C}C_d[B_d(\beta,N)]^{n-2}C^n\sum_{I\varsubsetneq [n]}\sum_{\substack{i\in I\\ j\in I^c}}|I|^{|I|-2}|I^c|^{|I^c|-2}
\\
&\quad \times \int_0^ts^{\frac{d}{2}-1}s^{-\frac{d}{2}}\rbr{s^{\frac{d}{2}-\frac{\beta}{8\pi}}}^{|I|}s^{-\frac{d}{2}}\rbr{s^{\frac{d}{2}-\frac{\beta}{8\pi}}}^{|I^c|}\ds.
\\
\end{split}
\end{equation}
The $s$-integral becomes
\begin{equation}
\int_0^ts^{-1-\frac{d}{2}+n\rbr{\frac{d}{2}-\frac{\beta}{8\pi}}}\ds=\frac{1}{n-1}\frac{1}{\frac{d}{2}-\frac{n}{n-1}\frac{\beta}{8\pi}}t^{-\frac{d}{2}}\rbr{t^{\frac{d}{2}-\frac{\beta}{8\pi}}}^n
\end{equation}
if 
\begin{equation}
-1-\frac{d}{2}+n\rbr{\frac{d}{2}-\frac{\beta}{8\pi}}>-1\Leftrightarrow\beta<\rbr{1-\frac{1}{n}}4d\pi
\end{equation}
for all $n> N+1$. The right-hand side is minimized by $n=N+2$, and since we assumed that $\beta\in [\beta_N,\beta_{N+2})$ if $d=1$ and $\beta\in[d2\pi,(d+1)2\pi)$ for $d\geq 2$ and $(d+1)2\pi\leq \beta_4$ for $d\geq 2$, this condition is indeed satisfied. Note that here the choice not to use \Cref{lem:estimate for the exponential in tilVtn} for the $n=N+1$ term would fail. 

Putting everything together yields
\begin{equation}
\begin{split}
\label{eq:last estimate for the case n leq N+1}
\norm{H_t^n(\cdot,\Si)h_t^n(\cdot,\Si)}_n&\leq \frac{\beta \tilde{C}C_d}{2(n-1)(\frac{d}{2}-
\frac{n}{n-1}\frac{\beta}{8\pi})}[B_d(\beta,N)]^{n-2}C^nt^{-\frac{d}{2}}\rbr{t^{\frac{d}{2}-\frac{\beta}{8\pi}}}^n\sum_{I\varsubsetneq [n]}|I|^{|I|-1}|I^c|^{|I^c|-1}
\\
&\leq
\beta\tilde{C} C_d\rbr{\sup_{n> N+1}{\frac{1}{\frac{d}{2}-\frac{n}{n-1}\frac{\beta}{8\pi}}}}[B_d(\beta,N)]^{n-2}C^nt^{-\frac{d}{2}}\rbr{t^{\frac{d}{2}-\frac{\beta}{8\pi}}}^n n^{n-2},
\end{split}
\end{equation}
where in the last step we used the combinatorial identity (see \cite[Lemma 4.2]{BrKe87a})
\begin{equation}
\label{eq:combinatorial sum giving n to the power n minu 2}
\sum_{I\varsubsetneq [n]}|I|^{|I|-1}|I^c|^{|I^c|-1}=\sum_{k=1}^{n-1}\binom{n}{k}k^{k-1}(n-k)^{n-k-1}=2(n-1)n^{n-2}.
\end{equation}
Estimate \Cref{eq:last estimate for the case n leq N+1} is of the correct form if $B_d(\beta,N)$ is chosen appropriately. Thus, choosing
\begin{equation}
B_d'(\beta,N)=\beta \tilde{C}C_d\sup_{n> N+1}\frac{1}{\frac{d}{2}-\frac{n}{n-1}\frac{\beta}{8\pi}}
\end{equation}
yields the desired result in the case $n>N+1$.

\begin{remark}
For the case $N=0$ the above analysis indeed yields the bound $\beta<2d\pi$. Together with our later analysis this shows that the first transition or collapse point scales with the dimension $d$ in the same way as the scaling from $d=1$ to $d=2$, and this persists for all $d>2$.
\end{remark}

\begin{remark}
We will make this type of inductive/recursive argument several times in \Cref{sec: the renormalized potential}, and the splitting of the spatial integrals at the inductive step together with the combinatorial part will always follow the same pattern. Thus, from now on we keep track only of the constants, the $t$ or mass $m$ dependence, and the convergence of the $s$-integrals.
\end{remark}

If $N=2$ we still need to consider the terms with $n=2$ or $n=3$ explicitly. First consider the $n=2$ term. Proceeding as above, and using symmetry together with \Cref{eq:estimate for V1} for the $\tilV_t^1$ terms, we have
\begin{equation}
\norm{H_t^2(\xx,\Si)h_t^2(\xx,\Si)}_2\leq \frac{\beta\bar{B}_2\tilde{C}}{t^{\frac{\beta}{8\pi}}}\int_0^ts^{\frac{\beta}{8\pi}-1}(Cs^{-\frac{\beta}{8\pi}})^2\esssup_{x_1\in \R^d} \int_{\R^d}f_d\rbr{\frac{|x_1-x_2|}{\sqrt{s}}}(H_t^2(\xx,\Si))^{1-\frac{\beta}{4\pi}}\dx_2\ds.
\end{equation}
For $\sigma_1= \sigma_2$ the above is bounded by
\begin{equation}
\beta \bar{B}_2\tilde{C}C_d C^2t^{-\frac{\beta}{8\pi}}\int_0^ts^{\frac{d}{2}-\frac{\beta}{8\pi}-1}\ds=\frac{\beta \bar{B}_2\tilde{C}C_d C^2}{\frac{d}{2}-\frac{\beta}{8\pi}}t^{-\frac{d}{2}}\rbr{t^{\frac{d}{2}-\frac{\beta}{8\pi}}}^2
\end{equation}
for all $\beta< 4d\pi$ and some constant $C_d>0$. This is of the expected form, noting that $n^{n-2}=1$ for $n=2$. 

Now let $\sigma_1\neq\sigma_2$. In this case we need to compute the $x_2$-integral (note that for $n=2$, $d(\xx)=|x_1-x_2|$):
\begin{equation}
\begin{split}
\label{eq:Spatial integral for n=2 case in higher dimensions}
\int_{\R^d}f_d\rbr{\frac{|x_1-x_2|}{\sqrt{s}}}\rbr{\frac{|x_1-x_2|}{\sqrt{t}}\wedge 1}^{1-\frac{\beta}{4\pi}}\dx_2&=t^{-\frac{1}{2}+\frac{\beta}{8\pi}}\int_{|x|<\sqrt{t}}f_d\rbr{\frac{|x|}{\sqrt{s}}}|x|^{1-\frac{\beta}{4\pi}}\dx+\int_{|x|\geq\sqrt{t}}f_d\rbr{\frac{|x|}{\sqrt{s}}}\dx
\\
&\leq \rbr{\frac{s}{t}}^{\half-\frac{\beta}{8\pi}}s^{\frac{d}{2}}\int_{|x|\leq \sqrt{t/s}}f_d(|x|)|x|^{1-\frac{\beta}{4\pi}}\dx+C_ds^{\frac{d}{2}}
\\
&\leq C_ds^{\frac{d}{2}}\rbr{\rbr{\frac{s}{t}}^{\half-\frac{\beta}{8\pi}}+1},
\end{split} 
\end{equation} 
where we first extended the region of integration of the second integral on the right-hand side of the first line to the whole space and then scaled to obtain the upper bound $s^{d/2}C_d$ for some constant $C_d>0$. As $s$ varies in $(0,t)$, the quantity $\sqrt{t/s}$ ranges from $1$ to $\infty$, and the last integral above is finite in the neighborhood of $\infty$ for all $\beta>0$ and in the neighborhood of the origin for all $\beta<(d+1)4\pi$. Therefore, we have bounded it by a constant, and again denoted the maximum of the constants above by $C_d$.

The contribution from the second term $C_ds^{d/2}$ on the last row of \Cref{eq:Spatial integral for n=2 case in higher dimensions} is treated exactly analogously to the whole contribution from the case $\sigma_1=\sigma_2$. Putting everything together for the first term, we obtain in this case
\begin{equation}
\beta\bar{B}_2\tilde{C}C_dC^2t^{-\half}\int_0^ts^{-1-\frac{\beta}{8\pi}+\frac{d}{2}+\half-\frac{\beta}{8\pi}}\ds=\frac{\beta\bar{B}_2\tilde{C}C_d}{\frac{d+1}{2}-\frac{\beta}{4\pi}}C^2t^{- \frac{d}{2}}\rbr{t^{\frac{d}{2}-\frac{\beta}{8\pi}}}^2
\end{equation}
if $-1/2+d/2-\beta/(4\pi)>-1$, that is, $\beta<(d+1)2\pi$. 
Thus, the $n=2$ term satisfies the claim by choosing (note that the sum is always greater than either of the terms individually for appropriate values of $\beta$)
\begin{equation}
B_d(\beta,2)=\beta\bar{B}_2\tilde{C}C_d\rbr{\frac{1}{\frac{d}{2}-\frac{\beta}{8\pi}}+\frac{1}{\frac{d+1}{2}-\frac{\beta}{4\pi}}}.
\end{equation}
   
Next consider the $n=3$ term. Below we denote $I_i:=[3]\setminus\{i\}\equiv\{k,l\}$ and use the symmetry of the kernels $C_s$ everywhere to simplify the sums. By using the definition \Cref{eq:def of htn for n leq N+1 and t leq m^{-2}} first for $h_t^3$ itself and then for $h_s^2$, and using \Cref{lem:estimate for GI function}, \Cref{eq:estimate for V1} and the fact that $H_t^3\equiv 1$, we obtain
\begin{equation}
\begin{split}
\norm{H_t^3(\cdot,\Si)h_t^3(\cdot,\Si)}_3&\leq \esssup_{x_1\in \R^d}\int_{\R^{2d}}\bigg\{\beta\sum_{i=1}^3\sum_{j\in I_i\equiv\{k,l\}}\int_0^t\bigg[ \tilde{C}H_s^2(\xx_{I_i},\Si_{I_i})\frac{1}{s}f_d\rbr{\frac{|x_i-x_j|}{\sqrt{s}}}Cs^{-\frac{\beta}{8\pi}}\bar{B}_2\rbr{\frac{s}{t}}^{\frac{\beta}{8\pi}}
\\
&\qquad\beta\int_0^s\frac{\tilde{C}}{r}f_d\rbr{\frac{|x_k-x_l|}{\sqrt{r}}}\bar{B}_2[H_s^2(\xx_{I_i},\Si_{I_i})]^{-\frac{\beta}{4\pi}}(Cr^{-\frac{\beta}{8\pi}})^2\rbr{\frac{r}{s}}^{\frac{\beta}{8\pi}}\dr\bigg]\ds\bigg\}\dx_2\dx_3
\\
&\leq \beta^2\bar{B}_2^2\tilde{C}^2C^3t^{-\frac{\beta}{8\pi}}\sum_{i=1}^3\sum_{j\in \{k,l\}}\int_0^t s^{-1-\frac{\beta}{8\pi}}\bigg\{\int_0^s r^{-1-\frac{\beta}{8\pi}}\bigg[
\\
&\qquad \times \esssup_{x_1\in \R^d}\int_{\R^{2d}}f_d\rbr{\frac{|x_i-x_j|}{\sqrt{s}}}f_d\rbr{\frac{|x_k-x_l|}{\sqrt{r}}}[H_s^2(\xx_{I_i},\Si_{I_i})]^{1-\frac{\beta}{4\pi}}\dx_2\dx_3\bigg]\dr\bigg\}\ds.
\end{split}
\end{equation}

If $\sigma_k=\sigma_l$, the spatial integrals decouple (say for $j=k$) by the changes of variables $x_i-x_k=x$ and $x_k-x_l=y$. Scaling then yields the upper bound
\begin{equation}
\begin{split}
6\beta^2\bar{B}_2^2\tilde{C}^2C_d^2C^3t^{-\frac{\beta}{8\pi}}\int_0^t s^{\frac{d}{2}-1-\frac{\beta}{8\pi}}\int_0^s r^{\frac{d}{2}-1-\frac{\beta}{8\pi}}\dr\ds&=\frac{6\beta^2\bar{B}_2^2\tilde{C}^2C_d^2}{\frac{d}{2}-\frac{\beta}{8\pi}}C^3t^{-\frac{\beta}{8\pi}}\int_0^t s^{d-1-\frac{\beta}{4\pi}}\ds
\\
&=3\rbr{\frac{\beta\bar{B}_2\tilde{C}C_d}{\frac{d}{2}-\frac{\beta}{8\pi}}}^2C^3t^{-\frac{d}{2}}\rbr{t^{\frac{d}{2}-\frac{\beta}{8\pi}}}^3,
\end{split}
\end{equation}
which is again of the correct form, since $n^{n-2}=3$ for $n=3$.  

For the case $\sigma_k\neq\sigma_l$ we obtain the spatial integral (after the decoupling translations)
\begin{equation}
\int_{\R^d}f_d\rbr{\frac{|x|}{\sqrt{s}}}\dx\int_{\R^d}f_d\rbr{\frac{|y|}{\sqrt{r}}}\rbr{\frac{|y|}{\sqrt{s}}\wedge 1}^{1-\frac{\beta}{4\pi}}\dy\leq s^{\frac{d}{2}}r^{\frac{d}{2}}C_d^2\rbr{\abr{\frac{r}{s}}^{\half-\frac{\beta}{8\pi}}+1},
\end{equation}
where we used the estimate \Cref{eq:Spatial integral for n=2 case in higher dimensions} from the $n=2$ case above. The $+1$ term on the right-hand side yields the same contribution as in the $\sigma_k=\sigma_l$ case. Thus, it remains to check the $t$-behavior and convergence of the $r$- and $s$-integrals for the first term above. We have
\begin{equation}
t^{-\frac{\beta}{8\pi}}\int_0^t s^{\frac{d}{2}-\frac{3}{2}}\int_0^s r^{\frac{d}{2}-\half-\frac{\beta}{4\pi}}\dr\ds=\frac{t^{-\frac{d}{2}}\rbr{t^{\frac{d}{2}-\frac{\beta}{8\pi}}}^3}{2\rbr{\frac{d}{2}-\frac{\beta}{8\pi}}\rbr{\frac{d+1}{2}-\frac{\beta}{4\pi}}}
\end{equation}
again for $\beta<(d+1)2\pi$. Putting these together, we are done for the case $d\geq 2$ if we choose
\begin{equation}
\label{eq:def of the constant BdbetaN for d larger than 2 and small t}
B_d(\beta,N):=\beta\tilde{C}C_d\max\rbr{\bar{B}_2\rbr{\frac{1}{\frac{d}{2}-\frac{\beta}{8\pi}}+\frac{1}{\frac{d+1}{2}-\frac{\beta}{4\pi}}},\sup_{n>3}\frac{1}{\frac{d}{2}-\frac{n}{n-1}\frac{\beta}{8\pi}}}.
\end{equation}

\end{subsubsection}

\begin{subsubsection}{Proof of {{\Cref{prop:induction statement}}} in the case $d=1$}
\label{sec:the case d=1 in the proof os induction prop}
We begin this section by proving the missing cases for estimates of the functions $h_t^n$ for $d=1$. This includes the cases $t<m^{-2}$ with $n\leq N+1$ and $t\geq m^{-2}$ with $n\leq N$. We consider the case $n\leq N+1$ and $t<m^{-2}$ first. Proceeding as in the case $n> N+1$ we obtain by \Cref{lem:estimate for GI function}
\begin{equation}
\begin{split}
\label{eq:first estimate for the case n<N+1 and t<m^{-2}}
\norm{H_t^n(\cdot,\Si)h_t^n(\cdot,\Si)}_n&\leq\frac{\beta\bar{B}_N\tilde{C}}{2t^\frac{\beta}{8\pi}}\sum_{I\varsubsetneq [n]}\sum_{\substack{i\in I\\j\in I^c}}\int_0^t\bigg\{
\\
&\quad\times \esssup_{x_1\in\R}\int_{\R^{n-1}}\bigg[s^{\frac{\beta}{8\pi}-1}f_1\rbr{\frac{|x_i-x_j|}{\sqrt{s}}}\rbr{H_t^n(\xx,\Si)}^{1-\frac{\beta}{4\pi}}
\\
&\qquad\qquad\qquad\times H_s^{|I|}(\xx_I,\Si_I)h_s^{|I|}(\xx_I,\Si_I)H_s^{|I^c|}(\xx_{I^c},\Si_{I^c})h_s^{|I^c|}(\xx_{I^c},\Si_{I^c})\bigg]\d\xx_{n-1}\bigg\}\ds
\end{split}
\end{equation}
Since we are working with $d=1$ and $\beta<4\pi$, we have $(H_t^n(\xx,\Si))^{1-\beta/(4\pi)}\leq 1$. Therefore, repeating the analysis in the case $n>N+1$ starting from \Cref{eq:splitting of the n norm} in the previous section we obtain 
\begin{equation}
\begin{split}
\norm{H_t^n(\cdot,\Si)h_t^n(\cdot,\Si)}_n&\leq \beta \bar{B}_N \tilde{C}C_1[B_1(\beta,N)]^{n-2} C^n(n-1)n^{n-2} t^{-\frac{\beta}{8\pi}}\int_0^ts^{\frac{\beta}{8\pi}-\frac{3}{2}}\rbr{s^{\half-\frac{\beta}{8\pi}}}^n\ds
\\ 
&\leq \frac{\beta \bar{B}_N \tilde{C}C_1}{\half-\frac{\beta}{8\pi}}[B_1(\beta,N)]^{n-2}C^nn^{n-2}t^{-\half}\rbr{t^{\half-\frac{\beta}{8\pi}}}^n,
\end{split}
\end{equation}
provided that $-1+(n-1)(1/2-\beta/8\pi)>-1$, that is, $\beta<4\pi$.

Thus, this part satisfies the claim with the choice $B_1(\beta,N)=\bar{B}_NC_1\tilde{C}\beta (\half-\beta/(8\pi))^{-1}$. The proof of the estimates for the functions $h_t^n$ in the case $t<m^{-2}$ for $d=1$ is completed by choosing 
\begin{equation}
\label{eq:def of constant BbetaN for d=1 and small t}
B_1(\beta,N)=B_1'(\beta,N):=\beta \tilde{C}C_1\max \rbr{\bar{B}_N\rbr{\frac{1}{2}-\frac{\beta}{8\pi}}^{-1}, \sup_{n> N+1}\rbr{\frac{1}{2}-\frac{n}{n-1}\frac{\beta}{8\pi}}^{-1}}.
\end{equation} 

Consider then the case $t\geq m^{-2}$. First we note that for $s\leq r$ we have $H_s^n(\xx,\Si)\geq H_r^n(\xx,\Si)$ for all $n$. Therefore, we have 
\begin{equation}
\label{eq:combining the Hn functions}
H_t^n(\xx,\Si)[H_{m^{-2}}^n(\xx,\Si)]^{-\frac{\beta}{4\pi}}\leq [H_{m^{-2}}^n(\xx,\Si)]^{1-\frac{\beta}{4\pi}}\leq 1
\end{equation}
since $\beta<4\pi$ for $d=1$. We will then proceed in a similar vein to the proof of the $t<m^{-2}$ case.

Then we use the definition \Cref{eq:def of htn for n leq N and t geq m^{-2}}, and proceed as in the case $t<m^{-2}$ using also \Cref{eq:combining the Hn functions}. More precisely, we first use \Cref{lem:estimate for GI function} and proceed analogously to \Cref{eq:First estimate for the n>N case in 1d} for both parts ($s\in (0,m^{-2})$ and $s\in (m^{-2},t)$). After this, splitting the norms as in \Cref{eq:splitting of the n norm}, using the induction hypothesis and the combinatorial identity \Cref{eq:combinatorial sum giving n to the power n minu 2} yields
\begin{equation}
\begin{split}
\label{eq:first estimate in the induction for large t and $d=1$}
\norm{H_t^n(\cdot,\Si)h_t^n(\cdot,\Si)}_{n}&\leq \beta\tilde{C}C_1[B_1(\beta,N)]^{n-2}C^n(n-1)n^{n-2}\bigg(\bar{B}_Nm^{\frac{\beta}{4\pi}}\int_0^{m^{-2}}s^{\frac{\beta}{8\pi}-\frac{3}{2}+n\rbr{\half-\frac{\beta}{8\pi}}}\ds 
\\
&\qquad\qquad\qquad\qquad\qquad\qquad\qquad\qquad\qquad
+(m^{-2})^{-1+n(\half-\frac{\beta}{8\pi})}\int_{m^{-2}}^t\frac{e^{-m^2s}}{\sqrt{s}}\ds\bigg)
\\
&\leq \beta\tilde{C}C_1\rbr{\frac{\bar{B}_N}{\half-\frac{\beta}{8\pi}}+(N-1)\Gamma\rbr{\half,1}}[B_1(\beta,N)]^{n-2}C^nn^{n-2}(m^{-2})^{-\half+n(\half-\frac{\beta}{8\pi})},
\end{split}
\end{equation}
where $\Gamma(z,a):=\int_a^\infty s^{z-1}e^{-s}\ds$ is the incomplete gamma function. 

Thus, we choose
\begin{equation}
\label{eq:def of B1betaN}
B_1(\beta,N)=\max\rbr{B_1'(\beta,N),\beta\tilde{C}C_1\abr{\frac{\bar{B}_N}{\half-\frac{\beta}{8\pi}}+(N-1)\Gamma\rbr{\half,1}}},
\end{equation}
where the constant $B_1'(\beta,N)$ came from the analysis of the $t<m^{-2}$ case and was defined in \Cref{eq:def of constant BbetaN for d=1 and small t}. This completes the proof of the $d=1$ case for all $t>0$ and thus the entire proof for the estimates of the $h_t^n$ functions. 

We will utilize the existence of the functions $h_t^n$ and proceed with the functions $g_t^n$ again inductively. The case $n=1$ is clear by \Cref{lem:uniform bound for V1}. 
Indeed,
\begin{equation}
\label{eq:mass derivative of Vt1}
\partial_m\tilV_t^1(m,\eps)=\frac{\beta m}{4\pi}\underbrace{\rbr{\int_{\eps^2}^te^{-m^2s}\ds}}_{\leq ( t \wedge m^{-2})}\tilV_t^1(m,\eps)\leq C m(t\wedge m^{-2})^{1-\frac{\beta}{8\pi}},
\end{equation}
for some universal constant $C>0$. Now assume that $\tilV_t^k$ is differentiable, and that the functions $g_t^k$ with the claimed properties exist for all $k<n$ for some $n\in\N$. We then formally differentiate through the $s$-integral in the definition \Cref{eq:def of Vn} of the functions $\tilV_t^n$ (the proof of \Cref{lem:differentiability of Vtn with respect to mass ETC} will justify this) to obtain
\begin{equation}
\begin{split}
\label{eq:mass derivative of Vtn}
\frac{\partial}{\partial m}\tilV_t^n(\xx,\Si|m,\eps)&=\beta m\sum_{I\varsubsetneq [n]}\int_{\eps^2}^t\bigg\{\rbr{-s+\frac{\beta}{2}\int_s^tr e^{-m^2r}\sum_{k,l=1}^n\sigma_{kl}C_r(k,l)\dr}e^{-m^2s}\bigg(\sum_{\substack{i\in I\\j\in I^c}}\sigma_{ij}C_s(i,j)\bigg)
\\
&\qquad\qquad\qquad\times  \tilV_s^{|I|}(\xx_I,\Si_I|m,\eps)\tilV_s^{|I^c|}(\xx_{I^c},\Si_{I^c}|m,\eps)e^{-\frac{\beta}{2}\sum_{k,l=1}^n\sigma_{kl}K_{s,t}^m(k,l)}\bigg\}\ds
\\
&\quad + \frac{\beta}{2}\sum_{I\varsubsetneq [n]}\int_{\eps^2}^t\bigg\{e^{-m^2s}\bigg(\sum_{\substack{i\in I\\j\in I^c}}\sigma_{ij}C_s(i,j)\bigg)e^{-\frac{\beta}{2}\sum_{k,l=1}^n\sigma_{kl}K_{s,t}^m(k,l)}
\\
&\qquad\qquad\qquad\times  \big([\partial_m\tilV_s^{|I|}(\xx_I,\Si_I|m,\eps)]\tilV_s^{|I^c|}(\xx_{I^c},\Si_{I^c}|m,\eps)
\\
&\qquad\qquad\qquad\quad+\tilV_s^{|I|}(\xx_I,\Si_I|m,\eps)[\partial_m\tilV_s^{|I^c|}(\xx_{I^c},\Si_{I^c}|m,\eps)]\big)\bigg\}\ds,
\end{split}
\end{equation}
where we used the shorthand $\partial_m:=\partial/\partial m$ on the right-hand side.

The most delicate term is the $r$-integral appearing inside the brackets of the first term. We begin by estimating this contribution. For $t<m^{-2}$ we have by the first estimate of \Cref{lem:estimate for the full sum of the massive heat kernels} and \Cref{lem:estimate for the exponential in tilVtn}
\begin{equation}
\begin{split}
\label{eq:estimate for the integral coming from the mass derivative of the integral for small t}
\bigg(\int_s^tre^{-m^2r}\sum_{k,l=1}^n\sigma_{kl}C_r(k,l)\dr\bigg)  e^{-\frac{\beta}{2}\sum_{k,l=1}^n\sigma_{kl}K_{s,t}^m(k,l)}&\leq\int_s^tr \rbr{\sum_{k,l=1}^n\sigma_{kl}C_r(k,l)}e^{-\frac{\beta}{2}\sum_{k,l=1}^n\sigma_{kl}K_{s,r}^m(k,l)}\dr
\\
&\leq\tilde{B}\bar{B}_N n^2\int_s^t\rbr{\frac{s}{r}}^{\frac{\beta}{8\pi}}\underbrace{[H_r^n(\xx,\Si)]^{1-\frac{\beta}{4\pi}}}_{\leq 1}\dr
\\
&\leq \frac{\tilde{B}\bar{B}_NN^2 }{1-\frac{\beta}{8\pi}}t^{1-\frac{\beta}{8\pi}}s^{\frac{\beta}{8\pi}},
\end{split}
\end{equation}
for some constant $\tilde{B}>0$. As with the other constants, we will henceforth denote the constant coming from the application of the first estimate of \Cref{lem:estimate for the full sum of the massive heat kernels} by $\tilde{B}$. We have also first estimated
\begin{equation} 
\begin{split}
\exp(-\tfrac{\beta}{2}\sum_{k,l=1}^n\sigma_{kl}K_{s,t}^m(k,l))&=\exp(-\tfrac{\beta}{2}\sum_{k,l=1}^n\sigma_{kl}K_{s,r}^m(k,l))\underbrace{\exp(-\tfrac{\beta}{2}\sum_{k,l=1}^n\sigma_{kl}K_{r,t}^m(k,l))}_{\leq 1}
\\
&\leq \exp(-\tfrac{\beta}{2}\sum_{k,l=1}^n\sigma_{kl}K_{s,r}^m(k,l))
\end{split}
\end{equation}
before applying \Cref{lem:estimate for the exponential in tilVtn}.

For $t\geq m^{-2}\geq s$ we obtain similarly
\begin{equation}
\begin{split}
\label{eq:estimate for the tricky term for gtn part 1}
\bigg(\int_s^tr& e^{-m^2r}\sum_{k,l=1}^n\sigma_{kl}C_r(k,l)\dr\bigg)  e^{-\frac{\beta}{2}\sum_{k,l=1}^n\sigma_{kl}K_{s,t}^m(k,l)}
\\
&\leq \int_s^{m^{-2}}r \rbr{\sum_{k,l=1}^n\sigma_{kl}C_r(k,l)}e^{-\frac{\beta}{2}\sum_{k,l=1}^n\sigma_{kl}K_{s,r}^m(k,l)}\dr
\\
&\quad +\int_{m^{-2}}^{t}r e^{-m^2r}\rbr{\sum_{k,l=1}^n\sigma_{kl}C_r(k,l)}e^{-\frac{\beta}{2}\sum_{k,l=1}^n\sigma_{kl}K_{s,m^{-2}}^m(k,l)}\dr
\\
&\leq\tilde{B}\bar{B}_N n^2\bigg(\int_s^{m^{-2}}\rbr{\frac{s}{r}}^{\frac{\beta}{8\pi}}e^{-m^2r}\underbrace{[H_r^n(\xx,\Si)]^{1-\frac{\beta}{4\pi}}}_{\leq 1}\dr
\\
&\qquad\qquad\qquad+\int_{m^{-2}}^t\rbr{\frac{s}{m^{-2}}}^{\frac{\beta}{8\pi}}e^{-m^2r}\underbrace{H_r^n(\xx,\Si)[H_{m^{-2}}^n(\xx,\Si)]^{-\frac{\beta}{4\pi}}}_{\leq 1 \text{ by \Cref{eq:combining the Hn functions}}}\dr\bigg)
\\
&\leq \tilde{B}\bar{B}_NN^2\rbr{\frac{1}{1-\frac{\beta}{8\pi}}+1}(m^{-2})^{1-\frac{\beta}{8\pi}}s^{\frac{\beta}{8\pi}}.
\end{split}
\end{equation}

For $t\geq s\geq m^{-2}$ we use
\begin{equation}
\label{eq:estimate for the tricky term for gtn part 2}
\bigg(\int_s^tr e^{-m^2r}\sum_{k,l=1}^n\sigma_{kl}C_r(k,l)\dr\bigg)  e^{-\frac{\beta}{2}\sum_{k,l=1}^n\sigma_{kl}K_{s,t}^m(k,l)}\leq n^2 \int_0^\infty e^{-m^2r}\dr\leq N^2m^{-2},
\end{equation}
where we used $rC_r(x,y)\leq 1$ for all $x,y\in\R$ and $r>0$. Now we have all the preliminary ingredients to consider the functions $g_t^n$.

Applying the induction hypothesis and using the established existence of the functions $h_t^n$, we define the functions $g_t^n$ recursively. First we have for $t<m^{-2}$
\begin{equation}
\begin{split}
\label{eq:def of g_t}
g_t^n(\xx,\Si|m)&:=\beta m\sum_{I\varsubsetneq [n]}\int_0^te^{-m^2s}\rbr{s+\frac{\beta}{2}\frac{\tilde{B}\bar{B}_NN^2}{1-\frac{\beta}{8\pi}}t^{1-\frac{\beta}{8\pi}}s^{\frac{\beta}{8\pi}}}G_s^I(\xx,\Si)
\\
&\qquad\qquad\qquad \times h_s^{|I|}(\xx_I,\Si_I|m)h_s^{|I^c|}(\xx_{I^c},\Si_{I^c}|m)\ds
\\
&\quad +\frac{\beta}{2}\sum_{I\varsubsetneq [n]}\int_0^te^{-m^2s}G_s^I(\xx,\Si)
\\
&\qquad\qquad\qquad \times [g_s^{|I|}(\xx_I,\Si_I|m)h_s^{|I^c|}(\xx_{I^c},\Si_{I^c}|m)+h_s^{|I|}(\xx_I,\Si_I|m)g_s^{|I^c|}(\xx_{I^c},\Si_{I^c}|m)]\ds
\\
&\,\,=:\sum_{i=1}^4g_t^{n,i}(\xx,\Si|m),
\end{split} 
\end{equation}
where the $i=1$ term corresponds to the $s$ in the brackets on the first line, $i=2$ corresponds to the remaining part, and the $i=3$ and $i=4$ terms are symmetric under $I\leftrightarrow I^c$, since the sum runs over all ordered bipartitions of $[n]$. If $t\geq m^{-2}$ only the $i=2$ term must be modified. In this case we define
\begin{equation}
\begin{split}
    g_t^{n,2}(\xx,\Si|m)&:=\frac{1}{2}\beta^2\tilde{B}\bar{B}_NN^2\rbr{\frac{1}{1-\frac{\beta}{8\pi}}+1}m(m^{-2})^{1-\frac{\beta}{8\pi}}
    \\
    &\qquad\qquad\qquad\qquad\!\!\times\sum_{I\varsubsetneq [n]}\int_0^{m^{-2}}s^{\frac{\beta}{8\pi}}e^{-m^2s}G_s^I(\xx,\Si)h_s^{|I|}(\xx_I,\Si_I|m)h_s^{|I^c|}(\xx_{I^c},\Si_{I^c}|m)\ds
    \\
    &\qquad+\half \beta^2N^2 m^{-1}\sum_{I\varsubsetneq [n]}\int_{m^{-2}}^te^{-m^2s}G_s^I(\xx,\Si)h_s^{|I|}(\xx_I,\Si_I|m)h_s^{|I^c|}(\xx_{I^c},\Si_{I^c}|m)\ds
    \end{split}
\end{equation}
Permutation invariance follows exactly as in the case of the functions $h_t^n$.

Let us begin the estimation with the case $t<m^{-2}$. Consider first the contribution from $g_t^{n,1}$. This can be estimated similarly to the case $n>N+1$ for the functions $h_t^n$. More precisely, starting analogously to \Cref{eq:First estimate for the n>N case in 1d} (with $H_t^n\to 1$), using \Cref{lem:estimate for GI function}, proceeding in a similar vein as in the analysis from \Cref{eq:splitting of the n norm} onwards, and using the combinatorial identity \Cref{eq:combinatorial sum giving n to the power n minu 2}, we have
\begin{equation}
\begin{split}
\norm{g_t^{n,1}(\cdot,\Si)}_{n,\Lambda} &\leq 2\beta \tilde{C}C_1[B_1(\beta,N)]^{n-2} C^n(n-1)n^{n-2}m\int_0^ts^{-\half+n(\half-\frac{\beta}{8\pi})}\ds
\\ 
&\leq 2\tilde{C}C_1\beta\rbr{\sup_{n\leq N}\frac{1}{\half-\frac{n}{n+1}\frac{\beta}{8\pi}}}[B_1(\beta,N)]^{n-2}C^nn^{n-2}mt^{\half}\rbr{t^{\half-\frac{\beta}{8\pi}}}^n,
\end{split}
\end{equation}
where the $s$-integral converges for all $\beta<(1+1/n)4\pi$. The last estimate on the right-hand side is of the correct form. 

For $g_t^{n,2}$ we obtain by proceeding similarly as with $g_t^{n,1}$ above
\begin{equation}
\begin{split}
\norm{g_t^{n,2}(\cdot,\Si)}_{n,\Lambda}&\leq \frac{\beta^2\tilde{B}\bar{B}_N\tilde{C}C_1n^2}{1-\frac{\beta}{8\pi}}[B_1(\beta,N)]^{n-2}C^n(n-1)n^{n-2}mt^{1-\frac{\beta}{8\pi}}\int_0^ts^{-\frac{3}{2}+\frac{\beta}{8\pi}}\rbr{s^{\half-\frac{\beta}{8\pi}}}^n\ds
\\
&\leq \frac{\beta^2\tilde{B}\bar{B}_N\tilde{C}C_1N^2}{\rbr{1-\frac{\beta}{8\pi}}\rbr{\half-\frac{\beta}{8\pi}}}[B_1(\beta,N)]^{n-2}C^nn^{n-2}mt^{\half}\rbr{t^{\half-\frac{\beta}{8\pi}}}^n,
\end{split}
\end{equation}
for all $\beta<4\pi$. Again, we obtained an estimate of the correct form.

Finally, as already mentioned, by symmetry we only need to treat $g_t^{n,3}$ of the remaining terms. We proceed similarly to the above, noting that after the step analogous to \Cref{eq:splitting of the n norm}, the term corresponding to $I$ will be $\norm{g_t^{|I|}(\cdot,\Si_I)}_{|I|}$ (noting that $H_s^{|I|}(\xx_I,\Si_I)\leq 1$) and we need to apply the induction hypothesis to it. Thus, we obtain   
\begin{equation}
\begin{split}
\norm{g_t^{n,3}(\cdot,\Si)}_{n,\Lambda}&\leq \beta\tilde{C}C_1[\tilde{B}_1(\beta,N)]^{|I|-1}[B_1(\beta,N)]^{|I^c|-1}C^n(n-1)n^{n-2}m\int_0^ts^{-\half}\rbr{s^{\half-\frac{\beta}{8\pi}}}^n\ds
\\
&\leq \beta\tilde{C}C_1\rbr{\sup_{n\leq N}\frac{1}{\half-\frac{n}{n+1}\frac{\beta}{8\pi}}}\max\rbr{\tilde{B}_1(\beta,N),B_1(\beta,N)}^{n-2} C^n n^{n-2}mt^{\half}\rbr{t^{\half-\frac{\beta}{8\pi}}}^n 
\end{split}
\end{equation}
The last estimate above is again of the correct form.
Thus, choosing 
\begin{equation}
\label{eq:def of Btilde1'}
\tilde{B}_1(\beta,N):=\tilde{B}_1'(\beta,N):=\max\rbr{B_1(\beta,N),4\beta\tilde{C}C_1\sup_{n\leq N}\frac{1}{\half-\frac{n}{n+1}\frac{\beta}{8\pi}}+\frac{\beta^2\tilde{C}C_1\tilde{B}\bar{B}_N N^2}{\rbr{1-\frac{\beta}{8\pi}}\rbr{\half-\frac{\beta}{8\pi}}}},
\end{equation}
where $B_1(\beta,N)$ is the constant from the estimates for the $h_t^n$ functions, yields the result for $t<m^{-2}$.

Consider then the case $t\geq m^{-2}$. In this case we need to split the $s$-integral at $m^{-2}$; otherwise the treatment is similar to the above, with appropriate modifications to the term with $s>m^{-2}$. First, we have
\begin{equation}
\begin{split}
\norm{g_t^{n,1}(\cdot,\Si)}_{n,\Lambda}
&\leq 2\beta \tilde{C}C_1[B_1(\beta,N)]^{n-2} C^n(n-1)n^{n-2}
\\
&\quad\times m\rbr{\int_0^{m^{-2}}s^{-\half+n(\half-\frac{\beta}{8\pi})}\ds+(m^{-2})^{-1+n\rbr{\half-\frac{\beta}{8\pi}}}\int_{m^{-2}}^t s^{\half}e^{-m^2s}\ds}
\\ 
&\leq 2C_1\tilde{C}\beta\rbr{\sup_{n\leq N}\frac{1}{\half-\frac{n}{n+1}\frac{\beta}{8\pi}}+(N-1)\Gamma\rbr{\frac{3}{2},1}}[B_1(\beta,N)]^{n-2}C^nn^{n-2}(m^{-2})^{n\rbr{\half-\frac{\beta}{8\pi}}}.
\end{split}
\end{equation}

The second case yields
\begin{equation}
\begin{split}
\norm{g_t^{n,2}(\cdot,\Si)}_{n,\Lambda}
&\leq \beta^2\tilde{C}C_1N^2[B_1(\beta,N)]^{n-2}C^n(n-1)n^{n-2}
\\
&\quad \times\bigg(\tilde{B}\bar{B}_N\rbr{\frac{1}{1-\frac{\beta}{8\pi}}+1}m(m^{-2})^{1-\frac{\beta}{8\pi}}\int_0^{m^{-2}}s^{-\frac{3}{2}+\frac{\beta}{8\pi}}s^{n\rbr{\half-\frac{\beta}{8\pi}}}\ds
\\
&\qquad\qquad\quad+m(m^{-2})(m^{-2})^{-1+n\rbr{\half-\frac{\beta}{8\pi}}}\int_{m^{-2}}^t\frac{e^{-m^2s}}{\sqrt{s}}\ds\bigg)
\\
&\leq\beta^2\tilde{C}C_1 N^2\rbr{\tilde{B}\bar{B}_N\rbr{\frac{1}{1-\frac{\beta}{8\pi}}+1}\frac{1}{\half-\frac{\beta}{8\pi}}+(N-1)\Gamma\rbr{\half,1}}
\\
&\quad\times [B_1(\beta,N)]^{n-2}C^nn^{n-2}(m^{-2})^{n\rbr{\half-\frac{\beta}{8\pi}}} 
\end{split}
\end{equation}

Finally, we have for the last case
\begin{equation}
\begin{split}
\norm{g_t^{n,3}(\cdot,\Si)}_{n,\Lambda}
&\leq \beta\tilde{C}C_1[\tilde{B}_1(\beta,N)]^{|I|-1}[B_1(\beta,N)]^{|I^c|-1}C^n(n-1)n^{n-2}
\\
&\quad\times\rbr{m\int_0^{m^{-2}}s^{-\half}\rbr{s^{\half-\frac{\beta}{8\pi}}}^n\ds+m(m^{-2})^{n\rbr{\half-\frac{\beta}{8\pi}}}\int_{m^{-2}}^t\frac{e^{-m^2s}}{\sqrt{s}}\ds}
\\
&\leq \beta\tilde{C}C_1\rbr{\sup_{n\leq N}\frac{1}{\half-\frac{n}{n+1}\frac{\beta}{8\pi}}+(N-1)\Gamma\rbr{\half,1}}
\\
&\quad \times[\max(\tilde{B}_1(\beta,N),B_1(\beta,N))]^{n-2} C^n n^{n-2}(m^{-2})^{n\rbr{\half-\frac{\beta}{8\pi}}}. 
\end{split}
\end{equation}

All of the above estimates are of the correct form. Thus, the result follows by choosing $\tilde{B}_1(\beta,N)$ to be the maximum of $\tilde{B}_1'(\beta,N)$ (the constant from the $t<m^{-2}$ case defined in \Cref{eq:def of Btilde1'}) and 
\begin{equation}
\begin{split}
\beta\tilde{C}C_1&\bigg[4\sup_{n\leq N}\frac{1}{\half-\frac{n}{n+1}\frac{\beta}{8\pi}}+\beta\tilde{B}\bar{B}_NN^2\rbr{\frac{1}{1-\frac{\beta}{8\pi}}+1}\frac{1}{\half-\frac{\beta}{8\pi}}
\\
&\quad+2(N-1)\rbr{\Gamma\rbr{\frac{3}{2},1}+\Gamma\rbr{\half,1}}+\beta N^2(N-1)\Gamma\rbr{\half,1}\bigg]
\end{split}
\end{equation}

\end{subsubsection}

\begin{subsubsection}{Proof of {{\Cref{cor:convergence of the Vinfty-Vt terms}}} for the case $d=1$}
\label{sec:proof of the convergence of Vinfty-Vt for d=1}
Recall that the case $d\geq 2$ was discussed after the statement of \Cref{cor:convergence of the Vinfty-Vt terms}. Thus, here we only need to prove the case $d=1$. We first establish \Cref{eq:majorant condition for Vinfty-Vt for d=1,eq:uniform bound for the majorant in d=1}, which shows that, if the limit \Cref{eq:limits of the difference Vinfty-Vt} exists, it is finite as mentioned in \Cref{rem:finiteness of the limit of the Vinfty-Vt terms}. We then derive an explicit expression for the error term in $R_t^n(m,\xx,\Si)$ in the decomposition in \Cref{eq:decomposition of the Vinfty-Vt terms}, and finally we prove that the appropriate integral of $|R_t^n(m,\xx,\Si)|$ vanishes in the limit $m\to 0$.

First, note that
\begin{equation}
    \begin{split}
    \tilV_\infty^n(\xx,\Si|m,\eps)-\tilV_t^n(\xx,\Si|m,\eps)
    &=\int_t^\infty\frac{\partial}{\partial s}\tilV_s^n(\xx,\Si|m,\eps)\ds
    \\
    &=\int_t^\infty \bigg[\frac{\beta}{2}\sum_{I\varsubsetneq [n]}\bigg(\sum_{\substack{i\in I\\ j\in I^c}}\sigma_{ij}C_s^m(i,j)\bigg)\tilV_s^{|I|}(\xx_I,\Si_I|m,\eps)\tilV_s^{|I^c|}(\xx_{I^c},\Si_{I^c}|m,\eps)
    \\
    &\qquad\quad-\frac{\beta}{2}\bigg(\sum_{k,l=1}^n\sigma_{kl}C_s^m(k,l)\bigg)\tilV_s^n(\xx,\Si|m,\eps)\bigg]\ds .
    \end{split}
\end{equation}
We then use the majorants $h_t^n$ from \Cref{prop:induction statement} together with the notation from \Cref{eq:def of GsI function} to obtain
\begin{equation}
\begin{split}
\label{eq:def of the majorant Ft(1n)}
|\tilV_\infty^n(\xx,\Si|m,\eps)-\tilV_t^n(\xx,\Si|m,\eps)|
&\leq \frac{\beta}{2}\sum_{I\varsubsetneq [n]}\int_t^\infty e^{-m^2s}G_s^I(\xx,\Si)h_s^{|I|}(\xx_I,\Si_I|m)h_s^{|I^c|}(\xx_{I^c},\Si_{I^c}|m)\ds
\\
&\quad+\frac{\beta}{2}\int_t^\infty e^{-m^2s}\rbr{\sum_{k,l=1}^n\sigma_{kl}C_s(k,l)}h_s^n(\xx,\Si|m)\ds
\\
&=: F_t^{1,n}(\xx,\Si|m).
\end{split}
\end{equation}
We now derive upper bounds for the terms above separately.

Consider the contribution from the first term. We may assume that $t<m^{-2}$ since if $t\geq m^{-2}$ we may replace the lower bound $t$ with $m^{-2}$ by positivity, in which case we only have the second term below. We again use \Cref{lem:estimate for GI function}, split the $s$-integral at $m^{-2}$, and treat the spatial integrals as in \Cref{eq:splitting of the n norm}. However, for the term with $s\in (t,m^{-2})$, instead of scaling the integral corresponding to $x_j$ in this step, we need to use
\begin{equation}
\label{eq:trading the measure of Lambda}
\int_{\Lambda}\underbrace{f_d\rbr{\frac{|x_j-x_i|}{\sqrt{s}}}}_{\leq 1}\dx_j\leq |\Lambda|.
\end{equation}
The term with $s\in (m^{-2},\infty)$ is treated as the terms with $s\in (m^{-2},t)$ in \Cref{eq:first estimate in the induction for large t and $d=1$}.
Thus, we have
\begin{equation}
\begin{split}
\label{eq:first term of the integral bound for |Vinfty-Vt|}
\sum_{\Si\in \{-1,1\}^n}\int_{\Lambda^n}\frac{\beta}{2}
&\sum_{I\varsubsetneq [n]}\int_t^\infty e^{-m^2s}G_s^I(\xx,\Si)h_s^{|I|}(\xx_I,\Si_I|m)h_s^{|I^c|}(\xx_{I^c},\Si_{I^c}|m)\ds\d \xx
\\
&\leq \beta\tilde{C}[B_1(\beta,N)]^{n-2}2^nC^n(n-1)n^{n-2}|\Lambda|
\\
&\quad \times\bigg(|\Lambda|\int_t^{m^{-2}}s^{-2}\rbr{s^{\half-\frac{\beta}{8\pi}}}^n\ds 
+ C_1(m^{-2})^{-1}(m^{-2})^{n\rbr{\half-\frac{\beta}{8\pi}}}\int_{m^{-2}}^\infty\frac{e^{-m^2s}}{\sqrt{s}}\ds\bigg)
\\
&\leq  \beta\tilde{C}[B_1(\beta,N)]^{n-2}2^nC^n(n-1)n^{n-2}|\Lambda|
\\
&\quad\times\rbr{|\Lambda|\frac{m^{n\rbr{\frac{\beta}{4\pi}-1}+2}}{(n-2)\half-\frac{n\beta}{8\pi}}
+ C_1m^{n\rbr{\frac{\beta}{4\pi}-1}+1}\Gamma\rbr{\half,1}} .
\end{split}
\end{equation}
Here, the factor $2^n$ accounts for the number of terms in the $\Si$-sum. The first factor $|\Lambda|$ arises from replacing the $x_1$-integral by a supremum, which allows us to use the norm estimates from \Cref{prop:induction statement} for the functions $h_t^n$. Note that if we did not use the estimate \Cref{eq:trading the measure of Lambda} and proceeded as before, the exponent of $s$ would be $-3/2+n(1/2-\beta/(8\pi))$. Although such an integral would still be convergent, it would lead to a logarithmic dependence on the mass in the case $\beta=\beta_N$ and $n=N$, which is unbounded for $m\in(0,1)$. For the second $s$-integral after the first inequality in \Cref{eq:first term of the integral bound for |Vinfty-Vt|}, we can scale the integral so that it becomes independent of the mass $m$. In this case, the $\beta=\beta_N$ and $n=N$ situation corresponds to the mass exponent being zero in the second term on the last line of \Cref{eq:first term of the integral bound for |Vinfty-Vt|}, which remains bounded for $m\in(0,1)$.

Consequently, we obtain a uniform (in $m$) upper bound for this term provided that the exponents of $m$ in the last line of \Cref{eq:first term of the integral bound for |Vinfty-Vt|} are non-negative. This condition yields the constraints
\begin{equation}
\label{eq:beta constrainsts from the integral bound for |Vinfty-Vt|}
\beta\geq \rbr{1-\frac{2}{n}}4\pi 
\qquad\text{and}\qquad 
\beta\geq \rbr{1-\frac{1}{n}}4\pi,
\quad \text{for all } n\leq N,
\end{equation}
which are satisfied under the assumption $\beta\in[\beta_N,\beta_{N+2})$. 

For the second term, we obtain using \Cref{prop:induction statement,lem:estimate for the full sum of the massive heat kernels} and similar manipulations as with the first term
\begin{equation}
\begin{split}
\label{eq:second term of the integral bound for |Vinfty-Vt|}
    &\sum_{\Si\in\{-1,1\}^n}\int_{\Lambda^n}\frac{\beta}{2}
    \int_t^\infty e^{-m^2 s}\rbr{\sum_{k,l=1}^n\sigma_{kl}C_s(k,l)}h_s^n(\xx,\Si|m)\ds\d\xx
    \\
    &\quad\leq [B_1(\beta,N)]^{n-1}2^nC^nn^n|\Lambda|
    \bigg(C(t,\Lambda)\int_t^{m^{-2}} s^{-2+n\rbr{\half-\frac{\beta}{8\pi}}}\ds
    +\tilde{B}(m^{-2})^{(n-1)\half-\frac{n\beta}{8\pi}}\int_{m^{-2}}^\infty\frac{e^{-m^2s}}{s}\ds\bigg).
\end{split}
\end{equation}

Here, we needed to use the second claim of \Cref{lem:estimate for the full sum of the massive heat kernels} for the part where $s\in (t,m^{-2})$. This provides the extra $s^{-1/2}$-factor, which makes the resulting mass exponents non-negative, but again enforces a stronger dependence on $\Lambda$. That is, the constant $C(t,\Lambda)$ plays the role of the second $|\Lambda|$-factor in the analysis of the first term above. For the part with $s\in (m^{-2},\infty)$ it sufficed to use the first claim of \Cref{lem:estimate for the full sum of the massive heat kernels}. The first term on the right-hand side of \Cref{eq:second term of the integral bound for |Vinfty-Vt|} has the same $s$-integral as the first term in the second row of \Cref{eq:first term of the integral bound for |Vinfty-Vt|} and the second $s$-integral in \Cref{eq:second term of the integral bound for |Vinfty-Vt|} is independent of mass by scaling. Thus, we obtain the same constraints for $\beta$ as in \Cref{eq:beta constrainsts from the integral bound for |Vinfty-Vt|}.

\begin{remark}
\label{rem:lambda dependence}
The strong $\Lambda$-dependence in the above estimates is caused by taking the limit $m\to 0$. Thus, the situation reflects a trade-off between keeping the mass non-zero and allowing stronger dependence of the estimates on the finite volume. Both mechanisms effectively act as an IR-cutoff.
\end{remark}

Since the precise form of the bound is unimportant, by maximizing all of the above estimates over $n\leq N$, we conclude that
\begin{equation}
\label{eq:bound for the integral of the majorant Ft1n}
    \sum_{\Si\in\{-1,1\}^n}\int_{\Lambda^n}F_t^{1,n}(\xx,\Si|m)\d\xx
    \leq C(\beta,t,N,\Lambda)<\infty
\end{equation}
for all $n\leq N$ and some constant $C(\beta,t,N,\Lambda)>0$ that is independent of $m$. This proves \Cref{eq:uniform bound for the majorant in d=1}.

The majorant in the above estimate is independent of $\eps$, so the limit $\eps\to 0$ can be taken through the spatial integrals, and the pointwise limit 
\begin{equation}
\lim_{\eps\to 0}(\tilV_\infty^n(\xx,\Si|m,\eps)-\tilV_t^n(\xx,\Si|m,\eps))=:\tilV_\infty^n(\xx,\Si|m)-\tilV_t^n(\xx,\Si|m)
\end{equation}
exists as explained in \Cref{rem:limits of the Vtn functions}. Thus, we turn to deriving a useful representation for the error function $R_t^n(m,\xx,\Si)$ such that
\begin{equation}
    \tilV_\infty^n(\xx,\Si|m)-\tilV_t^n(\xx,\Si|m)
    =\tilV_\infty^n(\xx,\Si)-\tilV_t^n(\xx,\Si)+R_t^n(m,\xx,\Si),
\end{equation}
where $\tilV_t^n(\xx,\Si):=\lim_{m\to 0}\lim_{\eps\to 0}\tilV_t^n(\xx,\Si|m,\eps)$ for $t\in (0,\infty]$ and these limits again exist by \Cref{rem:limits of the Vtn functions}. Furthermore, by \Cref{lem:differentiability of Vtn with respect to mass ETC} we also have $\tilV_t^n(\xx,\Si)=\tilV_t^n(\xx,\Si|m=0,\eps=0)$.

First we can write
\begin{equation}
\begin{split}
\label{eq:error form for Vinfty-Vt}
    \tilV_\infty^n(\xx,\Si|m)&-\tilV_t^n(\xx,\Si|m)
    \\
    &=\int_t^\infty\bigg\{\frac{\beta}{2}\sum_{I\varsubsetneq [n]}\bigg[\bigg(\sum_{\substack{i\in I\\ j\in I^c}}\sigma_{ij}C_s(i,j)\bigg)
    \rbr{\tilV_s^{|I|}(\xx_I,\Si_I)-\abr{\tilV_s^{|I|}(\xx_I,\Si_I)-e^{-\half m^2s}\tilV_s^{|I|}(\xx_I,\Si_I|m)}}
    \\
    &\qquad\qquad\quad\times
    \rbr{\tilV_s^{|I^c|}(\xx_{I^c},\Si_{I^c})-\abr{\tilV_s^{|I^c|}(\xx_{I^c},\Si_{I^c})-e^{-\half m^2s}\tilV_s^{|I^c|}(\xx_{I^c},\Si_{I^c}|m)}}\bigg]
    \\
    &\qquad -\frac{\beta}{2}\bigg(\sum_{k,l=1}^n\sigma_{kl}C_s(k,l)\bigg)
    \rbr{\tilV_s^{n}(\xx,\Si)-\abr{\tilV_s^{n}(\xx,\Si)-e^{-m^2s}\tilV_s^{n}(\xx,\Si|m)}}\bigg\}\ds
    \\
    &=\int_t^\infty\bigg\{ \frac{\beta}{2}\sum_{I\varsubsetneq [n]}
    \bigg[\bigg(\sum_{\substack{i\in I\\j\in I^c}}\sigma_{ij}C_s(i,j)\bigg)
    \tilV_s^{|I|}(\xx_I,\Si_I)\tilV_s^{|I^c|}(\xx_{I^c},\Si_{I^c})\bigg] 
    \\
    &\qquad\quad-\frac{\beta}{2}\rbr{\sum_{k,l=1}^n\sigma_{kl}C_s(k,l)}\tilV_s^n(\xx,\Si)\bigg\}\ds
    +R_t^n(m,\xx,\Si)
    \\
    &=\tilV_\infty^n(\xx,\Si)-\tilV_t^n(\xx,\Si)+R_t^n(m,\xx,\Si),
\end{split}
\end{equation}
where the error term can be written as
\begin{equation}
    \begin{split}
    \label{eq:def of the error term Rm}
    R_t^n(m,\xx,\Si)
    &:=-\frac{\beta}{2}\sum_{I\varsubsetneq [n]}\int_t^\infty\bigg\{
    \bigg(\sum_{\substack{i\in I\\j\in I^c}}\sigma_{ij}C_s(i,j)\bigg)
    \\
    &\qquad\qquad\qquad\quad\times \bigg(
    \abr{\tilde{W}_s^{|I|}(1/2,\xx_I,\Si_I)-\tilde{W}_s^{|I|}(1/2,\xx_I,\Si_I|m)}
    \tilde{W}_s^{|I^c|}(1/2,\xx_{I^c},\Si_{I^c}|m)
    \\
    &\qquad\qquad\qquad\qquad
    +\tilde{W}_s^{|I|}(1/2,\xx_I,\Si_I|m)
    \abr{\tilde{W}_s^{|I^c|}(1/2,\xx_{I^c},\Si_{I^c})-\tilde{W}_s^{|I^c|}(1/2,\xx_{I^c},\Si_{I^c}|m)}
    \\
    &\qquad\qquad\qquad\qquad
    +\abr{\tilde{W}_s^{|I|}(1/2,\xx_I,\Si_I)-\tilde{W}_s^{|I|}(1/2,\xx_I,\Si_I|m)}
    \\
    &\qquad\qquad\qquad\qquad\quad\times \abr{\tilde{W}_s^{|I^c|}(1/2,\xx_{I^c},\Si_{I^c})-\tilde{W}_s^{|I^c|}(1/2,\xx_{I^c},\Si_{I^c}|m)}
    \bigg)\bigg\}\ds
    \\
    &\qquad+\frac{\beta}{2}\int_t^\infty
    \rbr{\sum_{k,l=1}^n\sigma_{kl}C_s(k,l)}
    \abr{\tilde{W}_s^n(1,\xx,\Si)-\tilde{W}_s^n(1,\xx,\Si|m)}\ds .
    \end{split}
\end{equation}
Here $\tilde{W}_s^n(a,\xx,\Si|m):=e^{-am^2s}\tilV_s^n(\xx,\Si|m)$ and
$\tilde{W}_s^n(a,\xx,\Si):=\tilde{W}_s^n(a,\xx,\Si|m=0)\equiv\tilV_s^n(\xx,\Si)$ for any $a>0$.

The final goal is to prove \Cref{eq:vanishing of the error in the Vinfty-Vt terms}. For this we need good estimates for the error term.
First we can write 
\begin{equation}
\begin{split}
\label{eq:integral estimate for W-W(m)}
\tilde{W}_s^n(a,\xx,\Si|m)-\tilde{W}_s^n(a,\xx,\Si)&=\int_0^m\frac{\partial}{\partial\mu}\tilde{W}_s^n(a,\xx,\Si|\mu)\d\mu
\\
&=-2as\int_0^m\mu e^{-a\mu^2s}\tilV_s^n(\xx,\Si|\mu)\d\mu+\int_0^me^{-a\mu^2s}\frac{\partial}{\partial \mu}\tilV_s^n(\xx,\Si|\mu)\d\mu,
\end{split}
\end{equation}
which is justified by \Cref{lem:differentiability of Vtn with respect to mass ETC}.

Then, we use \Cref{eq:integral estimate for W-W(m)} and the majorant functions $h_t^n$ and $g_t^n$ together with their estimates from \Cref{prop:induction statement}, to bound $|R_t^n(m,\xx,\Si)|$. 

Thus, we first obtain
\begin{equation}
    \begin{split}
    \label{eq:first estimate for the absolute value of  R_tn}
    |R_t^n(m,\xx,\Si)|
    &\leq \half\beta \int_0^m \sum_{I\varsubsetneq [n]}\int_t^\infty    
    \!\!G_s^I(\xx,\Si)e^{-\half(m^2+\mu^2)s}\bigg\{ 
    \abr{\mu s h_s^{|I|}(\xx_I,\Si_I|\mu)+g_s^{|I|}(\xx_I,\Si_I|\mu)}
    h_s^{|I^c|}(\xx_{I^c},\Si_{I^c}|m)
    \\
    &\qquad\qquad\qquad\qquad\qquad
    +h_s^{|I|}(\xx_{I},\Si_{I}|m)
    \abr{\mu s h_s^{|I^c|}(\xx_{I^c},\Si_{I^c}|\mu)+g_s^{|I^c|}    
    (\xx_{I^c},\Si_{I^c}|\mu)}
    \bigg\}\ds\d\mu
    \\
    &\quad+\half\beta \int_0^m \int_0^m\sum_{I\varsubsetneq [n]}\int_t^\infty    
    \!\!G_s^I(\xx,\Si)e^{-\half(\mu^2+\lambda^2)s}\bigg\{\abr{\lambda s h_s^{|I|}(\xx_I,\Si_I|\lambda)+g_s^{|I|}(\xx_I,\Si_I|\lambda)}
    \\
    &\qquad\qquad\qquad\qquad\qquad\times\abr{\mu s h_s^{|I^c|}(\xx_{I^c},\Si_{I^c}|\mu)+g_s^{|I^c|}(\xx_{I^c},\Si_{I^c}|\mu)}\bigg\}\ds\d\mu\d\lambda
    \\
    &\quad+\half\beta \int_0^m\int_t^\infty e^{-\mu^2s}
    \rbr{\sum_{k,l=1}^n\sigma_{kl}C_s(k,l)}
    \abr{2\mu s h_s^n(\xx,\Si|\mu)+g_s^n(\xx,\Si|\mu)}\ds\d\mu
    \\
    &=:\int_0^m \!\![R_1(m,\mu,\xx,\Si)+R_2(m,\mu,\xx,\Si)]\d\mu+\int_0^m\int_0^m \!\! R_3(\mu,\lambda,\xx,\Si)
    \d\mu\d\lambda+  
    \int_0^m\!\! R_4(\mu,\xx,\Si)\d\mu,
    \end{split}
\end{equation}
where we have suppressed the $t$ and $n$ dependence from the terms $R_i$, $i=1,2,3,4$, since these will not be used elsewhere with different values, unlike the full error term $R_t^n$. Furthermore, the exchange of the mass scale integrals with respect to $\mu$ and $\lambda$, and the $s$-integrals above, as well as the exchange of any of these with the spatial integrals, is justified since everything is positive. We will always evaluate the mass scale integrals last. Furthermore, without loss of generality we assume that $t<m^{-2},\mu^{-2},\lambda^{-2}$, since otherwise some of the terms simply would not appear and the task would be simpler.

Consider the $R_4$ term first, since it involves only one mass scale. Using \Cref{prop:induction statement}, the second claim of \Cref{lem:estimate for the full sum of the massive heat kernels}, splitting the $s$-integral at the mass scale $\mu^{-2}$ and using $e^{-\mu^2s}\leq 1$ for the part with $s\in (t,\mu^{-2})$, we obtain
\begin{equation}
\begin{split}
\label{eq:estimate for R4}
    &\sum_{\Si\in\{-1,1\}^n}\int_{\Lambda^n}R_4(\mu,\xx,\Si)\d\xx
    \\
    &\quad \leq  \half \beta  C(t,\Lambda)|\Lambda|
    [\tilde{B}_1(\beta,N)]^{n-2}2^nC^n n^n
    \bigg[
    3\mu\int_t^{\mu^{-2}}s^{-1+n\rbr{\half-\frac{\beta}{8\pi}}}\ds
    \\
    &\qquad\qquad\qquad\qquad
    +2\mu(\mu^{-2})^{(n-1)\half-\frac{n\beta}{8\pi}}
    \int_{\mu^{-2}}^\infty\frac{e^{-\mu^2s}}{\sqrt{s}}\ds
    +(\mu^{-2})^{n\rbr{\half-\frac{\beta}{8\pi}}}
    \int_{\mu^{-2}}^\infty\frac{e^{-\mu^2s}}{s^{\frac{3}{2}}}\ds
    \bigg]
    \\
    &\quad\leq  \half\beta C(t,\Lambda)|\Lambda|
    [\tilde{B}_1(\beta,N)]^{n-2}2^nC^nn^n
    \rbr{\frac{3}{n}\frac{1}{\half-\frac{\beta}{8\pi}}
    +2\Gamma\rbr{\half,1}+\Gamma\rbr{-\half,1}}
    \mu^{n\rbr{\frac{\beta}{4\pi}-1}+1},
\end{split}
\end{equation}
where we used $\tilde{B}_1(\beta,N)\geq B_1(\beta,N)$ by definition. The factor $|\Lambda|$ again came from taking a supremum over $x_1$ and $2^n$ is just the number of terms in the $\sigma$-sum. Noting that $0\leq\mu\leq m<1$ we find that
$\mu^{n(\beta/(4\pi)-1)+1}\leq 1$ provided that
\begin{equation}
n\rbr{\frac{\beta}{4\pi}-1}+1\geq 0
\;\;\Leftrightarrow\;\;
\beta\geq \rbr{1-\frac{1}{n}}4\pi,
\end{equation}
which is satisfied for all $n\leq N$ under the assumption $\beta\in[\beta_N,\beta_{N+2})$. Hence, after using this, the integral $\int_0^m\d\mu$ only produces an overall factor $m$, which forces this contribution to vanish as $m\to 0$. Note that here we had to make the trade-off between non-zero mass and stronger dependence on the compact set $\Lambda$, as discussed in \Cref{rem:lambda dependence}, for all terms.

Consider next the terms with $R_1$ and $R_2$. These are symmetric under $I\leftrightarrow I^c$, so we only estimate the term with $R_1$. We can again use \Cref{prop:induction statement,lem:estimate for GI function}. Here we split the $s$-integral at both mass scales $m^{-2}$ and $\mu^{-2}$. We will also use the estimate \Cref{eq:trading the measure of Lambda} at the step analogous to \Cref{eq:splitting of the n norm}, for all terms. This produces again the second factor of $|\Lambda|$. Thus, we obtain
\begin{equation}
\begin{split}
\label{eq:estimate for R1}
\sum_{\Si\in\{-1,1\}^n}
&\int_{\Lambda^n}R_1(m,\mu,\xx,\Si)\d\xx
\\
&\leq \half \beta \tilde{C}[\tilde{B}_1(\beta,N)]^{n-2}2^nC^n|\Lambda|^2
\bigg\{
4(n-1)n^{n-2}\underbrace{\mu\int_t^{m^{-2}}
s^{-1+n\rbr{\half-\frac{\beta}{8\pi}}}\ds}_{I_1(\mu,m)}
\\
&\quad+\sum_{I\varsubsetneq [n]}|I|^{|I|-1}|I^c|^{|I^c|-1}
\bigg[
2\underbrace{\mu (m^{-2})^{-\half+|I^c|\rbr{\half-\frac{\beta}{8\pi}}}
\int_{m^{-2}}^{\mu^{-2}}e^{-\half(m^2+\mu^2)s}
s^{-\frac{1}{2}+|I|\rbr{\half-\frac{\beta}{8\pi}}}\ds}_{=:I_2(\mu,m)}
\\
&\qquad\qquad\qquad\qquad\qquad\qquad
+\underbrace{\mu(m^{-2})^{-\half+|I^c|\rbr{\half-\frac{\beta}{8\pi}}}
(\mu^{-2})^{-\half+|I|\rbr{\half-\frac{\beta}{8\pi}}}
\int_{\mu^{-2}}^\infty e^{-\half(m^2+\mu^2)s}\ds}_{=:I_3(\mu,m)}
\\
&\qquad\qquad\qquad\qquad\qquad\qquad
+\underbrace{(m^{-2})^{-\half+|I^c|\rbr{\half-\frac{\beta}{8\pi}}}
(\mu^{-2})^{|I|\rbr{\half-\frac{\beta}{8\pi}}}
\int_{\mu^{-2}}^\infty\frac{e^{-\half(m^2+\mu^2)s}}{s}\ds}_{=:I_4(\mu,m)}
\bigg]
\bigg\},
\end{split}
\end{equation}
where we have used $e^{-\half(\mu^2+m^2)s}\leq 1$ for the term $I_1(\mu,m)$ and combined the contributions coming from $h_s^{|I|}(\cdot,\cdot|\mu)$ and $g_s^{|I|}(\cdot,\cdot|\mu)$ on the scales $s<\mu^{-2}$, where their $\mu$ and $s$ dependence is the same (we will also do this later with the term involving $R_3$).

Since we are only interested in the dependence on $m$ and $\mu$ as $m\to 0$, we only estimate the $s$-integrals and $\mu$-integrals separately for each term with appropriate prefactors depending on $m$, and ignore all the constants and combinatorial factors. Thus, our task reduces to showing that 
\begin{equation}
\int_0^mI_i(\mu,m)\d\mu\overset{m\to 0}{\longrightarrow}0
\end{equation}
for all $i=1,2,3,4$.

For the first term the $\mu$- and $s$-integral decouple, and we obtain
\begin{equation}
\begin{split}
    \int_0^mI_1(\mu,m)\d\mu
    &\leq B m^2(m^{-2})^{n\rbr{\half-\frac{\beta}{8\pi}}}\leq Bm^{n\rbr{\frac{\beta}{4\pi}-1}+2}\overset{m\to 0}{\longrightarrow}0
\end{split}
\end{equation}
for all $\beta> (1-2/n)4\pi$, which is again satisfied for all $n\leq N$ by the assumption $\beta\in[\beta_N,\beta_{N+2})$. The constant $B>0$ is independent of $m$ and $\mu$, and it is allowed to change from line to line (we will continue this convention in future estimates).

For the second term we obtain
\begin{equation}
\begin{split}
  \int_0^mI_2(\mu,m)\d\mu
    &=B m^{|I^c|\rbr{\frac{\beta}{4\pi}-1}+1}\int_0^m\mu 
    [(m^2+\mu^2)^{-1}]^{\half+|I|\rbr{\half-\frac{\beta}{8\pi}}}
    \int_{\half+\half\frac{\mu^2}{m^2}}^{\half+\half\frac{m^2}{\mu^2}}
    e^{-u}u^{-\half+|I|\rbr{\half-\frac{\beta}{8\pi}}}\du\d\mu
    \\
    &\leq Bm^{|I^c|\rbr{\frac{\beta}{4\pi}-1}+1}\int_0^m
    \underbrace{\frac{\mu}{\sqrt{\mu^2+m^2}}}_{\leq 1}
    \underbrace{\rbr{\frac{1}{m^2+\mu^2}}^{|I|\rbr{\half-\frac{\beta}{8\pi}}}}_{\leq (m^{-2})^{|I|\rbr{\half-\frac{\beta}{8\pi}}}}\d\mu
    \\
    &\leq B
    \underbrace{m^{n\rbr{\frac{\beta}{4\pi}-1}+1}}_{\leq 1}\int_0^m\d\mu
    \overset{m\to 0}{\longrightarrow}0,
\end{split}
\end{equation}
where the $u$-integral is bounded uniformly in $m$ and $\mu$. The exponent of the mass term before the integral on the last line is non-negative for $\beta\geq (1-1/n)4\pi$ and this is guaranteed for all $n\leq N$ by the assumption $\beta\in [\beta_N,\beta_{N+2})$.  

The third term yields 
\begin{equation}
    \begin{split}
 \int_0^mI_3(\mu,m)\d\mu &\leq
    B m^{|I^c|\rbr{\frac{\beta}{4\pi}-1}+1}
    \int_0^m \mu^{|I|\rbr{\frac{\beta}{4\pi}-1}+2}\underbrace{(\mu^2+m^2)^{-1}}_{\leq m^{-2}}
    \underbrace{\int_{\half+\half\frac{m^2}{\mu^2}}^\infty e^{-u}\du}_{\leq 1}\d\mu
    \\
    &\leq Bm^{n\rbr{\frac{\beta}{4\pi}-1}+2}\overset{m\to 0}{\longrightarrow}0
    \end{split}
\end{equation}
where the last $\mu$-integral, after using the estimates written below the terms, is convergent for $\beta>(1-3/|I|)4\pi$, and the exponent of the mass on the last row is positive for $\beta>(1-2/n)4\pi$; these are satisfied for all $n\leq N$ and $\emptyset \neq I\varsubsetneq [n]$, again by the assumption $\beta\in[\beta_N,\beta_{N+2})$.

Finally, the contribution from the last term is 
\begin{equation}
\begin{split}
    \int_0^mI_4(\mu,m)\d\mu
    &= m^{|I^c|\rbr{\frac{\beta}{4\pi}-1}+1}\int_0^m\mu^{|I|\rbr{\frac{\beta}{4\pi}-1}}\int_{\half+\half\frac{m^2}{\mu^2}}^\infty e^{-u}\frac{\du}{u}\d\mu
    \\
    &\leq B m^{n\rbr{\frac{\beta}{4\pi}-1}+2}\overset{m\to 0}{\longrightarrow}0
\end{split}
\end{equation}
where the $u$-integral is again bounded uniformly in $m$ and $\mu$, the last $\mu$-integral is convergent for all $\beta>(1-1/|I|)4\pi$ and the exponent of the mass on the last row is positive for all $\beta>(1-2/n)4\pi$, and these are again satisfied for all relevant cases by the assumption $\beta\in [\beta_N,\beta_{N+2})$.

Let us continue with the term involving $R_3$. By the symmetry $I\leftrightarrow I^c$ we only need to consider the case $\mu\leq \lambda$, which means that the integration region for the mass scale integrals with respect to $(\lambda,\mu)$ is $[0,m]\times[0,\lambda]$. Analogously to \Cref{eq:estimate for R1} we obtain
\begin{equation}
\begin{split}
\sum_{\Si\in\{-1,1\}^n}&\int_{\Lambda^n}R_3(\mu,\lambda,\xx,\Si)\d\xx
	\\
	&\leq\half\beta\tilde{C}[\tilde{B}_1(\beta,N)]^{n-2}2^nC^n|\Lambda|^2\bigg\{8(n-1)n^{n-2}\mu\lambda\int_t^{\lambda^{-2}}s^{n\rbr{\half-\frac{\beta}{8\pi}}}\ds
	\\
	&\quad+\sum_{I\varsubsetneq [n]}|I|^{|I|-1}|I^c|^{|I^c|-1}\bigg[\int_{\lambda^{-2}}^{\mu^{-2}}\frac{e^{-\half(\lambda^2+\mu^2)s}}{s}\rbr{\lambda s(\lambda^{-2})^{-\half+|I|\rbr{\half-\frac{\beta}{8\pi}}}+(\lambda^{-2})^{|I|\rbr{\half-\frac{\beta}{8\pi}}}}
	\\
	&\qquad\qquad\qquad\qquad\qquad\qquad\times2\mu ss^{-\half+|I^c|\rbr{\half-\frac{\beta}{8\pi}}}\ds
	\\
	&\qquad\qquad\qquad\qquad\qquad\quad+\int_{\mu^{-2}}^\infty \frac{e^{-\half(\lambda^2+\mu^2)s}}{s}\rbr{\lambda s(\lambda^{-2})^{-\half+|I|\rbr{\half-\frac{\beta}{8\pi}}}+(\lambda^{-2})^{|I|\rbr{\half-\frac{\beta}{8\pi}}}}
	\\
	&\qquad\qquad\qquad\qquad\qquad\qquad\times\rbr{\mu s(\mu^{-2})^{-\half+|I^c|\rbr{\half-\frac{\beta}{8\pi}}}+(\mu^{-2})^{|I^c|\rbr{\half-\frac{\beta}{8\pi}}}}\ds\bigg]\bigg\}.
	\\
	\end{split}
\end{equation}

We may write for the part inside the $\{\}$-brackets
\begin{equation}
8(n-1)n^{n-2}J_1(\lambda,\mu)+\sum_{I\varsubsetneq [n]}|I|^{|I|-1}|I^c|^{|I^c|-1}\rbr{2J_2(\lambda,\mu)+J_3(\lambda,\mu)}
\end{equation}
with obvious identifications for $J_i$, $i=1,2,3$. Thus, the final task is to prove
\begin{equation}
\int_0^m\int_0^\lambda J_i(\lambda,\mu)\d\mu\d\lambda\overset{m\to 0}{\longrightarrow}0
\end{equation}
for all $i=1,2,3$. Recall that we assumed $\mu\leq \lambda$. 

First we have,
\begin{equation}
\begin{split}
\int_0^m\int_0^\lambda J_1(\lambda,\mu)\d\mu\d\lambda &\leq B\int_0^m\lambda(\lambda^{-2})^{n\rbr{\half-\frac{\beta}{8\pi}}+1}\int_0^\lambda\mu\d\mu\d\lambda\leq B\int_0^m\lambda^{n\rbr{\frac{\beta}{4\pi}-1}+1}\d\lambda\overset{m\to 0}{\longrightarrow}0
\end{split}
\end{equation}
if $\beta>(1-2/n)4\pi$, which is satisfied for all $n\leq N$. In the first step we have simply set $t=0$ in the $s$-integral since it is convergent for all $\beta<4\pi$ in the neighborhood of the origin. 

The second term yields 
\begin{equation}
\begin{split}
\int_0^m\int_0^\lambda &J_2(\lambda,\mu)\d\mu\d\lambda
\\
&\leq C\int_0^m\lambda^{|I|\rbr{\frac{\beta}{4\pi}-1}+2}\int_0^\lambda\mu([\lambda^2+\mu^2]^{-1})^{\frac{3}{2}+|I^c|\rbr{\half-\frac{\beta}{8\pi}}}\int_{\half+\half\frac{\mu^2}{\lambda^2}}^{\half+\half\frac{\lambda^2}{\mu^2}}e^{-u}u^{\half+|I^c|\rbr{\half-\frac{\beta}{8\pi}}}\du\d\mu\d\lambda
\\
&\quad+B\int_0^m\lambda^{|I|\rbr{\frac{\beta}{4\pi}-1}}\int_0^\lambda\mu([\lambda^2+\mu^2]^{-1})^{\frac{1}{2}+|I^c|\rbr{\half-\frac{\beta}{8\pi}}}\int_{\half+\half\frac{\mu^2}{\lambda^2}}^{\half+\half\frac{\lambda^2}{\mu^2}}e^{-u}u^{-\half+|I^c|\rbr{\half-\frac{\beta}{8\pi}}}\du\d\mu\d\lambda
\\
&\leq B \int_0^m\lambda^{|I|\rbr{\frac{\beta}{4\pi}-1}}\int_0^\lambda \mu^{|I^c|\rbr{\frac{\beta}{4\pi}-1}}\d\mu\d\lambda
\\
&\leq B m^{n\rbr{\frac{\beta}{4\pi}-1}+2}\overset{m\to 0}{\longrightarrow}0,
\end{split}
\end{equation}
if $\beta>(1-1/|I^c|)4\pi$, $\beta>(1-2/|I|)4\pi$ and $\beta>(1-2/n)4\pi$, and all of these are satisfied for all $n\leq N$ and $\emptyset \neq I\varsubsetneq [n]$ as before. The $u$-integrals are bounded uniformly in $\lambda$ and $\mu$, and we have used $\mu/\sqrt{\lambda^2+\mu^2}\leq 1$ and $(\lambda^2+\mu^2)^{-1}\leq \lambda^{-2}$ or $(\lambda^2+\mu^2)^{-1}\leq \mu^{-2}$ appropriately to get the same exponents for $\lambda$ and $\mu$ for both terms.

Finally, the last term yields
\begin{equation}
\begin{split}
\int_0^m\int_0^\lambda& J_3(\lambda,\mu)\d\mu\d\lambda
\\
&=\int_0^m\int_0^\lambda \int_{\mu^{-2}}^\infty e^{-\half(\lambda^2+\mu^2)s}(s\lambda^{|I|\rbr{\frac{\beta}{4\pi}-1}+2}+\lambda^{|I|\rbr{\frac{\beta}{4\pi}-1}})(s\mu^{|I^c|\rbr{\frac{\beta}{4\pi}-1}+2}+\mu^{|I^c|\rbr{\frac{\beta}{4\pi}-1}})\frac{\ds}{s}\d\mu\d\lambda.
\end{split}
\end{equation}
Each of the terms resulting from expanding the brackets above can be reduced to the cases treated in the estimation of the previous term by similar techniques. Thus, we are done.
\end{subsubsection}

\begin{subsubsection}{Proof of {{\Cref{lem:differentiability of Vtn with respect to mass ETC}}}}
In the proofs of \Cref{prop:induction statement,cor:convergence of the Vinfty-Vt terms} we simply assumed that the functions $\tilV_t^n(\xx,\Si|m,\eps)$ are differentiable with respect to the mass $m$ and that the mass derivative can be computed under the $s$-integral in the definition \Cref{eq:def of Vn}. 
In this section, we justify this assumption. 

We begin the proof with the estimates for the functions $\tilde{h}_t^n$. The proof proceeds similarly to that of \Cref{prop:induction statement}, and we therefore keep the exposition brief. For the base case, we can simply choose $\tilde{h}_t^1\equiv h_t^1$ for $t<1<m^{-2}$ and $\tilde{h}_t^1=C'$ from \Cref{eq:uniform bound for V1} in \Cref{lem:uniform bound for V1} for $t\geq 1$. We then assume that the claim holds for all $k<n\leq N$ and define $\tilde{h}_t^n\equiv h_t^n$ for $t<1<m^{-2}$ (since $m\in(0,1)$). For $t\geq 1$, we define
\begin{equation}
\begin{split}
    \tilde{h}_t^n(\xx,\Si)
    &:=\frac{\beta}{2}\bar{B}_N[H_1^n(\xx,\Si)]^{-\frac{\beta}{4\pi}}
    \sum_{I\varsubsetneq[n]}\int_0^1 s^{\frac{\beta}{8\pi}}
    G_s^I(\xx,\Si)
    \tilde{h}_s^{|I|}(\xx_I,\Si_I)
    \tilde{h}_s^{|I^c|}(\xx_{I^c},\Si_{I^c})\ds
    \\
    &\quad+\frac{\beta}{2}\sum_{I\varsubsetneq [n]}
    \int_1^t G_s^I(\xx,\Si)
    \tilde{h}_s^{|I|}(\xx_I,\Si_I)
    \tilde{h}_s^{|I^c|}(\xx_{I^c},\Si_{I^c})\ds,
\end{split}
\end{equation}

where we have split the integral at $s=1$ rather than at $s=m^{-2}$. The analysis for the case $t<1$ is identical to the analysis for $t<m^{-2}$ in the proof of \Cref{prop:induction statement}, modulo the substitution $m^{-2}\to 1$, and yields the same contribution to the constant $B_1(\beta,N,\Lambda)$, which also coincides with the one obtained for the first term below.

For $t\geq 1$ we first use \Cref{lem:estimate for GI function} and then proceed analogously to \Cref{eq:splitting of the n norm}, except that we again use \Cref{eq:trading the measure of Lambda} for the $x_j$-integral if $s\geq 1$ (this is necessary to obtain the correct $t$-dependence). We may assume that $t>e$ since otherwise the logarithmic part below would not appear at all and the resulting estimates are uniform in $t$. Then, after using the induction hypothesis, we need to further split the term with $s\in (1,t)$ at $s=e$ since the estimate of the claim changes at this value. Indeed, for $s\in (1,e)$ the induction hypothesis produces uniform estimates as opposed to the logarithmic ones for $s\geq e$. Thus, we obtain
\begin{equation}
\begin{split}
     \norm{H_t^n(\cdot,\Si)\tilde{h}_t^n(\cdot,\Si)}_{n,\Lambda}
     &\leq \beta\tilde{C}[B_1(\beta,N,\Lambda)]^{n-2}C^n(n-1)n^{n-2}
     \\
     &\quad\times\bigg(
     \bar{B}_NC_1\int_0^1 s^{\frac{\beta}{8\pi}-\frac{3}{2}
     +n\rbr{\half-\frac{\beta}{8\pi}}}\ds
     +|\Lambda|\abr{\int_1^e\frac{\ds}{s}+\int_e^t s^{-1}[\log(s)]^{n-2}\ds}
     \bigg)
     \\
     &\leq\beta \tilde{C}[B_1(\beta,N,\Lambda)]^{n-2}C^n(n-1)n^{n-2}
     \\
     &\quad\times\rbr{
     \frac{C_1\bar{B}_N}{(n-1)\rbr{\half-\frac{\beta}{8\pi}}}
     +|\Lambda|+|\Lambda|\frac{[\log(t)]^{n-1}}{n-1}}
     \\
     &\leq \beta \tilde{C}
     \rbr{\frac{C_1\bar{B}_N}{\half-\frac{\beta}{8\pi}}
     +N|\Lambda|}
     [B_1(\beta,N,\Lambda)]^{n-2}
     C^nn^{n-2}\max(1,[\log(t)]^{n-1}),
\end{split}
\end{equation}
where we also used the bound
\begin{equation}
\int_e^t s^{-1}[\log(s)]^{n-2}\ds
\leq \frac{1}{n-1}[\log(t)]^{n-1}.
\end{equation}
Thus, the claim follows by choosing
\begin{equation}
B_1(\beta,N,\Lambda)
:=\beta \tilde{C}
\rbr{\frac{C_1\bar{B}_N}{\half-\frac{\beta}{8\pi}}
+N|\Lambda|}.
\end{equation}

We now turn to the differentiability claim and the functions $\tilde{g}_t^n$.
The base case $n=1$ follows from \Cref{eq:mass derivative of Vt1} using $m\leq 1$ for the prefactor, $e^{-m^2s}\leq 1$ and $|\tilV_t^1(m,\eps)|\leq \tilde{h}_ t^1$. We then assume that $\tilV_t^k$ is differentiable for all $k<n\leq N$ and that the functions $\tilde{g}_t^k$ with the stated properties exist. Under this assumption, the formal computation in \Cref{eq:mass derivative of Vtn} is justified once the induction step is completed.

Next, we again mimic the proof of \Cref{prop:induction statement} and decompose
$\tilde{g}_t^n=\sum_{i=1}^4\tilde{g}_t^{n,i}$. The terms with $i=1,3,4$ are obtained directly from the functions $g_t^{n,i}$ by replacing
$h_s^{|I|}$ with $\tilde{h}_s^{|I|}$,
$g_s^{|I|}$ with $\tilde{g}_s^{|I|}$ (and analogously for $I^c$),
the prefactor $m$ in $g_t^{n,1}$ with $1$, and $e^{-m^2s}$ with $1$.
The analysis of these terms for $t<1$ is completely analogous to that of the terms $g_t^{n,i}$ with $t<m^{-2}$, producing the same result.

The term $\tilde{g}_t^{n,2}$ is obtained by the same transformations from $g_t^{n,2}$ in the case $t<1$. For $t\geq 1$, we modify the estimates
\Cref{eq:estimate for the tricky term for gtn part 1,eq:estimate for the tricky term for gtn part 2}
by splitting the integrals at $s=1$ rather than at $s=m^{-2}$. For $s\leq 1\leq t$, we obtain
\begin{equation}
\begin{split}
    \rbr{\int_s^t re^{-m^2r}\rbr{\sum_{k,l=1}^n\sigma_{kl}C_r(k,l)}\dr}
    e^{-\frac{\beta}{2}\sum_{kl}\sigma_{kl}K_{s,t}^m(k,l)}
    &\leq\tilde{B}\bar{B}_NN^2 s^{\frac{\beta}{8\pi}}
    \rbr{\int_s^1 r^{-\frac{\beta}{8\pi}}\dr+\int_1^t\dr}
    \\
    &\leq\tilde{B}\bar{B}_NN^2 s^{\frac{\beta}{8\pi}}
    \int_0^t \rbr{r^{-\frac{\beta}{8\pi}}+1}\dr
    \\
    &\leq \tilde{B}\bar{B}_NN^2
    s^{\frac{\beta}{8\pi}}t
    \rbr{\frac{1}{1-\frac{\beta}{8\pi}}+1},
\end{split}
\end{equation}
where, in the last line, we used the fact that $t^{-\frac{\beta}{8\pi}}\leq 1$ for $t\geq 1$.
For $t\geq s\geq 1$, we instead use
\begin{equation}
    \rbr{\int_s^t re^{-m^2r}\rbr{\sum_{k,l=1}^n\sigma_{kl}C_r(k,l)}\dr}
    e^{-\frac{\beta}{2}\sum_{kl}\sigma_{kl}K_{s,t}^m(k,l)}
    \leq n^2\int_s^t \dr \leq N^2 t.
\end{equation}

For $t\geq 1$, we then define
\begin{equation}
\begin{split}
    \tilde{g}_t^{n,2}(\xx,\Si)
    &:=\half\beta^2\tilde{B}\bar{B}_NN^2
    \rbr{\frac{1}{1-\frac{\beta}{8\pi}}+1}t
    \sum_{I\varsubsetneq [n]}
    \int_0^1 s^{\frac{\beta}{8\pi}}
    G_s^I(\xx,\Si)
    \tilde{h}_s^{|I|}(\xx_I,\Si_I)
    \tilde{h}_s^{|I^c|}(\xx_{I^c},\Si_{I^c})\ds
    \\
    &\quad+\half\beta^2N^2t
    \sum_{I\varsubsetneq [n]}
    \int_1^t G_s^I(\xx,\Si)
    \tilde{h}_s^{|I|}(\xx_I,\Si_I)
    \tilde{h}_s^{|I^c|}(\xx_{I^c},\Si_{I^c})\ds .
\end{split}
\end{equation}

We may again assume that $t\geq e$ since otherwise the logarithmic contribution will again not appear and the task becomes simpler. We can now estimate the contributions from the range $t\geq 1$. 

Analogously to the estimates for the functions $\tilde{h}_t^n$, we obtain for the first term
\begin{equation}
    \begin{split}
    \label{eq:estimate for gtilde1}
    \norm{\tilde{g}_t^{n,1}(\cdot,\Si)}_{n,\Lambda}
    &\leq 2\beta  \tilde{C}
    [B_1(\beta,N,\Lambda)]^{n-2}C^n(n-1)n^{n-2}
    \\
    &\quad\times
    \bigg(
    C_1\int_0^1 s^{-\half+n\rbr{\half-\frac{\beta}{8\pi}}}\ds
    +|\Lambda|\int_1^e\ds
    +|\Lambda|\int_e^t[\log(s)]^{n-2}\ds
    \bigg)
    \\
    &\leq 2\beta\tilde{C}
    \rbr{C_1\sup_{n\leq N}\frac{1}{\half-\frac{n}{n+1}\frac{\beta}{8\pi}}+N|\Lambda|}
    [B_1(\beta,N,\Lambda)]^{n-2}
    C^nn^{n-2}
    t\max(1,[\log(t)]^{n-1}),
    \end{split}
\end{equation}
where $B_1(\beta,N,\Lambda)$ is the constant obtained in the estimate for $\tilde{h}_t^n$. In this estimate, we multiplied the first term in the second row of \Cref{eq:estimate for gtilde1} by $t\geq 1$ and used
\begin{equation}
\int_1^e\ds+\int_e^t[\log(s)]^{n-2}\ds\leq \int_0^t\ds+t\int_1^ts^{-1}[\log(s)]^{n-2}\ds=\rbr{1+(n-1)^{-1}[\log(t)]^{n-1}}t.
\end{equation}
Note also that the first $s$-integral in the second row of \Cref{eq:estimate for gtilde1} is convergent for all $\beta<(1+1/n)4\pi$.

For the second case, we obtain analogously 
\begin{equation}
    \begin{split}
        \norm{\tilde{g}_t^{n,2}(\cdot,\Si)}_{n,\Lambda}
        &\leq \beta^2\tilde{C}N^2
        [B_1(\beta,N,\Lambda)]^{n-2}
        C^n(n-1)n^{n-2}t 
        \\
        &\quad\times\bigg(
        \bar{B}_N\tilde{B}C_1
        \rbr{\frac{1}{1-\frac{\beta}{8\pi}}+1}
        \int_0^1 s^{\frac{\beta}{8\pi}-\frac{3}{2}
        +n\rbr{\half-\frac{\beta}{8\pi}}}\ds
        \\
        &\qquad\qquad
        +|\Lambda|\rbr{\int_1^e\frac{\ds}{s}+\int_e^ts^{-1}[\log(s)]^{n-2}\ds}
        \bigg)
        \\
        &\leq \beta^2\tilde{C}N^2
        \abr{
        \tilde{B}\bar{B}_N
        \rbr{\frac{1}{1-\frac{\beta}{8\pi}}+1}
        \frac{1}{\half-\frac{\beta}{8\pi}}
        +|\Lambda|N}
        \\
        &\quad\times
        [B_1(\beta,N,\Lambda)]^{n-2}
        C^nn^{n-2}
        t\max(1,[\log(t)]^{n-1}),
    \end{split}
\end{equation}
where we used manipulations analogous to those in the estimate for $\tilde{h}_t^n$.
Finally, we have
\begin{equation}
    \begin{split}
        \norm{\tilde{g}_t^{n,3}(\cdot,\Si)}_{n,\Lambda}
        &\leq \beta\tilde{C}
        [\tilde{B}_1(\beta,N,\Lambda)]^{|I|-1}
        [B_1(\beta,N,\Lambda)]^{|I^c|-1}
        C^n(n-1)n^{n-2}
        \\
        &\quad\times\bigg(
        C_1\int_0^1 s^{-\half+n\rbr{\half-\frac{\beta}{8\pi}}}\ds
        +|\Lambda|\int_1^e\ds
        +|\Lambda|\int_e^t[\log(s)]^{n-2}\ds
        \bigg)
        \\
        &\leq\beta\tilde{C}
        \rbr{C_1\sup_{n\leq N}\frac{1}{\half-\frac{n}{n+1}\frac{\beta}{8\pi}}+N|\Lambda|}
        \abr{\max(B_1(\beta,N,\Lambda),
        \tilde{B}_1(\beta,N,\Lambda))}^{n-2}
        C^nn^{n-2}
        \\
        &\quad\times
        t\max(1,[\log(t)]^{n-1}),
    \end{split}
\end{equation}
where we used the same manipulations as for the term $\tilde{g}_t^{n,1}$.

We conclude by choosing $\tilde{B}_1(\beta,N,\Lambda)$ as the maximum of
$B_1(\beta,N,\Lambda)$, $\tilde{B}_1'(\beta,N)$ (the constant arising from the $t<1$ contribution, which is, as mentioned, the same as \Cref{eq:def of Btilde1'} in the proof of \Cref{prop:induction statement}), and
\begin{equation}
\begin{split}
4\beta\tilde{C}
\rbr{C_1\sup_{n\leq N}\frac{1}{\half-\frac{n}{n+1}\frac{\beta}{8\pi}}+N|\Lambda|}
+\beta^2\tilde{C}N^2
\abr{
\tilde{B}\bar{B}_N C_1
\rbr{\frac{1}{1-\frac{\beta}{8\pi}}+1}
\frac{1}{\half-\frac{\beta}{8\pi}}
+N|\Lambda|} .
\end{split}
\end{equation}

This completes the proof of differentiability and of the norm estimates. Note that the majorants $\tilde{g}_t^n$ are independent of $\eps$, so the estimates are also applicable in the case $\eps=0$. Thus, $\tilV_t^n(\xx,\Si|m):=\lim_{\eps\to 0}\tilV_t^n(\xx,\Si|m,\eps)$ is also differentiable with respect to $m$ and respects the same estimates. Differentiability on $(0,1)$ implies continuity on $(0,1)$, and by the discussion in \Cref{rem:limits of the Vtn functions}, continuity at $m=0$ holds by definition. This concludes the proof.

\end{subsubsection}
\end{subsection}

\end{section}

\begin{section}{Analysis of the renormalized partition function}
\label{sec:Analysis of the partition function}
The main tasks of this section are to prove \Cref{thm:renormalizability of the partition function} and to provide further tools for the proof of \Cref{thm:existence of the sine Gordon correlation functions,thm:existence of the field}. Recall from \Cref{sec:main results} that \Cref{thm:renormalizability of the partition function} is equivalent to parts 1--3 of \Cref{thm:convergence of the partition function} (with the choice $\eta=-z\1_{x\in \Lambda}$), which is stated below. Part 4 of \Cref{thm:convergence of the partition function} is needed in the proofs of \Cref{thm:existence of the sine Gordon correlation functions,thm:existence of the field}. The three subsections below provide preliminaries for the proof of \Cref{thm:convergence of the partition function}. In \Cref{sec:Uniform bounds for partition function} we prove uniform bounds for the renormalized partition function $\Zcal_R$, \Cref{sec:Expansion of the renormalized partition function} provides an everywhere convergent power series representation for $\Zcal_R$, and in \Cref{sec:convergence of the partition function} we prove that this expansion also has a limit as we first take $\eps\to 0$ and then $m\to 0$.

In what follows, we treat all dimensions on equal footing until the point where estimates derived in the previous section are applied. We first fix $\beta\in(0,4d\pi)$ and choose $N(\beta)\equiv N$ even and such that
$\beta\in[\beta_N,\beta_{N+2})$, where $\beta_N:=(1-1/N)4d\pi$. The restriction
$\beta\in(0,(d+1)2\pi)$ is also in force in this section for $d\geq 2$. Nevertheless, we perform all computations in arbitrary dimension. Thus, if we were to obtain more refined estimates, analogous to the results in \Cref{sec: the renormalized potential} and the techniques used in their proofs, in the full parameter range $(0,8\pi)$ for $d=2$, then all remaining arguments would carry over without difficulty as long as these estimates satisfy the assumptions of \Cref{lem:Gaussian tails} below. All of the results in this section have their analogues in \cite[Section 5]{BaWe24a}.

We begin by defining the renormalized generalized partition function analogously to
\Cref{eq:definition of the renormalized partition function of the sG model,eq:definition of the multiplicative counter terms}: Let $\eta\in L_c^\infty(\R^d\times\{-1,1\})$ and
\begin{equation}
\label{eq:def of renormalized general partition function}
\begin{split}
\Cal{Z}_R(\beta,\eta|m,\eps)
&:=\Cal{Z}_0(\beta,\eta|m,\eps)\prod_{k=1}^{N/2}Z_k(\eta)
\\
&:=\Cal{Z}_0(\beta,\eta|m,\eps)
\\
&\qquad\times\prod_{k=1}^{N/2}
e^{-\frac{1}{(2k)!}
\sum_{\Si\in\{-1,1\}^{2k}}
\1_{\sum_{i=1}^{2k}\sigma_i=0}
\int_{\R^{2kd}}
\Bigl(\prod_{i=1}^{2k}\eta_i\Bigr)
\E_{m,\eps}^{T}\abr{\eps^{-\frac{\beta}{4\pi}}
e^{i\sqrt{\beta}\sigma_j\vp_j}\mid j\in[2k]}
\d\xx}
\\
&\,\,=\E_{m,\sqrt{t}}\abr{e^{-S_t^R(\beta,\eta,\vp|m,\eps)}},
\end{split}
\end{equation}
where $\Zcal_0$ was defined in \Cref{eq:def of the generalized partition function}, and the last equality follows from \Cref{thm:the equivalence of St and Vt}. Recall that the renormalized series $S_t^R$ was defined in \Cref{eq:def of StR}. We have also used the shorthand notations $\eta_i:=\eta(x_i,\sigma_i)$ and $\vp_j:=\vp(x_j)$ as before.

We now state the main convergence result for the partition function.
\begin{theorem}[Convergence of the partition function]
\label{thm:convergence of the partition function}
Let $\beta\in(0,4\pi)$ if $d=1$ and $\beta\in(0,(d+1)2\pi)$ if $d\geq 2$, and let
$\eta\in L_c^\infty(\R^d\times\{-1,1\})$. Then the following statements hold.
\begin{enumerate}
\item The limits
\begin{equation}
\begin{split}
\lim_{\eps\to 0}\Zcal_R(\beta,\eta|m,\eps)
&=:\Zcal_R(\beta,\eta|m),
\\
\lim_{m\to 0}\lim_{\eps\to 0}\Zcal_R(\beta,\eta|m,\eps)
&=:\Zcal_R(\beta,\eta)
\end{split}
\end{equation}
exist and are finite.

\item The maps $z\mapsto \Zcal_R(\beta,z\eta|m)$ and $z\mapsto \Zcal_R(\beta,z\eta)$ define entire functions of $z\in\C$. Moreover, $\Zcal_R(\beta,z\eta)$ is an even function, i.e.,
$\Zcal_R(\beta,z\eta)=\Zcal_R(\beta,-z\eta)$.

\item If $\eta(x,1)=\overline{\eta(x,-1)}$ for a.a.\ $x\in\R^d$, then
$\Zcal_R(\beta,\eta|m)>0$ and $\Zcal_R(\beta,\eta)>0$.

\item Suppose that $\eta_\alpha$ depends on complex parameters $\alpha\in\C^k$ for some $k\in\N$, and that
$\eta_\alpha(\cdot|m,\eps)$ depends on $m,\eps\in(0,1)$ and the same parameters in such a way that, for every compact set $K\subset\C^k$, there exists a compact set $\Lambda\subset\R^d$ such that the following three statements hold
\begin{equation}
\bigcup_{\alpha\in K}\supp(\eta_\alpha)\subset \Lambda\times \{-1,1\} \quad \text{ and } \quad \bigcup_{\alpha\in K}\bigcup_{m,\eps\in (0,1)}\supp(\eta_\alpha(\cdot,\cdot|m,\eps))\subset \Lambda\times \{-1,1\},
\end{equation}
\begin{equation}
\sup_{\alpha\in K}\norm{\eta_\alpha}_{L^\infty(\R^d\times\{-1,1\})}<\infty \quad \text{ and } \quad \sup_{m,\eps\in (0,1)}\sup_{\alpha\in K}\norm{\eta_\alpha(\cdot,\cdot|m,\eps)}_{L^\infty(\R^d\times\{-1,1\})}<\infty,
\end{equation}
and 
\begin{equation}
\lim_{m\to 0}\limsup_{\eps\to 0}
\sup_{\alpha\in K}
\norm{\eta_\alpha(\cdot,\cdot|m,\eps)-\eta_\alpha}_{L^\infty(\R^d\times\{-1,1\})}
=0.
\end{equation}
Then
\begin{equation}
\lim_{m\to 0}\limsup_{\eps\to 0}\sup_{\alpha\in  K}
\bigl|\Zcal_R(\beta,\eta_\alpha(\cdot,\cdot|m,\eps)|m,\eps)
-\Zcal_R(\beta,\eta_\alpha)\bigr|
=0.
\end{equation}
An analogous statement also holds for fixed $m\in(0,1)$.
\end{enumerate}
\end{theorem}

For the proof of \Cref{thm:convergence of the partition function}, we require several preliminary results. As a first step, one could invoke \cite[Lemma~5.2]{BaWe24a}, since the proof given there readily extends to dimensions other than $d=2$. However, we anticipate that in higher dimensions, and for $\beta$ beyond $(d+1)2\pi$—in particular for $\beta\in[6\pi,8\pi)$ when $d=2$—additional input will be necessary. We therefore present a slightly more general statement in arbitrary dimension, postponing its proof to Appendix \ref{sec:proof of Gaussian tails}, since it closely follows the proof of \cite[Lemma~5.2]{BaWe24a}.

Recall that $\vp_{m,\sqrt{t}}$ is the centered Gaussian process on $\R^d$ with covariance
\begin{equation}
K_{\sqrt{t}}^m(x,y):=\int_t^\infty C_s^m(x,y)\ds,
\end{equation}
whose law is supported on $C^\infty(\R^d)$. We denote expectation with respect to this law by $\E_{m,\sqrt{t}}$. In the lemma below, we restrict $\vp_{m,\sqrt{t}}$ to a compact set $\Lambda\subset\R^d$, but for simplicity, we write
$\vp_{m,\sqrt{t}}|_\Lambda\equiv\vp_{m,\sqrt{t}}$.

For a sufficiently smooth function $f\colon\R^d\to\R$, we define
\begin{equation}
\begin{split}
p_{\nu,\Lambda}(f)
&:=\sup_{x\in\Lambda}\bigl|(\partial^\nu f)(x)\bigr|,
\end{split}
\end{equation}
where $\nu$ is a multi-index of order $k\geq 1$.

We then have the following generalization of \cite[Lemma~5.2]{BaWe24a}.

\begin{lemma}
\label{lem:Gaussian tails}
Let $\Lambda\subset\R^d$ be compact, $B>0$, and $1\leq\alpha<2$, $k\in\N$ and $t\in (0,\infty)$. Then there exists a constant
$C_{k,B,\alpha,t,\Lambda}<\infty$ such that
\begin{equation}
\E_{m,\sqrt{t}}\abr{e^{B[p_{\nu,\Lambda}(f)]^\alpha}}
\leq C_{k,B,\alpha,t,\Lambda},
\end{equation}
where the constant $C_{k,B,\alpha,t,\Lambda}>0$ is independent of $m$.
\end{lemma}

\begin{remark}
By \Cref{eq:estimate for derivatives of the massive heat kernel},
\begin{equation}
\label{eq:covariance for the derivative fields}
K_{\sqrt{t}}^{m,\nu}(x,y)
:=(\partial_1^\nu\partial_2^\nu K_{\sqrt{t}}^m)(x,y)
=\int_t^\infty(\partial_1^\nu\partial_2^\nu C_s^m)(x,y)\ds .
\end{equation}
Thus, $\partial^\nu\vp_{m,\sqrt{t}}$ is a smooth centered Gaussian field with covariance
\Cref{eq:covariance for the derivative fields}, by Kolmogorov--Chentsov type arguments (see, for example, \cite[Appendix B]{LaRhVa15a}).
\end{remark}

\begin{remark}
Note that $\lim_{t\to 0}C_{k,B,\alpha,t,\Lambda}$ is not finite, since the field ceases to be differentiable, or even an ordinary function, in this limit. Similarly, $\lim_{\alpha\to 2}C_{k,B,\alpha,t,\Lambda}$ cannot be finite for all values of $B$ by Gaussianity. Our proof also breaks down in the limit $|\Lambda|\to\infty$. However, none of these divergences is restrictive for our application.
\end{remark}

\begin{subsection}{Uniform bounds for $\Zcal_R$}
\label{sec:Uniform bounds for partition function}
We now turn to proving uniform bounds for $\Zcal_R$.

\begin{proposition}[Uniform bounds]
\label{prop:uniform bounds}
Fix a compact set $\Lambda\subset\R^d$ and let $\beta\in(0,4\pi)$ for $d=1$ and
$\beta\in(0,(d+1)2\pi)$ for $d\geq 2$. Then the following statements hold.
\begin{enumerate}
\item For any fixed $M>0$,
\begin{equation}
\sup_{m,\eps\in (0,1)}
\sup_{\substack{\eta\in L_c^\infty(\R^d\times \{-1,1\}):\\
\supp(\eta)\subset \Lambda\times \{-1,1\}\\
\norm{\eta}_\infty\leq M}}
|\Zcal_R(\beta,\eta|m,\eps)|<\infty .
\end{equation}

\item Let $\overline{L_c^\infty}(\R^d\times \{-1,1\})$ denote the set of
$\eta\in L_c^\infty(\R^d\times\{-1,1\})$ satisfying
$\overline{\eta(\cdot,1)}=\eta(\cdot,-1)$. Then $\Zcal_R$ is positive, and for all $M>0$ we have
\begin{equation}
 \inf_{m,\eps\in (0,1)}
 \inf_{\substack{\eta\in \overline{L_c^\infty}(\R^d\times \{-1,1\}):\\
 \supp(\eta)\subset \Lambda\times \{-1,1\}\\
 \norm{\eta}_\infty\leq M}}
 \Zcal_R(\beta,\eta|m,\eps)>0 .
\end{equation}

\item There exists $\delta\equiv\delta_{\Lambda,\beta}>0$, independent of $\eps$ and $m$, such that
\begin{equation}
\inf_{m,\eps\in (0,1)}
\inf_{\substack{\eta\in L_c^\infty(\R^d\times \{-1,1\}):\\
\supp(\eta)\subset \Lambda\times \{-1,1\}\\
\norm{\eta}_\infty\leq \delta}}
|\Zcal_R(\beta,\eta|m,\eps)|>0 .
\end{equation}
\end{enumerate}
\end{proposition}

\begin{proof}
For the first claim, fix $M>0$, choose $\beta$ in the appropriate interval, and let
$N\equiv N(\beta)$ as usual. Then, by the bounds \Cref{eq:uniform bound for Acaln,eq:uniform bound for Acalphi} derived in the proof of the first claim in
\Cref{thm:the equivalence of St and Vt}, we obtain
\begin{equation}
\begin{split}
|\Zcal_R(\beta,\eta|m,\eps)|
&=\abs{\E_{m,\sqrt{t}}\abr{e^{-S_t^R(\beta,\eta,\vp|m,\eps)}}}
\\
&\leq \E_{m,\sqrt{t}}\abr{e^{|S_t^R(\beta,\eta,\vp|m,\eps)|}}
\\
&\leq \exp\!\rbr{\sum_{n=1}^N|\Acal_n^\infty(\beta,t|m,\eps)|}
\E_{m,\sqrt{t}}\abr{\exp\!\rbr{\sum_{n=1}^\infty|\Acal_n^\vp(\beta,t|m,\eps)|}}
\\
&\leq \exp\!\rbr{NC(\beta,N,t,\eta)}
\\
&\quad\times
\E_{m,\sqrt{t}}\abr{\exp\!\rbr{
\bigl(2[\sqrt{\beta t}\norm{\grad \vp}_{L^\infty(B_\Lambda)}\vee 1]+1\bigr)
\frac{t^{-\frac{d}{2}}}{B_d(\beta,N)}
\sum_{n=1}^\infty
\bigl(\hat{B}(\beta,N)M t^{\frac{d}{2}-\frac{\beta}{8\pi}}\bigr)^n }} .
\end{split}
\end{equation}
The renormalized series $S_t^R$ and the summands $\Acal_n^\infty$ and $\Acal_n^\vp$
were defined in \Cref{eq:def of StR} and \Cref{eq:splitting of the counter terms},
respectively. The constants $B_d(\beta,N)$,
$\hat{B}(\beta,N)=2eB_d(\beta,N)C$, and $C(\beta,N,t,\eta)$ were derived in the proofs
of \Cref{prop:induction statement,cor:convergence of the Vinfty-Vt terms} and are
independent of $m$ and $\eps$. More precisely, $B_d$ is defined in \Cref{eq:def of B1betaN} and \Cref{eq:def of the constant BdbetaN for d larger than 2 and small t} for the cases $d=1$ and $d\geq 2$, respectively, and the precise form of $C(\beta,N,t,\eta)$ is irrelevant to our analysis. We also used the assumption $\norm{\eta}_\infty\leq M$. Recall that $B_\Lambda:=\overline{B_{2R(\Lambda)}(0)}$, where $R(\Lambda)$ is chosen so that
$\Lambda\subset B_{R(\Lambda)}(0)$ and $\Lambda$ is compact with
$\supp(\eta)\subset\Lambda$. We now choose $t>0$ such that
$\hat{B}(\beta,N)M t^{\frac{d}{2}-\frac{\beta}{8\pi}}<1$, so that the series above converges geometrically. The first claim then follows from \Cref{lem:Gaussian tails}.

For the second claim, observe that if
$\overline{\eta(\cdot,1)}=\eta(\cdot,-1)$, then the exponent in
$\Zcal_R=\E_{m,\sqrt{t}}[e^{-S_t^R}]$ is real-valued. Hence,
$|e^{-x}|\geq e^{-|x|}$, and arguing as above, together with
\Cref{lem:Gaussian tails} and Jensen’s inequality in the form
$\E[X^{-1}]\geq \E[X]^{-1}$, yields the desired lower bound.

Finally, we turn to the third claim. The proof is completely analogous to the proof of \cite[Proposition 5.3 (ii)]{BaWe24a}, but we reproduce it here for the convenience of the reader. First fix $\eta$ with $\norm{\eta}_\infty=1$ and consider
the function $z\mapsto \Zcal_R(\beta,z\eta|m,\eps)$. Note that
$\Zcal_0(\beta,\eta|m,\eps)=\E_{m,\eps}\!\rbr{\exp(-V_0(\beta,\eta,\vp|\eps))}$, where,
for fixed $m,\eps\in(0,1)$,
\begin{equation}
V_0(\beta,\eta,\vp_{m,\eps}|\eps)
:=\eps^{-\frac{\beta}{4\pi}}
\sum_{\sigma\in\{-1,1\}}
\int_{\R^d}\eta(x,\sigma)e^{i\sqrt{\beta}\sigma\vp_{m,\eps}(x)}\dx
\end{equation}
is a deterministically bounded random variable. By Fubini’s theorem and Morera’s
theorem, it follows that $z\mapsto \Zcal_R(\beta,z\eta|m,\eps)$ is entire (for every
$\omega$ in the underlying probability space).

Let $K\subset\C$ be compact and choose $R>0$ such that $K\subset B_R(0)$. Then, for
$z\in K$, the Cauchy integral formula yields
\begin{equation}
\der{z}\Zcal_R(\beta,z\eta|m,\eps)
=\frac{1}{2\pi i}
\oint_{\partial B_R(0)}
\frac{\Zcal_R(\beta,\zeta\eta|m,\eps)}{(\zeta-z)^2}\d\zeta .
\end{equation}
By the first part, $\Zcal_R(\beta,\zeta\eta|m,\eps)$ is uniformly bounded (in $m,\eps\in (0,1)$ and $\zeta\in \overline{B_R(0)}$), and the
denominator is uniformly bounded away from zero by the choice of $K$ and $R$. Therefore, the derivative of the map $z\mapsto \Zcal_R(\beta,z\eta|m,\eps)$ is uniformly
bounded on compact subsets of $\C$ and also in $m,\eps\in (0,1)$. Since $\Zcal_R(\beta,0|m,\eps)=1$, this implies
that there exists $\delta>0$, independent of $m$ and $\eps$, such that
\begin{equation}
\label{eq:uniform estimate for ZcalR before scaling}
\inf_{m,\eps\in (0,1)}
\inf_{|z|<\delta}
\inf_{\substack{\eta\in L_c^\infty(\R^d\times \{-1,1\}):\\
\supp(\eta)\subset \Lambda\times \{-1,1\}\\
\norm{\eta}_\infty=1}}
|\Zcal_R(\beta,z\eta|m,\eps)|>0 
\end{equation}
by continuity of the map $z\mapsto \Zcal_R(\beta,z\eta|m,\eps)$. By scaling, this implies the claim.
\end{proof}

Note that the first two claims of \Cref{prop:uniform bounds} are valid for arbitrary $\eta$, which here plays the role of the coupling constant, whereas the third claim holds only for sufficiently
small coupling.
\end{subsection}

\begin{subsection}{Expansion of $\Zcal_R$}
\label{sec:Expansion of the renormalized partition function}
The strategy for proving that the limit
$\lim_{m\to 0}\lim_{\eps\to 0}\Zcal_R(\beta,\eta|m,\eps)$ exists is to first establish the
existence of an everywhere convergent series expansion for
$\Zcal_R(\beta,z\eta|m,\eps)$. More precisely, we have the following result.

\begin{proposition}
\label{prop:expansion for the renormalized partition function}
For fixed $m,\eps\in(0,1)$ and $\eta\in L_c^\infty(\R^d\times\{-1,1\})$, the map
$z\mapsto \Zcal_R(\beta,z\eta|m,\eps)$ is an entire function with the Taylor
expansion
\begin{equation}
\begin{split}
\Zcal_R(\beta,z\eta|m,\eps)
&=\sum_{L=0}^\infty\frac{z^L}{L!}\Mcal_L(\beta,\eta|m,\eps)
\\
&=\sum_{L=0}^\infty\frac{z^L}{L!}
\sum_{\Si\in\{-1,1\}^L}
\int_{\R^{Ld}}
\Bigl(\prod_{i=1}^L\eta_i\Bigr)
\widetilde{\Mcal}_L(\xx,\Si,\beta|m,\eps)\d\xx ,
\end{split}
\end{equation}
where $\eta_i:=\eta(x_i,\sigma_i)$ as before. Recall also that $\vp_j:=\vp(x_j)$. The kernels $\widetilde{\Mcal}_L$ admit two representations. The first is
\begin{equation}
\begin{split}
\label{eq:first representation for tildeMcal}
\widetilde{\Mcal}_L(\xx,\Si,\beta|m,\eps)
&=
\sum_{\pi\in\S_L}
\sum_{\substack{n,j_1,j_2,\dots,j_{N/2}\geq 0\\ L=n+\sum_{k=1}^{N/2}2k j_k}}
\frac{(-1)^n}{n!}
\prod_{l=1}^{N/2}\frac{(-1)^{j_l}}{[(2l)!]^{j_l}j_l!}
\\
&\quad\quad\times
\Biggl(
\prod_{k=1}^{N/2}\prod_{l_k=1}^{j_k}
\1_{\bigl\{\sum_{j\in\pi([[\,n+(l_k-1)2k+1,\;n+2l_k k\,]])}\sigma_j=0\bigr\}}
\Biggr)
\\
&\quad\quad\times
\E_{m,\eps}\abr{
\prod_{a=1}^n
\eps^{-\frac{\beta}{4\pi}}
e^{i\sqrt{\beta}\sigma_{\pi(a)}\vp_{\pi(a)}}
}
\\
&\quad\quad\times
\prod_{k=1}^{N/2}\prod_{l_k=1}^{j_k}
\E_{m,\eps}^T\abr{
\eps^{-\frac{\beta}{4\pi}}
e^{i\sqrt{\beta}\sigma_{p_{l_k}}\vp_{p_{l_k}}}
\mid p_{l_k}\in
\pi([[\,n+2(l_k-1) k+1,\;n+2l_kk\,]])
},
\end{split}
\end{equation}
where $\S_L$ denotes the permutation group of $[L]$ and we set $[[k,l]]:=\{k,k+1,\dots,l\}$ for integers $k<l$.

The second representation is
\begin{equation}
\begin{split}
\label{eq:second representation for tildeMcal}
\widetilde{\Mcal}_L(\xx,\Si,\beta|m,\eps)
&=
\sum_{\substack{p,\, \text{$p$ even}\\ p+q=L}}
\binom{L}{p}
\bigg(
\sum_{\Pi_1\in \mathfrak{P}_{q}}
(-1)^{|\Pi_1|}
\prod_{\pi_1\in\Pi_1}
\tilV_t^{|\pi_1|}(\xx_{\pi_1},\Si_{\pi_1}|m,\eps)
\\
&\qquad\times
\E_{m,\sqrt{t}}\abr{
\prod_{\pi_1\in\Pi_1}
\Bigl(
e^{i\sqrt{\beta}\sum_{j\in\pi_1}\sigma_j\vp_j}
-\1_{|\pi_1|\leq N}
\1_{\sum_{l\in\pi_1}\sigma_l=0}
\Bigr)
}
\bigg)
\\
&\quad\times
\bigg(
\sum_{\Pi_2\in\mathfrak{P}_{[[q+1,L]],\,\mathrm{even}}^{\leq N/2}}
(-1)^{|\Pi_2|}
\prod_{\pi_2\in\Pi_2}
\1_{\sum_{i\in\pi_2}\sigma_i=0}
\\
&\qquad\qquad\times
\rbr{
\tilV_t^{|\pi_2|}(\xx_{\pi_2},\Si_{\pi_2}|m,\eps)
-\tilV_\infty^{|\pi_2|}(\xx_{\pi_2},\Si_{\pi_2}|m,\eps)
}
\bigg),
\end{split}
\end{equation}

where $\mathfrak{P}_A$ denotes the collection of all partitions of a set $A$, and
we write $\mathfrak{P}_q\equiv\mathfrak{P}_{[q]}$. Moreover,
\begin{equation}
\mathfrak{P}_{A,\mathrm{even}}^{\leq N}
:=
\Bigl\{
\Pi\text{ is a partition of }A
\;\Big|\;
\text{each block $\pi$ of $\Pi$ has cardinality $2k$ for some $\N\ni k\leq N/2$}
\Bigr\}.
\end{equation}

Finally, for every $M,\delta>0$ and compact $\Lambda\subset\R^d$, there exists a
constant $0<B(\delta,M,\Lambda)<\infty$, independent of $\eps$, $m$, $\eta$, and
$L$, such that
\begin{equation}
\label{eq:uniform bound for Mcal}
\sup_{m,\eps\in(0,1)}
\sup_{\substack{
\eta\in L_c^\infty(\R^d\times\{-1,1\})\\
\supp(\eta)\subset \Lambda\times\{-1,1\}\\
\norm{\eta}_\infty\leq M}}
|\Mcal_L(\beta,\eta|m,\eps)|
\leq B(\delta,M,\Lambda)\,\delta^L\,L! .
\end{equation}
\end{proposition}

\begin{proof}
To derive the first representation of the kernels, we expand the expression
\begin{equation}
\begin{split}
\Zcal_R(\beta,z\eta|m,\eps)
&=\E_{m,\eps}\abr{e^{-V_0(\beta,z\eta,\vp|m,\eps)}}
\\
&\quad\times
\prod_{k=1}^{N/2}
e^{-\frac{z^{2k}}{(2k)!}
\sum_{\Si\in\{-1,1\}^{2k}}
\1_{\sum_{i=1}^{2k}\sigma_i=0}
\int_{\R^{2kd}}
\Bigl(\prod_{i=1}^{2k}\eta_i\Bigr)
\E_{m,\eps}^{T}
\abr{
\eps^{-\frac{\beta}{4\pi}}
e^{i\sqrt{\beta}\sigma_j\vp_j}
\mid j\in[2k]}
\d\xx } .
\end{split}
\end{equation}
into Taylor series in the variable $z$. Since all relevant random variables are
deterministically bounded, each factor is an entire function (for all
$\omega\in\Omega$) by Fubini’s and Morera’s theorems. Consequently,
$\Zcal_R$ itself is an entire function. We may therefore expand each exponential
separately and collect powers of $z$. In the first factor we may also exchange the
expectation and summation. This yields
\begin{equation}
\begin{split}
\Zcal_R(\beta,z\eta|m,\eps)
&=
\bigg(
\sum_{n=0}^\infty\frac{(-z)^n}{n!}
\sum_{\Si\in\{-1,1\}^n}
\int_{\R^{nd}}
\Bigl(\prod_{i=1}^n\eta_i\Bigr)
\E_{m,\eps}
\abr{
\prod_{j=1}^n
\eps^{-\frac{\beta}{4\pi}}
e^{i\sqrt{\beta}\sigma_j\vp_j}
}
\d\xx
\bigg)
\\
&\quad\times
\prod_{k=1}^{N/2}
\bigg\{
\sum_{j_k=0}^\infty
\frac{(-1)^{j_k}z^{2kj_k}}{[(2k)!]^{j_k}j_k!}
\sum_{\Si\in\{-1,1\}^{2kj_k}}
\Bigl(\prod_{b=0}^{j_k-1}
\1_{\sum_{l_b\in[[2bk+1,(b+1)2k]]}\sigma_{l_b}=0}\Bigr)
\\
&\qquad\qquad\times
\int_{\R^{2kj_k d}}
\Bigl(\prod_{i=1}^{2kj_k}\eta_i\Bigr)
\prod_{a=0}^{j_k-1}
\E_{m,\eps}^{T}
\abr{
\eps^{-\frac{\beta}{4\pi}}
e^{i\sqrt{\beta}\sigma_{p_a}\vp_{p_a}}
\mid p_a\in[[2ak+1,(a+1)2k]]
}
\d\xx
\bigg\}
\\
&=:
\bigg(\sum_{n=0}^\infty F_{0,n}z^n\bigg)
\prod_{k=1}^{N/2}
\sum_{j_k=0}^\infty F_{k,j_k}z^{2kj_k},
\end{split}
\end{equation}
where we use the notation $[[k,l]]:=\{k,k+1,\dots,l\}$ for integers $k<l$. The final
line serves only to simplify notation when combining the sums.

We now sum over $L=n+\sum_{k=1}^{N/2}2kj_k$ to obtain
\begin{equation}
\sum_{L=0}^\infty z^L
\sum_{\substack{
L=n+\sum_{k=1}^{N/2}2kj_k\\
n,j_1,j_2,\dots,j_{N/2}\geq 0}}
F_{0,n}\prod_{k=1}^{N/2}F_{k,j_k}.
\end{equation}
Next, we extract the integrals from the factors $F_{0,n}$ and $F_{k,j_k}$, exchange
them with the inner sum, and symmetrize the resulting expression to introduce the
factor $1/L!$. This yields the formula
\begin{equation}
\begin{split}
\Zcal_R(\beta,z\eta|m,\eps)
&=\sum_{L=0}^\infty\frac{z^L}{L!}
\sum_{\Si\in\{-1,1\}^L}
\int_{\R^{Ld}}
\Biggl\{
\Bigl(\prod_{i=1}^L\eta_i\Bigr)
\sum_{\pi\in\S_L}
\sum_{\substack{
L=n+\sum_{k=1}^{N/2}2kj_k\\
n,j_1,j_2,\dots,j_{N/2}\geq 0}}
\\
&\qquad\quad\times
\frac{(-1)^n}{n!}
\rbr{\prod_{l=1}^{N/2}
\frac{(-1)^{j_l}}{[(2l)!]^{j_l}j_l!}}
\prod_{k=1}^{N/2}\prod_{l_k=1}^{j_k}
\1_{\bigl\{\sum_{j\in\pi([[\,n+(l_k-1)2k+1,\;n+2l_k k\,]])}\sigma_j=0\bigr\}}
\\
&\qquad\quad\times\E_{m,\eps}
\abr{
\prod_{a=1}^n
\eps^{-\frac{\beta}{4\pi}}
e^{i\sqrt{\beta}\sigma_{\pi(a)}\vp_{\pi(a)}}
}
\\
&\qquad\quad\times
\prod_{k=1}^{N/2}\prod_{l_k=1}^{j_k}
\E_{m,\eps}^{T}
\abr{
\eps^{-\frac{\beta}{4\pi}}
e^{i\sqrt{\beta}\sigma_{p_{l_k}}\vp_{p_{l_k}}}
\mid p_{l_k}\in
\pi([[\,n+2(l_k-1) k+1,\;n+2l_kk\,]])
}
\Biggr\}
\d\xx,
\end{split}
\end{equation}
from which the desired representation can be seen.

To derive the second representation, we write
\begin{equation}
\Zcal_R(\beta,z\eta|m,\eps)=F_1(z)F_2(z),
\end{equation}
where the entire functions $F_1$ and $F_2$ are given by
\begin{equation}
F_1(z)
:=\E_{m,\sqrt{t}}
\abr{
\exp(
-\widetilde{S}_R^t(\beta,z\eta,\vp|m,\eps)
)
},
\end{equation}
with
\begin{equation}
\widetilde{S}_R^t(\beta,z\eta,\vp|m,\eps):=
S_t(\beta,z\eta,\vp|m,\eps)
-\sum_{k=1}^{N/2}\frac{z^{2k}}{(2k)!}
\sum_{\Si\in\{-1,1\}^{2k}}
\1_{\sum_{i=1}^{2k}\sigma_i=0}
\int_{\R^{2kd}}
\Bigl(\prod_{j=1}^{2k}\eta_j\Bigr)
\tilV_t^{2k}(\xx,\Si|m,\eps)
\d\xx,
\end{equation}
where $S_t$ is the series defined in \Cref{eq:def of S},
and
\begin{equation}
\begin{split}
F_2(z)
&:=
\prod_{k=1}^{N/2}
\exp\!\bigg(
-\frac{z^{2k}}{(2k)!}
\sum_{\Si\in\{-1,1\}^{2k}}
\1_{\sum_{i=1}^{2k}\sigma_i=0}
\int_{\R^{2kd}}
\Bigl(\prod_{j=1}^{2k}\eta_j\Bigr)
\abr{
\tilV_t^{2k}(\xx,\Si|m,\eps)
-\tilV_\infty^{2k}(\xx,\Si|m,\eps)
}
\d\xx
\bigg)
\\
&=:\prod_{k=1}^{N/2}\exp(-z^{2k}A_{2k}),
\end{split}
\end{equation}
where $A_{2k}\equiv \Acal_{2k}^\infty(\beta,t|m,\eps)$ as defined in
\Cref{eq:splitting of the counter terms}. For the purposes of this computation we use this shorthand and suppress the dependence of $A_{2k}$ on the parameters.

Since $z\mapsto\Zcal_R(\beta,z\eta|m,\eps)$ is entire, we have
\begin{equation}
\Mcal_L(\beta,\eta|m,\eps)
=\frac{d^L}{dz^L}\bigg|_{z=0}\!\bigl(F_1(z)F_2(z)\bigr)
=\sum_{p=0}^L\binom{L}{p}F_1^{(L-p)}(0)F_2^{(p)}(0).
\end{equation}
The surviving terms in $F_2^{(p)}(0)$ are most easily identified using
Faà~di~Bruno’s formula,
\begin{equation}
\label{eq:Faa di Bruno}
\dder{z}{p} f(g(z))
=\sum_{\Pi\in\mathfrak{P}_p}
f^{(|\Pi|)}(g(z))
\prod_{\pi\in\Pi} g^{(|\pi|)}(z),
\end{equation}
where the sum runs over all partitions $\mathfrak{P}_p$ of $[p]$, and $|\cdot|$
denotes cardinality. In the present case,
\begin{align}
f(z)&=e^{-z},
\\
g(z)&=\sum_{k=1}^{N/2}A_{2k} z^{2k}.
\end{align}
We readily observe that
\begin{equation}
g^{(n)}(0)=
\begin{cases}
(2k)!A_{2k}, & \text{if } n=2k \text{ for some } k\leq N/2,\\
0, & \text{otherwise},
\end{cases}
\end{equation}
and that $f^{(n)}(g(0))=f^{(n)}(0)=(-1)^n$. Hence,
\begin{equation}
F_2^{(p)}(0)=
\begin{cases}
\displaystyle
\sum_{\Pi\in\mathfrak{P}_{p,\mathrm{even}}^{\leq N}}
(-1)^{|\Pi|}\prod_{\pi\in\Pi}|\pi|!\,A_{|\pi|},
& \text{if $p$ is even},\\[1ex]
0, & \text{otherwise}.
\end{cases}
\end{equation}
Recall that $\mathfrak{P}_{p,\mathrm{even}}^{\leq N}$ is the set of partitions of $[p]$ with blocks that have even number (less than or equal to $N$) of elements. Suppose now that $p$ is even. We may then rewrite
\begin{equation}
\begin{split}
F_2^{(p)}(0)
&=
\sum_{\Si\in\{-1,1\}^p}
\int_{\R^{pd}}
\Biggl\{
\Bigl(\prod_{j=1}^p\eta_j\Bigr)
\sum_{\Pi\in\mathfrak{P}_{p,\mathrm{even}}^{\leq N}}
(-1)^{|\Pi|}
\\
&\qquad\qquad\times
\prod_{\pi\in\Pi}
\1_{\sum_{i\in\pi}\sigma_i=0}
\abr{
\tilV_t^{|\pi|}(\xx_\pi,\Si_\pi|m,\eps)
-\tilV_\infty^{|\pi|}(\xx_\pi,\Si_\pi|m,\eps)
}
\Biggr\}
\d\xx .
\end{split}
\end{equation}

For $F_1$ we need the Cauchy integral representation of the derivatives to justify the interchange of differentiation and expectation. For any $R>0$ we have
\begin{equation}
\label{eq:Cauchy integral representation of the derivatives of F_1}
F_1^{(q)}(0)
=\frac{q!}{2\pi i}
\oint_{\partial B_R(0)}\frac{F_1(z)}{z^{q+1}}\dz .
\end{equation}
Furthermore, if $R<t^*$, where $t^*$ is defined in \Cref{eq:def of tstar}, then Fubini’s theorem
(justified by the proof of the first claim in \Cref{thm:the equivalence of St and Vt}, more precisely the discussion around \Cref{eq:uniform bound for Acalphi}, and
\Cref{lem:Gaussian tails}) yields
\begin{equation}
\begin{split}
F_1^{(q)}(0)
&=
\E_{m,\sqrt{t}}
\bigg[
\dder{z}{q}\bigg|_{z=0}
\exp\!\bigg(
-\sum_{k=1}^\infty\frac{z^k}{k!}
\\
&\qquad\qquad\quad\times
\sum_{\Si\in\{-1,1\}^k}
\int_{\R^{kd}}
\Bigl(\prod_{i=1}^k\eta_i\Bigr)
\abr{
e^{i\sqrt{\beta}\sum_{j=1}^k\sigma_j\vp_j}
-\1_{\sum_{l=1}^k\sigma_l=0}\1_{k\leq N}
}
\tilV_t^k(\xx,\Si|m,\eps)
\d\xx
\bigg)
\bigg],
\end{split}
\end{equation}
where we used the Cauchy representation in the reverse direction again after interchanging the expectation and the complex line integral from \Cref{eq:Cauchy integral representation of the derivatives of F_1}.
As in the computation for $F_2$, the derivatives can now be evaluated using
Faà~di~Bruno's formula, without restrictions on the partitions. Consequently, after computing the derivatives we can interchange the expectation with the spatial integrals by Fubini's theorem to obtain
\begin{equation}
\begin{split}
F_1^{(q)}(0)
&=
\sum_{\Si\in\{-1,1\}^q}
\int_{\R^{qd}}
\Bigl(\prod_{i=1}^q\eta_i\Bigr)
\sum_{\Pi\in\mathfrak{P}_q}
(-1)^{|\Pi|}
\prod_{\pi\in\Pi}
\tilV_t^{|\pi|}(\xx_\pi,\Si_\pi|m,\eps)
\\
&\qquad\qquad\times
\E_{m,\sqrt{t}}
\abr{
\prod_{\pi\in\Pi}
\Bigl(
e^{i\sqrt{\beta}\sum_{j\in\pi}\sigma_j\vp_j}
-\1_{|\pi|\leq N}\1_{\sum_{l\in\pi}\sigma_l=0}
\Bigr)
}
\d\xx .
\end{split}
\end{equation}

By a suitable relabelling of the integration variables, we finally arrive at
\begin{equation}
\begin{split}
\Mcal_L(\beta,\eta|m,\eps)
&=
\sum_{\Si\in\{-1,1\}^L}
\int_{\R^{Ld}}
\Biggl\{
\Bigl(\prod_{i=1}^L\eta_i\Bigr)
\sum_{\substack{q+p=L\\ \text{$p$ is even}}}
\binom{L}{p}
\\
&\quad\times
\bigg(
\sum_{\Pi_1\in\mathfrak{P}_q}
(-1)^{|\Pi_1|}
\prod_{\pi_1\in\Pi_1}
\tilV_t^{|\pi_1|}(\xx_{\pi_1},\Si_{\pi_1}|m,\eps)
\\
&\qquad\times
\E_{m,\sqrt{t}}
\abr{
\prod_{\pi_1\in\Pi_1}
\Bigl(
e^{i\sqrt{\beta}\sum_{j\in\pi_1}\sigma_j\vp_j}
-\1_{|\pi_1|\leq N}\1_{\sum_{l\in\pi_1}\sigma_l=0}
\Bigr)
}
\bigg)
\\
&\quad\times
\bigg(
\sum_{\Pi_2\in\mathfrak{P}_{[[q+1,L]],\,\mathrm{even}}^{\leq N}}
(-1)^{|\Pi_2|}
\prod_{\pi_2\in\Pi_2}
\1_{\sum_{i\in\pi_2}\sigma_i=0}
\\
&\qquad\qquad\times
\abr{
\tilV_t^{|\pi_2|}(\xx_{\pi_2},\Si_{\pi_2}|m,\eps)
-\tilV_\infty^{|\pi_2|}(\xx_{\pi_2},\Si_{\pi_2}|m,\eps)
}
\bigg)
\Biggr\}
\d\xx .
\end{split}
\end{equation}

Finally, the uniform bound on the integrals $\Mcal_L$ follows from the fact that
$z\mapsto\Zcal_R(\beta,z\eta|m,\eps)$ is entire, and previous uniform bounds. By the Cauchy integral formula, for
any $R>0$,
\begin{equation}
\Mcal_L(\beta,\eta|m,\eps)
=\frac{L!}{2\pi i}
\oint_{\partial B_R(0)}
\frac{\Zcal_R(\beta,w\eta|m,\eps)}{w^{L+1}}\d w .
\end{equation}
Choosing $R=\delta^{-1}$ for $\delta>0$, we obtain
\begin{equation}
|\Mcal_L(\beta,\eta|m,\eps)|
\leq
\frac{L!}{2\pi}
\oint_{\partial B_{\delta^{-1}}(0)}
\frac{|\Zcal_R(\beta,w\eta|m,\eps)|}{|w|^{L+1}}\d w
\leq B(\delta,M,\Lambda)\,\delta^L\,L!,
\end{equation}
where we used the first part of \Cref{prop:uniform bounds}. Furthermore, $M>0$ and the compact set $\Lambda\subset \R^d$ are such that $\norm{\eta}_{L^\infty(\R^d\times\{-1,1\})}<M$ and $\supp(\eta)\subset \Lambda\times\{-1,1\}$.
\end{proof}
\end{subsection}

\begin{subsection}{Convergence of the renormalized partition function: proof of {{\Cref{thm:convergence of the partition function}}}}
\label{sec:convergence of the partition function}
The last ingredient needed for the proof of \Cref{thm:convergence of the partition function}
is the existence of the pointwise limits of $\widetilde{\Mcal}_L$ at distinct points, that is,
for $x_i\neq x_j$ whenever $i\neq j$ in
$(\xx,\Si):=(x_1,x_2,\dots,x_L,\sigma_1,\sigma_2,\dots,\sigma_L)$, as well as the convergence
of $\Mcal_L(\beta,\eta|m,\eps)$ to the corresponding integrals of these pointwise limits.

More precisely, define
\begin{equation}
\begin{split}
\widetilde{\Mcal}_L(\xx,\Si,\beta)
&:=\sum_{\pi\in\S_L}
\sum_{\substack{L=n+\sum_{k=1}^{N/2}2kj_k\\ n,j_1,j_2,\dots,j_{N/2}\geq 0}}
\binom{L}{n,2j_1,4j_2,\dots,Nj_{N/2}}
(-1)^n
\rbr{\prod_{l=1}^{N/2}\frac{(-1)^{j_l}}{[(2l)!]^{j_l}j_l!}}
\\
&\quad\times
\rbr{
\prod_{k=1}^{N/2}\prod_{l_k=1}^{j_k}
\1_{\{\sum_{j\in\pi([[n+(l_k-1)2k+1,n+2l_k k]])}\sigma_j=0\}}
}
\E\abr{\prod_{a=1}^n :e^{i\sqrt{\beta}\sigma_{\pi(a)}\vp_{\pi(a)}}:}
\\
&\quad\times
\prod_{k=1}^{N/2}\prod_{l_k=1}^{j_k}
\E^T\abr{
:e^{i\sqrt{\beta}\sigma_{p_{l_k}}\vp_{p_{l_k}}}:
\mid p_{l_k}\in \pi([[n+2(l_k-1) k+1,n+2l_kk]])
},
\end{split}
\end{equation}

where
\begin{equation}
\E\abr{\prod_{j=1}^n :e^{i\sqrt{\beta}\sigma_j\vp(x_j)}:}
:=\lim_{m\to 0}\lim_{\eps\to 0}
\E_{m,\eps}\abr{
\prod_{j=1}^n
\eps^{-\frac{\beta}{4\pi}}e^{i\sqrt{\beta}\sigma_j\vp(x_j)}
}
\end{equation}
and the corresponding truncated expectation is defined via the analogous limit. $\widetilde{\Mcal}_L(\xx,\Si,\beta|m)$ is defined similarly by replacing $\E$ and $\E^T$
above with $\E_m$ and $\E_m^T$, respectively, and omitting the limit $m\to 0$.

\begin{lemma}
\label{lem:pointwise limits of tildeMcal and their integrals}
For $x_i\neq x_j$ whenever $i\neq j$, we have
\begin{equation}
\begin{split}
\lim_{\eps\to 0}\widetilde{\Mcal}_L(\xx,\Si,\beta|m,\eps)
&=\widetilde{\Mcal}_L(\xx,\Si,\beta|m),
\\
\lim_{m\to 0}\lim_{\eps\to 0}\widetilde{\Mcal}_L(\xx,\Si,\beta|m,\eps)
&=\widetilde{\Mcal}_L(\xx,\Si,\beta),
\end{split}
\end{equation}
and $\widetilde{\Mcal}_L(\xx,\Si,\beta)=0$ for odd $L$.

Furthermore, if $\beta\in(0,4\pi)$ for $d=1$, or $\beta\in(0,(d+1)2\pi)$ for $d\geq 2$,
the limit
\begin{equation}
\Mcal_L(\beta,\eta)
:=\lim_{m\to 0}\lim_{\eps\to 0}\Mcal_L(\beta,\eta|m,\eps)
=\sum_{\Si\in\{-1,1\}^L}\int_{\R^{Ld}}
\rbr{\prod_{i=1}^L\eta_i}
\widetilde{\Mcal}_L(\xx,\Si,\beta)\d\xx
\end{equation}
exists and is finite.

Finally, suppose that $\eta_\alpha(\cdot,\cdot)$ and $\eta_\alpha(\cdot,\cdot|m,\eps)$
depend on complex parameters $\alpha\in\C^k$, $k\in\N$, in a manner similar to that in the 4th statement of \Cref{thm:convergence of the partition function}. That is, for every compact $K\subset \C^k$, there exists a compact set $\Lambda\subset\R^d$ such that the following three statements hold:
\begin{equation}
\bigcup_{\alpha\in K}\supp(\eta_\alpha)\subset \Lambda\times \{-1,1\} \quad \text{ and } \quad \bigcup_{\alpha\in K}\bigcup_{m,\eps\in (0,1)}\supp(\eta_\alpha(\cdot,\cdot|m,\eps))\subset \Lambda\times\{-1,1\},
\end{equation}
\begin{equation}
\sup_{\alpha\in K}\norm{\eta_\alpha}_{L^\infty(\R^d\times\{-1,1\})}<\infty \quad \text{ and } \quad \sup_{m,\eps\in (0,1)}\sup_{\alpha\in K}\norm{\eta_\alpha(\cdot,\cdot|m,\eps)}_{L^\infty(\R^d\times\{-1,1\})}<\infty,
\end{equation}
and 
\begin{equation}
\lim_{m\to 0}\limsup_{\eps\to 0}
\sup_{\alpha\in K}
\norm{\eta_\alpha(\cdot,\cdot|m,\eps)-\eta_\alpha}_{L^\infty(\R^d\times\{-1,1\})}
=0.
\end{equation}
Then 
\begin{equation}
\lim_{m\to 0}\lim_{\eps\to 0}\sup_{\alpha\in K}|\Mcal_L(\beta,\eta_\alpha(\cdot,\cdot|m,\eps)|m,\eps)-\Mcal_L(\beta,\eta_\alpha)|=0
\end{equation}
with
\begin{equation}
\Mcal_L(\beta,\eta_\alpha)
:=\lim_{m\to 0}\lim_{\eps\to 0}
\Mcal_L(\beta,\eta_\alpha(\cdot,\cdot|m,\eps)|m,\eps)
=\sum_{\Si\in\{-1,1\}^L}\int_{\R^{Ld}}
\rbr{\prod_{i=1}^L\eta_\alpha(x_i,\sigma_i)}
\widetilde{\Mcal}_L(\xx,\Si,\beta)\d\xx,
\end{equation}
and this quantity is finite. Analogous results hold for fixed $m$.
\end{lemma}

\begin{proof}
The first claim follows directly from \Cref{lem:pointwise charge correlations of the free field,rem:pointwise truncated charge correlation functions of the free field},
 and the representation
\Cref{eq:first representation for tildeMcal}. The parity claim also follows from
\Cref{lem:pointwise charge correlations of the free field} and from the fact that
$n$ in the expression is odd whenever $L$ is odd. Indeed, we have
\begin{equation}
\E\abr{\prod_{a=1}^n :e^{i\sqrt{\beta}\sigma_{a}\vp_{a}}:}=0
\end{equation}
unless $\sum_{a=1}^n\sigma_a=0$.

We now turn to the remaining claims and first consider the case $d\geq 2$.
The strategy in this case is to derive an integrable majorant and then apply
dominated convergence. Using the second representation
\Cref{eq:second representation for tildeMcal} of $\widetilde{\Mcal}_L$, we see that
$\abs{\widetilde{\Mcal}_L}$ is bounded by a finite combinatorial sum of terms of the
form
\begin{equation}
\begin{split}
\prod_{i}
&\abs{\tilV_\infty^{|A_i|}(\xx_i,\Si_i|m,\eps)
-\tilV_t^{|A_i|}(\xx_i,\Si_i|m,\eps)}
\\
&\times\prod_{j}
\abs{\tilV_t^{|\tilde{A}_j|}
(\xx_{\tilde{A}_j},\Si_{\tilde{A}_j}|m,\eps)}
\E_{m,\sqrt{t}}\abr{\prod_{j}
\abs{
e^{i\sqrt{\beta}\sum_{l_j\in\tilde{A}_j}\sigma_{l_j}\vp_{l_j}}
-\1_{|\tilde{A}_j|\leq N}
\1_{\sum_{k_j\in\tilde{A}_j}\sigma_{k_j}=0}
}
},
\end{split}
\end{equation}
where
\begin{equation}
\rbr{\biguplus_{i}A_i}\biguplus\rbr{\biguplus_{j}\tilde{A}_j}=[L]
\end{equation}
and $\uplus$ denotes disjoint union. 

Recall that we work within a compact set (since $\eta$ is compactly supported). An integrable majorant is then obtained using the bounds \Cref{eq:uniform bound for Vinfty minus Vt for d greater than 2 part 1,eq:uniform bound for Vinfty minus Vt for d greater than 2 part 2} from \Cref{cor:convergence of the Vinfty-Vt terms} for the terms corresponding to the sets $A_i$, and using \Cref{lem:estimate for the field dependent part of the counter term} followed by Hölder inequality applied to the expectation, the functions $h_t^{|\tilde{A}_j|}$ from \Cref{prop:induction statement} and \Cref{lem:Gaussian tails} for the terms corresponding to the sets $\tilde{A}_j$. \Cref{lem:Gaussian tails} ensures that $\E_{m,\sqrt{t}}[\norm{\grad \vp}_{L^\infty(B_\Lambda)}]$ is uniformly bounded in $m,\eps\in (0,1)$ for any compact $F\subset\R^d$. Since the sets $A_i$ and $\tilde{A}_j$ are distinct the above argument yields an integrable majorant for each of the sets separately. Consequently, we may take the pointwise limit in the first representation
\Cref{eq:first representation for tildeMcal} and obtain the second claim from the first one. The same argument applies to the final claim with the parameters $\alpha\in \C^k$ due to the assumptions.

For the case $d=1$, more care is required, since there is no integrable majorant
for the terms corresponding to the sets $A_i$ in the above discussion for $d\geq 2$, as discussed after
\Cref{cor:convergence of the Vinfty-Vt terms}. Nevertheless, even without taking absolute values, the integrals factor into those corresponding to the sets $A_i$ and those corresponding to $\tilde{A}_j$ as above, and so do the limits $m,\eps\to 0$. Furthermore, the part corresponding to the sets $A_i$ factors further for each separate set $A_i$. Then \Cref{cor:convergence of the Vinfty-Vt terms} implies that the integral
corresponding to each set $A_i$ converges. The limit is given by the integral in which we simply take pointwise limits of the integrands, and these pointwise limits exist by the discussion in
\Cref{rem:limits of the Vtn functions}.

Moreover, the terms corresponding to the
sets $\tilde{A}_j$ do admit integrable majorants by the same argument as in the
case $d\geq 2$ above. Note that here we can work with $t<m^{-2}$. It therefore remains to show that the pointwise limits of these terms exist.
The pointwise limits of the factors $\tilV_t^{|\tilde{A}_j|}$ exist by
\Cref{rem:limits of the Vtn functions}. For the remaining term, we expand the
product in order to define the limit. Specifically, we have
\begin{equation}
\begin{split}
\lim_{m\to 0}
\E_{m,\sqrt{t}}
&\abr{
\prod_{j=1}^M
\Bigl(
e^{i\sqrt{\beta}\sum_{l_j\in\tilde{A}_j}\sigma_{l_j}\vp_{l_j}}
-\1_{|\tilde{A}_j|\leq N}
\1_{\sum_{a_j\in\tilde{A}_j}\sigma_{a_j}=0}
\Bigr)
}
\\
&:=
\sum_{F\uplus G=[M]}
(-1)^{|F|}\rbr{
\prod_{k\in F}
\1_{|\tilde{A}_k|\leq N}
\1_{\sum_{i\in\tilde{A}_k}\sigma_i=0}
}
\lim_{m\to 0}
\E_{m,\sqrt{t}}
\abr{
e^{i\sqrt{\beta}
\sum_{l\in G}\sum_{a_l\in\tilde{A}_l}\sigma_{a_l}\vp_{a_l}}
},
\end{split}
\end{equation}
where the sum runs over all ordered bipartitions $\{F,G\}$ of the set $[M]$,
including the trivial ones $\{[M],\emptyset\}$ and $\{\emptyset,[M]\}$.
The final limit exists by
\Cref{lem:pointwise charge correlation functions of the scale decomposed free field}.

By linearity, we therefore obtain
\begin{equation}
\lim_{m\to 0}\lim_{\eps\to 0}
\sum_{\Si\in\{-1,1\}^L}
\int_{\R^{Ld}}
\rbr{\prod_{i=1}^L\eta_i}
\widetilde{\Mcal}_L(\xx,\Si,\beta|m,\eps)\d\xx
=
\sum_{\Si\in\{-1,1\}^L}
\int_{\R^{Ld}}
\rbr{\prod_{i=1}^L\eta_i}
\widetilde{\Mcal}_L^*(\xx,\Si,\beta)\d\xx,
\end{equation}
where $\widetilde{\Mcal}_L^*(\xx,\Si,\beta)$ denotes the limit of
$\widetilde{\Mcal}_L(\xx,\Si,\beta|m,\eps)$ taken termwise in the second
representation \Cref{eq:second representation for tildeMcal}, as described above.
Since the two representations agree for every $m,\eps\in(0,1)$, their limits must
also coincide, that is, $\widetilde{\Mcal}_L^*\equiv\widetilde{\Mcal}_L$. This completes the proof. The case involving parameters $\alpha\in\C^k$ follows again
immediately from the assumptions as in the case $d\geq 2$.
\end{proof}

With this at hand, we can finally prove \Cref{thm:convergence of the partition function}. Now that we have our estimates in place, the proof is analogous to that of \cite[Theorem 5.1]{BaWe24a}.

\begin{proof}[Proof of \Cref{thm:convergence of the partition function}]
It is enough to consider only the statements with both limits, since the cases with fixed $m$ are analogous, and it also suffices to consider the case with parameters $\alpha$. We can modify the proof of the bound \Cref{eq:uniform bound for Mcal} to the case with $\eta_\alpha(\cdot,\cdot|m,\eps)$ by using the first claim of \Cref{prop:uniform bounds} with
\begin{equation}
\label{eq:uniform bound for eta with parameters}
\sup_{m,\eps\in (0,1)}\sup_{\alpha\in  K}\norm{\eta_\alpha(\cdot,\cdot|m,\eps)}_{L^\infty(\R^d\times\{-1,1\})}<M
\end{equation}
instead of $\norm{\eta}_{L^\infty(\R^d\times\{-1,1\})}<M$ for $\alpha\in K$ with fixed compact set $K\subset \C^k$. Note that the condition \Cref{eq:uniform bound for eta with parameters} makes sense by our assumptions. Then for all $\delta>0$ and $M>0$ such that \Cref{eq:uniform bound for eta with parameters} holds we have the bound $\sup_{\alpha\in K}|\Mcal_L(\beta,\eta_\alpha)|\leq C(\delta,M,K)\delta^LL!$. In particular, by choosing $0<\delta\equiv \delta (z)$ small enough we obtain that
\begin{equation}
\label{eq:limit candidate for Zcal}
\Zcal_R(\beta,z\eta_\alpha):=\sum_{L=0}^\infty\frac{z^L}{L!}\Mcal_L(\beta,\eta_\alpha),
\end{equation}
where $\Mcal_L(\beta,\eta_\alpha)$ is defined in \Cref{lem:pointwise limits of tildeMcal and their integrals}, converges absolutely (and uniformly in $\alpha\in K$) for all $z\in \C$, and so defines an analytic function such that for fixed $r>0$ we also have
\begin{equation}
\sup_{|z|<r}\sup_{\alpha\in K}|\Zcal_R(\beta,z\eta_\alpha)|<\infty.
\end{equation}
The proof of the second claim is finished by noting that $\Zcal_R(\beta,z\eta_\alpha)$ is even, because $\Mcal_L(\beta,\eta_\alpha)$ is zero if $L$ is odd. The positivity in the third claim is implied by the second part of \Cref{prop:uniform bounds} using \Cref{eq:uniform bound for eta with parameters}.

Finally, to obtain convergence, let $A\geq 1$ and fix again a compact $K\subset \C^k$ and assume $\alpha\in K$. Then, using \Cref{eq:limit candidate for Zcal} as a Taylor expansion at zero with $z=1$, the Cauchy formula for derivatives, and simple triangle-inequality estimates, we have
\begin{equation}
\begin{split}
|\Zcal_R&(\beta,\eta_\alpha(\cdot|m,\eps)|m,\eps)-\Zcal_R(\beta,\eta_\alpha)|
\\
&\leq \sum_{L=0}^A\frac{1}{L!}|\Mcal_L(\beta,\eta_\alpha(\cdot|m,\eps)|m,\eps)-\Mcal_
L(\beta,\eta_\alpha)|
\\
&\quad+\sum_{L=A+1}^\infty\bigg(\oint_{\partial B_r(0)}\frac{|\Zcal_R(\beta,w\eta_\alpha(\cdot|m,\eps)|m,\eps)|}{|w|^{L+1}}\frac{\d|w|}{2\pi}+\oint_{\partial B_r(0)}\frac{|\Zcal_R(\beta,w\eta_\alpha)|}{|w|^{L+1}}\frac{\d|w|}{2\pi}\bigg)
\\
&\leq \sum_{L=1}^A\frac{1}{L!}|\Mcal_L(\beta,\eta_\alpha(\cdot|m,\eps)|m,\eps)-\Mcal_L(\beta,\eta_\alpha)|
\\
&\quad+\rbr{\sup_{m,\eps\in (0,1)}\sup_{|w|=r}|\Zcal_R(\beta,w\eta_\alpha(\cdot|m,\eps)|m,\eps)|+\sup_{|w|=r}|\Zcal_R(\beta,w\eta_\alpha)|}r^{-A}
\end{split}
\end{equation}
for any $r>1$. The uniform convergence follows from the uniform bounds for the partition function and the uniform convergence of $\Mcal_L$ in \Cref{prop:uniform bounds} (augmented with \Cref{eq:uniform bound for eta with parameters}) and \Cref{lem:pointwise limits of tildeMcal and their integrals}, respectively, by first taking the limits $m,\eps\to 0$ and then $A\to \infty$. This proves the last claim.
\end{proof}

\end{subsection}

\end{section}

\begin{section}{Correlation functions and the existence of the field}
\label{sec:The analysis of the correlation functions and the existence of the field and t}
In this section, we consider our actual sine-Gordon model, that is, we put
$\eta=-z\1_\Lambda$, where $\Lambda\subset \R^d$ is a compact set, in the generalized formalism of \Cref{sec: the renormalized potential,sec:Analysis of the partition function}. After the crucial
Cameron--Martin--Girsanov transformation, we choose an $\eta$ depending on certain complex
parameters and write everything in terms of $\Zcal_R$ so that we can utilize results
from the previous section. Indeed, after this, everything will follow from
\Cref{thm:convergence of the partition function} and the results in
\Cref{sec: the renormalized potential} and Appendix
\ref{sec:Correlation functions of the reference field.}.

We can carry out the preliminary analysis for the proofs of
\Cref{thm:existence of the field,thm:existence of the sine Gordon correlation functions}
simultaneously. Consider the mappings $\C^{n+k+1}\to\C$
\begin{equation}
\begin{split}
&(\boldsymbol{\mu},\boldsymbol{\nu},z)\mapsto  \log(\E_{\mathrm{sG}(\beta,z,\Lambda|m,\eps)}\abr{\exp(\sum_{j=1}^n\mu_j\eps^{-\frac{\beta}{4\pi}}e^{i\sqrt{\beta}\sigma_j\vp}(f_j)+\sum_{l=1}^k\nu_l[D_l\vp](g_l))})
\\
&:=\log(\!\![\Zcal_0(\beta,z,\Lambda|m,\eps)]^{-1}\E_{m,\eps}\!\! \abr{\exp(\sum_{j=1}^n\mu_j\eps^{-\frac{\beta}{4\pi}}e^{i\sqrt{\beta}\sigma_j\vp}(f_j)+\sum_{l=1}^k\nu_l[D_l\vp](g_l))e^{2z\int_{\Lambda}\eps^{-\frac{\beta}{4\pi}}\cos(\sqrt{\beta}\vp(x))\dx}}),
\end{split}
\end{equation}
and $\C^2\to \C$ 
\begin{equation}
(w,z)\mapsto \E_{\mathrm{sG}(\beta,z,\Lambda|m,\eps)}\abr{e^{w\vp(f)}}:=[\Zcal_0(\beta,z,\Lambda|m,\eps)]^{-1}\E_{m,\eps}\abr{e^{w\vp(f)}e^{2z\int_{\Lambda}\eps^{-\frac{\beta}{4\pi}}\cos(\sqrt{\beta}\vp(x))\dx}}
\end{equation}

Since the field $\vp_{m,\eps}$, its derivatives, and the integral against a test
function $\vp(f)$ are Gaussian, and other relevant random variables are bounded,
these functions are analytic separately in each variable $\mu_j,\nu_l,w$ in some
open disc centered at the origin of $\C$, which may a priori depend on $m,\eps,z$, and $\Lambda$. Indeed, the argument of the logarithm in the first map is analytic and especially continuous. Furthermore, it equals $1$ at $\mu_j=0=\nu_l$ for all $j,l$ and any fixed $z\in \C$.  Therefore, for small enough discs around the origin in all the variables $\mu_k,\nu_l$ it is especially strictly positive so that we may choose an analytic branch of the logarithm.  Thus, the first
mapping is also a jointly analytic function of $(\boldsymbol{\mu},\boldsymbol{\nu})$ in some neighborhood of the origin in
$\C^{n+k}$.  Differentiating it with respect to the parameters $\mu_j,\nu_l$ and
setting them to zero yields the truncated correlation functions (cumulants).

We wish to find a larger domain, independent of $m$ and $\eps$, in which the
functions are also analytic, and such that the convergence obtained by first
taking $\eps\to 0$ and then $m\to 0$ is uniform in all variables
$\boldsymbol{\mu},\boldsymbol{\nu},w$, and $z$. This implies that the limiting
function, together with all of its derivatives, is also analytic, now also with
respect to $z$.

We now apply the Cameron--Martin--Girsanov theorem as in the proofs of
\cite[Lemma 2.6, Theorem 3.1]{BaWe24a}. The application here is analogous to the
proof of \cite[Theorem 3.1]{BaWe24a}. For justification of its use in the complex case, see also
the proof of \cite[Lemma 2.6]{BaWe24a}. We need it in the following form. Let
$(\vp(x))_{x\in\R^d}$ be a smooth Gaussian field and $Y$ a complex-valued Gaussian random variable
that is measurable with respect to $\sigma(\vp)$, that is, the sigma-algebra
generated by the field $\vp$. Then, for a suitable function $F$ of the field for
which the expectation exists, we have
\begin{equation}
\E\abr{F(\vp)e^{Y-\half\E[Y^2]}}=\E\abr{F(\vp+\E(Y\vp))}.
\end{equation}
An obvious generalization holds for more than one random variable $Y$. Above,
$Y\vp$ denotes the random function $x\mapsto Y\vp(x)$. This case is covered by
\cite[Theorem 2.8]{Da06a}.

Thus, we obtain for the expectation in the first mapping
\begin{equation}
\begin{split}
\label{eq:rewriting the exponential in the sg correlation functions}
\E_{\mathrm{sG}(\beta,z,\Lambda|m,\eps)}&\abr{\exp(\sum_{j=1}^n\mu_j\eps^{-\frac{\beta}{4\pi}}e^{i\sqrt{\beta}\sigma_j\vp}(f_j)+\sum_{l=1}^k\nu_l[D_l\vp](g_l))}
\\
&=\exp(\half\E_{m,\eps}\abr{\rbr{\sum_{l=1}^k\nu_l[D_l\vp](g_l)}^2})\Zcal_0(\beta,z, \Lambda|m,\eps)^{-1}
\\
&\quad\times\E_{m,\eps}\bigg[ \exp\bigg(\int_{\R^d} 
\bigg\{z\1_{\Lambda}(x)\sum_{\sigma\in\{-1,1\}}\eps^{-\frac{\beta}{4\pi}}e^{i\sqrt{\beta}\sigma\vp(x)}e^{i\sqrt{\beta}\sigma\sum_{l=1}^k\nu_l\E_{m,\eps}[\vp(x)[D_l\vp](g_l)]}
\\
&\qquad\qquad+\sum_{j=1}^n\mu_jf_j(x)\eps^{-\frac{\beta}{4\pi}}e^{i\sqrt{\beta}\sigma_j\vp(x)}e^{i\sqrt{\beta}\sigma_j\sum_{l=1}^k\nu_l\E_{m,\eps}[\vp(x)[D_l\vp](g_l)]} \bigg\}\dx
 \bigg)\bigg],
\end{split}
\end{equation}
where we have set
\begin{equation}
Y:=\sum_{l=1}^k\nu_l[D_l\vp](g_l)
\end{equation}
and
\begin{equation}
F(\vp):=\exp(\int_{\R^d}\eps^{-\frac{\beta}{4\pi}}
\rbr{\sum_{j=1}^n\mu_jf_j(x)e^{i\sqrt{\beta}\sigma_j\vp_{m,\eps}(x)}
+\sum_{\sigma\in\{-1,1\}}z\1_\Lambda(x) e^{i\sqrt{\beta}\sigma\vp_{m,\eps}(x)}}\dx).
\end{equation}

We can now rewrite the exponential inside the expectation in the last factor of
\Cref{eq:rewriting the exponential in the sg correlation functions} in a form that
allows us to use the last part of \Cref{thm:convergence of the partition function},
namely,
\begin{equation}
\exp(-\sum_{\sigma\in\{-1,1\}}\int_{\R^d}
\eta_{\boldsymbol{\mu},\boldsymbol{\nu},z,\Lambda}(x,\sigma|m,\eps)
\eps^{-\frac{\beta}{4\pi}}e^{i\sqrt{\beta}\sigma\vp_{m,\eps}(x)}\dx),
\end{equation}
where
\begin{equation}
\label{eq:def of eta with parameters for correlation functions}
\eta_{\boldsymbol{\mu},\boldsymbol{\nu},z,\Lambda}(x,\sigma|m,\eps):=-z\1_\Lambda(x) e^{i\sqrt{\beta}\sigma\sum_{l=1}^k\nu_l\E_{m,\eps}[\vp(x)[D_l\vp](g_l)]}-\sum_{j=1}^n\delta_{\sigma,\sigma_j}\mu_jf_j(x)e^{i\sqrt{\beta}\sigma_j\sum_{l=1}^k\nu_l\E_{m,\eps}[\vp(x)[D_l\vp](g_l)]}.
\end{equation}

For the second mapping, using $Y:=w\vp_{m,\eps}(f)$ and
$F(\vp):=\exp(z\sum_{\sigma\in\{-1,1\}}\int_{\Lambda}
\eps^{-\frac{\beta}{4\pi}}e^{i\sigma\sqrt{\beta}\vp_{m,\eps}(x)}\dx)$, we obtain similarly
\begin{equation}
\begin{split}
\E_{\mathrm{sG}(\beta,z,\Lambda|m,\eps)}\abr{e^{w\vp(f)}}
&=\exp(\half w^2\E_{m,\eps}[(\vp(f))^2])
\Zcal_0(\beta,z,\Lambda|m,\eps)^{-1}
\\
&\quad\times\E_{m,\eps}\abr{\exp(-\sum_{\sigma\in\{-1,1\}}\int_{\R^d}
\eps^{-\frac{\beta}{4\pi}}\eta_{w,z,\Lambda}(x,\sigma|m,\eps)
e^{i\sigma\sqrt{\beta}\vp(x)}\dx)},
\end{split}
\end{equation}
where
\begin{equation}
\label{eq:eta with parameters for existence of the field}
\eta_{w,z,\Lambda}(x,\sigma|m,\eps)
:=-z\1_{\Lambda}(x)e^{i\sqrt{\beta}\sigma w\E_{m,\eps}[\vp(x)\vp(f)]}.
\end{equation}
We are now ready for the proofs of
\Cref{thm:existence of the sine Gordon correlation functions,thm:existence of the field}.

\begin{subsection}{Proof of {{\Cref{thm:existence of the sine Gordon correlation functions}}}}
\label{sec:Proof of existence of correlation functions}
We have
\begin{equation}
\begin{split}
&\E_{\mathrm{sG}(\beta,z,\Lambda|m,\eps)}^T\abr{\prod_{j=1}^n\eps^{-\frac{\beta}{4\pi}}e^{i\sqrt{\beta}\sigma_j\vp}(f_j)\prod_{l=1}^k[D_l\vp](g_l)}
\\
&=\prod_{j=1}^n\partd{\mu_j}\bigg|_{\mu_j=0}\prod_{l=1}^k\partd{\nu_l}\bigg|_{\nu_l=0}\log\cbo{0.7cm}\Zcal_0(\beta,z,\Lambda|m,\eps)^{-1}
\\
&\quad\times\E_{m,\eps}\abr{\exp(-\!\!\sum_{\sigma\in\{-1,1\}}\int_{\R^d}\eta_{\boldsymbol{\mu},\boldsymbol{\nu},z,\Lambda}(x,\sigma|m,\eps)\eps^{-\frac{\beta}{4\pi}}e^{i\sqrt{\beta}\sigma\vp(x)}\dx)}
\exp(\half\E_{m,\eps}\bigg[\bigg(\sum_{l=1}^k\nu_l[D_l\vp](g_l)\bigg)^2\bigg])\cbc{0.7cm}.
\end{split}
\end{equation}
We can also rewrite the argument of the logarithm above in terms of the renormalized partition function $\Zcal_R$
\begin{equation}
\begin{split}
\label{eq:the argument of log in the correlation functions in terms of ZcalR}
&\frac{\Zcal_R(\beta,\eta_{\boldsymbol{\mu},\boldsymbol{\nu},z,\Lambda}(\cdot,\cdot|m,\eps)|m,\eps)}{\Zcal_R(\beta,z,\Lambda|m,\eps)}
\\
&\times\prod_{k=1}^{N/2}e^{\frac{1}{(2k)!}\sum_{\Si\in\{-1,1\}^{2k}}\1_{\sum_{i=1}^{2k}\sigma_i=0}\int_{\R^{2kd}}\rbr{\prod_{i=1}^{2k}\eta_{\Bold{\mu},\Bold{\nu},z,\Lambda}(x_i,\sigma_i|m,\eps)}\E_{m,\eps}^{T}\abr{\eps^{-\frac{\beta}{4\pi}}e^{i\sqrt{\beta}\sigma_j\vp(x_j)}\mid j\in[2k]}\d\xx}
\\
&\times\prod_{k=1}^{N/2}e^{\frac{-z^{2k}}{(2k)!}\sum_{\Si\in\{-1,1\}^{2k}}\1_{\sum_{i=1}^{2k}\sigma_i=0}
\int_{\Lambda^{2k}}\E_{m,\eps}^{T}\abr{\eps^{-\frac{\beta}{4\pi}}e^{i\sqrt{\beta}\sigma_j\vp(x_j)}\mid j\in[2k]}\d\xx}
\\
&\times\exp(\half\E_{m,\eps}\bigg[\bigg(\sum_{l=1}^k\nu_l[D_l\vp](g_l)\bigg)^2\bigg]). 
\end{split}
\end{equation}

To use \Cref{thm:convergence of the partition function} we need the limit
\[
\eta_{\Bold{\mu},\Bold{\nu},z,\Lambda}(x,\sigma):=\lim_{m\to 0}\lim_{\eps\to 0}\eta_{\Bold{\mu},\Bold{\nu},z,\Lambda}(x,\sigma|m,\eps),
\]
which exists, and is given by making the replacement 
\begin{equation}
\E_{m,\eps}\abr{\vp(x)[D_l\vp](g_l)}\to\E\abr{\vp(x)[D_l\vp](g_l)}
=\frac{1}{2\pi}\int_{\R^d}[D_lg_l](y)\log(|x-y|)\dy,
\end{equation}
to \Cref{eq:def of eta with parameters for correlation functions}. The fixed $m>0$ case is analogous.

The convergence of $\eta_{\boldsymbol{\mu},\boldsymbol{\nu},z,\Lambda}$ is in
$L^\infty(\R^d\times\{-1,1\})$ by \Cref{lem:basic correlations of the derivatives of the free field}
and it is uniform in $\Bold{\mu},\Bold{\nu},z\in K\subset \C^{n+k+1}$ for compact $K$. Furthermore, we have
\begin{equation}
\bigcup_{(\Bold{\mu},\Bold{\nu},z)\in K}\supp(\eta_{\Bold{\mu},\Bold{\nu},z,\Lambda}(\cdot,\cdot|m,\eps))\subset \overline{\Lambda\cup\bigcup_{j=1}^n\supp(f_j)},
\end{equation}
which is compact, and
\begin{equation}
\sup_{(\Bold{\mu},\Bold{\nu},z)\in K}\norm{\eta_{\Bold{\mu},\Bold{\nu},z,\Lambda}(\cdot,\cdot|m,\eps)}_{L^\infty(\R^d\times\{-1,1\})}\leq\sup_{(\Bold{\mu},\Bold{\nu},z)\in K}\rbr{|z|+\sum_{j=1}^n|\mu_j|\norm{f_j}_{L^\infty(\R^d\times\{-1,1\})}}<\infty.
\end{equation}
Therefore, by \Cref{thm:convergence of the partition function} we have
\begin{equation}
\Zcal_R(\beta,\eta_{\Bold{\mu},\Bold{\nu},z,\Lambda}(\cdot,\cdot|m,\eps)|m,\eps)
\;\longrightarrow\;
\Zcal_R(\beta,\eta_{\Bold{\mu},\Bold{\nu},z,\Lambda}(\cdot,\cdot)).
\end{equation}
Furthermore, by the same argument as in the beginning of this section, that is,
boundedness of relevant random variables, the function
$\Zcal_R(\eta_{\Bold{\mu},\Bold{\nu},z,\Lambda}(\cdot,\cdot|m,\eps)|m,\eps)$
extends to an entire function in $\C^{n+k+1}$.  
The uniform convergence implies that the limit
$\Zcal_R(\beta,\eta_{\Bold{\mu},\Bold{\nu},z,\Lambda}(\cdot,\cdot))$ is also entire. 

For $n>N$, the only relevant contribution in \Cref{eq:the argument of log in the correlation functions in terms of ZcalR}  comes from the term
$\Zcal_R(\eta_{\Bold{\mu},\Bold{\nu},z,\Lambda}(\cdot,\cdot|m,\eps)|m,\eps)$
since other terms vanish when we take logarithmic derivatives and set the variables to zero. 
Thus, we have
\begin{equation}
\begin{split}
\E_{\mathrm{sG}(\beta,z,\Lambda|m,\eps)}^T
&\abr{\prod_{j=1}^n\eps^{-\frac{\beta}{4\pi}}e^{i\sqrt{\beta}\sigma_j\vp}(f_j)
\prod_{l=1}^k(D_l\vp)(g_l)}
\\
&=\prod_{j=1}^n\partd{\mu_j}\bigg|_{\mu_j=0}
\prod_{l=1}^k\partd{\nu_l}\bigg|_{\nu_l=0}
\log\!\left(\Zcal_R(\beta,\eta_{\boldsymbol{\mu},\boldsymbol{\nu},z,\Lambda}(\cdot,\cdot|m,\eps)|m,\eps)\right).
\end{split}
\end{equation}
Part 3 of \Cref{prop:uniform bounds} implies that
$\Zcal_R(\beta,\eta_{0,0,z,\Lambda}(\cdot,\cdot|m,\eps)|m,\eps)>0$
uniformly in $m$ and $\eps$ and for all $|z|<M$ with arbitrary fixed $M>0$.  
Therefore, there exists some open neighborhood $U$ of
the real axis $\R$ in $\C$ such that
\begin{equation}
\log\!\left(\Zcal_R(\beta,\eta_{\Bold{\mu},\Bold{\nu},z,\Lambda}(\cdot,\cdot|m,\eps)|m,\eps)\right)
\;\longrightarrow\;
\log\!\left(\Zcal_R(\beta,\eta_{\Bold{\mu},\Bold{\nu},z,\Lambda}(\cdot,\cdot))\right),
\end{equation}
when we first take $\eps\to 0$ and then $m\to 0$,
uniformly in $[D(K)]^{n+k}\times K$, where $K\subset U$ is compact and $D(K)$ is some open disk in $\C$ centered at the origin.  

This also implies that the limit is analytic in $[B_{r(R)}(0)]^{n+k}\times (B_R(0)\cap U)\subset\C^{n+k+1}$ for some $r(R)>0$ with all $R>0$.  
The derivatives are then also analytic and converge in the same manner in the same domain.
Thus, we have proven that the correlation functions corresponding to $n>N$ and $k\in\Z_{\geq 0}$
converge and are given by
\begin{equation}
\Ccal_{\mathrm{sG}(\beta,z,\Lambda)}^{n,k}(\V{f},\V{g})
=\prod_{j=1}^n\partd{\mu_j}\bigg|_{\mu_j=0}
\prod_{l=1}^k\partd{\nu_l}\bigg|_{\nu_l=0}
\log\!\left(\Zcal_R(\beta,\eta_{\boldsymbol{\mu},\boldsymbol{\nu},z,\Lambda})\right).
\end{equation}
Furthermore, they are analytic in all of $U$ since the above argument holds for all $R>0$.

We can also compute the derivative of arbitrary order with respect to $z$ at zero of these
correlation functions.  
Indeed, since
$\log(\Zcal_R(\beta,\eta_{\Bold{\mu},\Bold{\nu},z,\Lambda}(\cdot,\cdot|m,\eps)|m,\eps))$
is analytic, the derivatives exist, and by retracing our argument we can compute these
directly using the representation of cumulants in terms of moments
\Cref{eq:cumulants to moments} and the boundedness of the relevant random variables.
Thus, we obtain
\begin{equation}
\begin{split}
\dder{z}{q}\bigg|_{z=0}
&\E_{\mathrm{sG}(\beta,z,\Lambda|m,\eps)}^T
\abr{\prod_{j=1}^n\eps^{-\frac{\beta}{4\pi}}e^{i\sqrt{\beta}\sigma_j\vp}(f_j)
\prod_{l=1}^k[D_l\vp](g_l)}
\\
&=\sum_{\Bold{\tau}\in\{-1,1\}^q}
\E_{m,\eps}^T\abr{
\prod_{j=1}^n\eps^{-\frac{\beta}{4\pi}}e^{i\sqrt{\beta}\sigma_j\vp}(f_j)
\prod_{l=1}^k[D_l\vp](g_l)
\prod_{a=1}^q \eps^{-\frac{\beta}{4\pi}}e^{i\sqrt{\beta}\tau_a\vp}(\1_{\Lambda})
}
\\
&\overset{m\to 0,\;\eps\to 0}{\longrightarrow}
\E^T\abr{
\prod_{j=1}^n :e^{i\sqrt{\beta}\sigma_j\vp}(f_j):
\prod_{l=1}^k [D_l\vp](g_l)
\left(
:e^{i\sqrt{\beta}\vp}(\1_\Lambda):
+
:e^{-i\sqrt{\beta}\vp}(\1_\Lambda):
\right)^q
},
\end{split}
\end{equation}
where the expectation is with respect to the law of the free field and the convergence
follows by \Cref{prop:convergence of the smeared mixed correlation functions of the reference field in d=1}.  
Thus, we are done with the proof of the three parts of
\Cref{thm:existence of the sine Gordon correlation functions} for $n>N$.

Let us then consider the trickier cases $n\leq N$. The case $d\geq 2$ and $N=2$ is
completely analogous to the proof of \cite[Theorem 3.1 (i)–(iii)]{BaWe24a} for the cases
$n=0,1,2$, provided we use our
\Cref{lem:basic correlations of the derivatives of the free field,prop:convergence of the smeared mixed correlation functions of the reference field in d=1}
instead of  \cite[Lemmas 2.9 and 2.10]{BaWe24a}. When we consider the case $d=1$, we retain all
notations for general dimension $d$, even though the computations are only relevant for
$d=1$ in this paper. The computations remain valid in higher dimensions, in particular for
$d=2$ and $\beta\in[6\pi,8\pi)$ before taking the limit. Thus, if in the future the
convergence in
\Cref{prop:convergence of the smeared mixed correlation functions of the reference field in d=1}
can be improved to a wider range of $\beta$, our computations here will yield the
correlation functions also in higher dimensions.

Now we consider the case $d=1$. When $1\leq n\leq N$, we must also take logarithmic
derivatives of the factor on the second row of 
\Cref{eq:the argument of log in the correlation functions in terms of ZcalR}. For $n=0$, the factor on the last row of \Cref{eq:the argument of log in the correlation functions in terms of ZcalR} also contributes. However, using \Cref{lem:basic correlations of the derivatives of the free field}, we readily see that the
contribution from this last factor has a finite limit. 

We first treat the cases $n=1,2,\dots,N$. We need
to compute
\begin{equation}
\begin{split}
\prod_{j=1}^{n}&\partd{\mu_j}\bigg|_{\mu_j=0}
\prod_{l=1}^{k}\partd{\nu_l}\bigg|_{\nu_l=0}
\sum_{p=1}^{N/2}\frac{1}{(2p)!}
\sum_{\Si\in\{-1,1\}^{2p}}\1_{\sum_{i=1}^{2p}\sigma_i=0}
\\
&\qquad\qquad\times
\int_{\R^{2pd}}
\Bigl(\prod_{a=1}^{2p}
\eta_{\Bold{\mu},\Bold{\nu},z,\Lambda}(x_a,\sigma_a|m,\eps)\Bigr)
\E_{m,\eps}^{T}\abr{
\eps^{-\frac{\beta}{4\pi}}e^{i\sqrt{\beta}\sigma_b\vp(x_b)}
\mid b\in[2p]}
\d\xx.
\end{split}
\end{equation}
Since for fixed $m,\eps\in(0,1)$ everything is bounded and $\eta_{\Bold{\mu},\Bold{\nu},z,\Lambda}$ has compact support, the
main task is to compute
\begin{equation}
\begin{split}
\prod_{j=1}^{n}&\partd{\mu_j}\bigg|_{\mu_j=0}
\prod_{l=1}^{k}\partd{\nu_l}\bigg|_{\nu_l=0}
\prod_{a=1}^{2p}
\eta_{\Bold{\mu},\Bold{\nu},z,\Lambda}(x_a,\sigma_a|m,\eps)
\\
&=
\prod_{j=1}^{n}\partd{\mu_j}\bigg|_{\mu_j=0}
\prod_{l=1}^{k}\partd{\nu_l}\bigg|_{\nu_l=0}
\sum_{J\subset[2p]}
\Bigl(\prod_{a\in J}
z\1_{\Lambda}(x_a)
e^{i\sqrt{\beta}\sigma_a\sum_{l=1}^k\nu_l
\E_{m,\eps}[\vp(x_a)[D_l\vp](g_l)]}\Bigr)
\\
&\qquad\qquad\times
\Bigl(\prod_{b\in J^c}
\sum_{j=1}^n
\delta_{\sigma_b,\sigma_j}\mu_j f_j(x_b)
e^{i\sqrt{\beta}\sigma_j\sum_{l=1}^k\nu_l
\E_{m,\eps}[\vp(x_b)[D_l\vp](g_l)]}\Bigr).
\end{split}
\end{equation}
Unless $|J^c|=n$, the resulting term is zero. Moreover, each $\mu_j$ must appear exactly
once. Computing the $\mu$-derivatives first, we obtain
\begin{equation}
\begin{split}
\prod_{j=1}^{n}\partd{\mu_j}\bigg|_{\mu_j=0}
&\Bigl(
\prod_{b\in J^c,\ |J^c|=n}
\sum_{j=1}^n
\delta_{\sigma_b,\sigma_j}\mu_j f_j(x_b)
e^{i\sqrt{\beta}\sigma_j\sum_{l=1}^k\nu_l
\E_{m,\eps}[\vp(x_b)[D_l\vp](g_l)]}
\Bigr)
\\
&=
\sum_{\substack{\tau\in\S_{J^c}\\ J^c=(b_1,\dots,b_n)}}
\prod_{j=1}^n
\delta_{\sigma_{\tau(b_j)},\sigma_j}
f_j(x_{\tau(b_j)})
e^{i\sqrt{\beta}\sigma_j\sum_{l=1}^k\nu_l
\E_{m,\eps}[\vp(x_{\tau(b_j)})[D_l\vp](g_l)]}.
\end{split}
\end{equation}
Next we compute the $\nu$-derivatives to obtain
\begin{equation}
\begin{split}
\prod_{j=1}^{n}&\partd{\mu_j}\bigg|_{\mu_j=0}
\prod_{l=1}^{k}\partd{\nu_l}\bigg|_{\nu_l=0}
\prod_{a=1}^{2p}
\eta_{\Bold{\mu},\Bold{\nu},z,\Lambda}(x_a,\sigma_a|m,\eps)
\\
&=
z^{2p-n}
\sum_{\substack{J\subset[2p]\\ |J^c|=n}}
\sum_{\substack{\tau\in\S_{J^c}\\ J^c=(b_1,\dots,b_n)}}
\Bigl(\prod_{a\in J}\1_{\Lambda}(x_a)\Bigr)
\Bigl(\prod_{j=1}^n
\delta_{\sigma_{\tau(b_j)},\sigma_j} f_j(x_{\tau(b_j)})\Bigr)
\\
&\qquad\times
\prod_{l=1}^k
i\sqrt{\beta}
\Bigl(
\sum_{a\in J}\sigma_a
\E_{m,\eps}[\vp(x_a)[D_l\vp](g_l)]
+
\sum_{j=1}^n\sigma_j
\E_{m,\eps}[\vp(x_{\tau(b_j)})[D_l\vp](g_l)]
\Bigr).
\end{split}
\end{equation}

Finally, by linearity and
\Cref{prop:convergence of the smeared mixed correlation functions of the reference field in d=1},
for $1\leq n\leq N$ we obtain
\begin{equation}
\begin{split}
\Ccal_{\mathrm{sG}(\beta,z,\Lambda)}^{n,k}(\V{f},\V{g})
&=
\prod_{j=1}^n\frac{\partial}{\partial\mu_j}\bigg|_{\mu_j=0}
\prod_{l=1}^k\partd{\nu_l}\bigg|_{\nu_l=0}
\log\!\bigl(\Zcal_R(\beta,\eta_{\boldsymbol{\mu},\boldsymbol{\nu},z,\Lambda})\bigr)
\\
&\quad+
\sum_{p=1}^{N/2}
\frac{z^{2p-n}\1_{2p-n\ge0}}{(2p)!}
\sum_{\substack{J\subset[2p]\\ |J^c|=n}}
\sum_{\substack{\tau\in\S_{J^c}\\ J^c=(b_1,\dots,b_n)}}
\sum_{\Si\in\{-1,1\}^{2p}}
\1_{\sum_{i=1}^{2p}\sigma_i=0}
\\
&\qquad\times
\int_{\R^{2pd}}
\E^{T}\abr{
:e^{i\sqrt{\beta}\sigma_q\vp(x_q)}:
\mid q\in[2p]}
\Bigl(\prod_{a\in J}\1_{\Lambda}(x_a)\Bigr)
\Bigl(\prod_{j=1}^n
\delta_{\sigma_{\tau(b_j)},\sigma_j} f_j(x_{\tau(b_j)})\Bigr)
\\
&\qquad\qquad\times
\prod_{l=1}^k i\sqrt{\beta}
\Bigl(
\sum_{a\in J}\sigma_a\E[\vp(x_a)[D_l\vp](g_l)]
+
\sum_{j=1}^n\sigma_j\E[\vp(x_{\tau(b_j)})[D_l\vp](g_l)]
\Bigr)
\d\xx.
\end{split}
\end{equation}
The term starting from the second row above is just a polynomial in $z$ with complicated coefficients, hence it is entire in $z$. The
third claim follows exactly as in the case $n>N$.

For $n=0$, the factor on the second row of \Cref{eq:the argument of log in the correlation functions in terms of ZcalR} contributes
\begin{equation}
\begin{split}
\prod_{l=1}^k&\partd{\nu_l}\bigg|_{\nu_l=0}
\Biggl(
\sum_{p=1}^{N/2}\frac{1}{(2p)!}
\sum_{\Si\in\{-1,1\}^{2p}}
\1_{\sum_{i=1}^{2p}\sigma_i=0}
\\
&\qquad\qquad\times
\int_{\R^{2pd}}\biggl\{
\rbr{\prod_{a=1}^{2p}
(-z)\1_\Lambda(x_a)
e^{i\sqrt{\beta}\sigma_a\sum_{l=1}^k\nu_l
\E_{m,\eps}[\vp(x_a)[D_l\vp](g_l)]}}
\\
&\qquad\qquad\qquad\quad\times
\E_{m,\eps}^T\abr{
\eps^{-\frac{\beta}{4\pi}}e^{i\sqrt{\beta}\sigma_b\vp(x_b)}
\mid b\in[2p]}\biggr\}
\d\xx
\Biggr)
\\
&=
\sum_{p=1}^{N/2}
\frac{z^{2p}}{(2p)!}
\sum_{\Si\in\{-1,1\}^{2p}}
\1_{\sum_{i=1}^{2p}\sigma_i=0}
\int_{\Lambda^{2p}}\biggl\{
\E_{m,\eps}^T\abr{
\eps^{-\frac{\beta}{4\pi}}e^{i\sqrt{\beta}\sigma_b\vp(x_b)}
\mid b\in[2p]}
\\
&\qquad\qquad\qquad\qquad\qquad\qquad\qquad\qquad\times
\prod_{l=1}^{k}
i\sqrt{\beta}
\sum_{a=1}^{2p}
\sigma_a
\E_{m,\eps}[\vp(x_a)[D_l\vp](g_l)]
\biggr\}
\d\xx
\end{split}
\end{equation}
since 
\begin{equation}
\eta_{0,\nu,z,\Lambda}(x,\sigma|m,\eps)=-z\1_{\{x\in\Lambda\}}e^{i\sqrt{\beta}\sigma\sum_{l=1}^k\nu_l\E_{m,\eps}[\vp(x)[D_l\vp](g_l)]}.
\end{equation}

The term on the last row of \Cref{eq:the argument of log in the correlation functions in terms of ZcalR} contributes
\begin{equation}
\prod_{l=1}^k\partd{\nu_l}\bigg|_{\nu_l=0}
\frac12
\E_{m,\eps}\Bigl[
\Bigl(\sum_{a=1}^k\nu_a[D_a\vp](g_a)\Bigr)^2
\Bigr]
=
\1_{k=2}\E_{m,\eps}[[D_1\vp](g_1)[D_2\vp](g_2)],
\end{equation}
which converges by
\Cref{lem:basic correlations of the derivatives of the free field} as already mentioned.  

Thus, again by linearity and
\Cref{prop:convergence of the smeared mixed correlation functions of the reference field in d=1},
the $n=0$ correlation function becomes
\begin{equation}
\begin{split}
\E_{\mathrm{sG}(\beta,z,\Lambda)}^T\abr{\prod_{l=1}^k[D_l\vp](g_l)}
&=
\prod_{l=1}^k\partd{\nu_l}\bigg|_{\nu_l=0}
\log(\Zcal_R(\beta,\eta_{0,\boldsymbol{\nu},z,\Lambda}))
\\
&\quad+
\sum_{p=1}^{N/2}
\frac{z^{2p}}{(2p)!}
\sum_{\Si\in\{-1,1\}^{2p}}
\1_{\sum_{i=1}^{2p}\sigma_i=0}
\\
&\qquad\times
\int_{\Lambda^{2p}}
\E^T\abr{:
e^{i\sqrt{\beta}\sigma_a\vp(x_a)}:
\mid a\in[2p]}
\prod_{l=1}^{k}
i\sqrt{\beta}
\sum_{b=1}^{2p}
\sigma_b
\E[\vp(x_b)[D_l\vp](g_l)]
\d\xx
\\
&\quad+
\1_{k=2}\E[[D_1\vp](g_1)[D_2\vp](g_2)].
\end{split}
\end{equation}

The analyticity and the third claim follow analogously to the previous cases.  
This concludes the proof.

\end{subsection}
\begin{subsection}{Proof of {{\Cref{thm:existence of the field}}}}
\label{sec:proof of existence of the field}
Analogously to the case of the correlation functions, we have
\begin{equation}
\begin{split}
\label{eq:the expectation of wvp(f) written in terms of the renormalized partition function}
\E_{\mathrm{sG}(\beta,z,\Lambda|m,\eps)}\abr{e^{w\vp(f)}}&=\frac{\Zcal_R(\beta,\eta_{w,z,\Lambda}(\cdot,\cdot|m,\eps)|m,\eps)}{\Zcal_R(\beta,z,\Lambda|m,\eps)}
\\
&\times\prod_{k=1}^{N/2}e^{\frac{1}{(2k)!}\sum_{\Si\in\{-1,1\}^{2k}}\1_{\sum_{i=1}^{2k}\sigma_i=0}\int_{\R^{2kd}}\rbr{\prod_{i=1}^{2k}\eta_{w,z,\Lambda}(x_i,\sigma_i|m,\eps)}\E_{m,\eps}^{T}\abr{\eps^{-\frac{\beta}{4\pi}}e^{i\sqrt{\beta}\sigma_j\vp(x_j)}\mid j\in[2k]}\d\xx}
\\
&\times\prod_{k=1}^{N/2}e^{\frac{-z^{2k}}{(2k)!}\sum_{\Si\in\{-1,1\}^{2k}}\1_{\sum_{i=1}^{2k}\sigma_i=0}
\int_{\Lambda^{2k}}\E_{m,\eps}^{T}\abr{\eps^{-\frac{\beta}{4\pi}}e^{i\sqrt{\beta}\sigma_j\vp(x_j)}\mid j\in[2k]}\d\xx}
\\
&\times\exp\!\left(\frac{w^2}{2}\E_{m,\eps}[(\vp(f))^2]\right).
\end{split}
\end{equation}
Here $\eta_{w,z,\Lambda}(x,\sigma|m,\eps)$ is given by \Cref{eq:eta with parameters for existence of the field}. We consider the limit $m,\eps\to 0$ with $\eps\to 0$ taken first; the fixed-mass case is analogous. 

For $f\in C_c^\infty(\R^d)$ with $\int f=0$, the factor on the last row of the right-hand side of \Cref{eq:the expectation of wvp(f) written in terms of the renormalized partition function} converges again by \Cref{lem:basic correlations of the derivatives of the free field}. Moreover, by the same lemma, the convergence $\eta_{w,z,\Lambda}(x,\sigma|m,\eps)\to\eta_{w,z,\Lambda}(x,\sigma)$ holds in $L^\infty(\R^d\times\{-1,1\})$, uniformly for $(w,z)$ in compact subsets of $\C^2$,
\begin{equation}
\supp(\eta_{w,z,\Lambda}|m,\eps)=\Lambda\times\{-1,1\},\,\,\forall w,z\in \C \text{ and } m,\eps\in (0,1)
\end{equation} 
and 
\begin{equation}
\sup_{(w,z)\in K}\norm{\eta_{w,z,\Lambda}(\cdot,\cdot|m,\eps)}_{L^\infty(\R^d\times\{-1,1\})}\leq \sup_{(w,z)\in K}|z|<\infty
\end{equation}
for all compact $K\subset\C^2$.
The limit $\eta_{w,z,\Lambda}$ is defined by \Cref{eq:eta with parameters for existence of the field} with the replacement
\begin{equation}
\E_{m,\eps}[\vp(x)\vp(f)]\to \E[\vp(x)\vp(f)]:=\lim_{m\to 0}\lim_{\eps\to 0}\E_{m,\eps}[\vp(x)\vp(f)]
=-\frac{1}{2\pi}\int_{\R}f(y)\log(|x-y|)\dy,
\end{equation}
which exists by \Cref{lem:basic correlations of the derivatives of the free field}. Thus, as in the previous section, $\Zcal_R(\beta,\eta_{w,z,\Lambda}(\cdot,\cdot|m,\eps)|m,\eps)\to\Zcal_R(\beta,\eta_{w,z,\Lambda})$ by \Cref{thm:convergence of the partition function}, and the limit is an entire function of $(w,z)\in\C^2$.

Recalling the definition of $\eta_{w,z,\Lambda}$ in \Cref{eq:eta with parameters for existence of the field} and combining the second and third rows on the right-hand side of \Cref{eq:the expectation of wvp(f) written in terms of the renormalized partition function} yields
\begin{equation}
\label{eq:relevant part of the existence of the field}
\prod_{k=1}^{N/2}e^{\frac{z^{2k}}{(2k)!}\sum_{\Si\in\{-1,1\}^{2k}}\1_{\sum_{i=1}^{2k}\sigma_i=0}\int_{\Lambda^{2k}}\rbr{e^{i\sqrt{\beta}\sum_{l=1}^{2k}\sigma_l w\E_{m,\eps}[\vp(x_l)\vp(f)]}-1}\E_{m,\eps}^{T}\abr{\eps^{-\frac{\beta}{4\pi}}e^{i\sqrt{\beta}\sigma_j\vp(x_j)}\mid j\in[2k]}\d\xx}.
\end{equation}
The key step in proving that \Cref{eq:relevant part of the existence of the field} has a limit is to show that the exponent converges. For $d\ge 2$, the argument is analogous to the proof of \cite[Theorem 3.2]{BaWe24a}, since $N=2$.

For $d=1$, we use results from \Cref{sec: the renormalized potential}. The strategy parallels the proof of \Cref{cor:convergence of the Vinfty-Vt terms}. First, for each $k=1,\dots,N/2$, we show that the integral in the exponent, including the sum over $\Si$, is uniformly bounded for $m,\eps\in(0,1)$. Then we decompose the integral into a limit candidate and an error term, and we prove that the error term vanishes. The limit candidate is
\begin{equation}
\sum_{\Si\in\{-1,1\}^{2k}}\1_{\sum_{i=1}^{2k}\sigma_i=0}\int_{\Lambda^{2k}}\rbr{e^{i\sqrt{\beta}\sum_{l=1}^{2k}\sigma_l w\E[\vp(x_l)\vp(f)]}-1}\E^{T}\abr{:e^{i\sqrt{\beta}\sigma_j\vp(x_j)}:\mid j\in[2k]}\d\xx.
\end{equation}
Here
\begin{equation}
\begin{split}
\E^{T}\abr{:e^{i\sqrt{\beta}\sigma_j\vp(x_j)}:\mid j\in[2k]}&:=\lim_{m\to 0}\lim_{\eps\to 0}\E_{m,\eps}^{T}\abr{\eps^{-\frac{\beta}{4\pi}}e^{i\sqrt{\beta}\sigma_j\vp(x_j)}\mid j\in[2k]},
\end{split}
\end{equation}
and the limit exists by \Cref{lem:pointwise charge correlations of the free field,rem:pointwise truncated charge correlation functions of the free field}. An analogous statement holds without taking the limit $m\to 0$.

To control the error we write everything in terms of the functions $\tilV_{t}^n$ defined in \Cref{eq:def of V1,eq:def of Vn}. By the third statement of \Cref{thm:the equivalence of St and Vt} we may write
\begin{equation}
\begin{split}
\label{eq:splitting of the regularized charge correlations}
\E_{m,\eps}^{T}\abr{\eps^{-\frac{\beta}{4\pi}}e^{i\sqrt{\beta}\sigma_j\vp(x_j)}\mid j\in[2k]}&=(-1)^{2k-1}\tilV_\infty^{2k}(\xx,\Si|m,\eps)
\\
&=(-1)^{2k-1}\rbr{\tilV_1^{2k}(\xx,\Si|m,\eps)+\abr{\tilV_\infty^{2k}(\xx,\Si|m,\eps)-\tilV_1^{2k}(\xx,\Si|m,\eps)}}.
\end{split}
\end{equation}

Note that $\E_{m,\eps}[\vp(x)\vp(f)]$ is a smooth function. However, to use an estimate similar to \Cref{eq:field dependent counterterm estimate} for
\begin{equation}
e^{i\sqrt{\beta}w\sum_{l=1}^{2k}\sigma_l\E_{m,\eps}[\vp(x_l)\vp(f)]}-1
\end{equation}
we need
\begin{equation}
\label{eq:sup norm of the derivative of the once smeared variance}
\norm{(\E_{m,\eps}[\vp(\cdot)\vp(f)])'}_{L^\infty(B_\Lambda)}<B,
\end{equation}
for some constant $B>0$ that may depend on $f$ but is independent of $m$ and $\eps$. Above, $(f(\cdot,\dots))'$ denotes differentiation with respect to the variable in place of $\cdot$ and $B_\Lambda$ is a closed ball with radius $2R(\Lambda)$, where $R(\Lambda)$ is such that $\Lambda\subset B_{R(\Lambda)}(0)$. To derive this uniform estimate, we write
\begin{equation}
\begin{split}
\der{x}\E_{m,\eps}[\vp(x)\vp(f)]
&=\der{x}\int_{\R^d}f(y)\int_{\eps^2}^\infty C_s^m(x,y)\ds\dy
\\
&=\frac{1}{2\pi}\int_\R f'(x-u)K_0(m\abs{u})\du-\int_{\R}f'(x-u)\int_0^{\eps^2}e^{-m^2s-\frac{u^2}{4s}}\frac{ds}{4\pi s}\du\\
&=:E_m^{1}(x)+E_{m,\eps}^2(x),
\end{split}
\end{equation}
where in the first step we have made the translation $u=x-y$,  differentiated under the integral sign, and then added and subtracted $\int_0^{\eps^2}e^{-m^2s-u^2/(4s)}/(4\pi s)\ds$. The differentiation under the integral sign is permitted for non-zero $m,\eps$ since then $\int_{\eps^2}^\infty C_s^m(0)\ds<\infty$.
The modified Bessel function $K_0$ has the well-known asymptotics
\begin{equation}
\label{eq:asymptotics of K0}
K_0(x)=-\gamma-\log(\half x)+\bigO(x).
\end{equation}
This yields
\begin{equation}
E_m^1(x)=-\frac{1}{2\pi}\int_{\R}f'(x-u)\log(\abs{u})\du+\int_\R f'(x-u)\bigO(m\abs{u})\du,
\end{equation}
where the term $(-\gamma-\log(m/2))\int_\R f'(x-u)\du$ vanished since $f$ has compact support. Since $m\in(0,1)$, $f(x-u)$ vanishes outside a compact set for $x\in B_\Lambda$, $\sup_{x\in B_\Lambda,u\in\R}|f'(x-u)|<\infty$, and the maps $x\mapsto\log(x)$ and $x\mapsto x$ are locally integrable, $E_m^1$ is uniformly bounded in $m\in(0,1)$ and $x\in B_\Lambda$.

For the second term, we may use size estimates, since the logarithmic mass divergence is contained in the first term. We have
\begin{equation}
\begin{split}
\label{eq:bound for E2m}
|E_{m,\eps}^2(x)|
&\leq\frac{1}{4\pi}\int_{\R}|f'(x-u)|\int_0^1e^{-\frac{u^2}{4s}}\frac{ds}{s}\du
\\
&\leq\frac{1}{4\pi}\int_\R |f'(x-u)|\rbr{\gamma+\abs{\log(\frac{u^2}{4})}+\bigO(u^2)}\du
\\
&\leq \bigO(1).
\end{split}
\end{equation}
Here we have used 
\begin{equation}
\label{eq:exponential integral estimate}
\int_0^1e^{-\frac{u^2}{4s}}\frac{\ds}{s}=\int_{\frac{u^2}{4}}^\infty e^{-r}\frac{\dr}{r}=E_1\rbr{\frac{u^2}{4}}=-\gamma-\log(\frac{u^2}{4})+\int_0^{\frac{u^2}{4}}\frac{1-e^{-r}}{r}\dr\leq\gamma+\abs{\log(\frac{u^2}{4})}+\bigO(u^2),
\end{equation} 
where $E_1$ is the exponential integral and we have used its well-known properties. The implied constant bound on the last row of \Cref{eq:bound for E2m} is independent of $m\in(0,1)$ and $x\in B_\Lambda$ by the same type of argument as with $E_m^1$ above. Thus, \Cref{eq:sup norm of the derivative of the once smeared variance} indeed holds.

Therefore, we may use an estimate similar to \Cref{eq:field dependent counterterm estimate} for
$e^{i\sqrt{\beta}w\sum_{l=1}^{2k}\sigma_l\E_{m,\eps}[\vp(x_l)\vp(f)]}-1$
with $t=1$ for the term corresponding to $\tilV_1^{2k}$. For the term corresponding to $\tilV_\infty^{2k}-\tilV_1^{2k}$, we simply use
\[
\abs{e^{i\sqrt{\beta}w\sum_{l=1}^{2k}\sigma_l\E_{m,\eps}[\vp(x_l)\vp(f)]}-1}\leq 2.
\]
Then \Cref{prop:induction statement,cor:convergence of the Vinfty-Vt terms} yield the required uniform boundedness for the terms corresponding to $\tilV_1^{2k}$ and $\tilV_\infty^{2k}-\tilV_1^{2k}$, respectively. Furthermore, these results also imply that these terms admit integrable majorants independent of $\eps$. Therefore, passing to the limit $\eps\to 0$ through the integrals is again justified by dominated convergence. In fact, since $1<m^{-2}$ for $m\in(0,1)$, the above arguments show that we can also take the limit $m\to 0$ through the integrals for the term corresponding to $\tilV_1^{2k}$. Indeed, the majorant for $\tilV_1^{2k}$ given by \Cref{prop:induction statement} is independent of $m$. The pointwise limits of the functions $\tilV_t^{2k}$ exist by \Cref{rem:limits of the Vtn functions}.

Thus, we only need to carry out the error analysis for the term corresponding to $\tilV_\infty^{2k}-\tilV_1^{2k}$. By the above argument, we may assume in the analysis below that the limit $\eps\to 0$ has already been taken. Recall that we denote
\begin{equation}
\tilV_t^l(\xx,\Si|m):=\lim_{\eps\to 0}\tilV_t^l(\xx,\Si|m,\eps)
\quad \text{and} \quad
\tilV_t^l(\xx,\Si):=\lim_{m\to 0}\lim_{\eps\to 0}\tilV_t^l(\xx,\Si|m,\eps).
\end{equation}

Then we can write
\begin{equation}
\begin{split}
&\sum_{\Si\in\{-1,1\}^{2k}}\1_{\sum_{i=1}^{2k}\sigma_i=0}\int_{\Lambda^{2k}}
\rbr{e^{i\sqrt{\beta}\sum_{l=1}^{2k}\sigma_l w\E_m[\vp(x_l)\vp(f)]}-1}
(-1)^{2k-1}\rbr{\tilV_\infty^{2k}(\xx,\Si|m)-\tilV_1^{2k}(\xx,\Si|m)}\d\xx\\
&\quad=
-\sum_{\Si\in\{-1,1\}^{2k}}\1_{\sum_{i=1}^{2k}\sigma_i=0}\int_{\Lambda^{2k}}
\rbr{e^{i\sqrt{\beta}\sum_{l=1}^{2k}\sigma_l w\E[\vp(x_l)\vp(f)]}-1}
\rbr{\tilV_\infty^{2k}(\xx,\Si|m)-\tilV_1^{2k}(\xx,\Si|m)}\d\xx\\
&\qquad-
\sum_{\Si\in\{-1,1\}^{2k}}\1_{\sum_{i=1}^{2k}\sigma_i=0}\int_{\Lambda^{2k}}
\rbr{e^{i\sqrt{\beta}\sum_{l=1}^{2k}\sigma_l w\E_m[\vp(x_l)\vp(f)]}
-e^{i\sqrt{\beta}\sum_{l=1}^{2k}\sigma_l w\E[\vp(x_l)\vp(f)]}}\\
&\qquad\qquad\qquad\qquad\qquad\qquad\times
\rbr{\tilV_\infty^{2k}(\xx,\Si|m)-\tilV_1^{2k}(\xx,\Si|m)}\d\xx\\
&\quad=: I_1(m)+I_2(m).
\end{split}
\end{equation}

Since $e^{i\sqrt{\beta}\sum_{l=1}^{2k}\sigma_l w\E[\vp(x_l)\vp(f)]}-1$ is uniformly bounded, we obtain from the proof of \Cref{cor:convergence of the Vinfty-Vt terms} that
\begin{equation}
\lim_{m\to 0}I_1(m)=
-\sum_{\Si\in\{-1,1\}^{2k}}\1_{\sum_{i=1}^{2k}\sigma_i=0}\int_{\Lambda^{2k}}
\rbr{e^{i\sqrt{\beta}\sum_{l=1}^{2k}\sigma_l w\E[\vp(x_l)\vp(f)]}-1}
\rbr{\tilV_\infty^{2k}(\xx,\Si)-\tilV_1^{2k}(\xx,\Si)}\d\xx.
\end{equation}
Indeed, by \Cref{eq:error form for Vinfty-Vt} we can write
\begin{equation}
\tilV_\infty^{2k}(\xx,\Si|m)-\tilV_1^{2k}(\xx,\Si|m)
=\tilV_\infty^{2k}(\xx,\Si)-\tilV_1^{2k}(\xx,\Si)+R_1^{2k}(m,\xx,\Si),
\end{equation}
where $R_1^{2k}(m,\xx,\Si)$ is the error term from \Cref{eq:def of the error term Rm} with $t=1$. By the uniform boundedness of the $\exp(\cdot)-1$ term and the error analysis in the proof of \Cref{cor:convergence of the Vinfty-Vt terms},  the error term converges to zero.

It therefore remains to show that $\lim_{m\to 0}I_2(m)=0$. Following the proof of \Cref{cor:convergence of the Vinfty-Vt terms}, we write
\begin{equation}
\begin{split}
I_2(m)
&=
-\sum_{\Si\in\{-1,1\}^{2k}}\1_{\sum_{i=1}^{2k}\sigma_i=0}\int_{\Lambda^{2k}}
\rbr{e^{i\sqrt{\beta}\sum_{l=1}^{2k}\sigma_l w\E_m[\vp(x_l)\vp(f)]}
-e^{i\sqrt{\beta}\sum_{l=1}^{2k}\sigma_l w\E[\vp(x_l)\vp(f)]}}\\
&\qquad\qquad\qquad\qquad\qquad\qquad\times
\rbr{\tilV_\infty^{2k}(\xx,\Si)-\tilV_1^{2k}(\xx,\Si)}\d\xx\\
&\quad-
\sum_{\Si\in\{-1,1\}^{2k}}\1_{\sum_{i=1}^{2k}\sigma_i=0}\int_{\Lambda^{2k}}
\rbr{e^{i\sqrt{\beta}\sum_{l=1}^{2k}\sigma_l w\E_m[\vp(x_l)\vp(f)]}
-e^{i\sqrt{\beta}\sum_{l=1}^{2k}\sigma_l w\E[\vp(x_l)\vp(f)]}}R_1^{2k}(m,\xx,\Si)\d\xx\\
&=:I_{2,1}(m)+I_{2,2}(m).
\end{split}
\end{equation}
Both terms satisfy $\lim_{m\to 0}I_{2,1}(m)=0=\lim_{m\to 0}I_{2,2}(m)$ by dominated convergence and arguments similar to those in the analysis of $I_1$ above, since the difference of the exponentials is uniformly bounded and converges to zero pointwise almost everywhere by \Cref{lem:basic correlations of the derivatives of the free field}.

This completes the proof that \Cref{eq:relevant part of the existence of the field} has a finite limit when first letting $\eps\to 0$ and then $m\to 0$, and thus also completes the proof of \Cref{thm:existence of the field}.
\end{subsection}
\end{section}

\appendix

\begin{section}{Definition and properties of cumulants}
\label{sec:Def and prop of cumulants}
Let $\{X_i\}_{i\in I}$ be an arbitrary collection of random variables on a common probability space. Then we can define the joint cumulants recursively. For $(i_1,i_2,\dots,i_N)=:A\subset I$ with $N\in\N$ we put
\begin{equation}
\begin{split}
\E^T\abr{X_i\mid i\in A}
:=\E\abr{\prod_{i\in A}X_i}
-\sum_{\Pi\in \mathfrak{P}_A'}\prod_{\pi\in \Pi}\E^T\abr{X_i\mid i\in \pi},
\end{split}
\end{equation}
where $\mathfrak{P}_A'$ denotes the set of proper partitions of $A$, that is, all partitions except the trivial one consisting of $A$ itself. If
\begin{equation}
\log(\E\abr{e^{\sum_{l=1}^Nt_lX_{l}}})
\end{equation}
exists and is analytic (for real random variables, the existence of exponential moments suffices), then for cumulants involving only $\{X_l\}_{l=1}^N$ (or analogously for arbitrary subsets of the index set $I$) we may write
\begin{equation}
\begin{split}
\E^T\abr{X_{n}\mid n\in [N]}
:=\partial_{1}\partial_{2}\dots\partial_{N}
\log(\E\abr{e^{\sum_{l=1}^Nt_lX_{l}}})
\bigg|_{t_i=0,\forall i\in A},
\end{split}
\end{equation}
with $[N]:=\{1,2,\dots,N\}$. Above we have written $\partial_{l}:=\partial/\partial t_{l}$.

Cumulants and moments satisfy the so‑called Möbius inversion relations:
\begin{align}
\label{eq:moments to cumulants}
\E\abr{\prod_{i\in A}X_i}
&=\sum_{\Pi\in\mathfrak{P}_A}\prod_{\pi\in \Pi}\E^T\abr{X_i\mid i\in \pi},
\\
\label{eq:cumulants to moments}
\E^T\abr{X_i\mid i\in A}
&=\sum_{\Pi\in\mathfrak{P}_A}(|\Pi|-1)!(-1)^{|\Pi|-1}
\prod_{\pi\in \Pi}\E\abr{\prod_{i\in \pi}X_i},
\end{align}
where $\mathfrak{P}_A$ is now the set of all partitions of $A$ (including the trivial one), and $|B|$ denotes the cardinality of the set $B$. Thus, for $B=\Pi$, it counts the number of blocks $\pi$ in the partition $\Pi$.
\end{section}

\begin{section}{Correlation functions of the free field}
\label{sec:Correlation functions of the reference field.}
In this section, we collect statements concerning the correlation functions of the free field. These correlation functions play a crucial role in the proofs of our main theorems. In addition to the usual smeared correlation functions, we will also consider pointwise correlation functions of the free field. Finally, besides the full free field, we will need certain correlation functions of the free field restricted to different scales, as explained below.

We begin by recalling that our regularized free field is a smooth centered Gaussian field $\vp_{m,\eps}$ with covariance
\begin{equation}
\E[\vp_{m,\eps}(x)\vp_{m,\eps}(y)]\equiv\E_{m,\eps}[\vp(x)\vp(y)]:=\int_{\eps^2}^\infty C_s^m(x,y)\ds,
\end{equation}
where $C_s^m(x,y):=\exp(-m^2s-|x-y|^2/4s)/(4\pi s)$. Furthermore, we denote by $\vp_m$ (with corresponding expectation $\E_m$) the massive free field, that is, a centered Gaussian field with covariance
\begin{equation}
\E[\vp_m(x)\vp_m(y)]\equiv\E_m[\vp(x)\vp(y)]:=\int_0^\infty C_s^m(x,y)\ds=\frac{1}{2\pi}K_0(m|x-y|),
\end{equation}
where $K_0$ is the modified Bessel function of the second kind of order zero. Lastly, as in the analysis of \Cref{sec: the renormalized potential}, we can express the regularized free field as a sum of two independent components. More precisely, let $\vp_{m,(\eps^2,t)}$ and $\vp_{m,\sqrt{t}}$ be independent Gaussian fields (defined on the same probability space as $\vp_{m,\eps}$) with covariances
\begin{equation}
\E_{m,(\eps^2,t)}[\vp(x)\vp(y)]:=\int_{\eps^2}^t C_s^m(x,y)\ds \quad\text{and}\quad \E_{m,\sqrt{t}}[\vp(x)\vp(y)]:=\int_t^\infty C_s^m(x,y)\ds
\end{equation}
respectively. Then the regularized free field admits the decomposition
\begin{equation}
\vp_{m,\eps}\overset{d}{=}\vp_{m,(\eps^2,t)}+\vp_{m,\sqrt{t}}
\end{equation}
The field $\vp_{m,(\eps^2,t)}$ contains (in the limit $\eps\to 0$) the UV-rough component on scales below $t$, while $\vp_{m,\sqrt{t}}$ is UV-regular.

The correlation functions we are ultimately interested in are the truncated mixed charge and gradient correlation functions
\begin{align}
\label{eq:pointwise mixed correlation function of massive free field}
\Ccal_m^{n,k}(\xx,\Si,\zz)
&:=\E_m^T\abr{\prod_{j=1}^n:e^{i\sqrt{\beta}\sigma_j\vp(x_j)}:\prod_{l=1}^k(D_l\vp)(z_l)}
:=\lim_{\eps\to 0}\E_{m,\eps}^T\abr{\prod_{j=1}^n\eps^{-\frac{\beta}{4\pi}}e^{i\sqrt{\beta}\sigma_j\vp(x_j)}\prod_{l=1}^k(D_l\vp)(z_l)},
\\
\label{eq:pointwise mixed correlation function of massless free field}
\Ccal^{n,k}(\xx,\Si,\zz)
&:=\E^T\abr{\prod_{j=1}^n:e^{i\sqrt{\beta}\sigma_j\vp(x_j)}:\prod_{l=1}^k(D_l\vp)(z_l)}
:=\lim_{m\to 0}\lim_{\eps\to 0}\E_{m,\eps}^T\abr{\prod_{j=1}^n\eps^{-\frac{\beta}{4\pi}}e^{i\sqrt{\beta}\sigma_j\vp(x_j)}\prod_{l=1}^k(D_l\vp)(z_l)}.
\end{align}
We also consider the corresponding smeared correlation functions
\begin{align}
\label{eq:smeared mixed correlation function of massive free field}
\Ccal_m^{n,k}(\V{f},\V{g})
&:=\E_m^T\abr{\prod_{j=1}^n:e^{i\sqrt{\beta}\sigma_j\vp}:(f_j)\prod_{l=1}^k(D_l\vp)(g_l)}
:=\lim_{\eps\to 0}\E_{m,\eps}^T\abr{\prod_{j=1}^n\eps^{-\frac{\beta}{4\pi}}e^{i\sqrt{\beta}\sigma_j\vp}(f_j)\prod_{l=1}^k(D_l\vp)(g_l)},
\\
\label{eq:smeared mixed correlation function of massless free field}
\Ccal^{n,k}(\V{f},\V{g})
&:=\E^T\abr{\prod_{j=1}^n:e^{i\sqrt{\beta}\sigma_j\vp}:(f_j)\prod_{l=1}^k(D_l\vp)(g_l)}
:=\lim_{m\to 0}\lim_{\eps\to 0}\E_{m,\eps}^T\abr{\prod_{j=1}^n\eps^{-\frac{\beta}{4\pi}}e^{i\sqrt{\beta}\sigma_j\vp}(f_j)\prod_{l=1}^k(D_l\vp)(g_l)}.
\end{align}
Here $x_j,z_l\in\R^d$ are pairwise distinct points, $f_j\in L_c^\infty(\R^d)$ for $j=1,2,\dots,n$, and $g_l\in C_c^\infty(\R^d)$ for $l=1,2,\dots,k$. Moreover, for each fixed $l$, the operator $D_l$ denotes either a partial derivative $\partial_a:=\partial/\partial x_a$ for some $a=1,2,\dots,d$, or, in the case $d=1$, simply the derivative $d/dx$ (we will denote $\vp'(x):=d/dx \vp(x)$). We also note that truncated correlation functions coincide with the cumulants of the corresponding random variables.

The results collected in this appendix are analogous to those in \cite[Section~2]{BaWe24a} for $d=2$, with $D_l=\partial$ or $\overline{\partial}$ (the Wirtinger derivatives). The pointwise correlation functions in arbitrary dimension follow directly from the computations therein, with the only modification being the form of the differential operators, which is straightforward to account for. The analysis of the smeared correlation functions requires some additional work in the case $d=1$.

First, we record the following lemma concerning basic correlation functions.
\begin{lemma}
\label{lem:basic correlations of the derivatives of the free field}
For arbitrary dimension $d$ we have
\begin{equation}
\begin{split}
\lim_{\eps \to 0}\E_{m,\eps}[(\vp(x))^2]+\frac{1}{2\pi}\log(\eps)&=-\frac{1}{2\pi}\log(m)-\frac{\gamma}{4\pi},
\\
\lim_{\eps \to 0}\E_{m,\eps}[\vp(x)\vp(y)]&=-\frac{1}{2\pi}\log(m) - \frac{1}{2\pi}\log(\half|x-y|)-\frac{\gamma}{2\pi}+\bigO(m|x-y|),
\end{split}
\end{equation}
where $\gamma$ denotes the Euler--Mascheroni constant.

Moreover, for dimensions $d\geq 2$ we have
\begin{equation}
\label{eq:derivative covariances in d geq 2}
\begin{split}
\lim_{m\to 0}\lim_{\eps\to 0}\E_{m,\eps}\abr{\vp(x)(\partial_j\vp)(y)}
&=\frac{1}{2\pi}\frac{x_j-y_j}{|x-y|^2}
=-\lim_{m\to 0}\lim_{\eps\to 0}\E_{m,\eps}\abr{(\partial_j\vp)(x)\vp(y)},
\\
\lim_{m\to 0}\lim_{\eps\to 0}\E_{m,\eps}\abr{(\partial_i\vp)(x)(\partial_j\vp)(y)}
&=\frac{1}{2\pi}\rbr{\frac{\1_{i=j}}{|x-y|^2}-2\frac{(x_i-y_i)(x_j-y_j)}{|x-y|^4}}
=:\frac{1}{2\pi}A_{ij}(x-y).
\end{split}
\end{equation}
For $d=1$ we instead obtain
\begin{equation}
\label{eq:derivative covariances in d is 1}
\begin{split}
\lim_{m\to 0}\lim_{\eps\to 0}\E_{m,\eps}\abr{\vp(x)\vp'(y)}
&=\frac{1}{2\pi}\frac{1}{x-y}
=\frac{1}{2\pi}\frac{\operatorname{sgn}(x-y)}{|x-y|}
=\frac{1}{2\pi}\frac{x-y}{|x-y|^2}
=-\lim_{m\to 0}\lim_{\eps\to 0}\E_{m,\eps}\abr{\vp'(x)\vp(y)},
\\
\lim_{m\to 0}\lim_{\eps\to 0}\E_{m,\eps}\abr{\vp'(x)\vp'(y)}
&=-\frac{1}{2\pi}\frac{1}{|x-y|^2}.
\end{split}
\end{equation}
All limits are uniform on compact subsets of $\R^{d}\times\R^d$ avoiding the diagonal $\{(x,y)\in\R^d\times\R^d\mid x=y\}$. Here $\operatorname{sgn}$ denotes the sign function, taking the value $1$ for nonnegative arguments and $-1$ for negative ones.

For $g\in L_c^\infty(\R^d)$ we further have
\begin{equation}
\begin{split}
\lim_{m\to 0}\lim_{\eps\to 0}\E_{m,\eps}[\vp(u)(\partial_j\vp)(g)]
&:=\int_{\R^d}g(x)\frac{1}{2\pi}\frac{u_j-x_j}{|u-x|^2}\dx,
\qquad (d\geq 2),
\\
\lim_{m\to 0}\lim_{\eps\to 0}\E_{m,\eps}[\vp(u)\vp'(g)]
&:=\int_{\R}g(x)\frac{1}{2\pi}\frac{u-x}{|u-x|^2}\dx,
\qquad (d=1),
\end{split}
\end{equation}
(with differentiation understood in the classical sense) uniformly on compact subsets of $u\in\R^d$ for $d\geq 2$ and of $u\in\R\setminus\supp(g)$ for $d=1$.

Similarly, for $f,g\in L_c^\infty(\R^d)$ with disjoint supports,
\begin{equation}
\begin{split}
\lim_{m\to 0}\lim_{\eps\to 0}\E_{m,\eps}[(\partial_i\vp)(f)(\partial_j\vp)(g)]
&:=\int_{\R^{2d}}f(x)g(y)\frac{1}{2\pi}A_{ij}(x-y)\dx\dy,
\qquad (d\geq 2),
\\
\lim_{m\to 0}\lim_{\eps\to 0}\E_{m,\eps}[\vp'(f)\vp'(g)]
&:=-\int_{\R^{2d}}f(x)g(y)\frac{1}{2\pi}\frac{1}{|x-y|^2}\dx\dy,
\qquad (d=1).
\end{split}
\end{equation}

To relax the assumption that $u$ avoids the support of $g$ in the one-dimensional case, or that $f$ and $g$ have disjoint supports, we may instead interpret derivatives in the distributional sense. In this case, for $f,g\in C_c^\infty(\R^d)$ we have
\begin{equation}
\label{eq:principal value integral statement}
\lim_{m\to 0}\lim_{\eps\to 0}\E_{m,\eps}[\vp(u)\vp'(g)]
:=\frac{1}{2\pi}\int_{\R}g'(y)\log(|u-y|)\dy
=\PV\int_{\R}g(x)\frac{1}{2\pi}\frac{1}{u-x}\dx,
\end{equation}
uniformly on compact subsets of $u\in\R$, and
\begin{equation}
\begin{split}
\lim_{m\to 0}\lim_{\eps\to 0}\E_{m,\eps}[(\partial_i\vp)(f)(\partial_j\vp)(g)]
&:=-\int_{\R^{2d}}(\partial_i f)(x)(\partial_j g)(y)\frac{1}{2\pi}\log(|x-y|)\dx\dy
\\
\lim_{m\to 0}\lim_{\eps\to 0}\E_{m,\eps}[\vp'(f)\vp'(g)]
&:=-\int_{\R^{2d}}f'(x)g'(y)\frac{1}{2\pi}\log(|x-y|)\dx\dy,
\qquad (d=1).
\end{split}
\end{equation}

Finally, for $f\in L_c^\infty(\R^d)$ satisfying $\int f=0$, we also have
\begin{equation}
\begin{split}
\lim_{m\to 0}\lim_{\eps\to 0}\E_{m,\eps}[\vp(u)\vp(f)]
&=-\int_{\R^d}f(x)\frac{1}{2\pi}\log(|x-u|)\dx,
\\
\lim_{m\to 0}\lim_{\eps\to 0}\E_{m,\eps}[(\vp(f))^2]
&=-\int_{\R^{2d}}f(x)f(y)\frac{1}{2\pi}\log(|x-y|)\dx\dy,
\end{split}
\end{equation}
uniformly on compact subsets of $u\in\R^d$.
\end{lemma}
\begin{remark}
\label{rem: basic correlations of the free field}
Analogous limits with fixed $m$ also exist and enjoy the same local uniformity. These can be computed explicitly using the differentiation rules for the modified Bessel functions.
\end{remark}
\begin{proof}
The basic ingredients are the convergence
\begin{equation}
\E_{m,\eps}[\vp(x)\vp(y)]\overset{\eps\to 0}{\longrightarrow}\frac{1}{2\pi}K_0(m|x-y|)
\end{equation}
uniformly on compact subsets of $\R^d\times\R^d$ avoiding the diagonal, together with the well-known asymptotics \Cref{eq:asymptotics of K0} of the modified Bessel function $K_0$. The first fact allows us to interchange the limit $\eps\to 0$ with differentiation, while the asymptotic behavior renders the computation of the derivatives straightforward. The uniformity of the limits follows exactly as in the proof of \cite[Lemma~2.4]{BaWe24a}, and in particular justifies the use of dominated convergence.

The statements involving distributional derivatives follow by integration by parts and basic distribution theory of the principal value of the Cauchy kernel in \ref{eq:principal value integral statement}.
\end{proof}

Then we have:
\begin{lemma}
\label{lem:pointwise charge correlations of the free field}
For any $\beta>0$ and distinct $x_i\in\R^d$, $i=1,2,\dots,n$, the limits
\begin{equation}
\begin{split}
\label{eq:pointwise charge correlations of the free field}
\E_m\abr{\prod_{j=1}^n:e^{i\sqrt{\beta}\sigma_j\vp(x_j)}:}
&:=\lim_{\eps\to 0}\E_{m,\eps}\abr{\prod_{j=1}^n\eps^{-\frac{\beta}{4\pi}}e^{i\sqrt{\beta}\sigma_j\vp(x_j)}},
\\
\E\abr{\prod_{j=1}^n:e^{i\sqrt{\beta}\sigma_j\vp(x_j)}:}
&:=\lim_{m\to 0}\lim_{\eps\to 0}\E_{m,\eps}\abr{\prod_{j=1}^n\eps^{-\frac{\beta}{4\pi}}e^{i\sqrt{\beta}\sigma_j\vp(x_j)}}
\end{split}
\end{equation}
exist. Moreover, the latter is given explicitly by
\begin{equation}
\E\abr{\prod_{j=1}^n:e^{i\sqrt{\beta}\sigma_j\vp(x_j)}:}
=\1_{\sum_{i=1}^n\sigma_i=0}(4e^{-\gamma})^{\frac{n\beta}{8\pi}}
\prod_{1\leq i<j\leq n}|x_i-x_j|^{\sigma_{ij}\frac{\beta}{2\pi}},
\end{equation}
where $\sigma_{ij}:=\sigma_i\sigma_j$.
\end{lemma}

\begin{remark}
\label{rem:pointwise truncated charge correlation functions of the free field}
The corresponding truncated correlation functions are obtained by applying the cumulants-to-moments formula \Cref{eq:cumulants to moments}. We denote these truncated correlations by
\begin{align}
\label{eq:pointwise truncated charge correlation function of the massive free field}
\Ccal_m^n(\xx,\Si)
&:=\E_m^T\abr{\prod_{j=1}^n:e^{i\sqrt{\beta}\sigma_j\vp(x_j)}:}
=\lim_{\eps\to 0}\E_{m,\eps}^T\abr{\prod_{j=1}^n\eps^{-\frac{\beta}{4\pi}}e^{i\sqrt{\beta}\sigma_j\vp(x_j)}},
\\
\label{eq:pointwise truncated charge correlation function of the massless free field}
\Ccal^n(\xx,\Si)
&:=\E^T\abr{\prod_{j=1}^n:e^{i\sqrt{\beta}\sigma_j\vp(x_j)}:}
=\lim_{m\to 0}\lim_{\eps\to 0}\E_{m,\eps}^T\abr{\prod_{j=1}^n\eps^{-\frac{\beta}{4\pi}}e^{i\sqrt{\beta}\sigma_j\vp(x_j)}}.
\end{align}
\end{remark}

\begin{proof}[Proof of \Cref{lem:pointwise charge correlations of the free field}.]
The proof is exactly the same as that of \cite[Lemma~2.5]{BaWe24a}. A straightforward Gaussian computation yields
\begin{equation}
\E_{m,\eps}\abr{\prod_{j=1}^n\eps^{-\frac{\beta}{4\pi}}e^{i\sqrt{\beta}\sigma_j\vp(x_j)}}
=\eps^{-\frac{n\beta}{4\pi}}e^{-\half n\beta K_\eps^m(0)}
e^{-\beta\sum_{1\leq i<j\leq n}\sigma_{ij}K_\eps^m(i,j)},
\end{equation}
where we use the notation introduced in \Cref{eq:some notations for K and the heat kernel}. The claim then follows from the first two statements of \Cref{lem:basic correlations of the derivatives of the free field}.
\end{proof}

We also require charge correlation functions for the fields $\vp_{m,(\eps^2,t)}$ and $\vp_{m,\sqrt{t}}$. More precisely, we have the following result.

\begin{lemma}
\label{lem:pointwise charge correlation functions of the scale decomposed free field}
For any $\beta,t>0$ and distinct $x_i\in\R^d$, $i=1,2,\dots,n$, the limits
\begin{equation}
\lim_{m\to 0}\lim_{\eps \to 0}\E_{m,(\eps^2,t)}\abr{\prod_{i=1}^n\eps^{-\frac{\beta}{4\pi}}e^{i\sqrt{\beta}\sigma_i\vp(x_i)}}
\end{equation}
and
\begin{equation}
\lim_{m\to 0}\E_{m,\sqrt{t}}\abr{\prod_{i=1}^n e^{i\sqrt{\beta}\sigma_i\vp(x_i)}}
\end{equation}
exist and are finite. Analogous statements hold for the corresponding truncated correlation functions by the cumulants to moments formula in \ref{eq:cumulants to moments}. All claims remain valid for fixed $m$.
\end{lemma}

\begin{proof}
By independence, we may write
\begin{equation}
\begin{split}
\lim_{m\to 0}\lim_{\eps\to 0}\E_{m,\eps}\abr{\prod_{i=1}^n\eps^{-\frac{\beta}{4\pi}}e^{i\sqrt{\beta}\sigma_i\vp(x_i)}}
&=\rbr{\lim_{m\to 0}\lim_{\eps\to 0}\E_{m,(\eps^2,t)}\abr{\prod_{i=1}^n\eps^{-\frac{\beta}{4\pi}}e^{i\sqrt{\beta}\sigma_i\vp(x_i)}}}
\\
&\quad\times\lim_{m\to 0}\E_{m,\sqrt{t}}\abr{\prod_{i=1}^n e^{i\sqrt{\beta}\sigma_i\vp(x_i)}}.
\end{split}
\end{equation}
Since the limit on the left-hand side exists by \Cref{lem:pointwise charge correlations of the free field}, it suffices to show that at least one of the limits on the right-hand side is finite and non-zero. We therefore establish the first limit.

A direct Gaussian computation yields
\begin{equation}
\label{eq:charge moments at scale t}
\E_{m,(\eps^2,t)}\abr{\prod_{j=1}^n\eps^{-\frac{\beta}{4\pi}}e^{i\sqrt{\beta}\sigma_j\vp(x_j)}}
=\eps^{-\frac{n\beta}{4\pi}}e^{-\half n\beta K_{\eps^2,t}^m(0)}
e^{-\beta\sum_{1\leq i<j\leq n}\sigma_{ij}K_{\eps^2,t}^m(i,j)},
\end{equation}
where we again use the notation from \Cref{eq:some notations for K and the heat kernel}. We first consider
\begin{equation}
K_{\eps^2,t}^m(0)
=\int_{\eps^2}^t e^{-m^2s}\frac{\ds}{4\pi s}
\overset{u=m^2 s}{=}\frac{1}{4\pi}\int_{m^2\eps^2}^{m^2 t}e^{-u}\frac{\du}{u}
=\frac{1}{4\pi}\rbr{E_1(m^2\eps^2)-E_1(m^2t)},
\end{equation}
where $E_1$ is the exponential integral.
Using \Cref{eq:exponential integral estimate} we obtain
\begin{equation}
K_{\eps^2,t}^m(0)
=-\frac{1}{4\pi}\log\!\rbr{\frac{\eps^2}{t}}
+\bigO(m^2t)+\bigO((m\eps)^2)
=-\frac{1}{2\pi}\log(\eps)+\frac{1}{4\pi}\log(t)+\bigO(m^2t)+\bigO((m\eps)^2).
\end{equation}
Substituting this into the right-hand side of \Cref{eq:charge moments at scale t} yields for the first two factors
\begin{equation}
\lim_{m\to 0}\lim_{\eps\to 0}\eps^{-\frac{n\beta}{4\pi}}e^{-\half n\beta K_{\eps^2,t}^m(0)}=t^{-\frac{n\beta}{8\pi}}
\end{equation}

Next, we have
\begin{equation}
\begin{split}
K_{\eps^2,t}^m(i,j)
&\leq\frac{1}{4\pi}\int_0^t
e^{-\frac{|x_i-x_j|^2}{4s}}\frac{\ds}{s}=\frac{1}{4\pi}E_1\rbr{\frac{|x_i-x_j|^2}{4t}}<\infty
\end{split}
\end{equation}
for fixed $t>0$. Thus, by dominated convergence we obtain
\begin{equation}
\lim_{m\to 0}\lim_{\eps\to 0}\E_{m,(\eps^2,t)}\abr{\prod_{j=1}^n\eps^{-\frac{\beta}{4\pi}}e^{i\sqrt{\beta}\sigma_j\vp(x_j)}}=t^{-\frac{n\beta}{8\pi}}\exp(-\frac{\beta}{4\pi}\sum_{1\leq i<j\leq n}\sigma_{ij}E_1\rbr{\frac{|x_i-x_j|^2}{4t}})<\infty
\end{equation}
for fixed $t>0$. This completes the proof.
\end{proof}

Now we can state the result for the pointwise truncated mixed charge and gradient correlation functions.

\begin{lemma}
\label{lem:pointwise mixed correlation function of the free field}
Let $\beta>0$, $n\geq 1$, $k\geq 0$, $x_i,z_j\in\R^d$ distinct, and $\Si=(\sigma_1,\sigma_2,\dots,\sigma_n)\in\{-1,1\}^n$. Then the limits defining $\Ccal_m^{n,k}(\xx,\Si,\zz)$ and $\Ccal^{n,k}(\xx,\Si,\zz)$ in \Cref{eq:pointwise mixed correlation function of massive free field} and \Cref{eq:pointwise mixed correlation function of massless free field} respectively, exist uniformly on compact subsets of $(\xx,\zz)\in\R^{(n+k)d}$, where all $d$-dimensional components are distinct.

Furthermore, we have
\begin{equation}
\label{eq:exact form of pointwise mixed truncated correlation functions of massless free field}
\Ccal^{n,k}(\xx,\Si,\zz)
=\E^T\abr{\prod_{j=1}^n:e^{i\sqrt{\beta}\sigma_j\vp(x_j)}:}\prod_{l=1}^k\rbr{i\sqrt{\beta}\sum_{q=1}^n\sigma_q F_{d,l}(x_q-z_l)},
\end{equation} 
where
\begin{equation}
\label{eq:def F_{d,q}}
F_{d,l}(x):=\frac{1}{2\pi}
\begin{cases}
\frac{x_i}{|x|^2}, &\text{ if $d\geq 2$ and $D_l=\partial_i$,}
\\
\frac{x}{|x|^2}, &\text{ if $d=1$.}
\end{cases}
\end{equation}
\end{lemma}

\begin{remark}
\label{rem:pointwise gradient correlations}
If $n=0$ in the above lemma, by Gaussianity of the derivative fields only the $k=2$ terms are non-zero; for the massless case, these are given in \Cref{lem:basic correlations of the derivatives of the free field}.
\end{remark}

\begin{proof}
The proof is exactly as in \cite[Lemma 2.6]{BaWe24a}, taking into account the difference in the differential operators we use. Since we will need it later, we record the result of the Cameron-Martin-Girsanov transformation used in the proof of \cite[Lemma 2.6]{BaWe24a} before taking the limits. We have
\begin{equation}
\begin{split}
\label{eq:the mixed correlations of free field after Girsanov}
\E_{m,\eps}^T&\abr{\prod_{j=1}^n\eps^{-\frac{\beta}{4\pi}}e^{i\sqrt{\beta}\sigma_j\vp(x_j)}\prod_{l=1}^k(D_l\vp)(z_l)}
\\
&=\E_{m,\eps}^T\abr{\prod_{j=1}^n\eps^{-\frac{\beta}{4\pi}}e^{i\sqrt{\beta}\sigma_j\vp(x_j)}}\prod_{l=1}^k
\Big(i\sqrt{\beta}\sum_{q=1}^n\sigma_q \E_{m,\eps}[\vp(x_q)(D_l\vp)(z_l)]\Big).
\end{split}
\end{equation}
The result then follows from \Cref{lem:basic correlations of the derivatives of the free field,lem:pointwise charge correlations of the free field,rem:pointwise truncated charge correlation functions of the free field}.
\end{proof}

Next, we consider the smeared correlation functions. For these, we again use distributional derivatives.

\begin{proposition}
\label{prop:local integrability of the charge correlation function in dimension 1}
If $d=1$, fix $\beta\in (0,4\pi)$, and if $d\geq 2$, fix $\beta\in (0,(d+1)2\pi)$. Let $f_i\in L_c^\infty(\R^d\times \{-1,1\})$, $i=1,2,\dots,n$. Then for $n=1$ or $n>N$, where $N\equiv N(\beta)$ is even and such that $\beta\in[\beta_N,\beta_{N+2})$, we have 
\begin{equation}
\begin{split}
\Ccal^n(\V{f})&:=\lim_{m\to 0}\lim_{\eps\to 0}\E_{m,\eps}^T\abr{\prod_{j=1}^n \eps^{-\frac{\beta}{4\pi}} e^{i\sqrt{\beta}\sigma_j\vp}(f_j)}
\\
&=\sum_{\Si\in\{-1,1\}^n}\int_{\R^{dn}} \rbr{\prod_{i=1}^n f_i(x_i,\sigma_i)} \Ccal^n(\xx,\Si) \d\xx < \infty,
\end{split}
\end{equation}
where $\Ccal^n$ is the pointwise truncated charge correlation function defined in \Cref{rem:pointwise truncated charge correlation functions of the free field}.

If $\supp(f_i)\cap\supp(f_j)=\emptyset$ whenever $i\neq j$, the same statement holds for all $n\geq 1$ and $\beta>0$. Similar claims hold for fixed $m\in (0,1)$.
\end{proposition}

The proof is completely analogous to that of \cite[Lemma 2.8]{BaWe24a}, but we need to use the kernels $\widetilde{\Mcal}$ defined in \Cref{prop:expansion for the renormalized partition function}, and \Cref{lem:pointwise limits of tildeMcal and their integrals} for the limits of their integrals. Their proof can be readily extended to arbitrary dimension $d\geq 1$ and to suitable kernels. Moreover, \cite[Lemma 5.6]{BaWe24a} can be modified to our setting since it does not depend explicitly on the precise form of the kernels, and the Fourier argument for uniqueness applies equally well in arbitrary dimension.

Finally, we can state the result for the mixed smeared correlation functions. 
\begin{proposition}
\label{prop:convergence of the smeared mixed correlation functions of the reference field in d=1}
Fix $\beta\in (0,4\pi)$ for $d=1$ and $\beta\in (0,(d+1)2\pi)$ for $d\geq 2$, and let $N\equiv N(\beta)$ be such that $\beta\in [\beta_N,\beta_{N+2})$. Also suppose $f_j\in L_c^\infty(\R^d \times \{-1,1\})$, $j=1,2,\dots,n$, and $g_l\in C_c^\infty(\R^d)$, $l=1,2,\dots,k$, and assume that if $n=2,3,\dots,N$, then $\supp(f_j)\cap\supp(f_{j'})=\emptyset$ whenever $j\neq j'$. Then the limits defining the mixed correlation functions $\Ccal^{n,k}(\V{f},\V{g})$ and $\Ccal_m^{n,k}(\V{f},\V{g})$ defined in \Cref{eq:smeared mixed correlation function of massless free field,eq:smeared mixed correlation function of massive free field}, respectively, exist and are finite. Furthermore, we have for $d\geq 2$
\begin{equation}
\begin{split}
\label{eq:limit of smeared mixed correlations of the free field}
\Ccal^{n,k}(\V{f},\V{g})&:=\lim_{m\to 0}\lim_{\eps\to 0}\E_{m,\eps}^T\abr{\prod_{j=1}^n\eps^{-\frac{\beta}{4\pi}}e^{i\sqrt{\beta}\sigma_j\vp}(f_j)\prod_{l=1}^k(D_l\vp)(g_l)}
\\
&\,\,=\sum_{\Si\in\{-1,1\}^n}\int_{\R^{(n+k)d}}\bigg(\prod_{j=1}^nf_j(x_j,\sigma_j)\bigg)\bigg(\prod_{l=1}^kg_l(z_l)\bigg)\Ccal^{n,k}(\xx,\Si,\zz)\d\zz\d\xx<\infty,
\end{split}
\end{equation}
where $\Ccal^{n,k}$ is the pointwise mixed truncated charge and gradient correlation function defined in \Cref{eq:pointwise mixed correlation function of massless free field} and given in closed form in \Cref{lem:pointwise mixed correlation function of the free field}. For $d=1$ we have
\begin{equation}
\begin{split}
\Ccal^{n,k}(\V{f},\V{g})&:=\lim_{m\to 0}\lim_{\eps\to 0}\E_{m,\eps}^T\abr{\prod_{j=1}^n\eps^{-\frac{\beta}{4\pi}}e^{i\sqrt{\beta}\sigma_j\vp}(f_j)\prod_{l=1}^k\vp'(g_l)}
\\
&\,\,=\sum_{\Si\in\{-1,1\}^n}\int_{\R^n}\bigg(\prod_{j=1}^nf_j(x_j,\sigma_j)\bigg)\E^T\abr{:e^{i\sqrt{\beta}\sigma_j\vp(x_j)}:\,\big|\, j\in [n]}\prod_{l=1}^k\abr{i\sqrt{\beta}\sum_{q=1}^n\sigma_q\E[\vp(x_q)\vp'(z_l)]}\d\xx
\\
&\,\,=\sum_{\Si\in\{-1,1\}^n}\int_{\R^n}\bigg(\prod_{j=1}^nf_j(x_j,\sigma_j)\bigg)\E^T\abr{:e^{i\sqrt{\beta}\sigma_j\vp(x_j)}:\,\big|\, j\in [n]}
\\
&\qquad\qquad\qquad\qquad\times\prod_{l=1}^k\abr{i\sqrt{\beta}\sum_{q=1}^n\sigma_q\PV\int_\R\frac{1}{2\pi}\frac{g_l(z_l)}{x_q-z_l}\dz_l}\d\xx,
\end{split}
\end{equation}
where 
\begin{equation}
\E^T\abr{:e^{i\sqrt{\beta}\sigma_j\vp(x_j)}:\,\big|\, j\in [n]}\quad \text{ and } \quad \E[\vp(x_q)\vp'(z_l)] 
\end{equation}
are defined by the limits in \Cref{eq:pointwise charge correlations of the free field} and \Cref{eq:derivative covariances in d is 1} respectively. 

Furthermore, if $k\geq 1$, the claims above hold for $n=2,3,\dots,N$ without assuming that the test functions have disjoint supports.
\end{proposition}

\begin{proof}
The cases with $n=1,\,N+1,N+2,\dots$ or when the test functions have disjoint supports are direct consequences of \Cref{lem:basic correlations of the derivatives of the free field,rem: basic correlations of the free field,lem:pointwise mixed correlation function of the free field,prop:local integrability of the charge correlation function in dimension 1}. 

The proof of the case $n=2=N$ for $d\geq 2$ is completely analogous to the proof of the $n=2$ case in \cite[Lemma 2.10]{BaWe24a}, and the constraint $\beta<(d+1)2\pi$ can easily be deduced from  \cite[equation (2.68)]{BaWe24a} in its proof with the replacement $q+q'\to k$. 

For $d=1$ the proof follows a similar strategy to the last part of the proof of \Cref{thm:existence of the field}. First, we have by the third statement of \Cref{thm:the equivalence of St and Vt} and \Cref{eq:the mixed correlations of free field after Girsanov}
\begin{equation}
\begin{split}
\E_{m,\eps}^T&\abr{\prod_{j=1}^n\eps^{-\frac{\beta}{4\pi}}e^{i\sqrt{\beta}\sigma_j\vp}(f_j)\prod_{l=1}^k(D_l\vp)(g_l)}
\\
&=\sum_{\Si\in\{-1,1\}^n}\int_{\R^n}\rbr{\prod_{j=1}^nf_j(x_j,\sigma_j)}(-1)^{n-1}\abr{\tilV_1^n(\xx,\Si|m,\eps)+[\tilV_\infty^n(\xx,\Si|m,\eps)-\tilV_1^n(\xx,\Si|m,\eps)]}
\\
&\qquad\qquad\qquad\times(-1)^k\prod_{l=1}^k\abr{i\sqrt{\beta}\sum_{q=1}^n\sigma_q\int_\R g_l'(z_l)\E_{m,\eps}[\vp(x_q)\vp(z_l)]\dz_l}\d\xx.
\end{split}
\end{equation}
We have applied integration by parts before and after using \Cref{eq:the mixed correlations of free field after Girsanov} to obtain the final form. 

Assume that the configuration is neutral and, without loss of generality, that the first $p=n/2$ charges are different from the rest. Then we can write
\begin{equation}
\begin{split}
\label{eq:pairing in the proof of prop B.9 for d=1}
\sum_{q=1}^n\sigma_q\int_\R g_l'(z_l)\E_{m,\eps}[\vp(x_q)\vp(z_l)]\dz_l&=\sum_{q=1}^{n/2}\sigma_q\int_\R [g_l'(x_q-z_l)-g_l'(x_{n/2+\rho(q)}-z_l)]E_{m,\eps}(z_l)\dz_l,
\\
&=\sum_{q=1}^{n/2}\sigma_q\bigg[(x_q-x_{n/2+\rho(q)})\int_\R g_l''(x_{n/2+\rho(q)}-z_l)E_{m,\eps}(z_l)\dz_l
\\
&\qquad\qquad\quad+\int_\R R(x_q,x_{n/2+\rho(q)},z_l)E_{m,\eps}(z_l)\dz_l\bigg],
\end{split}
\end{equation}
where we have used Taylor's theorem to first order and $R$ denotes the resulting error term. Above, $\rho\in\S_{p}$ is a permutation, and we may choose it to minimize the Wasserstein distance \Cref{eq:def Wasserstein distance}. We have also denoted
\begin{equation}
E_{m,\eps}(z)=\int_{\eps^2}^\infty e^{-m^2s-\frac{z^2}{4s}}\frac{\ds}{4\pi s}.
\end{equation}
Then we note that $\int_\R R(x_q,x_{n/2+\rho(q)},z_l)\dz_l=0$ (to see this, write $R$ as the difference of the other terms in the Taylor estimate), and furthermore we can always write 
\begin{equation}
R(x_q,x_{n/2+\rho(q)},z_l)=(x_q-x_{n/2+\rho(q)})\tilde{R}(x_q,x_{n/2+\rho(q)},z_l),
\end{equation}
where $\tilde{R}$ is uniformly bounded for $x_q,x_{n/2+\rho(q)}$ restricted to a compact set and locally integrable in $z_l$. If the charge configuration is not neutral or if $d(\xx)\geq 1$, we do not need to do the pairing on the right-hand side of \Cref{eq:pairing in the proof of prop B.9 for d=1}, and we can simply bound each integral uniformly. By these facts and the estimates in the proof of \Cref{eq:sup norm of the derivative of the once smeared variance}, we may now write for $\xx\in \prod_{j=1}^n\supp(f_j)$
\begin{equation}
\abs{(-1)^k\prod_{l=1}^k\abr{i\sqrt{\beta}\sum_{q=1}^n\sigma_q\int_\R g_l'(z_l)\E_{m,\eps}[\vp(x_q)\vp(z_l)]\dz_l}}\leq C [H_1^n(\xx,\Si)]^k\leq C H_1^n(\xx,\Si)\leq C,
\end{equation}
for some constant $C>0$ independent of $m,\eps\in (0,1)$. We can use the constant bound for the term corresponding to $\tilV_\infty^n-\tilV_1^n$ and the second-to-last bound for the term corresponding to $\tilV_1^n$. Then \Cref{prop:induction statement,cor:convergence of the Vinfty-Vt terms} yield integrable majorants independent of $\eps$, so that the $\eps\to 0$ limit can be taken through the $x$-integrals in both terms. In the term corresponding to $\tilV_1^n$, even the $m\to 0$ limit can be taken through the $x$-integrals since $1<m^{-2}$ so that the majorant provided by \Cref{prop:induction statement} is independent of $m$. Furthermore, \Cref{cor:convergence of the Vinfty-Vt terms} shows that the term corresponding to $\tilV_\infty^n-\tilV_1^n$ is bounded uniformly in $m,\eps\in (0,1)$. Thus, we are left to perform the explicit error analysis for the term corresponding to $\tilV_\infty^n-\tilV_1^n$ analogous to the proof of \Cref{cor:convergence of the Vinfty-Vt terms}. Recall that $R_1^n(m,\xx,\Si)$ denotes the error term from the proof of \Cref{cor:convergence of the Vinfty-Vt terms} with $t=1$, that is defined in \Cref{eq:def of the error term Rm}.

We can write
\begin{equation}
\begin{split}
\sum_{\Si\in\{-1,1\}^n}&\int_{\R^n}\rbr{\prod_{j=1}^nf_j(x_j,\sigma_j)}(-1)^{n-1}\abr{\tilV_\infty^n(\xx,\Si|m)-\tilV_1^n(\xx,\Si|m)}
\\
&\qquad\times (-1)^k\prod_{l=1}^k\abr{i\sqrt{\beta}\sum_{q=1}^n\sigma_q\int_\R g_l'(z_l)\E_m[\vp(x_q)\vp(z_l)]\dz_l}\d\xx
\\
&=\sum_{\Si\in\{-1,1\}^n}\int_{\R^n}\rbr{\prod_{j=1}^nf_j(x_j,\sigma_j)}(-1)^{n-1}\abr{\tilV_\infty^n(\xx,\Si)-\tilV_1^n(\xx,\Si)}
\\
&\qquad\times (-1)^k\prod_{l=1}^k\abr{i\sqrt{\beta}\sum_{q=1}^n\sigma_q\int_\R g_l'(z_l)\E[\vp(x_q)\vp(z_l)]\dz_l}\d\xx
\\
&\quad+\sum_{\Si\in\{-1,1\}^n}\int_{\R^n}\rbr{\prod_{j=1}^nf_j(x_j,\sigma_j)}\abr{\Rcal_{m}^1(\xx,\Si)+\Rcal_{m}^2(\xx,\Si)}\d\xx,
\end{split}
\end{equation}
where again $\tilV_t^n(\xx,\Si):=\lim_{m\to 0}\lim_{\eps\to 0}\tilV_t^n(\xx,\Si|m,\eps)$, $\tilV_t^n(\xx,\Si|m):=\lim_{\eps\to 0}\tilV_t^n(\xx,\Si|m,\eps)$ (and these exist by \Cref{rem:limits of the Vtn functions}),
\begin{equation}
\begin{split}
\Rcal_m^1(\xx,\Si)=(-i\sqrt{\beta})^k(-1)^{n-1}R_1^n(m,\xx,\Si)\prod_{j=1}^k\abr{\sum_{q=1}^n\sigma_q\int_\R g_j'(z_j)\E_m[\vp(x_q)\vp(z_j)]\dz_j}
\end{split}
\end{equation}
and
\begin{equation}
\begin{split}
\Rcal_m^2(\xx,\Si)&=(-i\sqrt{\beta})^k(-1)^{n-1}[\tilV_\infty^n(\xx,\Si)-\tilV_1^n(\xx,\Si)]
\\
&\qquad\times\sum_{\substack{F\uplus G=[k]\\ G\neq \emptyset}}\bigg\{\rbr{\prod_{i\in F}\abr{\sum_{q=1}^n\sigma_q\int_\R g_i'(z_i)\E[\vp(x_q)\vp(z_i)]\dz_i}}
\\
&\qquad\qquad\qquad\quad\times\rbr{\prod_{j\in G}\abr{\sum_{q=1}^n\sigma_q\int_\R g_j'(z_j)\rbr{\E_m[\vp(x_q)\vp(z_j)]-\E[\vp(x_q)\vp(z_j)]}\dz_j}}\bigg\}.
\end{split}
\end{equation}
The above sum goes through the ordered bipartitions $\{F,G\}$ of $[k]$, except $\{[k],\emptyset\}$. Recall that $\E[\vp(x)\vp(y)]=-\frac{1}{2\pi}\log(|x-y|)$ when integrated against a compactly supported smooth test function with vanishing integral. The integral of $g_l'$ indeed vanishes since $g_l$ has compact support. Furthermore, $\E_m[\vp(x)\vp(y)]=\frac{1}{2\pi}K_0(m|x-y|)$. By the asymptotics \Cref{eq:asymptotics of K0} and the facts that the functions $g_l$ have compact support and the logarithm is locally integrable, each integral above is uniformly bounded in $m\in (0,1)$ and $x_q\in \supp(f_q)$. Furthermore, $\E_m[\vp(x_q)\vp(z_l)]-\E[\vp(x_q)\vp(z_l)]\to 0$ as $m\to 0$ pointwise for almost every $z_l$. Therefore, by the error analysis in the proof of \Cref{cor:convergence of the Vinfty-Vt terms} starting from \Cref{eq:error form for Vinfty-Vt}, we have
\begin{equation}
\lim_{m\to 0}\sum_{\Si\in\{-1,1\}^n}\int_{\R^n}\rbr{\prod_{j=1}^nf_j(x_j,\sigma_j)}\abr{|\Rcal_{m}^1(\xx,\Si)|+|\Rcal_{m}^2(\xx,\Si)|}\d\xx=0.
\end{equation}

This completes the proof since the alternative forms follow directly from \Cref{lem:basic correlations of the derivatives of the free field}.
\end{proof}

\end{section}

\begin{section}{Proof of {{\Cref{lem:Gaussian tails}}}}
\label{sec:proof of Gaussian tails}
First, we introduce the same tools that were used in \cite[Lemma 5.2]{BaWe24a}, namely Dudley’s theorem (see, e.g., \cite[Section~1.3]{AdTa07a}) and the Borel--TIS inequality (see, e.g., \cite[Section~2.1]{AdTa07a}). These are
\begin{theorem}[Dudley’s theorem]
Let $T$ be a metric space (such that $(T,d_X)$ as defined below is compact) and let $X$ be a centered real-valued Gaussian process on $T$. Define another (pseudo) metric on $T$ by
\[
d_X(t,s):=\sqrt{\E[\{X(t)-X(s)\}^2]}.
\]
Then there exists a universal constant $C>0$ such that
\begin{align}
\E\abr{\sup_{t\in T}X(t)}\leq C\int_0^{\diam(T)/2}\sqrt{H_X(\eps)}\d\eps,
\end{align}
where $H_X$ is the log-entropy corresponding to the metric $d_X$. It is defined by $H_X(\eps):=\log(N_X(\eps))$, where $N_X(\eps)$ is the minimal number of (closed) $d_X$-balls with radius $\eps$ needed to cover $T$.
\end{theorem}
and
\begin{theorem}[Borel--TIS inequality]
Let $T$ be as above and let $X$ be an almost surely bounded, centered Gaussian process on $T$. Then 
\begin{align}
\E[\sup_{t\in T}X(t)]<\infty \quad \text{and}\quad \P(\{\sup_{t\in T}X(t)-\E[\sup_{t\in T}X(t)]>u\})\leq e^{-\frac{u^2}{2\sigma_T^2}},
\end{align}
where $\sigma_T^2:=\sup_{t\in T}\E[(X(t))^2]$.
\end{theorem}
\begin{remark}
\label{rem:corollary to Borel-TIS}
The following facts are also stated in \cite[Section~2.1]{AdTa07a}. The Borel--TIS inequality has the simple corollary
\begin{align}
\P\rbr{\cbr{\sup_{t\in T}X(t)>u}}
\leq e^{-\frac{\rbr{u-\E\abr{\sup_{t\in T}X(t)}}^2}{2\sigma_T^2}},
\end{align}
and by symmetry we have
\begin{align}
\P\rbr{\cbr{\sup_{t\in T}\abs{X(t)}>u}}
\leq 2\,\P\rbr{\cbr{\sup_{t\in T}X(t)>u}}.
\end{align}
\end{remark}

In the proof below, we simplify notation by writing $\vp\equiv \vp_{m,\sqrt{t}}$, and similarly we drop subscripts on expectations and probabilities, since there is no possibility for confusion.
\begin{proof}
[Proof of \Cref{lem:Gaussian tails}]
First, we write
\begin{equation}
\label{eq:first estimate for Gaussian tales proposition}
\begin{split}
\E\abr{e^{B[p_{\nu,\Lambda}(\vp)]^\alpha}}
&=\int_0^\infty\P\rbr{\cbr{e^{B[p_{\nu,\Lambda}(\vp)]^\alpha}>u}}\du\\
&=\underbrace{\int_0^1\P\rbr{\cbr{e^{B[p_{\nu,\Lambda}(\vp)]^\alpha}>u}}\du}_{=1}
+\int_1^\infty\P\rbr{\cbr{p_{\nu,\Lambda}(\vp)>\rbr{\frac{\log(u)}{B}}^{1/\alpha}}}\du\\
&\leq 1+C(\alpha,B)\int_0^\infty r^{\alpha-1}e^{Br^\alpha}\P\rbr{\cbr{p_{\nu,\Lambda}(\vp)>r}}\dr,
\end{split}
\end{equation}
where the first term equals $1$ since the exponent is positive, and $C(\alpha,B)>0$ is a finite constant.

Next, consider the induced (pseudo) metric $d_\nu\equiv d_{\partial^\nu\varphi}$ from Dudley's Theorem. We have by \Cref{eq:estimate for derivatives of the massive heat kernel} together with the bounds $e^{-m^2s}\leq 1$ and $f_d\leq 1$,
\begin{equation}
\begin{split}
d_\nu(x,y)
&:=\sqrt{\E\abr{\cbr{(\partial_\nu\vp)(x)-(\partial_\nu\vp)(y)}^2}}\\
&=\sqrt{2\bigl(K_{\sqrt{t}}^{m,\nu}(0)-K_{\sqrt{t}}^{m,\nu}(x,y)\bigr)}\\
&\leq \sqrt{C_{t,\nu}}\,|x-y|^{\half},
\end{split}
\end{equation}
where $C_{t,\nu}>0$ is a constant independent of $m$ and $K_{\sqrt{t}}^{m,\nu}(0)=K_{\sqrt{t}}^{m,\nu}(x,x)$ analogously to \Cref{eq:some notations for K and the heat kernel}. Similarly to the metric above, we have denoted $N_{\partial^\nu\vp}\equiv N_\nu$ and $H_{\partial^\nu\vp}\equiv H_\nu$. 

Therefore,
\begin{equation}
\{y\in\Lambda:\sqrt{C_{t,\nu}}\,|x-y|^{1/2}<\eps\}\subset\{y\in\Lambda:d_\nu(x,y)<\eps\},
\end{equation}
so the number of $d_\nu$-balls of radius $\eps$ needed to cover $\Lambda$ is bounded by the number of Euclidean balls of radius $\eps^2/C_{t,\nu}$ needed to cover $\Lambda$. Since $(\Lambda,\cdot)$, with $\cdot$ denoting the Euclidean metric, is compact, so is $(\Lambda,d_\nu)$. If $\Lambda\subset[-L,L]^d$ for some $L\in\N$, then the number of Euclidean balls of radius $\eps$ needed to cover $\Lambda$ is at most $(2L/\eps)^d$. Moreover, in this case $\Lambda\subset B_{\sqrt{d}L}(0)$, where $B_{\sqrt{d}L}(0)$ denotes the Euclidean ball of radius $\sqrt{d}L$ centered at the origin. Consequently, $N_\nu(\eps)=1$ for $\eps>dL^2/C_{t,\nu}$.

It follows that
\begin{equation}
\begin{split}
\int_0^{\diam(\Lambda)/2}\sqrt{\log(N_\nu(\eps))}\d\eps
&\leq \int_0^{dL^2/C_{t,\nu}}\sqrt{\log\!\rbr{\abr{\frac{2L}{\eps^2/C_{t,\nu}}}^d}}\d\eps\\
&=\sqrt{4dLC_{t,\nu}}\int_0^{\frac{d}{\sqrt{2}}\rbr{\frac{L}{C_{t,\nu}}}^{3/2}}
\sqrt{\log(u^{-1})}\du,
\end{split}
\end{equation}
where $u=\eps/\sqrt{2LC_{t,\nu}}$. Substituting $u=e^{-x^2}$ yields
\begin{equation}
2\sqrt{4dLC_{t,\nu}}\int_{B_{t,\nu}(L)}^{\infty}x^2e^{-x^2}\dx<\infty,
\end{equation}
with
\begin{equation}
B_{t,\nu}(L):=\sqrt{\log\rbr{\frac{\sqrt{2}}{d}\abr{\frac{C_{t,\nu}}{L}}^{3/2}}}.
\end{equation}

By Dudley’s theorem, we conclude that
\begin{align}
\label{eq: bound for E_t,nu}
E_{m,t,\nu,\Lambda}(\vp)
:=\E\abr{\sup_{x\in\Lambda}(\partial^\nu\varphi)(x)}
\leq 2\sqrt{4dLC_{t,\nu}}\int_{B_{t,\nu}(L)}^{\infty}x^2e^{-x^2}\dx=: \tilde{C}_{t,\nu,\Lambda}<\infty,
\end{align}
where $\tilde{C}_{t,\nu,\Lambda}$ is manifestly independent of $m$. In particular, $\partial^\nu\varphi$ is almost surely bounded, and we may apply the Borel--TIS inequality. Using \Cref{rem:corollary to Borel-TIS}, we obtain
\begin{equation}
\P\rbr{\cbr{p_{\nu,\Lambda}(\vp)>r}}
\leq 2\exp\!\rbr{-\frac{\abr{r-E_{m,t,\nu,\Lambda}(\vp)}^2}{2\sigma_\nu^2}},
\end{equation}
where
\begin{equation}
\sigma_\nu^2=\E\abr{(\partial^\nu\vp(x))^2}
=\int_t^\infty(\partial_1^\nu\partial_2^\nu C_s^m)(0)\ds
\leq \bar{B}_{t,\nu}<\infty
\end{equation}
for some constant $\bar{B}_{t,\nu}>0$ independent of $m$ again by \Cref{eq:estimate for derivatives of the massive heat kernel}. 

Thus, for $r>2\tilde{C}_{t,\nu,\Lambda}$ we have
\begin{equation}
\frac{(r-E_{m,t,\nu,\Lambda}(\vp))^2}{2\sigma_\nu^2}
\geq \frac{(r-\tilde{C}_{t,\nu,\Lambda})^2}{2\bar{B}_{t,\nu}}
\geq \frac{r^2}{8\bar{B}_{t,\nu}},
\end{equation}
and consequently
\begin{equation}
\P\rbr{\cbr{p_{\nu,\Lambda}(\vp)>r}}
\leq
\begin{cases}
1, & \text{if } r<2\tilde{C}_{t,\nu,\Lambda},\\[0.3em]
e^{-\frac{r^2}{8\bar{B}_{t,\nu}^2}}, & \text{otherwise}.
\end{cases}
\end{equation}

Finally, we estimate the last integral in \Cref{eq:first estimate for Gaussian tales proposition}:
\begin{align*}
\int_0^\infty r^{\alpha-1}e^{Br^\alpha}\P\rbr{\cbr{p_{\nu,\Lambda}(\vp)>r}}\dr
&\leq \int_0^{2\tilde{C}_{t,\nu,\Lambda}} r^{\alpha-1}e^{Br^\alpha}\dr\\
&\quad+\int_{2\tilde{C}_{t,\nu,\Lambda}}^\infty r^{\alpha-1}e^{Br^\alpha}
e^{-\frac{r^2}{8\bar{B}_{t,\nu}}}\dr\\
&<\infty,
\end{align*}
where we used that probabilities are bounded by $1$ in the first term. Since $2>\alpha\geq 1$, the first integral is finite at the origin and the second at infinity. Thus, we obtain an upper bound independent of $m$.
\end{proof}
\end{section}

\end{document}